\documentclass[10pt, letterpaper]{article}

\usepackage[utf8]{inputenc} 
\usepackage[T1]{fontenc}
\usepackage{microtype}
\usepackage[margin=1in, top=0.8in]{geometry}
\usepackage{mathpazo} 
\usepackage[round]{natbib}

\usepackage{titlesec}
\titleformat{\section}{\large\bfseries\raggedright}{\thesection}{1em}{}
\titlespacing*{\section}{0pt}{1.2ex plus 0.3ex minus 0.2ex}{1.0ex plus 0.2ex minus 0.2ex}
\titleformat{\subsection}{\normalsize\bfseries\raggedright}{\thesubsection}{1em}{}
\titlespacing*{\subsection}{0pt}{1.0ex plus 0.3ex minus 0.2ex}{0.5ex plus 0.2ex}
\titleformat{\subsubsection}{\normalsize\bfseries\raggedright}{\thesubsubsection}{1em}{}
\titlespacing*{\subsubsection}{0pt}{0.8ex plus 0.2ex minus 0.2ex}{0.3ex plus 0.1ex}
\titleformat{\paragraph}[runin]{\normalsize\bfseries}{\theparagraph}{1em}{}
\titlespacing*{\paragraph}{0pt}{0.8ex plus 0.2ex minus 0.2ex}{0.5em}

\renewenvironment{abstract}%
{%
  \vspace*{0.075in}%
  \begin{center}%
    {\large \bfseries Abstract}%
  \end{center}%
  \vspace{-2em}%
  \begin{quote}%
}
{
  \end{quote}%
  \vspace{1ex}%
}

\makeatletter
\def\blfootnote{\gdef\@thefnmark{}\@footnotetext}
\makeatother

\usepackage{url}           
\usepackage{booktabs}
\usepackage{nicefrac}  
\usepackage[table]{xcolor}
\usepackage{enumitem}   
\usepackage{amsmath, amsthm, amsfonts, amssymb, mathtools}
\usepackage{graphicx}
\usepackage{mdframed}
\usepackage{float}
\usepackage{caption}
\usepackage[toc,page,header]{appendix}
\usepackage{wrapfig}
\usepackage{algorithm, algpseudocode}
\usepackage{multicol}
\usepackage{etoolbox}
\usepackage{multirow}
\usepackage[scaled]{beramono} 

\usepackage{notation}

\usepackage[
    linktocpage=true,
    plainpages=false,
    pdfpagelabels
]{hyperref}
\hypersetup{
  colorlinks=true,
  citecolor=blue,
  urlcolor=cyan,
  linkcolor=magenta,
}

\definecolor{light-gray}{gray}{0.93}

\setlist{itemsep=0.1em, topsep=0em, leftmargin=2.5em}

\usepackage{thmtools, thm-restate}

\declaretheoremstyle[
  spaceabove=1.5ex plus 0.2ex minus 0.2ex, 
  spacebelow=0.5ex,                        
  headfont=\bfseries,                     
  bodyfont=\normalfont                  
]{tighterrem}

\declaretheoremstyle[
  spaceabove=1.5ex plus 0.2ex minus 0.2ex, 
  spacebelow=0.5ex,                        
  headfont=\bfseries,                     
  bodyfont=\it                     
]{tighterthm}

\newtheorem*{theorem*}{Theorem}
\newtheorem*{informaltheorem*}{Informal Theorem}

\theoremstyle{definition}

\renewenvironment{proof}{\noindent\textbf{Proof.}\,}{\qed}

\title{  
  \vspace*{-1em}
    \bf
    When and Why is Optimistic Multiplicative Weights Slow?\\ 
    The Geometry of Energy Dissipation
}

\date{}

\author{
  \begin{tabular}{c}
    John Lazarsfeld$^{1}$
  \end{tabular}
  \hfil
  \hspace*{1em}
  \begin{tabular}{c}
    Anas Barakat$^{1}$
  \end{tabular}
  \hfil
  \hspace*{1em}
  \begin{tabular}{c}
    Georgios Piliouras$^{2}$
  \end{tabular}
  \vspace{0.2em} \\ 
  \begin{tabular}{c}
    Antonios Varvitsiotis$^{1, 3, 4}$
  \end{tabular}
  \hspace*{1em}
  \begin{tabular}{c}
    Andre Wibisono$^{5}$
  \end{tabular}
}

\begin{document}
\maketitle
\blfootnote{
    Affiliations:\;\;    
    $^{1}$ SUTD \;\;
    $^{2}$ Google DeepMind \;\;
    $^{3}$ NUS CQT\;\;
    $^{4}$ Archimedes/Athena RC \;\;
    $^{5}$ Yale \;\;
  \vspace*{0.1em}
}
\blfootnote{%
Contact: \qquad
\texttt{jlazarsfeld@gmail.com,
    anas\_barakat@sutd.edu.sg, 
    gpil@deepmind.com, \\
    \hspace*{8em}
    antonios@sutd.edu.sg, 
    andre.wibisono@yale.edu.
}}

\vspace*{-2em}
\begin{abstract}

This paper studies the convergence of the 
Optimistic Multiplicative Weights Update algorithm 
(OMWU) in two-player zero-sum games. 
Recent works have identified instances on which 
the last-iterate of OMWU can converge arbitrarily slowly,
but understanding when and why this 
slow convergence occurs has remained open.
In this work, we develop a new analysis framework 
that gives sharp, quantitative
explanations for this behavior.
Our analysis is based on 
viewing the algorithm's dual iterates
as an \textit{optimistic skew-gradient descent}
with respect to an energy function.
We prove over the dual iterates that
energy is \textit{dissipative},
and by establishing tight bounds 
on the magnitude of dissipation, our 
analysis quantifies the geometric bottlenecks that arise 
when the corresponding primal iterates are close 
to the simplex boundary.
This further translates into a new linear last-iterate 
convergence rate in KL divergence
on games with a unique and interior Nash equilibrium.
Compared to prior work, this new rate contains a much sharper
dependence on game-specific constants,
and we prove this dependence is optimal.
Moreover, these geometric insights further translate
into new separations on \textit{uniform} convergence
rates for OMWU. On the one hand, we prove \textit{constant lower bounds}
on the uniform \textit{best-iterate} convergence rate in 
KL divergence and total variation distance from Nash. 
On the other hand, 
we establish for the $2\times 2$ setting 
a new $\widetilde O(T^{-1/2})$ 
best-iterate rate in duality gap,
improving substantially over prior work.
Together, this shows in general that uniform
convergence rate guarantees do not transfer across
different measures of distance to Nash. 

\bigskip

\captionsetup[figure]{labelfont=footnotesize}
\begin{figure}[hb!]
\centering
\includegraphics[width=1\textwidth]{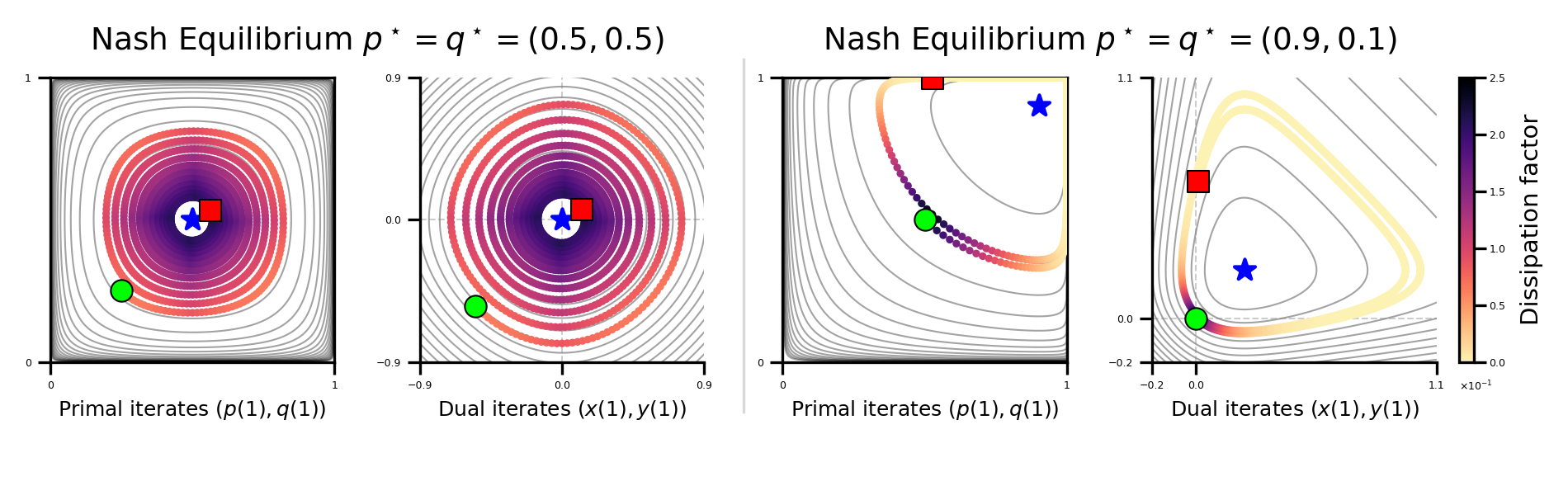}
\vspace*{-3.5em}
\label{fig:abstract}
\caption{%
    \footnotesize
    \textbf{The primal and dual iterates of OMWU} on two instances
    of a $2\times 2$ zero-sum game.
    The algorithm is run for $T=1500$ iterations with 
    stepsize $\eta = 0.2$ in both instances. 
    The Nash equilibrium (and corresponding dual equilibrium point) 
    is denoted by the blue star; the OMWU initializations by
    the green circle; the $T$'th primal
    and dual iterates by the red square. 
    In the left instance, the Nash equilibrium (NE)
    is $p^\star = q^\star = (0.5, 0.5)$
    and the primal initialization is $(0.25, 0.25)$.
    In the right instance, the NE is 
    $p^\star = q^\star = (0.9, 0.1)$  
    and the primal initialization
    is $(0.5, 0.5)$. Note that the algorithm's dual iterates are 
    completely determined by the primal iterates. 
    In both instances, the left subplot shows the 
    trajectory of the \textbf{primal iterates overlayed on the levelsets 
    of KL divergence from NE},
    and the right subplot shows the trajectory of the
    \textbf{dual iterates overlayed on the levelsets 
    of the log-sum-exp energy function}.
    In the right instance, the last-iterate convergence to NE 
    is significantly slower due to the lower energy dissipation
    over the dual iterates.
    }    
\end{figure}

\end{abstract}

\newpage

\section{Introduction}
\label{sec:intro}

Optimistic Multiplicative Weights Update (OMWU)
is among the most prominent algorithms for online
learning in games. 
While the standard Multiplicative Weights Update (MWU)
algorithm diverges to the simplex boundary
in two-player zero-sum games~\citep{bailey2018multiplicative},
the optimistic variant \textit{converges in last-iterate 
to a Nash equilibrium}~\citep{DaskalakisP19, wei2021linear}.
This desirable convergence property
has led OMWU to gain increasing relevance
in machine learning 
applications ranging from adversarial training
\citep{daskalakis-et-al18training-gans,mertikopoulos-et-al18},
to zero-sum Markov games~\citep{cen-et-al23faster},
and  Nash learning from human feedback
\citep{munos2024nash,tiapkin-et-al26pp-nlhf}.

However, recent work has identified instances on which
 the last-iterate convergence of OMWU
can be arbitrarily slow~\citep{FGKLLZ24},
raising new doubts on
the algorithm's reliability for fast learning in games. 
Moreover, this finding has highlighted the fact that,
despite an extensive recent literature studying the
\textit{average-iterate} convergence behavior of 
OMWU (\citet{syrgkanis2015fast, daskalakis2021near,anagnostides2022near}), 
and despite the long history of optimistic algorithms in general,\footnote{%
    OMWU is the canonical instantiation
    of the Optimistic Follow-the-Regularized Leader (OFTRL)
    family with
    entropic regularization~\citep{rakhlin2013optimization},
    and this family itself has its roots in 
    optimistic and extragradient methods
    for solving monotone variational inequalities
    (see, e.g., \cite{korpelevich1976extragradient},
    \cite{popov1980modification},
    \cite{nemirovski2004prox}, \cite{nesterov2007dual}).
}
obtaining a fine-grained
and quantitative understanding of the
\textit{last-iterate} behavior of Optimistic MWU
has still remained surprisingly elusive,
even in the most fundamental constrained setting 
of two-player bilinear zero-sum games.

By way of background,
the \textit{asymptotic} last-iterate convergence
of OMWU to a unique Nash equilibrium
in zero-sum games was first proven 
by~\cite{DaskalakisP19}.
This result was later strengthened by~\cite{wei2021linear},
who proved a non-asymptotic \textit{linear} 
last-iterate convergence rate
for the same setting,
but with a bound containing constant
factors that depend on the game's smallest
non-zero equilibrium mass $\delta$.
Such bounds are called \textit{universal}
convergence rates, and they imply that the algorithm's 
performance can vary dramatically depending on 
the magnitude of $\delta$ in a specific game instance. 
However, whether this potentially poor scaling in their analysis 
accurately captures the 
true iterate-to-iterate behavior of the algorithm
was left unanswered.

Towards this direction,~\cite{FGKLLZ24}
recently gave new evidence that OMWU \textit{can indeed}
converge arbitrarily slowly:
for every time horizon $T$,
they proved that the last iterate 
of OMWU on a $2 \times 2$ instance 
remains an absolute constant
distance away from equilibrium in duality gap. 
While this confirms for fixed time horizons
that the last-iterate of OMWU can remain far from Nash,
it does not explain what \textit{features} of a specific game instance
and what \textit{algorithmic properties} of OMWU govern its
possibly slow iterate-to-iterate behavior.
Moreover, whether the analysis of~\cite{wei2021linear}
-- which involves tracking distances between
\textit{primal} iterates -- 
and its dependence on $\delta$ is quantitatively 
sharp still remains unresolved. 
Thus, the central focus of our paper is to 
address the following fundamental question:
\begin{center}
\vspace*{-0.1em}
\textit{When and why can the last-iterate convergence
of OMWU be arbitrarily slow in zero-sum games?}
\end{center}

\subsection{Our Contributions and Techniques}
Our main contributions provide sharp, quantitative 
answers to this question.

\paragraph{New Analysis Framework via Dual Energy Dissipation.}
For zero-sum games with 
a unique and interior Nash equilibrium,
we develop a new framework for analyzing
the last-iterate convergence of OMWU.
We analyze the \textit{dual iterates}
of OMWU in the unconstrained space of payoff vectors, 
where we show the algorithm follows an 
\textit{optimistic skew-gradient descent} 
with respect to the \textit{log-sum-exp energy function}.
This energy function is the convex conjugate of 
the negative entropy regularizer, 
and we show that minimizing energy over the 
dual iterates corresponds to minimizing KL divergence
from Nash equilibrium 
over the primal iterates
(Proposition~\ref{prop:change-kl-energy-equiv}
and Proposition~\ref{prop:calZ-minimizer}).

Under this dual perspective, 
we prove that \textit{energy is dissipative} 
over the iterates of OMWU, 
which is in stark contrast to the behavior of the
standard MWU variant.
For OMWU, we obtain \textit{tight} upper and 
lower bounds on the one-step dissipation 
in energy, the strength of which is 
quantified in the local geometry of the 
energy function Hessian
(Lemma~\ref{lem:energy-one-step-full}).
Specifically, we relate the magnitude of dual energy dissipation
to the primal KL divergence from Nash 
by proving a novel, general 
\textit{skew-gradient domination} inequality
(Proposition~\ref{prop:dissipation-kl-bound})
that may be of independent interest.
Together, these results give a new and more 
clear description of the geometric properties
that lead to the algorithm's slow behavior:
when the primal iterate is near the simplex boundary, 
the dual iterate moves in directions 
over which the energy function is extremely flat,
thus leading to vanishingly small magnitudes
of energy dissipation.
As the decrease in energy is equivalent to
decrease in KL divergence from Nash,
a geometric bottleneck arises when the 
primal iterates spend many steps near the boundary
of the simplex, a property that emerges
when the game's equilibrium is 
also close to the boundary.

\paragraph{Tight Universal Last-Iterate Convergence Rates.}
Our results on energy dissipation further
translate into a new 
and optimal \textit{linear last-iterate convergence rate}
in KL divergence for zero-sum games with 
a unique and interior Nash equilibrium
(Theorem~\ref{thm:kl-last-unified}).
This new universal rate has a significantly sharper
dependence on $\delta$
(the minimum non-zero Nash equilibrium coordinate)
compared to the prior rate of~\cite{wei2021linear}.
In particular, it improves by a leading multiplicative factor of 
$\exp(\frac{1}{\delta})$, a term that is unbounded as $\delta \to 0$.
While our new bound can still decay 
slowly as $\delta$ grows small, we prove a lower bound 
on the last-iterate convergence rate
with a matching dependence on $\delta$ (Theorem~\ref{thm:kl-last-lower}), establishing in general that our new rate and analysis is tight.
Note also that upper bounds on convergence rates
in KL divergence also imply converges rates in duality gap
(see expression~\eqref{eq:ne-distances}).

\captionsetup[table]{labelfont=small}
\begin{table}[t!]
    \centering
    \small
    \renewcommand{\arraystretch}{1.6}
    \begin{tabular}{c | c | c | c }
        & \textbf{\textbf{KL Divergence}}
        & \textbf{\textbf{TV Distance}}
        & \textbf{\textbf{Duality Gap}} \\ \hline
        \multirow{2}{*}{Universal Last-Iterate}
        & $O(\exp(\frac{1}{\delta}) \cdot 
         \exp(- \exp(\frac{-1}{\delta}) \cdot T))$ 
         $^{\diamond}$ 
        & 
        (via~\eqref{eq:ne-distances})
        $^{\diamond}$     
        & 
        (via~\eqref{eq:ne-distances})
        $^{\diamond}$ 
        \\ \cline{2-4} 
        &
        {\cellcolor{light-gray}
        $\Theta(\exp(- \exp(\frac{-1}{\delta}) \cdot T))$ 
        {\color{magenta}$^{\star}$}
        {\color{magenta}$^{\blacktriangle}$}}
        & 
        (via~\eqref{eq:ne-distances})
         \cellcolor{light-gray}{\color{magenta}$^{\star}$}
        & 
        (via~\eqref{eq:ne-distances})
        \cellcolor{light-gray}{\color{magenta}$^{\star}$} \\
        \hline \hline
        Uniform Last-Iterate 
        &  $\Omega(1)$ $^{\dagger}$ 
        &  $\Omega(1)$ $^{\dagger}$ 
        &  $\Omega(1)$ $^{\dagger}$ \\ \hline
        & \cellcolor{light-gray}
        & \cellcolor{light-gray}
        & $O(T^{-1/6})$ $^{\dagger \dagger}$ \\ \cline{4-4} 
        \multirow{-2}{*}{Uniform Best-Iterate}
        &\multirow{-2}{*}{\cellcolor{light-gray}
        $\Omega(1)$ {\color{magenta}$^{\star \star}$}}
        &\multirow{-2}{*}{\cellcolor{light-gray}
        $\Omega(1)$ {\color{magenta}$^{\star \star}$}}
        & \cellcolor{light-gray}  
        $\widetilde O(T^{-1/2})$ 
        {\color{magenta}$^{\star \star \star}$}
        \\ \hline
    \end{tabular}
  \vspace*{0.5em}
  \caption{%
    \small
    \textbf{The landscape of universal and 
    uniform convergence rates
    for OMWU}, initialized from the uniform distribution and with 
    constant stepsize $\eta$,
    on zero-sum games with a unique and interior
    Nash equilibrium $w^\star$ with minimum 
    coordinate $\delta > 0$.
    Our results are highlighted in gray. 
    $^{\diamond}$: \cite{wei2021linear}.
    $^{\star}$:  upper bound in Theorem~\ref{thm:kl-last-unified},
    and $^{\blacktriangle}$: lower bound in
    Theorem~\ref{thm:kl-last-lower}.
    Both $(^{\star})$ and $(^{\diamond})$ imply rates
    in $\TV$ and duality gap due to~\eqref{eq:ne-distances}.
    Note that Theorem~\ref{thm:kl-last-unified} improves
    over the bound of~\cite{wei2021linear} by
    an $\exp(1/\delta)$ factor, a term that grows
    large as $\delta \to 0$.
    $^{\dagger}$:  proven by \cite{FGKLLZ24} in duality gap,
    which also implies the same lower bound in $\TV$
    and $\KL$ due to~\eqref{eq:ne-distances}.
    \; $^{\star \star}$: Theorem~\ref{thm:uniform-lb-main}. 
    $^{\star \star \star}$: Theorem~\ref{thm:dg-best-2x2},
    which holds only for the $2 \times 2$ setting.
    $^{\dagger \dagger}$: \cite{C25},
    also only for the $2\times 2$ setting.
    \vspace*{-1em}
    }
  \label{table:results}
\end{table}

\paragraph{New Separations on Uniform Convergence Guarantees.}
Our geometric insights additionally yield new
contrasting bounds on \textit{uniform convergence rates}
for OMWU.
Uniform rates hold simultaneously over all instances 
of a fixed dimension and 
have no dependence on game-specific
constants like $\delta$.
Importantly, a potentially slow \textit{universal}
last-iterate convergence rate does not preclude 
the possibility of obtaining
\textit{uniform} rates under
weaker notions of convergence.
While the result of~\cite{FGKLLZ24}
confirms that no uniform \textit{last-iterate}
rates are attainable, their followup work proved for 
the $2\times2$ setting that OMWU 
does have
a uniform $O(T^{-1/6})$ 
\textit{best-iterate}  rate in duality gap
\citep{C25}.
Note that upper bounds on KL divergence imply upper bounds
in duality gap (see the relationships of 
expression~\eqref{eq:ne-distances}),
but the reverse direction does not necessarily hold.
Thus, given our improved geometric understanding
and new universal last-iterate analysis
in KL, it is natural to ask:
\begin{center}
\vspace*{-0.5em}
\textit{Can we obtain any stronger uniform convergence 
guarantees for OMWU?}
\vspace*{-0.5em}
\end{center}
In particular, can we strengthen the guarantee of a
uniform best-iterate convergence rate 
to also hold 
in $\KL$ divergence or total variation ($\TV$) distance
from Nash?
And, can we obtain any \textit{faster} rates than 
the current $O(T^{-1/6})$ bound?

Our work additionally gives answers to these questions.
On the negative side, we prove a new \textit{separation} between
the uniform guarantees attainable in
duality gap versus $\KL$ divergence and $\TV$ distance. 
Specifically, we prove that OMWU has 
\textit{no uniform best-iterate convergence rate}
in $\KL$ or $\TV$ (Theorem~\ref{thm:uniform-lb-main}),
meaning that on certain hard instances, \textit{all} iterates 
of OMWU over a finite time horizon remain far
from Nash under these measures.
On the other hand, in duality gap, we prove a new 
$\widetilde O(T^{-1/2})$ best-iterate 
convergence rate for the $2\times 2$ setting
(Theorem~\ref{thm:dg-best-2x2}), a bound that
substantially improves over the prior rate 
of~\cite{C25} for this same regime.
Together, these results indicate that uniform 
convergence guarantees do not
translate between different measures of distance to Nash,
a distinction that arises due to the inherent relationship
between OMWU and $\KL$ divergence,
and the fundamental geometric differences between 
the level sets of $\KL$ divergence, $\TV$ distance, and
duality gap over the simplex
(see Section~\ref{sec:2x2-levelsets-compare}
for a visual comparison).

Ultimately, our new techniques
and results open the door to
a broader, quantitative
study of last-iterate learning in games
that is both physics-informed
and geometry-aware.
We summarize our new rates in Table~\ref{table:results},
and we defer a more detailed discussion
of other related work to Section~\ref{sec:related-work}.


\section{Preliminaries}
\label{sec:prelims}

\paragraph{Notations.} 
For $k \in \mathbb{N}$, let $[k] = \{1, \dots, k\}$. 
Let $\Delta_k$ denote the probability simplex in $\R^k$,
and let $\relint(\Delta_k) = \{ p \in \Delta_k : 0 < p(i) <1\}$
denote its relative interior.
For $p \in \relint(\Delta_k)$, we say that
$p$ is \textit{interior} or \textit{fully mixed}.
For $u, v \in \R^k$, we write 
$\langle u , v \rangle = u^\top v = \sum_{i=1}^k u(i) \cdot v(i)$
to denote the $\ell_2$ inner product. 
We write $\1_{k} \in \R^k$ to denote the
vector of all ones.

\paragraph{Online learning in zero-sum games.}
We study the setting where
two players repeatedly play a zero-sum game with 
payoff matrix $A \in \R^{m \times n}$, which yields the min-max optimization problem:  
\begin{equation*}
	\textstyle
	\min_{p \in \Delta_m} \, \max_{q \in \Delta_n} \,
	\langle p, A q \rangle \;.
\end{equation*}
At each step $t \ge 0$ the players choose strategies $p_t \in \Delta_m$
and $q_t \in \Delta_n$, incur losses $\langle p_t, A q_t \rangle$
and $- \langle q_t, A^\top p_t \rangle$, and then 
observe gradient feedback payoff vectors $Aq_t \in \R^m$
and $-A^\top p_t \in \R^n$, respectively. 
The players seek to converge to a Nash equilibrium (NE) of $A$. 
For this, let $\calW := \Delta_m \times \Delta_n$.
For $\eps \ge 0$, an ($\eps$-approximate) Nash equilibrium
$w = (p, q) \in \calW$ of $A$ satisfies
\begin{equation}	
	\textstyle
	\max_{q' \in \Delta_m} \, \langle q', A^\top p \rangle
	- \eps
	\le 
	\langle p, A q \rangle 
	\le
	\min_{p' \in \Delta_n} \, \langle p', A q \rangle + \eps .
	\label{eq:nash-prelims}
\end{equation}
Every payoff matrix $A$ has at least one 
($\eps=0$) Nash equilibrium
\citep{v1928theorie}.

\paragraph{Distances to Nash equilibrium.}
This work focuses on games with an \textit{interior} Nash equilibrium. 
Let $\relint(\calW) := \relint(\Delta_m) \times \relint(\Delta_n)$.
For an interior NE $w^\star = (p^\star, q^\star) \in \relint(\calW)$
of $A$, we measure convergence to $w^\star$ under several different
distances. For any $w = (p, q) \in \relint(\calW)$, we define: 
\begin{itemize}[
	topsep=0em,
	itemsep=0em,
	leftmargin=1em,
]
	\item
	\textbf{Duality gap (DG)}:
	$\DG(w) =
	 \max_{q' \in \Delta_m} \langle q', A^\top p\rangle
	 - \min_{p' \in \Delta_n} \langle p', Aq \rangle$.

	\item
	\textbf{Kullback-Leibler divergence (KL)}:
	Define the components
	$\KL(p^\star, p) = \sum_{i=1}^m p^\star(i) 
	\log \big(p^\star(i)/p(i)\big)$ and 
	$\KL(q^\star, q) = \sum_{j=1}^n q^\star(j) 
	\log \big(q^\star(j)/q(j)\big)$.
	Then $\KL(w^\star, w) = \KL(p^\star, p) + \KL(q^\star, q)$.

	\item
	\textbf{Total variation distance (TV)}:
	Define the components $\TV(p^\star,p) = \frac{1}{2}\|p^\star - p\|_1$
	and $\TV(q^\star, q) = \frac{1}{2} \|q^\star - q\|_1$.
	Then $\TV(w^\star, w) = \TV(p^\star, p) + \TV(q^\star, q)$.
\end{itemize}

If $\DG(w) \le \eps$, then $w$ is an $\eps$-approximate NE.
Moreover, letting $a_{\max} := \max_{(i, j) \in [m] \times [n]} |A(i, j)|$ 
denote the largest absolute entry of $A$, the following
relationships hold:
\begin{equation}
	\textstyle
	\frac{1}{(2a_{\max})^2} \cdot 
	\DG(w)^2 
	\le
	\TV(w^\star, w)^2 
	\le
	\KL(w^\star, w) \;.
	\label{eq:ne-distances}
\end{equation}
Thus, upper bounds on $\KL(w^\star, w)$ or $\TV(w^\star, w)$ 
also correspond to approximate Nash equilibria 
in the sense of~\eqref{eq:nash-prelims}.
We give the proof of~\eqref{eq:ne-distances}
and additional preliminaries in 
Section~\ref{app:zsg-prelims}.

\paragraph{Additional preliminaries.}
For a symmetric matrix $M=M^\top \in \R^{k \times k}$, 
we write $M \succeq 0$ when $M$ is positive semi-definite (PSD)
and $M \succ 0$ when $M$ is positive definite (PD).
For a (not-necessarily-symmetric) matrix $M$, 
let $\sigma_{\max}(M) := \max_{x \neq 0} \|Mx\|_2 / \|x\|_2$
denote its spectral norm. 
For a linear subspace $\calV$, let $\Pi_{\calV}(\cdot)$
denote the orthogonal  projection operator onto $\calV$,
and for a matrix $M$, let 
$\sigma_{\min}(M, \calV) 
= \inf_{v \in \calV \setminus \{0\}} \|\Pi_{\calV}(Mv)\|_2 / \|v\|_2$
denote its minimum singular value restricted to $\calV$.
A matrix $M \in \R^{k \times k}$ is \textit{skew-symmetric}
when $M^\top = -M$, which implies
$\langle v, M v \rangle = 0$ for all $v \in \R^k$.
For an integer $k \ge 2$ and $x \in \R^k$, 
we write $\softmax_{k}(x) \in \relint(\Delta_k)$ to denote 
the distribution with coordinates
$\softmax_k(x)(i) = \exp(x(i)) / (\sum_{j=1}^k \exp(x(j)))$
for all $i \in [k]$.
For readability,  we further use the notation
$\softmax_k(x)(i) \propto \exp(x(i))$ for $i \in [k]$.
For a distribution $p \in \relint(\Delta_k)$, the discrete entropy of $p$
is given by $\ent_k(p) = \sum_{i=1}^k p(i) \log (1/p(i))$.
Finally, for continuously differentiable function $f: \calX \to \R$, 
recall the Bregman divergence 
$D_f: \calX \times \relint(\calX) \to \R$
is given by
$D_f(x', x) = f(x') - f(x) - \langle \nabla f(x), x'-x\rangle$
for all $x\in \calX, x' \in \relint(\calX)$.


\section{OMWU as Optimistic Skew-Gradient Descent}
\label{sec:omwu}

In this section, we show that the dual iterates of
OMWU can be viewed as optimistic skew-gradient descent. 
We start with preliminaries on OMWU and 
defer a review of standard  MWU to Section~\ref{app:omwu-prelims}. 

\paragraph{Primal OMWU update rule.}
The algorithm is initialized from an arbitrary
$w_0 = (p_0,  q_0) \in \relint(\calW)$.
Then for $t \ge 1$, using a fixed stepsize $\eta > 0$, 
the iterate  $w_t = (p_t, q_t) \in \relint(\calW)$ has coordinates
\begin{equation}
	\begin{aligned}
		p_t(i)
		&\;\propto\;
		p_{t-1}(i)
		\cdot
		\exp
		\big(
			- \eta \big(2 A q_{t-1} - A q_{t-2}\big)(i)
		\big) &&\text{$\forall i \in [m]$}\,, \\
		q_t(j)
		&\;\propto\;
		q_{t-1}(j) 
		\cdot 
		\exp
		\big(	
			\eta \big( 2 A^\top p_{t-1} - A^\top p_{t-2}\big)(j)
		\big)
		&&\text{$\forall j \in [n]$} \;.
	\end{aligned}
	\label{eq:omwu-primal}
	\tag{OMWU}
\end{equation}	
For notational convenience, we assume
that $w_{-1} = w_0 \in \relint(\calW)$. 

\paragraph{Dual iterates and OMWU as Optimistic FTRL.}
OMWU can also be derived as an instantiation 
of the Optimistic Follow-the-Regularized-Leader (OFTRL)
family using the negative entropy regularizer
\citep{rakhlin2013optimization,syrgkanis2015fast}.
Specifically, under OFTRL, the primal iterates are chosen
by applying a \textit{regularized} best-response map
to a sequence of \textit{dual iterates}. 
In the zero-sum game setting,
these primal and dual iterates
have a compact representation 
that we describe here:

First, given $A \in \R^{m \times n}$,
let $J \in \R^{(m+n) \times (m+n)}$ be the
block skew-symmetric matrix
\begin{equation}
	J   
    =
	\begin{pmatrix}
	0 & A \\
	- A^\top & 0 
	\end{pmatrix}
    = -J^\top
    \,.
	\label{eq:J-def}
\end{equation}
Under OFTRL, the primal 
iterates $\{w_t\}$ are initialized at 
$w_0 = (p_0, q_0) \in \relint(\calW)$.
For the dual iterates,
initialize $x_0 \in \R^m$
and $y_0 \in \R^n$, 
and let $z_0 = z_{-1} = (x_0, y_0) \in \R^{m+n}$.
At time $t=1$, let $z_1 = z_0 - \eta J w_0$.
Then at time $t+1\ge 2$, the dual iterate 
$z_{t+1} = (x_{t+1}, y_{t+1}) \in \R^{m+n}$
is
\begin{equation}
	\textstyle
	z_{t+1} = z_0 - \eta 
	\big(
	\sum_{k=0}^{t}
	Jw_k + J w_{t}
	\big) \;,
	\label{eq:zt-wt-def} 
\end{equation}
which is a sum that depends only on 
$w_0, \dots, w_{t}$.
The primal iterate
at time $t+1$ is defined as follows:
first, let $R_m : \Delta_m \to \R$ and $R_n: \Delta_n \to \R$
be a pair of strictly convex regularizers, and let $R$ be the joint,
separable regularizer defined by $R(w) = R_m(p) + R_n(q)$ 
for $w = (p, q) \in \calW$.
Then OFTRL sets
\begin{equation}
	\textstyle
	w_{t+1} = \argmin_{w = (p, q) \in \calW}\,
	\big\{  
		\langle w, z_{t+1} \rangle 
        + R(w)
	\big\} \;.
	\label{eq:omwu-oftrl}
\end{equation}
For OMWU, the regularizer $R$ is 
instantiated by the 
\textit{negative entropy} functions 
$R_m = -\ent_m$ and $R_n = - \ent_n$. 
Then, the first-order optimality 
conditions of~\eqref{eq:omwu-oftrl}
under this setting of $R$
yield the closed-form expression
$w_{t} = (p_{t}, q_{t}) = 
(\softmax_m(x_{t}), \softmax_n(y_{t})) \in \relint(\calW)$ for all $t \ge 1$.
This recovers the update rule of~\eqref{eq:omwu-primal}.
We review the full derivation of this equivalence
in Section~\ref{app:omwu-prelims}.

\subsection{Dual OMWU Iterates as Optimistic Skew-Gradient Descent}

Our analysis of OMWU is based on 
viewing the dual iterates $\{z_t\}$ 
as an \textit{optimistic skew-gradient descent}
with respect to an energy function $F$.
We derive this connection here:

\paragraph{Energy function.}
Let $\LSE_m$ and $\LSE_n$ denote 
the $m$ and $n$-dimensional \textit{log-sum-exp} functions.
For $x \in \R^{m}$ and $y \in \R^n$, 
$\LSE_m(x) = \log(\sum_{i=1}^m \exp(x(i)))$
and 
$\LSE_n(y) = \log(\sum_{j=1}^n \exp(y(j)))$.
Together, these functions define the convex,  
continuously differentiable, and separable \textit{energy function}
$F: \R^{m+n} \to \R$.
For $z = (x, y) \in \R^{m+n}$, we let
\begin{equation}
	\textstyle
	F(z) = \LSE_m(x) + \LSE_n(y) \;.
	\label{eq:energy}
\end{equation}
Specifically, $F$ is the \textit{dual function}
(convex conjugate) of the negative 
entropy regularizer $R$,
and the gradient of $F$ maps
from the dual to the primal space.
For $z = (x,y) \in \R^{m+n}$, define
$\nabla F (z) = (\nabla_x F(z), \nabla_y F(z))$.
Then for $z \in \R^{m+n}$
and $w = (p, q) \in \calW$ such
that $p= \softmax_m(x)$ and $q = \softmax_n(y)$,
we have $\nabla F(z) = w$. 
The central consequence of this property
is that the dual iterates of~\eqref{eq:zt-wt-def} 
can be written as a modified \textit{skew-gradient descent} 
with respect to $F$. Formally:
\begin{restatable}[Optimistic Skew-Gradient Descent]
{prop}{propskewgrad}
	\label{prop:omwu-skew-grad}
	Let $\{w_t\}$ be the iterates of~\eqref{eq:omwu-primal} with 
	stepsize $\eta > 0$ initialized from $w_0 \in \relint(\calW)$. 
	Let $\{z_t\}$ be the dual iterates 
	of~\eqref{eq:zt-wt-def} initialized from $z_0 \in \R^{m+n}$
	such that $\nabla F(z_0) = w_0$. 
	Then for all $t \ge 0$, 
	it holds that $w_{t} = \nabla F(z_{t})$ and
	\begin{equation}
		z_{t+1} = 
		\underbrace{z_t  - \eta J \nabla F(z_{t})}_{\text{
			skew-gradient descent}}
		-
		\underbrace{
			\eta (J \nabla F(z_{t}) -J \nabla F(z_{t-1}))
		}_{\text{optimistic skew-gradient correction}}\;.
		\label{eq:omwu-dual} 
		\hspace*{-1.5em}
		\tag{OMWU Dual}
	\end{equation} 
\end{restatable}
Proposition~\ref{prop:omwu-skew-grad} 
(proof in Section~\ref{app:omwu-prelims:skew-grad})
shows that~\eqref{eq:omwu-dual} consists of 
a \textit{skew-gradient descent} step with 
an additional \textit{optimistic corrective term}.  
As shown by~\cite{wibisono2022alternating}, 
the dual iterates following skew-gradient descent correspond 
to the standard MWU algorithm,
which is a forward Euler discretization of 
the continuous-time skew-gradient flow 
$\dot z (t) = - J \nabla F(z(t))$.
Under skew-gradient flow, energy $F$ is conserved along 
the trajectory, while under the forward discretization,
energy is non-decreasing
(see Section~\ref{app:omwu-prelims:omwu-related} 
for a review).
In contrast, for OMWU, we prove the correction term
leads to strict \textit{energy dissipation},
which will yield last-iterate convergence
to NE.

\subsection{Primal Convergence 
in KL via Dual Dissipation of Energy}
\label{sec:omwu-main:primal-dual}

In this section, we show that 
the change in energy over the dual iterates 
is exactly the change in $\KL(w^\star, w_t)$
over the primal iterates when $A \in \R^{m \times n}$ has an interior NE $w^\star \in \relint(\calW)$. 
This equivalence is due to a more general 
relationship between $F$ and 
$\KL(w^\star, \cdot)$
that holds over the \textit{effective} space of dual iterates. 
To state this relationship, we first require
introducing the linear subspaces
$\calZ, \calS, \calS^\bot \subset \R^{m+n}$
and several additional preliminaries. First, we
define the subspaces:
\begin{equation}
	\calZ = \Span(J\calW)\,,\;
	\calS = \Span\big(
		\big(\begin{smallmatrix}\1_m \\  0\end{smallmatrix}\big),
		\big(\begin{smallmatrix} 0 \\  \1_n \end{smallmatrix}\big)
	\big)\,,
	\;\text{and}\;
	\calS^\bot 
	= 
	\{ v \in \R^{m+n} :  \text{$\forall s \in \calS$, } 
	\langle v, s \rangle = 0 \}.
	\label{eq:calZ-calS-def}
\end{equation}
We refer to $\calZ$ as the \textit{effective dual space},
as, by~\eqref{eq:zt-wt-def}, the dual 
iterates satisfy $z_t \in \calZ$ for all $t \ge 1$.
Moreover, $\calS$ is the subspace of constant shift directions,
and $\calS^\bot$ is its orthogonal complement.

\paragraph{Geometry of the effective dual space.}
A key property is that $\calZ$ is \textit{orthogonal} to
interior Nash equilibria of~$A$.
This is because for any interior NE $w^\star \in \relint(\calW)$, $Ap^\star$ and $A^\top q^\star$
are the constant vectors 
$Ap^\star = c \1_m$ and $A^\top q^\star = -c\1_n$ 
for some $c \in \R$
(see Proposition~\ref{prop:interior-NE}). 
Together with the skew-symmetry of $J$,
this implies $\langle z, w^\star \rangle = 0$ for all $z \in \calZ$.
See Proposition~\ref{prop:dual-subspace-property}.

\paragraph{Primal-dual coupling despite
lack of strict convexity.}
The energy function is not globally strictly
convex, as $\nabla F$ is invariant under
shifts $s \in \calS$
(see Proposition~\ref{prop:energy-linear}).
However, the dual relationship of 
$R$ and $F$ and the structure of the 
gradient map $\nabla R$
still allows for establishing the following key relationship:
for $z \in \R^{m+n}$ and $w = \nabla F(z)$, it holds that
$\KL(\widetilde w, w) = R(\widetilde w) + F(z) - 
\langle z, \widetilde w \rangle$ for
all $\widetilde w \in \relint(\calW)$.
We state and prove this formally 
in Proposition~\ref{prop:restricted-fenchel}.

Together, these properties lead to the following
relationship beween $\KL(w^\star, \cdot)$
and $F$ over $\calZ$:

\begin{restatable}[Equivalence between energy and KL differences]
	{prop}{propchangeklenergy}
	\label{prop:change-kl-energy-equiv}
	Let $A \in \R^{m \times n}$ have an interior NE
	$w^\star \in \relint(\calW)$.
	Fix $z, z' \in \calZ$,
    and let $w = \nabla F(z) \in \relint(\calW)$
    and $w' = \nabla F(z')\in \relint(\calW)$. 
	Then it holds that
	\begin{equation*}
		\KL(w^\star, w') - \KL(w^\star, w) 
		= F(z') - F(z) \;.
	\end{equation*}
\end{restatable}

The proof of the proposition is in
Section~\ref{app:energy-kl-equiv:prop-proof}.
As an immediate corollary, 
note that as the  OMWU iterates satisfy
$z_t \in \calZ$ 
and $w_t = \nabla F(z_t)$ for all $t \ge 1$, it holds that
$\KL(w^\star, w_{t+1}) - \KL(w^\star, w_t) = F(z_{t+1}) - F(z_t)$
for all such $t$.
Moreover, we prove in 
Section~\ref{app:energy-kl-equiv:energy-suboptimality} 
that the map $\nabla F: \calZ \to \R^{m+n}$ 
is \textit{surjective} 
(Proposition~\ref{prop:calZ-grad-map}), and that 
$\KL(w^\star, w) = F(z) - \min_{z' \in \calZ} F(z')$
for all 
$z \in \calZ$ and $w = \nabla F(z)$
when $w^\star$ is unique and interior
(Proposition~\ref{prop:calZ-minimizer}).
Taken together, these propositions establish
an exact equivalence between energy dissipation
in the dual space and minimizing $\KL$ divergence to
NE in the primal space.


\section{Tight Bounds on Energy Dissipation}
\label{sec:energy-dissipation}

In this section, we establish tight bounds on
energy dissipation under~\eqref{eq:omwu-dual}.
Due to the relationships discussed in 
Section~\ref{sec:omwu-main:primal-dual},
this leads to a sharp last-iterate convergence 
rate in $\KL$.

\subsection{One-Step Change in Energy under OMWU}
\label{subsec:one-step-change-energy}

We show a dissipative one-step change
in energy under~\eqref{eq:omwu-dual}
that is measured using the local norm 
induced by the Hessian of $F$.
For this, we first introduce the following notation:

\paragraph{Local norms induced by energy Hessian.}
The energy function $F$
is convex and continuously differentiable,
and thus for all $z \in \R^{m+n}$ its 
Hessian matrix $\nabla^2 F(z)$ is PSD
and induces an inner product and norm. 
Specifically, we define 
$\langle u, v \rangle_z := \langle u, \nabla^2 F(z) v \rangle$ and 
$\|v\|_{z} := \sqrt{\langle v, \nabla^2 F(z) v \rangle}$
for $u, v \in \R^{m+n}$.
We refer to $\|\cdot\|_z$ as a \textit{local norm}, 
as the spectrum of $\nabla^2 F(z)$ 
is dependent on the state $z$. 
We give further preliminaries on the 
energy Hessian in Section~\ref{app:energy-gsc-prelims}.

Under the following constant stepsize condition,
we prove our key energy dissipation lemma. 

\begin{restatable}[Stepsize]
	{ass}{assstepsize}
    \label{ass:stepsize}
    $0 < \eta \le \frac{1}{4 (54 \sigma_{\max} + 9)}$,
    where $\sigma_{\max} = \|A \|_2$.
\end{restatable}

\begin{restatable}[Energy Dissipation]
{lem}{lemenergyonestep}
    \label{lem:energy-one-step-full}
    For~\eqref{eq:omwu-dual} with 
    $\eta$ satisfying Assumption~\ref{ass:stepsize},
    for all $t \ge 1$:
    \begin{equation*}
    	\textstyle
        -\frac{5}{4} \cdot 
        \eta^2 \|J \nabla F(z_t) \|^2_{z_t}
        \;\le\;
        F(z_{t+1}) - F(z_t)
        \;\le\;
        - \frac{1}{20} \cdot \eta^2 {}
         \| J \nabla F(z_t) \|^2_{z_t} \;.
    \end{equation*}
\end{restatable}
The significance of Lemma~\ref{lem:energy-one-step-full}
is due to its \textit{exact} characterization of
the one-step change in energy,
with matching upper and lower bounds.
We make several further remarks on
its interpretation:

\paragraph{Dependence on local geometry.}
The one-step change in energy is always 
negative and proportional
to a \textit{dissipation term} $\|J \nabla F(z_t)\|^2_{z_t}$.
This term measures the magnitude of the skew-gradient
payoff vector under the \textit{local geometry}
of $\nabla^2 F(z_t)$.
When the energy Hessian has low curvature,
the dissipation term is small
and leads to a slow one-step change in energy.
Note that the \textit{flatness} of $F$ is dual to the 
\textit{steepness} of the entropy regularizer $R$,
and thus dissipation will grow small
when $w_t = \nabla F(z_t)$ approaches the boundary 
of the simplex, where $R$ is steepest.

\paragraph{Primal characterization of dissipation.} 
The dissipation term can also be expressed as the 
\textit{variance} of the payoff vector $Jw_t$ under $w_t = (p_t, q_t) \in \relint(\calW)$.
Specifically, $\|J \nabla F(z_t)\|^2_{z_t} = 
\Var_{p_t}(Aq_t) + \Var_{q_t}(A^\top p_t)$.
This characterization stems from the structure of $\nabla^2 F$ and 
the relation $\nabla F(z_t) = w_t$.
See Section~\ref{app:energy-gsc-prelims:energy-hessian}.

\paragraph{Asymptotic last-iterate convergence in KL.}
Together with Proposition~\ref{prop:change-kl-energy-equiv},
the negative upper bound in Lemma~\ref{lem:energy-one-step-full}
also immediately implies that OMWU has \textit{asymptotic} 
last-iterate convergence to NE.
Formally, we have the following result
(see Section~\ref{app:kl-last-iterate} for the short proof):

\begin{restatable}[Asymptotic Last-Iterate Convergence]
{thm}{thmasymptotic}
    \label{thm:asymptotic-convergence}
    Let $A \in \R^{m \times n}$ have a unique and interior 
    NE~$w^\star$, 
    and let $\{w_t\}$ be the iterates of~\eqref{eq:omwu-primal} 
    on $A$ with $\eta$ satisfying Assumption~\ref{ass:stepsize}.
    Then the sequence of iterates $\{w_t\}$ converges, and 
    $\lim_{t \to \infty} w_t = w^\star$.
\end{restatable}

\subsection{Proof Overview of Lemma~\ref{lem:energy-one-step-full}}
\label{sec:energy-dissipation:sketch}

We now give the main intuitions
for proving the lemma.
The full proof is developed in Section~\ref{app:energy-dissipation}.

\paragraph{Leading dissipation term.}
To see how the term $\|J \nabla F(z_t)\|^2_{z_t}$
arises, recall from~\eqref{eq:omwu-dual} that
\begin{equation*}
	z_{t+1} - z_t = - \eta J \nabla F(z_t) 
	- \eta J  (\nabla F(z_t) - \nabla F(z_{t-1})) .
\end{equation*}
Then using a first-order Taylor expansion, the difference of
energy gradients in the corrective term can be 
expressed exactly as
\begin{equation}
	\textstyle
	\nabla F(z_t) - \nabla F(z_{t-1})
	= 
	\nabla^2 F(z_t) (z_t - z_{t-1}) + G_F(z_t, z_{t-1}) \;,
	\label{eq:energy-main-01}
\end{equation}
where $G_F(z_t, z_{t-1})$ is a
remainder defined precisely in 
Proposition~\ref{prop:grad-diff-general}.
By using the dual update rule at time $t$ and repeating 
a similar expansion, we can further express
$z_t - z_{t-1} = - \eta J \nabla F(z_t)  - Q_t$
for an error term $Q_t$ 
(defined precisely in Proposition~\ref{prop:energy-one-step-expand}).
Substituting this difference into~\eqref{eq:energy-main-01},
and substituting~\eqref{eq:energy-main-01} back 
into the dual update rule at time $t+1$, we can ultimately write
\begin{equation}
	\textstyle
	z_{t+1} - z_t 
	= - \eta J \nabla F(z_t) 
	+ \eta^2 J \nabla^2 F(z_t) J \nabla F(z_t)
	+ \eta W_t \;,
	\label{eq:energy-main-02}
\end{equation}
where $W_t = J \nabla^2 F(z_t) Q_t - J G_F(z_t, z_{t-1})$.
Then, using the definition of the Bregman 
divergence~$D_F$,
we substitute~\eqref{eq:energy-main-02} to quantify the one-step change in energy:
\begin{align}
	F(z_{t+1}) 
	- F(z_t)
	&=
	\langle
	\nabla F(z_t), z_{t+1} - z_t
	\rangle 
	+ D_F(z_{t+1}, z_t)
	\nonumber \\
	&= 
	- \eta^2 \|J \nabla F(z_t)\|^2_{z_t}
	+ 
	\eta \langle \nabla F(z_t), W_t \rangle
	+ D_F(z_{t+1}, z_t)\;.
	\label{eq:energy-main-03}
\end{align}
Here, the final equality follows 
by the skew-symmetry of $J = -J^\top$
and by definition of $\|\cdot \|_{z_t}$.

\paragraph{Controlling the error terms.}
The expansion of $F(z_{t+1}) - F(z_t)$
in~\eqref{eq:energy-main-03} is exact,
and thus the remaining technical task is
to control the magnitude of the latter two terms.
In Propositions~\ref{prop:DH-error},~\ref{prop:RH-inner-bound},
and~\ref{prop:error-third-order},
we prove that 
$\eta | \langle \nabla F(z_t), W_t \rangle + D_F(z_{t+1}, z_t)|
\le C \eta^2 \|J \nabla F(z_t)\|^2_{z_t}$
for an absolute constant $C < 1$ when
$\eta$ satisfies Assumption~\ref{ass:stepsize}.
This allows the leading dissipation term to absorb the
error terms, resulting in the statement of the lemma. 
Our analysis for controlling these error terms
relies on the following
\textit{local Hessian stability} (LHS) property of 
the energy function $F$:
\begin{restatable}[LHS Property]
{prop}{propenergygsc}
	\label{prop:energy-gsc}
	For any $z, z' \in \R^{m+n}$ 
	and $\alpha > 0$, 
	if $\|z-z'\|_\infty \le \alpha$, then:
	\begin{equation}
		\exp(-2\alpha) \cdot \nabla^2 F(z)
		\preceq
		\nabla^2 F(z') 
		\preceq
		\exp(2\alpha) \cdot \nabla^2 F(z)  \;.
		\label{H-gsc}
		\tag{LHS}
	\end{equation}
\end{restatable} 
Roughly speaking, the~\eqref{H-gsc} property
is a local and relative notion of third-order smoothness.
The property states that the spectra of $\nabla^2 F(z')$ 
and $\nabla^2 F(z)$ are 
similar, up to a multiplicative factor depending
on the closeness of $z, z'$.
The dual OMWU iterates all satisfy
$\|z_{t+1} - z_t\|_{\infty} \le 3 \eta \sigma_{\max}$
(see Proposition~\ref{prop:omwu-dual-diff}), so
for a small enough constant stepsize $\eta$,
the~\eqref{H-gsc} property 
allows for establishing the bounds on
the error terms above.
The~\eqref{H-gsc} property is also
related to \textit{self-concordance},
with the former often derived as a consequence of the latter. 
Note that a recent work of \citet{freund2026second}
established self-concordance of the log-sum-exp function
(which implies Proposition~\ref{prop:energy-gsc} 
by definition of $F$). 
For completeness, we give a short direct proof of
Proposition~\ref{prop:energy-gsc} in 
Section~\ref{app:energy-gsc-prelims:energy-hessian}.

\subsection{Universal Linear Last-Iterate Convergence Rate}

In addition to the asymptotic last-iterate convergence of
Theorem~\ref{thm:asymptotic-convergence},
our new dual energy-based analysis leads to a sharp
linear last-iterate convergence rate.
This relies on a structural inequality relating
$\|J \nabla F(z)\|^2_z$ to $\KL(w^\star, w)$.
For this, let 
$\sigma_{\min} 
:= \sigma_{\min}(J, \calS^\bot)
= \inf_{v \in \calS^\bot \setminus \{0\}} 
\|\Pi_{\calS^\bot}(Jv)\|_2 / \|v\|_2$
denote the minimum singular value of $J$
restricted to $\calS^\bot$.
When $A$ has a unique and interior NE,
then $\sigma_{\min} > 0$ (see Proposition~\ref{prop:gammapositive}). 
We then prove the following general inequality:
\begin{restatable}[Non-uniform skew-gradient domination]
{prop}{propdissipationklbound}
	\label{prop:dissipation-kl-bound}
	Let $A \in \R^{m \times n}$ have a unique and interior NE
	$w^\star \in \relint(\calW)$.
	Fix $z \in \calZ$, and let 
	$w = (p, q) = \nabla F(z) \in \relint(\calW)$. 
	Let $p_{\min} = \min_{i \in[m]} p(i)$
	and $q_{\min} = \min_{j \in [n]} q(j)$, 
	and define $w_{\min} := \min\{p_{\min}, q_{\min}\}$. 
 	Then $\sigma_{\min} > 0$, and moreover 
	\begin{equation*}
		\|J \nabla F(z)\|^2_z
		\ge
		\sigma_{\min}^2 \cdot w_{\min}^2 \cdot \KL(w^\star, w) \;.
	\end{equation*}
\end{restatable}

\noindent
The proof of the inequality is given in 
Section~\ref{app:dissipation-kl}.
Here, we make several remarks:

\begin{restatable}[Skew-gradient domination]
	{rem}{remskewgradientdomination}
	\label{remark:skew-gradient-domination}
	As
	$\KL(w^\star, w) = F(z) - \min_{z^\star \in \calZ} F(z^\star)$
	(Proposition~\ref{prop:calZ-minimizer}),
	note that Proposition~\ref{prop:dissipation-kl-bound}
	establishes a structural property akin to 
	\textit{gradient domination} in optimization,
	but with key differences arising 
	from the non-Euclidean geometry 
	induced by $\|\cdot \|_z$ 
	and the concern of \textit{skew}-gradients.
	In particular, the relationship is constrained by 
	a non-uniform and state-dependent factor $w^2_{\min}$, a term that captures a bottleneck arising from the local geometry of $F$. 
	Specifically,
	$w_{\min}$ is exactly the smallest positive
	eigenvalue of $\nabla^2 F(z)$,
	and thus the curvature of $F$
	flattens when the corresponding primal variable 
	is near the simplex boundary.
	In these regions, the term $\|J \nabla F(z)\|^2_z$
	is small even when $w$ is far from Nash in $\KL$.
	In Proposition~\ref{prop:kl-dissipation-upper-2x2},
	we give an \textit{upper bound} on $\|J \nabla F(z)\|^2_z$
	showing that this non-uniform dependence on $w_{\min}$
	is necessary. 
\end{restatable}

\begin{restatable}[Independence of bottlenecks]
	{rem}{remindependencebottlenecks}
	\label{remark:independence-bottlenecks}
	A second bottleneck arises when
	$J$ is ill-conditioned and has a small
	restricted singular value $\sigma_{\min}$.
	Intuitively, this means that large primal perturbations 
	$w-w' \in \calS^\bot$ lead to only small relative changes in 
	$J(w-w')$, which makes OMWU
	less reactive to meaningful signal from payoff vectors. 
	Importantly, $\sigma_{\min}$ is a property of $J$
	that is completely independent from the geometry
	of $F$ and the proximity of the Nash $w^\star$ to 
	the simplex boundary: we show in 
	Propositions~\ref{prop:scaled-mp-rvs-example}
	and~\ref{prop:small-nash-large-rsv}
	that there are families of matrices
	where (i) $w^\star$ is uniform but $\sigma_{\min}$ is arbitrarily small, 
	and (ii) where $\sigma_{\min} = \tfrac{1}{2}$,
	but $w^\star$ is arbitrarily close to a vertex.
	See Section~\ref{sec:sigma-wmin-relationship}.
\end{restatable}

Recall by Proposition~\ref{prop:change-kl-energy-equiv}
that $\KL(w^\star, w_{t+1}) - \KL(w^\star, w_t) 
= F(z_{t+1}) - F(z_t)$ along the OMWU iterates.
Then combining the energy dissipation
bound of Lemma~\ref{lem:energy-one-step-full}
and the skew-gradient domination property of
Proposition~\ref{prop:dissipation-kl-bound},
we obtain a new, linear last-iterate 
convergence rate in $\KL$:

\begin{restatable}[Last-Iterate Convergence Rate in KL]
{thm}{thmkllastunified}
	\label{thm:kl-last-unified}
	Let $A \in \R^{m\times n}$ have 
	a unique and interior NE $w^\star = (p^\star, q^\star)$.
	Let $\{w_t\}$ denote the iterates of~\eqref{eq:omwu-primal} on $A$
	with $\eta$ satisfying Assumption~\ref{ass:stepsize}, initialized from $w_0 \in \relint(\calW)$.
	Let $\delta_p = \min_{i \in [m]} p^\star(i)$, 
	$\delta_q = \min_{j \in [n]} q^\star(j)$, and 
 	$\delta := \min\{\delta_p, \delta_q\}$.
	For every $t$,
	let $p_{t, \min} = \min_{i \in [m]} p_t(i)$, 
	$q_{t, \min} = \min_{j \in [n]} q_t(j)$, and
	$w_{t, \min} := \min\{p_{t, \min}, q_{t, \min}\}$.
	Let
	$\Lambda := \crossent(w^\star, w_0) 
	= \KL(w^\star, w_0) - R(w^\star)$.
	Then for all $t \ge 1$,
	the following hold:
	\begin{alignat*}{2}
		&\text{(1)}
		\quad
		\KL(w^\star, w_{t+1})
		\le
		\KL(w^\star, w_t) \cdot 
		\left(
		1 - \tfrac{1}{20} \cdot \eta^2 \sigma_{\min}^2
		\cdot w_{t, \min}^2
		\right) \;,
		\\
		&\text{(2)}
		\quad
		\KL(w^\star, w_{t+1})
		\le
		2 \KL(w^\star, w_0)
		\cdot 
		\exp\left(
			- \tfrac{1}{20} 
			\cdot \eta^2 \sigma^2_{\min}
			\cdot \exp\left(\tfrac{-2\Lambda}{\delta}\right)
			\cdot t 
		\right) \;.
	\end{alignat*}
\end{restatable}

\noindent
The proof of the theorem is given
in Section~\ref{app:kl-last-iterate}.
Here, we make the following remarks:

\begin{restatable}[Trajectory-dependence]
{rem}{remkllast}
	\label{remark:kl-last}
	The dependence on $w_{t, \min}$ in Part (1) 
	of the theorem highlights the influence of the 
	iterates' trajectory on the rate of convergence. 
	The trajectory depends on the initialization $w_0$,
	and more importantly on the location of the NE.
	When $w^\star$ is close to the simplex boundary
	(meaning $\delta$ is small),
	the iterates $\{w_t\}$ are also more frequently
	near the boundary.
	In these regions $w_{t, \min}$ is small,
	meaning the one-step contraction in $\KL$ is also small.
	This compounds in slower overall convergence.
	Note that the sequence $\{w_{t, \min}\}$ is not monotonic.
	However, 
	by establishing a \textit{uniform lower bound} of
	$w_{t, \min} \ge \exp(\frac{-2\Lambda}{\delta})$
    (Lemma~\ref{lem:wtmin-bound-uniform})
	we obtain the rate in Part (2). 
\end{restatable}

\begin{restatable}[Comparison with~\cite{wei2021linear}]
{rem}{remarkratecomparisonwei}
	The linear rate of Wei et al. (2021, Theorem 3)
	also depends on $\exp(\frac{-\Lambda}{\delta})$.
	However, explicitly tracking this dependence in their proof 
	yields the much slower bound of
	$\KL(w^\star, w_t)
	\le 
	O(\exp(\frac{\Lambda}{\delta}) \cdot \exp(-\eta^2 \exp(\frac{-\Lambda}{\delta})\cdot t))
	$. 
	Part (2) of Theorem~\ref{thm:kl-last-unified} 
	is thus sharper by at least a multiplicative $\exp(\frac{1}{\delta})$ 
	factor, a term that grows unbounded as $\delta \to 0$.
	Moreover, our new energy-based proof technique 
	offers a more transparent and quantitative understanding
	of  the geometric bottlenecks that contribute to slow convergence.
	See Section~\ref{app:wei-comparisons} 
	for further comparisons.
\end{restatable}

\paragraph{Matching Last-Iterate Lower Bound in KL.}
We further prove that the rate in 
Theorem~\ref{thm:kl-last-unified} is \textit{optimal}.
In particular, we prove a \textit{lower bound} on 
the last-iterate convergence rate in $\KL$ 
with a matching dependence on $\delta$, 
establishing that our new analysis is tight.
Formally:

\begin{restatable}[Lower bound in KL]
{thm}{thmkllastlower}
	\label{thm:kl-last-lower}
	Assume the definitions of 
	Theorem~\ref{thm:kl-last-unified}. 
	Fix any $T \ge 3$.
	Then there exists $A \in [-1, 1]^{2\times 2}$
	with an interior NE
	$w^\star = (p^\star, q^\star)$
	with $\delta := \delta_p = \delta_q$
	such that, from a positive measure set of initializations $w_0$,
	each with $\KL(w^\star, w_0)\le 6$,
	the OMWU iterates $\{w_t\}$
	with $\eta$ as in Assumption~\ref{ass:stepsize} satisfy
    \begin{equation*}
	\KL(w^\star, w_t)
		\ge 
		\KL(w^\star, w_0)
		\cdot 
		\exp(- 40 \eta^2 \exp(\tfrac{-1}{\delta}) \cdot t
		)
	\;\;\;\text{for all $t \in [T]$.}
    \end{equation*}
\end{restatable}

Together with Theorem~\ref{thm:kl-last-unified},
this lower bound indicates that the most severe 
geometric bottleneck in convergence arises when
$w^\star$ is near a \textit{vertex} of the simplex.
We give the proof in Section~\ref{app:kl-last-lower}.


\section{New Bounds on Uniform Convergence Rates}
\label{sec:uniform-bounds}

We use the geometric insights
developed in Theorems~\ref{thm:kl-last-unified} and~\ref{thm:kl-last-lower}
to additionally
prove new bounds on \textit{uniform} convergence rates under OMWU.
For concreteness, fix the dimensions $m$ and $n$,
and recall that an algorithm has a \textit{uniform last-iterate} or 
\textit{uniform best-iterate} convergence
rate of order $f(T)$ 
under a distance measure $D$ if the following holds:
\textit{for all $T \ge 1$} and \textit{for all} $A \in \R^{m \times n}$: 
\begin{alignat*}{3}
	&\text{(uniform last-iterate)}:\;\;
	&&D(w_T) &&\le O(f(T))  \\
	\text{or}\;\;
	&\text{(uniform best-iterate)}:\;\;
	\min\nolimits_{t \in [T]}
	&&D(w_t) &&\le O(f(T))  \;.
\end{alignat*}
Here, we emphasize that $f$ must be a function only
of the time horizon $T$ and the dimensions
$m$ and $n$, but with \textit{no dependence} on the 
minimum coordinate $\delta$ of the NE in a particular instance.
Proving such uniform rates thus allows
quantitative guarantees to hold in an instance-agnostic manner
(see also~\cite{C25} for further discussion on 
these different notions). 

While~\cite{FGKLLZ24} proved for OMWU that no
uniform last-iterate convergence rates are attainable,
the follow-up of \cite{C25} proved for 
the class of $2\times 2$ games with an interior NE,
that, when initialized at the 
uniform distribution, OMWU does have a
$O(T^{-1/6})$ \textit{uniform best-iterate convergence rate}
in \textit{duality gap}.
Unfortunately, we prove that similar uniform best-iterate 
rates are \textit{not} attainable under 
the stronger measures of $\TV$ distance and $\KL$ divergence.

\paragraph{Lower Bounds on Best-Iterate in KL and TV.}
We prove the following \textit{constant lower bounds}:

\begin{restatable}[Uniform Best-Iterate Lower Bounds]
{thm}{thmuniformlbmain}
	\label{thm:uniform-lb-main}
	For every $T \ge 2$, 
	there exists $A \in [-1, 1]^{2\times 2}$
	with an interior Nash equilibrium $w^\star$ such that,
	for the iterates $\{w_t\}$ of~\eqref{eq:omwu-primal} 
	on $A$ with $\eta$ satisfying Assumption~\ref{ass:stepsize},
	from a positive measure set of initializations that includes 
	the joint uniform distributions, the following hold:
    \begin{equation*}
    	\min\nolimits_{t \in [T]} \TV(w^\star, w_t) \ge \tfrac{1}{3}
        \;\;\text{and}\;\;
	      \min\nolimits_{t \in [T]} \KL(w^\star, w_t) \ge \tfrac{1}{9} \;.
    \end{equation*}
\end{restatable}

Note that these lower bounds in best-iterate also 
clearly imply a lower bound on last-iterate.
We give the proof and a full discussion in 
Section~\ref{app:lowerbound-details},
where we also establish an additional uniform best-iterate
lower bound in $\DG$ when the initialization
is not the uniform distribution
(see Theorem~\ref{thm:uniform-lb-secondary}).

\paragraph{Fast Best-Iterate Convergence Rate
in Duality Gap for 2x2 Games.}
In contrast, for $2\times 2$ games,
we obtain in \textit{duality gap}
a new upper bound on \textit{uniform best-iterate convergence rate}
that substantially improves over the prior $O(T^{-1/6})$
rate of~\cite{C25} that was proven for the same setting:

\begin{restatable}[Uniform Best-Iterate Convergence Rate in DG]
{thm}{thmdgbest}
	\label{thm:dg-best-2x2}
	Let $A \in [-1, 1]^{2\times 2}$ be a zero-sum game
	with an interior Nash equilibrium $w^\star$, and let $\{w_t\}$ 
	denote the iterates of running~\eqref{eq:omwu-primal}
	on $A$ with constant stepsize $\eta$ satisfying 
	Assumption~\ref{ass:stepsize}, initialized at the
	uniform distributions. 
	Then for any $T \ge 1$, it holds that 
    \begin{equation*}
        \min\nolimits_{t \in [T]}
        \DG(w_t) \le O(T^{-1/2} \cdot \sqrt{\log T}) \;.
    \end{equation*}
\end{restatable}

See Section~\ref{app:dg-upper}
for the proof of the theorem and a full discussion. 
Notably, our proof relies on the new, universal last-iterate 
analysis in KL divergence from Theorem~\ref{thm:kl-last-unified},
along with an exact control of the OMWU trajectory 
in the $2\times 2$ setting (see Section~\ref{app:omwu-2x2}).
While we prove the result only for this low-dimensional regime,
we conjecture that the same bound holds in higher dimensions.
We expect our techniques will be useful for obtaining such 
a result, but we leave this for future work.


\section{Conclusion}
Our work obtains tight new guarantees
on the iterate-to-iterate behavior of OMWU
in two-player zero-sum games,
and it offers sharp geometric explanations 
for the algorithm's potentially slow convergence.
We believe our new energy-based analysis framework
can be naturally extended to study the
last-iterate convergence properties of other instantiations 
of Optimistic FTRL,
in other settings (such as under bandit feedback),
and for other more general classes of games.
We leave these as directions for future work. 

\medskip

\paragraph{Acknowledgements.} 
JL thanks JJ for inspiration.
AV, AB, and JL are supported by the MOE Tier 2 Grant (MOE-T2EP20223-0018),  the CQT++ Core Research Funding Grant (SUTD) (RS-NRCQT-00002), the National Research Foundation Singapore and DSO National Laboratories under the AI Singapore Programme (Award Number: AISG2-RP-2020-016), and partially by Project MIS 5154714 of the National Recovery and Resilience Plan, Greece 2.0, funded by the European Union under the NextGenerationEU Program. AW is supported by NSF awards CCF-2403391 and CAREER CCF-2443097.

\newpage
\bibliographystyle{apalike}
\bibliography{references}

\begin{thebibliography}{}

\bibitem[Abernethy et~al., 2021]{abernethy2021fast}
Abernethy, J., Lai, K.~A., and Wibisono, A. (2021).
\newblock Fast convergence of fictitious play for diagonal payoff matrices.
\newblock In {\em Proceedings of the 2021 ACM-SIAM Symposium on Discrete
  Algorithms (SODA 2021)}.

\bibitem[Anagnostides et~al., 2022a]{anagnostides2022near}
Anagnostides, I., Daskalakis, C., Farina, G., Fishelson, M., Golowich, N., and
  Sandholm, T. (2022a).
\newblock Near-optimal no-regret learning for correlated equilibria in
  multi-player general-sum games.
\newblock In {\em Proceedings of the 54th Annual ACM SIGACT Symposium on Theory
  of Computing (STOC 2022)}.

\bibitem[Anagnostides et~al., 2022b]{APFS22}
Anagnostides, I., Panageas, I., Farina, G., and Sandholm, T. (2022b).
\newblock On last-iterate convergence beyond zero-sum games.
\newblock In {\em International Conference on Machine Learning (ICML 2022)}.

\bibitem[Arora et~al., 2012]{arora2012multiplicative}
Arora, S., Hazan, E., and Kale, S. (2012).
\newblock The multiplicative weights update method: a meta-algorithm and
  applications.
\newblock {\em Theory of computing}, 8(1):121--164.

\bibitem[Azizian et~al., 2024]{azizian2024rate}
Azizian, W., Iutzeler, F., Malick, J., and Mertikopoulos, P. (2024).
\newblock The rate of convergence of bregman proximal methods: Local geometry
  versus regularity versus sharpness.
\newblock {\em SIAM Journal on Optimization}, 34(3):2440--2471.

\bibitem[Bach, 2010]{bach2010self}
Bach, F. (2010).
\newblock Self-concordant analysis for logistic regression.
\newblock {\em Electronic Journal of Statistics}, 4:384--414.

\bibitem[Bailey and Piliouras, 2019a]{bailey2019fast}
Bailey, J. and Piliouras, G. (2019a).
\newblock Fast and furious learning in zero-sum games: Vanishing regret with
  non-vanishing step sizes.
\newblock {\em Advances in Neural Information Processing Systems (NeurIPS
  2019)}.

\bibitem[Bailey et~al., 2020]{bailey2020finite}
Bailey, J.~P., Gidel, G., and Piliouras, G. (2020).
\newblock Finite regret and cycles with fixed step-size via alternating
  gradient descent-ascent.
\newblock In {\em Conference on Learning Theory (COLT 2020)}.

\bibitem[Bailey and Piliouras, 2018]{bailey2018multiplicative}
Bailey, J.~P. and Piliouras, G. (2018).
\newblock Multiplicative weights update in zero-sum games.
\newblock In {\em Proceedings of the 2018 ACM Conference on Economics and
  Computation (EC 2018)}.

\bibitem[Bailey and Piliouras, 2019b]{bailey2019multi}
Bailey, J.~P. and Piliouras, G. (2019b).
\newblock Multi-agent learning in network zero-sum games is a {H}amiltonian
  system.
\newblock In {\em Proceedings of the 18th International Conference on
  Autonomous Agents and MultiAgent Systems (AAMAS 2019)}.

\bibitem[Boyd and Vandenberghe, 2004]{boyd2004convex}
Boyd, S. and Vandenberghe, L. (2004).
\newblock {\em Convex optimization}.
\newblock Cambridge university press.

\bibitem[Cai et~al., 2024]{FGKLLZ24}
Cai, Y., Farina, G., Grand{-}Cl{\'{e}}ment, J., Kroer, C., Lee, C., Luo, H.,
  and Zheng, W. (2024).
\newblock Fast last-iterate convergence of learning in games requires forgetful
  algorithms.
\newblock In {\em Advances in Neural Information Processing Systems 38 (NeurIPS
  2024)}.

\bibitem[Cai et~al., 2025]{C25}
Cai, Y., Farina, G., Grand{-}Cl{\'{e}}ment, J., Kroer, C., Lee, C., Luo, H.,
  and Zheng, W. (2025).
\newblock On separation between best-iterate, random-iterate, and last-iterate
  convergence of learning in games.
\newblock {\em CoRR}, abs/2503.02825.

\bibitem[Cai et~al., 2022]{CZ22}
Cai, Y., Oikonomou, A., and Zheng, W. (2022).
\newblock Finite-time last-iterate convergence for learning in multi-player
  games.
\newblock In {\em Advances in Neural Information Processing Systems 35 (NeurIPS
  2022)}.

\bibitem[Cen et~al., 2023]{cen-et-al23faster}
Cen, S., Chi, Y., Du, S.~S., and Xiao, L. (2023).
\newblock Faster last-iterate convergence of policy optimization in zero-sum
  markov games.
\newblock In {\em The Eleventh International Conference on Learning
  Representations (ICLR 2023)}.

\bibitem[Cesa-Bianchi and Lugosi, 2006]{cesa2006prediction}
Cesa-Bianchi, N. and Lugosi, G. (2006).
\newblock {\em Prediction, learning, and games}.
\newblock Cambridge University Press.

\bibitem[Chen and Peng, 2020]{chen2020hedging}
Chen, X. and Peng, B. (2020).
\newblock Hedging in games: Faster convergence of external and swap regrets.
\newblock {\em Advances in Neural Information Processing Systems (NeurIPS
  2020)}.

\bibitem[Daskalakis et~al., 2021]{daskalakis2021near}
Daskalakis, C., Fishelson, M., and Golowich, N. (2021).
\newblock Near-optimal no-regret learning in general games.
\newblock {\em Advances in Neural Information Processing Systems (NeurIPS
  2021)}.

\bibitem[Daskalakis et~al., 2018]{daskalakis-et-al18training-gans}
Daskalakis, C., Ilyas, A., Syrgkanis, V., and Zeng, H. (2018).
\newblock Training {GAN}s with optimism.
\newblock In {\em International Conference on Learning Representations (ICLR
  2018)}.

\bibitem[Daskalakis and Panageas, 2019]{DaskalakisP19}
Daskalakis, C. and Panageas, I. (2019).
\newblock Last-iterate convergence: Zero-sum games and constrained min-max
  optimization.
\newblock In {\em 10th Innovations in Theoretical Computer Science Conference,
  (ITCS 2019)}.

\bibitem[Facchinei and Pang, 2003]{facchinei2003finite}
Facchinei, F. and Pang, J.-S. (2003).
\newblock {\em Finite-dimensional variational inequalities and complementarity
  problems}.
\newblock Springer.

\bibitem[Fasoulakis et~al., 2022]{fasoulakis-et-al22flbr}
Fasoulakis, M., Markakis, E., Pantazis, Y., and Varsos, C. (2022).
\newblock Forward looking best-response multiplicative weights update methods
  for bilinear zero-sum games.
\newblock In {\em International Conference on Artificial Intelligence and
  Statistics (AISTATS 2022)}.

\bibitem[Feng et~al., 2024]{feng-et-al24lastiterate}
Feng, Y., Li, P., Panageas, I., and Wang, X. (2024).
\newblock Last-iterate convergence separation between extra-gradient and
  optimism in constrained periodic games.
\newblock In {\em The 40th Conference on Uncertainty in Artificial Intelligence
  (UAI 2024)}.

\bibitem[Freund et~al., 2026]{freund2026second}
Freund, Y., Harvey, N.~J., Portella, V.~S., Qi, Y., and Wang, Y.-X. (2026).
\newblock A second order regret bound for normalhedge.
\newblock {\em arXiv preprint arXiv:2602.08151}.

\bibitem[Freund and Schapire, 1999]{freund1999adaptive}
Freund, Y. and Schapire, R.~E. (1999).
\newblock Adaptive game playing using multiplicative weights.
\newblock {\em Games and Economic Behavior}, 29(1-2):79--103.

\bibitem[Gao and Pavel, 2017]{gao-et-al17softmax}
Gao, B. and Pavel, L. (2017).
\newblock On the properties of the softmax function with application in game
  theory and reinforcement learning.
\newblock {\em arXiv preprint arXiv:1704.00805}.

\bibitem[Gidel et~al., 2019]{gidel2018a}
Gidel, G., Berard, H., Vignoud, G., Vincent, P., and Lacoste-Julien, S. (2019).
\newblock A variational inequality perspective on generative adversarial
  networks.
\newblock In {\em International Conference on Learning Representations (ICLR
  2019)}.

\bibitem[Golowich et~al., 2020]{golowich2020tight}
Golowich, N., Pattathil, S., and Daskalakis, C. (2020).
\newblock Tight last-iterate convergence rates for no-regret learning in
  multi-player games.
\newblock {\em Advances in Neural Information Processing Systems (NeurIPS
  2020)}.

\bibitem[Gorbunov et~al., 2022]{gorbunov-et-al22mvi}
Gorbunov, E., Taylor, A., and Gidel, G. (2022).
\newblock Last-iterate convergence of optimistic gradient method for monotone
  variational inequalities.
\newblock {\em Advances in Neural Information Processing Systems (NeurIPS
  2022)}.

\bibitem[Katona et~al., 2026]{katona2024symplectic}
Katona, J.~E., Wang, X., and Wibisono, A. (2026).
\newblock A symplectic analysis of alternating mirror descent.
\newblock {\em Journal of Machine Learning Research}, 27(44):1--61.

\bibitem[Korpelevich, 1976]{korpelevich1976extragradient}
Korpelevich, G.~M. (1976).
\newblock The extragradient method for finding saddle points and other
  problems.
\newblock {\em Matecon}, 12:747--756.

\bibitem[Lazarsfeld et~al., 2025a]{lazarsfeld2025optimism}
Lazarsfeld, J., Piliouras, G., Sim, R., and Skoulakis, S. (2025a).
\newblock Optimism without regularization: Constant regret in zero-sum games.
\newblock In {\em The Thirty-ninth Annual Conference on Neural Information
  Processing Systems (NeurIPS 2025)}.

\bibitem[Lazarsfeld et~al., 2025b]{lazarsfeld2025fp}
Lazarsfeld, J., Piliouras, G., Sim, R., and Wibisono, A. (2025b).
\newblock Fast and furious symmetric learning in zero-sum games: Gradient
  descent as fictitious play.
\newblock {\em Conference on Learning Theory (COLT 2025)}.

\bibitem[Lee et~al., 2021]{lee2021last}
Lee, C.-W., Kroer, C., and Luo, H. (2021).
\newblock Last-iterate convergence in extensive-form games.
\newblock {\em Advances in Neural Information Processing Systems (NeurIPS
  2021)}.

\bibitem[Legacci et~al., 2024]{legacci2024no}
Legacci, D., Mertikopoulos, P., Papadimitriou, C., Piliouras, G., and
  Pradelski, B. (2024).
\newblock No-regret learning in harmonic games: Extrapolation in the face of
  conflicting interests.
\newblock {\em Advances in Neural Information Processing Systems (NeurIPS
  2024)}.

\bibitem[Lei et~al., 2021]{lei-et-al21}
Lei, Q., Nagarajan, S.~G., Panageas, I., et~al. (2021).
\newblock Last iterate convergence in no-regret learning: constrained min-max
  optimization for convex-concave landscapes.
\newblock In {\em International Conference on Artificial Intelligence and
  Statistics (AISTATS 2021)}.

\bibitem[Liang and Stokes, 2019]{liang2019interaction}
Liang, T. and Stokes, J. (2019).
\newblock Interaction matters: A note on non-asymptotic local convergence of
  generative adversarial networks.
\newblock In {\em The 22nd International Conference on Artificial Intelligence
  and Statistics (AISTATS 2019)}.

\bibitem[Mertikopoulos et~al., 2019]{mertikopoulos-et-al18}
Mertikopoulos, P., Lecouat, B., Zenati, H., Foo, C.-S., Chandrasekhar, V., and
  Piliouras, G. (2019).
\newblock Optimistic mirror descent in saddle-point problems: Going the
  extra(-gradient) mile.
\newblock In {\em International Conference on Learning Representations (ICLR
  2019)}.

\bibitem[Mertikopoulos et~al., 2018]{mertikopoulos2018cycles}
Mertikopoulos, P., Papadimitriou, C., and Piliouras, G. (2018).
\newblock Cycles in adversarial regularized learning.
\newblock In {\em Proceedings of the twenty-ninth annual ACM-SIAM Symposium on
  Discrete Algorithms (SODA 2018)}.

\bibitem[Mokhtari et~al., 2020]{mokhtari-et-al20unified}
Mokhtari, A., Ozdaglar, A., and Pattathil, S. (2020).
\newblock A unified analysis of extra-gradient and optimistic gradient methods
  for saddle point problems: Proximal point approach.
\newblock In {\em International Conference on Artificial Intelligence and
  Statistics (AISTATS 2020)}.

\bibitem[Munos et~al., 2024]{munos2024nash}
Munos, R., Valko, M., Calandriello, D., Gheshlaghi~Azar, M., Rowland, M., Guo,
  Z.~D., Tang, Y., Geist, M., Mesnard, T., Fiegel, C., Michi, A., Selvi, M.,
  Girgin, S., Momchev, N., Bachem, O., Mankowitz, D.~J., Precup, D., and Piot,
  B. (2024).
\newblock {N}ash learning from human feedback.
\newblock In {\em Proceedings of the 41st International Conference on Machine
  Learning (ICML 2024)}.

\bibitem[Nemirovski, 2004]{nemirovski2004prox}
Nemirovski, A. (2004).
\newblock Prox-method with rate of convergence o (1/t) for variational
  inequalities with lipschitz continuous monotone operators and smooth
  convex-concave saddle point problems.
\newblock {\em SIAM Journal on Optimization}, 15(1):229--251.

\bibitem[Nesterov, 2007]{nesterov2007dual}
Nesterov, Y. (2007).
\newblock Dual extrapolation and its applications to solving variational
  inequalities and related problems.
\newblock {\em Mathematical Programming}, 109(2):319--344.

\bibitem[Nocedal and Wright, 2006]{nocedal2006numerical}
Nocedal, J. and Wright, S.~J. (2006).
\newblock {\em Numerical optimization}.
\newblock Springer.

\bibitem[Ota and Fujimoto, 2025]{ota-et-al25hamiltonian}
Ota, T. and Fujimoto, Y. (2025).
\newblock Hamiltonian of polymatrix zero-sum games.
\newblock {\em arXiv preprint arXiv:2505.12609}.

\bibitem[Popov, 1980]{popov1980modification}
Popov, L.~D. (1980).
\newblock A modification of the arrow-hurwicz method for search of saddle
  points.
\newblock {\em Mathematical notes of the Academy of Sciences of the USSR},
  28(5):845--848.

\bibitem[Rakhlin and Sridharan, 2013]{rakhlin2013optimization}
Rakhlin, S. and Sridharan, K. (2013).
\newblock Optimization, learning, and games with predictable sequences.
\newblock {\em Advances in Neural Information Processing Systems (NeurIPS
  2013)}.

\bibitem[Shalev-Shwartz et~al., 2012]{shalev2012online}
Shalev-Shwartz, S. et~al. (2012).
\newblock Online learning and online convex optimization.
\newblock {\em Foundations and Trends{\textregistered} in Machine Learning},
  4(2):107--194.

\bibitem[Soleymani et~al., 2025]{soleymani2025faster}
Soleymani, A., Piliouras, G., and Farina, G. (2025).
\newblock Faster rates for no-regret learning in general games via cautious
  optimism.
\newblock In {\em Proceedings of the 57th Annual ACM Symposium on Theory of
  Computing (STOC 2025)}.

\bibitem[Sun and Tran-Dinh, 2019]{sun2019generalized}
Sun, T. and Tran-Dinh, Q. (2019).
\newblock Generalized self-concordant functions: a recipe for newton-type
  methods.
\newblock {\em Mathematical Programming}, 178(1):145--213.

\bibitem[Syrgkanis et~al., 2015]{syrgkanis2015fast}
Syrgkanis, V., Agarwal, A., Luo, H., and Schapire, R.~E. (2015).
\newblock Fast convergence of regularized learning in games.
\newblock {\em Advances in Neural Information Processing Systems (NeurIPS
  2015)}.

\bibitem[Tiapkin et~al., 2026]{tiapkin-et-al26pp-nlhf}
Tiapkin, D., Calandriello, D., Belomestny, D., Moulines, E., Naumov, A., Rasul,
  K., Valko, M., and Menard, P. (2026).
\newblock Proximal point nash learning from human feedback.
\newblock {\em arXiv preprint arXiv:2505.19731}.

\bibitem[Tran-Dinh et~al., 2015]{tran2015composite}
Tran-Dinh, Q., Li, Y.-H., and Cevher, V. (2015).
\newblock Composite convex minimization involving self-concordant-like cost
  functions.
\newblock In {\em Proceedings of the 3rd International Conference on Modelling,
  Computation and Optimization in Information Systems and Management Sciences
  (MCO 2015)}.

\bibitem[von Neumann, 1928]{v1928theorie}
von Neumann, J. (1928).
\newblock Zur theorie der gesellschaftsspiele.
\newblock {\em Mathematische Annalen}, 100(1):295--320.

\bibitem[Wang et~al., 2026]{wang2026lastiterate}
Wang, G., Acharya, K., Lakshmikanthan, L., Ziani, J., and Muthukumar, V.
  (2026).
\newblock Last-iterate convergence for symmetric, general-sum, 2x2 games under
  the exponential weights dynamic.
\newblock In {\em 37th International Conference on Algorithmic Learning Theory
  (ALT 2026)}.

\bibitem[Wei et~al., 2021]{wei2021linear}
Wei, C.-Y., Lee, C.-W., Zhang, M., and Luo, H. (2021).
\newblock Linear last-iterate convergence in constrained saddle-point
  optimization.
\newblock In {\em International Conference on Learning Representations (ICLR
  2021)}.

\bibitem[Wibisono et~al., 2022]{wibisono2022alternating}
Wibisono, A., Tao, M., and Piliouras, G. (2022).
\newblock Alternating mirror descent for constrained min-max games.
\newblock {\em Advances in Neural Information Processing Systems (NeurIPS
  2022)}.

\end{thebibliography}

\newpage
\appendix

\setcounter{tocdepth}{1}
\renewcommand*\contentsname{Table of Contents}
{\setlength{\parskip}{-0.7em}\tableofcontents}


\section{Other Related Work}
\label{sec:related-work}

In this section, we give additional discussion on other related works.

\paragraph{Optimistic algorithms for min-max optimization and
variational inequalities.}
The OMWU algorithm is more broadly related to 
optimistic and extra-gradient methods
in monotone 
variational inequalities~\citep{facchinei2003finite}.
Such algorithms date back to the 1970s and
80s~\citep{korpelevich1976extragradient, popov1980modification},
and within the past decade have received growing theoretical interest,
in part due to their applications in machine learning settings
(see, e.g.,~\cite{nemirovski2004prox},~\cite{nesterov2007dual},
~\cite{daskalakis-et-al18training-gans},
~\cite{gidel2018a},~\cite{mertikopoulos-et-al18},
~\cite{liang2019interaction},
~\citet{mokhtari-et-al20unified}, 
~\cite{golowich2020tight},~\cite{gorbunov-et-al22mvi},
~\cite{CZ22}, and references therein).

\paragraph{Optimistic MWU for online learning and learning in games.}
The Optimistic Mirror Descent and Optimistic Follow-the-Regularized-Leader 
families for adversarial online learning, 
of which OMWU is the canonical instantiation with negative 
entropy regularizer, were introduced by~\cite{rakhlin2013optimization}.
Since then, OMWU has emerged as a premier algorithm
for online learning in games, and it has been proven
to achieve optimal or near-optimal regret bounds
(which correspond to time-average convergence guarantees)
in many classes of games, including
zero-sum and general-sum normal-form games
(see, e.g.,~\cite{syrgkanis2015fast, chen2020hedging, 
daskalakis2021near, anagnostides2022near, soleymani2025faster}, 
and references therein). 

\paragraph{Last-iterate Convergence of OMWU in zero-sum games.}
The asymptotic last-iterate convergence of OMWU in two-player 
bilinear zero-sum games was first established 
by~\cite{DaskalakisP19} on games with a unique Nash equilibrium,
and \cite{lei-et-al21} generalized this result to hold 
also for convex-concave zero-sum games. 
\cite{wei2021linear} later obtained a quantitative (universal) linear 
last-iterate convergence rate in the bilinear case. 
Unlike Optimstic Gradient-Descent-Ascent (OGDA),
which has a uniform $O(T^{-1/2})$ last-iterate convergence 
rate in this setting~\citep{gorbunov-et-al22mvi,CZ22},
OMWU was shown by~\cite{FGKLLZ24} to have no such 
uniform last-iterate convergence rate  in duality gap
(i.e., that does not depend on any game-dependent 
constant such as the minimum non-zero Nash equilibrium mass).
\cite{FGKLLZ24} attributed this behavior of OMWU 
to a notion they  call \textit{forgetfulness}, as the algorithm's update rule 
(as is the case for all OFTRL instantiations) inherently depends 
on the accumulation of all previous gradients. 
Our paper reinforces the main takeaways of their result: 
the last-iterate of OMWU (and, under some measures of distance to Nash,
the best-iterate convergence, 
c.f.~Theorem~\ref{thm:uniform-lb-main})
can be arbitrarily slow. However, our results also give 
new and more precise quantitative explanations for this slow
convergence that are related to geometric properties
of the algorithm. 

\paragraph{Energy-based perspective for learning in games.}
Several prior works have also relied on using
a dual or energy-based perspective for analyzing 
the convergence of various algorithms for learning in games.
See, e.g.,~\cite{mertikopoulos2018cycles},~\cite{bailey2018multiplicative},
~\cite{bailey2019fast},~\cite{bailey2019multi},~\cite{bailey2020finite},
\cite{abernethy2021fast},~\cite{wibisono2022alternating}, 
~\cite{azizian2024rate},
~\cite{legacci2024no},~\cite{ota-et-al25hamiltonian},~\cite{lazarsfeld2025fp},
\cite{lazarsfeld2025optimism}~\cite{katona2024symplectic},
and the references therein. However, our work is the 
first to prove quantitative last-iterate 
convergence rates for the 
bilinear zero-sum setting by proving bounds
on \textit{strictly dissipating} energy.

\paragraph{Other related works on learning in games.}
The extra-gradient variant of MWU has been studied in 
bilinear zero-sum games and proven to
have asymptotic last-iterate convergence~\citep{fasoulakis-et-al22flbr}.
In the unconstained min-max setting, optimistic 
and extra-gradient methods have also been studied
in time-varying and periodic games, where~\cite{feng-et-al24lastiterate}
established a separation between the convergence guarantees 
of the extra-gradient and optimistic variants. 
\cite{lee2021last} proved last-iterate convergence
for a variant of OMWU in extensive-form games. 
\cite{APFS22} proved a $O(T^{-1/2})$ best-iterage convergence 
rate in duality gap using Optimistic Mirror Descent on games with non-negative
sums of regret (which is satisfied by zero-sum games).
However, this result assumes a smoothness condition
on the regularizer that is not satisfied by the negative entropy function.
Thus, the $\widetilde O(T^{-1/2})$ best-iterate convergence
rate in duality gap that we prove in
Theorem~\ref{thm:dg-best-2x2} for the $2 \times 2$ setting
is the fastest-known best-iterate result for OMWU. 
Moreover, this result also adds to a growing body of work
that has established various convergence results for algorithms
in the setting of $2\times 2$ games, which often still requires
intricate proof techniques despite the low dimension
\citep{bailey2019fast,chen2020hedging,FGKLLZ24,
lazarsfeld2025optimism,C25, wang2026lastiterate}.

\section{Preliminaries on Zero-Sum Games}
\label{app:zsg-prelims}

This Section gives additional details on
the preliminaries of zero-sum games introduced 
in Section~\ref{sec:prelims}.

\paragraph{Oragnization of Section.}
This section is organized as follows:
\begin{itemize}[
	leftmargin=1em
]
\item
\textbf{Section~\ref{app:zsg-prelims:nash}}
recalls several important properties
of zero-sum games with an interior NE.

\item
\textbf{Section~\ref{app:zsg-prelims:distances}}
derives the relationships between
distances to Nash introduced in expression~\eqref{eq:ne-distances}.

\item
\textbf{Section~\ref{app:zsg-prelims:spectral}}
states and proves two useful properties related
to the spectrum of the payoff matrix $A$.
\end{itemize}

The claims of this section are mostly standard
but included for completeness.

\subsection{Properties of Zero-Sum Games with Interior Equilibria}
\label{app:zsg-prelims:nash}

Recall that for $A \in \R^{m \times n}$ and $\eps \ge 0$, 
an \textit{$\eps$-approximate Nash equilibrium (NE)}
$w^\star = (p^\star, q^\star) \in \calW$ for the 
min-max game $\min_{p \in \Delta_m} \max_{q \in \Delta_n} \langle p, A q \rangle$ 
satisfies the inequalities (see also~\eqref{eq:nash-prelims})
\begin{equation}
	\textstyle
	\max_{q \in \Delta_m} \, \langle q, A^\top p^\star \rangle
	- \eps
	\le 
	\langle p^\star, A q^\star \rangle 
	\le
	\min_{p \in \Delta_n} \, \langle p, A q^\star \rangle + \eps \;.
	\label{eq:nash-app}
\end{equation}
When $\eps = 0$, we say $w^\star$ is a Nash equlibrium.

If a game has an interior Nash equilibrium,
then the following key properties hold: 

\begin{restatable}{prop}{propinteriornash}
	\label{prop:interior-NE}
	Fix $A \in \R^{m \times n}$. Then the following properties hold:
	\begin{enumerate}[
		label={(\roman*)},
	]
	\item
		Suppose $w^\star = (p^\star, q^\star) \in \relint(\calW)$
		is an interior NE of $A$ with 
		$\langle p^\star, A q^\star \rangle = c\in \R$.
		Then
		$A q^\star = c \1_m$
		and $A^\top p^\star = c\1_n$.
 	\item	
 		Fix $w = (p, q) \in \calW$.
 		Suppose $A q = d \1_m$ 
 		and $A^\top p = d' \1_n$ for $d, d' \in \R$.
 		Then $d = d'$ and 
 		$w$ is a Nash equilibrium of $A$.
	\end{enumerate}
\end{restatable}

\begin{proof} 
	We prove the two claims separately:

	\noindent
	\textbf{Claim (i)}:
	By the definition of an NE from~\eqref{eq:nash-app},
	and since $\langle p^\star, A q^\star \rangle = c$,
	we must have
	\begin{equation}
		\textstyle
		c = \langle p^\star, A q^\star \rangle 
		\le \min_{p \in \Delta_m}  \, \langle p, A q^\star \rangle  \;.
		\label{eq:neintprop}
	\end{equation}
	We will first show that 
	both $(Aq^\star)(i) \ge c$ and $(Aq^\star)(i) \le c$ for all $i \in [m]$.
	For the first direction, suppose by way of contradiction that
	there exists $j \in [m]$ such that $(Aq^\star)(j) < c$.
	Then construct $p' \in \Delta_m$ as follows:
	set $p'(j) = 1$, and $p'(i) = 0$ for all other $i \neq j \in [m]$.
	Then observe that 
	\begin{equation*}
		\textstyle
		\min_{p  \in \Delta_m}
		\langle p, Aq^\star \rangle 
		\le \langle p', A q^\star \rangle 
		= \sum_{k=1}^m p'(k) \cdot (Aq^\star)(k)
		= (Aq^\star)(j) < c = \langle p^\star, A q^\star \rangle \;.
	\end{equation*}
	However, this contradicts the fact that
	$\langle p^\star, A q^\star \rangle 
	\le \min_{p \in \Delta_m} \langle p, Aq^\star \rangle$
	from~\eqref{eq:neintprop}. 
	Thus we must have $(Aq^\star)(i) \ge c$
	for all coordinates $i \in [m]$.

	For the other direction, suppose by contradiction
	that there exists $j \in [m]$ where
	$(Aq^\star)(j) > c$. Note that if $(Aq^\star)(j) > c$
	for \textit{all} coordinates $j \in [m]$,
	then clearly $\langle p^\star, Aq^\star \rangle > c$,
	which contradicts~\eqref{eq:neintprop}.
	Thus suppose instead there exists coordinate
	$i \neq j \in [m]$ with $(Aq^\star)(i) \le c$.
	Now construct $p' \in \Delta_m$ as follows:
	set $p'(k) = p^\star(k)$ for 
	all $k \neq i \neq j$, 
	set $p'(j) = 0$, and set $p'(i) = p^\star(i)+p^\star(j)$.
	Observe by construction that this implies
	\begin{equation*}
		\textstyle
		p'(j) \big((Aq^\star)(j)\big) + p'(i) \big( (Aq^\star)(i)\big)
		< 
		p^\star(j) \big(Aq^\star(j)\big) + p^\star(i) \big(Aq^\star)(i)\big) \;.
	\end{equation*}
	Then this further implies by construction that
	\begin{equation*}
		\textstyle
		\min_{p \in \Delta_m} \langle p, Aq^\star \rangle
		\le
		\langle p', Aq^\star \rangle
		< 
		\langle p^\star, A q^\star \rangle
		= c \;,
	\end{equation*}
	which contradicts~\eqref{eq:neintprop}.
	Thus we also have $(Aq^\star)(i) \le c$
	for all coordinates $i \in [m]$,
	and thus $Aq^\star = c\1_m$.
	Repeating identical arguments, we also have
	$A^\top p^\star = c\1_n$, which completes the claim. 

	\smallskip 

	\noindent
	\textbf{Claim (ii).}
	By the assumptions of the claim that 
	$Aq = d\1_m$ and $A^\top p = d' \1_n$, 
	observe for any $u \in \Delta_m$ that 
	$\langle u, Aq \rangle = \langle u, d\1_m \rangle = d$,
	and for any $v \in \Delta_n$ that 
	$\langle v, A^\top p \rangle = \langle v, d' \1_n \rangle = d'$.
	Then since $\langle p, A q \rangle = \langle q, A^\top p \rangle$
	and both $p \in \Delta_m$ and $q \in \Delta_n$,
	we must have $d= d' = \langle p, A q \rangle$.
	To further show that $(p, q)$ is a Nash equilibrium,
	observe that since $Aq = d \1_m$
	and $A^\top p = d \1_n$, then
	\begin{equation*}
		\textstyle
		\max_{q' \in \Delta_n} 
		\langle q', A^\top p \rangle
		= 
		\langle q, A p \rangle 
		=
		\min_{p' \in \Delta_m}
		\langle p', A q \rangle \;.
	\end{equation*}
	Thus by definition in~\eqref{eq:nash-app},
	$w = (p, q)$ is a Nash equilibrium of $A$.
\end{proof}

\subsection{Distances to Nash Equilibrium}
\label{app:zsg-prelims:distances}

Recall from Section~\ref{sec:prelims} that we 
consider several measures of distance to Nash;
duality gap ($\DG$), total variation distance ($\TV$),
and KL divergence ($\KL$). 
We first restate these definitions.
Let $w^\star = (p^\star, q^\star) \in \relint(\calW)$
be a Nash equilibrium of $A$. 
Then for any $w = (p,q) \in \calW$, we define:
\begin{itemize}[
	leftmargin=1em,
]
	\item
	\textbf{Duality gap (DG)}:
	$\DG(w) =
	 \max_{q' \in \Delta_m} \langle q', A^\top p\rangle
	 - \min_{p' \in \Delta_n} \langle p', Aq \rangle$.

	\item
	\textbf{Kullback-Leibler divergence (KL)}:
	Define the components
	$\KL(p^\star, p) = \sum_{i=1}^m p^\star(i) 
	\log \big(p^\star(i)/p(i)\big)$ and 
	$\KL(q^\star, q) = \sum_{j=1}^n q^\star(j) 
	\log \big(q^\star(j)/q(j)\big)$.
	Then $\KL(w^\star, w) = \KL(p^\star, p) + \KL(q^\star, q)$.

	\item
	\textbf{Total variation distance (TV)}:
	Define the components $\TV(p^\star,p) = \frac{1}{2}\|p^\star - p\|_1$
	and $\TV(q^\star, q) = \frac{1}{2} \|q^\star - q\|_1$.
	Then $\TV(w^\star, w) = \TV(p^\star, p) + \TV(q^\star, q)$.
\end{itemize}
We also define the chi-squared divergence
$\chi^2(w^\star, w)$ as follows:
\begin{itemize}[
	leftmargin=1em,
]
	\item
	\textbf{Chi-squared divergence:}
	First define the components
	$\chi^2(p^\star, p) = \sum_{i=1}^m \frac{(p^\star(i)- p(i))^2}{p(i)}$
	and 
	$\chi^2(q^\star, q) = \sum_{j=1}^n \frac{(q^\star(j)- q(j))^2}{q(j)}$.
	Then $\chi^2(w^\star, w) = \chi^2(p^\star, p) + \chi^2(q^\star, q)$.

\end{itemize}

\subsubsection{Bounded Duality Gap Implies Approximate Nash Equilibrium}

A standard fact is that a bounded duality gap implies 
an approximate Nash equilibrium. Formally:

\begin{restatable}[Small DG implies Nash]
	{prop}{propdgnash}
	\label{prop:dg-nash}
	Let $A \in \R^{m \times n}$ be a zero-sum game.
	Fix $w = (p, q) \in \calW$ and $\eps \ge 0$.
	If $DG(w) \le \eps$, then $w$ is an $\eps$-approximate Nash
	equilibrium of $A$.
\end{restatable}

\begin{proof}
	Using the definition of duality gap and adding 
	and subtracting 
	$\langle p, Aq \rangle = \langle q, A^\top p \rangle$, 
	we can write and expand
	\begin{align*}
		\DG(w)
		&= 
		\max\nolimits_{q' \in \Delta_m}
		\langle q', A^\top p \rangle
		- 
		\min\nolimits_{p' \in \Delta_n}
		\langle p', A q \rangle \\
		&= 
		\big(
		\max\nolimits_{q' \in \Delta_m}
		\langle q', A^\top p \rangle
		- 
		\langle q, A^\top p \rangle
		\big)
		+ 
		\big(
		\langle
			p, A q
		\rangle
		- \min\nolimits_{p' \in \Delta_n}
		\langle p', A q \rangle
		\big) \;.
	\end{align*}
	Observe that both terms in the above expression
	are non-negative. Thus, using the assumption that
	$\DG(w) \le \eps$, it follows by rearranging that
	\begin{equation*}
		\max\nolimits_{q' \in \Delta_m}
		\langle q', A^\top p \rangle
		- \eps \le \langle q, A^\top p \rangle 
		\;\;\;\;\text{and}\;\;\;\;
		\langle
			p, A q
		\rangle
		\le 
		\min\nolimits_{p' \in \Delta_n}
		\langle p', A q \rangle + \eps \;.
	\end{equation*}
	Thus $w = (p, q)$ satisfies the inequalities
	in~\eqref{eq:nash-app} and is therefore an $\eps$-approximate NE
	of $A$.
\end{proof}

\subsubsection{Relationships between Duality Gap, KL, and TV}
\label{sec:prelims-zsg:dg-kl-tv-relations}

In this section, we prove a series of inequalities 
relating the distances $\DG$, $\TV$, $\KL$, and $\chi^2$.
These relationships were introduced
in expression~\eqref{eq:ne-distances}) 
from Section~\ref{sec:prelims}.
Here, we give a proof of each component separately,
most of which are standard, but included for completness.
Futher below, we restate the full sequence of
inequalities from~\eqref{eq:ne-distances}
in Corollary~\eqref{cor:distance-relations}.

\smallskip

\paragraph{Relating Duality Gap and TV Distance.}

\begin{restatable}[Duality-Gap vs. TV distance]
	{prop}{propdgtv}
	\label{prop:dg-tv-relationship}
	Let $A \in \R^{m \times n}$ be a zero-sum game 
	with Nash equilibrium $w^\star = (p^\star, q^\star) \in \calW$. 
	Let $a_{\max} = \max_{(i, j) \in [m] \times [n]} |A(i, j)|$ denote
	the maximum absolute entry of $A$. Then for any 
	$w = (p, q) \in \calW$, it holds that:
	\begin{equation*}
		\DG(w) \le 2 a_{\max} \TV(w^\star, w) \;.
	\end{equation*}
\end{restatable}

\begin{proof}
	Let $\widetilde w = (\w p, \w q) \in \calW$ be a 
	pair of distributions such that 
	$\max_{q' \in \Delta_m} \langle q', A^\top p \rangle 
	=  \langle \w q, A^\top p \rangle$
	and 
	$\min_{p' \in \Delta_n} \langle p', Aq \rangle
	= \langle \w p, A q \rangle$.
	Then by definition of a Nash equilibrium $w^\star = (p^\star, q^\star)$
	from~\eqref{eq:nash-app}, it follows that
	\begin{equation*}
		\textstyle
		\langle \w q, A^\top p^\star \rangle
		- \langle q^\star, A^\top p^\star \rangle 
		\le 0 
		\;\text{and}\;
		\langle \w p, A q^\star \rangle
		-
		\langle p^\star, A q^\star \rangle 
		\ge 0 \;.
	\end{equation*}
	Let $c = \langle p^\star, A q^\star \rangle$,
	and thus the above inequalities imply
	$\langle \w q, A^\top p^\star \rangle - c \le 0$
	and 
	$\langle \w p, A q^\star \rangle - c \ge 0$.
	Then it follows by definition of 
	 $\DG(w)$, that we can write and further simplify
	\begin{align}
		\DG(w) 
		&= 
		\max\nolimits_{q' \in \Delta_m} \langle q', A^\top p \rangle
		- \min\nolimits_{p' \in \Delta_n} \langle p', Aq \rangle \nonumber \\
		&= 
		\langle \w q, A^\top p \rangle 
		- \langle \w p, A q \rangle \nonumber \\
		&\le 
		\langle \w q, A^\top p \rangle 
		- (\langle \w q, A^\top p^\star \rangle - c)
		- \langle \w p, A q \rangle
		+ (\langle \w p, A q^\star \rangle - c) \nonumber \\
		&= 
		\langle \w q, A^\top (p - p^\star) \rangle
		+ 
		\langle \w p, A(q^\star - q) \rangle \nonumber \\
		&\le
		\|A \w q\|_{\infty} \cdot \|p - p^\star \|_1
		+ 
		\|A^\top \w p \|_{\infty} \cdot \|q - q^\star \|_1 \;,
		\label{eq:dge-01}
	\end{align}
	where the final inequality comes from 
	H{\"o}lder's inequality.
	As $(\w p, \w q) \in \calW$, it follows that both
	$\|A \w q \|_{\infty}, \|A^\top \w p \|_{\infty} \le a_{\max}$.
	Moreover, by definition of $\TV(w^\star, w)$, 
	we conclude from~\eqref{eq:dge-01} that
	\begin{equation*}
		\DG(w) 
		\le 
		a_{\max}(\|p^\star- p\|_1 + \|q^\star-q\|_1)
		= 2 a_{\max} \TV(w^\star, w)\;,
	\end{equation*}
	which concludes the proof. 
\end{proof}

\smallskip
\paragraph{Relating TV Distance and KL Divergence.}

\smallskip

\begin{restatable}[Pinsker-like Inequality for TV and KL]
	{prop}{proppinskermod}
	\label{prop:pinsker-mod}
	Fix any $w^\star = (p^\star, q^\star) \in \calW$
	and any $w = (p, q) \in \calW$. Then it holds that
	$\TV(w^\star, w)^2 \le \KL(w^\star, w)$.
\end{restatable}

\begin{proof}
	Recall by our definition of $\TV$ and $\KL$ that
	$\TV(w^\star, w) = \TV(p^\star, p) + \TV(q^\star, q)$
	and 
	$\KL(w^\star, w) = \KL(p^\star, p) + \KL(q^\star, q)$.
	By Pinsker's inequality, we have that for each component that
	\begin{equation*}
		2 \cdot \TV(p^\star, p)^2 \le \KL(p^\star, p) 
		\;\text{and}\; 
		2 \cdot \TV(q^\star, q)^2 \le \KL(q^\star, q)  \;.
	\end{equation*}
	Thus together this implies 
	\begin{equation*}
		\KL(w^\star, w)
		\ge 
		2 \cdot \big(\TV(p^\star, p)^2 + \TV(q^\star, q)^2\big)
		\ge 
		\big(\TV(p^\star, p) + \TV(q^\star, q))^2
		= 
		\TV(w^\star, w)^2 \;.
	\end{equation*}
	Here, the second inequality comes from the fact
	that $\sqrt{(a^2 + b^2) /2} \ge (a+b)/2$
	and thus also $2(a^2 + b^2) \ge a+b$ for all $a, b \in \R$.
\end{proof}

\smallskip
\paragraph{Relating KL Divergence and Chi-Squared Divergence.}

\smallskip

\begin{restatable}[KL Divergence vs. $\chi^2$ Divergence]
	{prop}{propklchisq}
	\label{prop:kl-chisq}
	Fix $w^\star = (p^\star, q^\star) \in \calW$
	and $w = (p, q) \in \calW$. Then it holds that
	$\KL(w^\star, w) \le \chi^2(w^\star, w)$.
\end{restatable}

\begin{proof}
We have by definition of $\KL(p^\star, p)$ that
\begin{equation*}
	\textstyle
	\KL(p^\star, p)
	= 
	\sum_{i=1}^m p^\star(i) \cdot \log\big(\tfrac{p^\star(i)}{p(i)}\big)
	\le 
	\sum_{i=1}^m p^\star(i) \cdot \big(\tfrac{p^\star(i)}{p(i)} - 1 \big) 
	= 
	\sum_{i=1}^m \tfrac{p^\star(i)^2}{p(i)} - 1
	\;,
\end{equation*}
where the inequality comes from the 
fact that $\log(u) \le u-1$ for all $u > 0$.
On the other hand, we have by definition of $\chi^2(w^\star, w)$
that  
\begin{equation*}
	\textstyle
	\chi^2(p^\star, p)
	= \sum_{i=1}^m \frac{(p^\star(i) - p(i))^2}{p(i)}
	= \sum_{i=1}^m \frac{p^\star(i)^2}{p(i)} 
	+ \frac{p(i)^2}{p(i)} - \frac{2p(i) p^\star(i)}{p(i)}
	= \sum_{i=1}^m \frac{p^\star(i)^2}{p(i)} - 1 \;.
\end{equation*}
Combining the two expressions then yields
$\KL(p^\star, p) \le \chi^2(p^\star, p)$.
An identical calculation also gives
$\KL(q^\star, q) \le \chi^2(q^\star, q)$,
which then implies by definition that 
$\KL(w^\star, w) \le \chi^2(w^\star, w)$.
\end{proof}

\paragraph{Full Sequence of Relationships.}

The following corollary summarizes the relationships
between the measures of distance to Nash
from Propositions~\ref{prop:dg-tv-relationship},
~\ref{prop:pinsker-mod}, and~\ref{prop:kl-chisq}:

\smallskip

\begin{restatable}[Relationships of Distances to Nash]
{cor}{cordistancerelations}
	\label{cor:distance-relations}
	Let $A \in \R^{m \times n}$ be a zero-sum game 
	with a Nash equilibrium $w^\star = (p^\star, q^\star) \in \calW$. 
	Let $a_{\max} = \max_{(i, j) \in [m] \times [n]} |A(i, j)|$ denote
	the maximum absolute entry of $A$. Then for any 
	$w = (p, q) \in \calW$, it holds that:
	\begin{equation*}
	\tfrac{1}{(2a_{\max})^2} \cdot \DG(w)^2
	\le 
	\TV(w^\star, w)^2
	\le 
	\KL(w^\star, w)
	\le \chi^2(w^\star, w) \;.
	\end{equation*}
\end{restatable}

Note that, due to Proposition~\ref{prop:dg-nash},
Corollary~\ref{cor:distance-relations} implies
that upper bounds on $\TV$, $\KL$, and $\chi^2$
each also imply an approximate Nash
in the sense of~\eqref{eq:nash-app}.

\subsection{Properties of Spectrum of Payoff Matrix}
\label{app:zsg-prelims:spectral}

Here, we state and prove two useful properties
related to the spectrum of the payoff matrix $A$.

\begin{restatable}[Relationship
between max entry and 
max singular value]
{prop}{propsigmaxamax}
	\label{prop:amax-sigmamax}
	Fix $A \in \R^{m \times n}$.
	Let 
	$a_{\max} = \max_{(i, j) \in [m] \times [n]}
	|A(i, j)|$
	and let $\sigma_{\max} = \|A\|_2$.
	Then:
	\begin{equation*}
		a_{\max} \le \sigma_{\max}
		\le \sqrt{mn} \cdot a_{\max} \;.
	\end{equation*}
\end{restatable}

\begin{proof}
	For the first inequality, let $e_i \in \R^m$,
	$e_j \in \R^n$ for $i \in [m]$, $j\in [n]$
	denote the standard $m$ and $n$-dimensional 
	basis vectors. 
	Then by Cauchy-Schwarz and the definition
	of $\sigma_{\max}$, we have 
	\begin{equation*}
		\textstyle
		\max_{i, j}\, |A(i, j)| 
		= 
		\max_{i, j}\, |\langle e_i, A e_j \rangle|
		\le \|e_i\|_2 \cdot \|A e_j\|_2 
		\le \sigma_{\max}  \;.
	\end{equation*}
	For the second inequality, 
	recall that the Frobenius norm of $A$ is given by
	$\|A\|_F = (\sum_{(i, j) \in [m] \times [n]} A(i, j))^{1/2}$, 
	and that $\|A\|_2 \le \|A\|_F$.
	Then we can further bound
	\begin{equation*}
		\textstyle
		\|A\|^2_2 \le \|A\|^2_F 
		= 
		\sum_{(i, j)\in [m]\times [n]} A(i, j)^2 
		\le mn \cdot a_{\max} \;.
	\end{equation*}
	Taking square roots yields the desired inequality. 
\end{proof}

In the next proposition, we show that
$\|J \|_2 = \|A\|_2$, where $J$ is the 
skew-symmetric  matrix from~\eqref{eq:J-def}.

\begin{restatable}[Spectral norms of $J$ and $A$]
{prop}{propspectralnormJA}
\label{prop:spectral-norm-J-A}
Fix $A \in \R^{m \times n}$, 
and let $J =  ((0, A), (-A^\top, 0)) \in\R^{(m+n)\times (m+n)}$
be the block skew-symmetric matrix from~\eqref{eq:J-def}.
Then $\|J\|_2 = \|A\|_2$.
\end{restatable}

\begin{proof}
	For a symmetric matrix $M \in \R^{k \times k}$, 
	let $\lambda_{\max}(M)$ denote the largest eigenvalue of $M$.
	Then recall by definition of the matrix spectral norm that,
	for a matrix $P$, we have 
	$\sigma_{\max}(P) = \| P \|_2 = \sqrt{\lambda_{\max}(P^\top P)}$.
	For the skew-symmetric $J$, we can compute
	\begin{equation*}
		J^\top J 
		= 
		\begin{pmatrix}
		0 & A \\ 
		-A^\top & 0
		\end{pmatrix}
		\begin{pmatrix}
		0 & -A \\ 
		A^\top & 0
		\end{pmatrix}
		= 
		\begin{pmatrix}
		AA^\top & 0 \\ 
		0 & A^\top A
		\end{pmatrix} \;.
	\end{equation*}
	Given the block structure of $J^\top J$,
	it follows that 
	$\lambda_{\max}(J^\top J) = 
	\max(\lambda_{\max}(AA^\top), \lambda_{\max}(A^\top A))$.
	However, $\lambda_{\max}(AA^\top) = \lambda_{\max}(A^\top A)$,
	and thus 
	\begin{equation*}
		\|J\|_2 
		= \sqrt{\lambda_{\max}(J^\top J)}
		= 
		\sqrt{\lambda_{\max}(A^\top A)}
		= 
		\|A\|_2 \;,
	\end{equation*}
	which yields the desired claim.
\end{proof}

\section{Details on Optimistic MWU Algorithm}
\label{app:omwu-prelims}

In this section, we review preliminaries on the OMWU algorithm
that were introduced in Section~\ref{sec:omwu}.

\paragraph{Organization of Section.}
This section is organized as follows:
\begin{itemize}[
	leftmargin=1em
]
\item
\textbf{Section~\ref{app:omwu-prelims:iterates}}
reviews the derivation of the primal and dual 
OMWU iterates (and also reviews of the standard MWU algorithm).
In particular, Section~\ref{app:omwu-prelims:oftrl}
gives details on OMWU as the instantiation of
Optimistic FTRL using the negative entropy regularizer.

\item
\textbf{Section~\ref{app:omwu-prelims:skew-grad}}
gives the proof of Proposition~\ref{prop:omwu-skew-grad},
which establishes the perspective 
of the dual OMWU iterates as an optimistic skew-gradient
descent on the log-sum-exp energy function. 

\end{itemize}

\subsection{Primal and Dual OMWU Iterates}
\label{app:omwu-prelims:iterates}

\subsubsection{Primal Update Rule}

\paragraph{Standard MWU.}
We first recall the primal update rule of the standard
Multiplicative Weights Update algorithm (MWU)
(see, e.g.,~\cite{freund1999adaptive},
\cite{arora2012multiplicative}),
which is the basis of the optimistic variant.
The MWU iterates are initialized
from $w_0 = (p_0, q_0) \in \relint(\calW)$. 
Then using a fixed stepsize $\eta > 0$,
for $t \ge 1$,  the MWU iterates $w_t = (p_t, q_t) \in \relint(\calW)$ 
are given by
\begin{equation}
	\begin{aligned}
	p_t(i) 
	&\propto
	p_{t-1}(i) \cdot \exp\big(
		- \eta (Aq_{t-1})(i)
	\big)
	\;\;\text{for all $i \in [m]$}
	\\
	\text{and}\quad
	q_t(j)
	&\propto
	q_{t-1}(j) \cdot \exp\big(
		\eta (A^\top p_{t-1})(j)
	\big)
	\;\;\text{for all $j \in [n]$} \;.
	\end{aligned}
	\label{eq:standard-mwu-primal}
	\tag{MWU}
\end{equation}

\paragraph{Optimistic MWU.}
The Optimistic MWU update rule is similar 
to the standard~\eqref{eq:standard-mwu-primal} update, 
but with an additional recency bias. 
This leads to increases
in the weights of coordinates corresponding
to smaller magnitude losses in the most recent gradient vector. 
Specifically, recall from~\eqref{eq:omwu-primal} that 
the primal iterates $\{w_t\}$ are initialized
from $w_0 = (p_0, q_0) \in \relint(\calW)$. 
Then using a fixed stepsize $\eta > 0$,
for $t +1 \ge 1$,  the iterates $w_t = (p_t, q_t) \in \relint(\calW)$ 
are given by
\begin{equation}
	\begin{aligned}
		p_{t+1}(i) 
        &\propto p_{t}(i) 
		\cdot
		\exp\big(- \eta (2Aq_{t} - Aq_{t-1})(i)\big)
		&\text{for all $i \in [m]$}  \\
        \text{and}\quad
		q_{t+1}(j) 
        &\propto q_{t}(j) 
		\cdot 
		\exp\big(\eta (2A^\top p_{t} - A^\top p_{t-1})(j)\big)
		&\text{for all $j \in [n]$}\;.
	\end{aligned}
\label{eq:omwu-primal-components}
\end{equation}
We assume for notational convenience that
$w_{-1} = (p_{-1}, q_{-1}) = w_0 \in \relint(\calW)$.
This means at time $t=1$ that 
$p_1(i) \propto p_0(i) \cdot \exp(- \eta(Aq_0)(i))$
for $i \in[m]$
and  $q_1(j) \propto q_0(j) \cdot \exp(\eta(A^\top p_0)(j))$
for $j \in [n]$.

\subsubsection{Dual Iterates and Optimistic FTRL Instantiation}
\label{app:omwu-prelims:oftrl}

\paragraph{Component-Wise Dual Iterates.}
We start by introducing the player-wise components 
of the dual iterates $\{z_t\}$ from~\eqref{eq:zt-wt-def}. 
For this, let $x_0 = x_{-1} \in \R^{m}$ and
$y_0 = y_{-1} \in \R^n$ be 
initial dual vectors.
Let $\{w_t\}$ denote the sequence of 
primal~\eqref{eq:omwu-primal} iterates, 
where each $w_t = (p_t, q_t) \in \relint(\calW)$. 
Then for stepsize $\eta  > 0$, 
the dual iterates $x_t \in\R^m$,
$y_t \in \R^n$ are defined for all $t \ge 1$ as follows:
first, at time $t=1$, let 
$x_1 = x_0 - \eta Aq_0$
and $y_1 = y_0 + \eta A^\top p_0$.
Then for all $t \ge 2$, we define
\begin{equation}
	\begin{aligned}
		\begin{cases}
		x_t = x_0 - \eta \cdot 
		\big(\sum_{k=0}^{t-1} Aq_k + Aq_{t-1}\big)  \\
		y_t = y_0 + \eta \cdot 
		\big(\sum_{k=0}^{t-1} A^\top p_k + A^\top p_{t-1} \big) \;.
		\end{cases}
	\end{aligned}
	\label{eq:omwu-dual-components}
\end{equation}

\paragraph{Optimistic FTRL Using Negative-Entropy Regularizer.}
As introduced in Section~\ref{sec:omwu}, 
the primal OMWU iterates can also be derived as 
an instantiation of Optimstic FTRL (OFTRL) using
negative entropy regularization
\citep{rakhlin2013optimization,syrgkanis2015fast}.
For this, recall that we define
the strictly convex regularizers $R_m : \Delta_m \to \R$
and $R_n : \Delta_n \to \R$ as 
the discrete negative entropy functions:
\begin{equation}
	\begin{aligned}
	R_n(p) &:= - \ent_m(p) 
	= \ssum\nolimits_{i=1}^m p(i) \log(p(i)) \\
	\text{and}\quad
	R_m(q) &:= - \ent_m(q) 
	= \ssum\nolimits_{j=1}^n q(i) \log(q(i)) \;.
	\end{aligned}
	\label{eq:regularizer-components}
\end{equation}
Here and throughout, we use the standard convention
that $0 \log 0 = 0$.

Under Optimistic FTRL, the primal iterates 
$p_t \in \Delta_m$, $q_t \in \Delta_n$ at time $t\ge1$
are chosen via a regularized best-response map
applied to the dual iterates $x_t \in \R^m$ and
$y_t \in \R^n$, respectively.
Specifically, 
let the OFTRL primal iterates be initialized
at $p_0 \in \relint(\Delta_m)$ and $q_0 \in \relint(\Delta_n)$.
At time $t + 1 \ge 1$, suppose
$x_{t+1} \in \R^m$ and $y_{t+1} \in \R^n$
are the dual iterates from~\eqref{eq:omwu-dual-components}
defined in terms of $\{q_0, \dots, q_t\}$
and $\{p_0, \dots, p_t\}$, respectively.
Then at time $t+1$, the primal iterates of 
OFTRL instantiated with the regularizers $R_m$ and $R_n$
are given by the objectives
\begin{equation}
\begin{aligned}
	p_{t+1} 
	&:= \argmin\nolimits_{p \in \Delta_m}
	\big\{
	\langle p, x_{t+1} \rangle + R_m(p)
	\big\} \\
	q_{t+1} 
	&:= \argmin\nolimits_{q \in \Delta_n}\,
	\big\{
	\langle q, y_{t+1} \rangle + R_n(q)
	\big\} \;.
\end{aligned}
\label{eq:oftrl-update-components}
\end{equation}
Here, observe that the dual variables $(x_{t+1}, y_{t+1})$
from~\eqref{eq:omwu-dual-components} are
already scaled by the stepsize $\eta > 0$.
Thus, smaller values of $\eta$ correspond
to more weight on the regularization terms
in~\eqref{eq:oftrl-update-components}, and vice-versa.
Moreover, under the negative entropy regularizers,
the variables $(p_{t+1}, q_{t+1})$ have the following
closed-form solution: 

\smallskip 

\begin{restatable}[Softmax update rule]
	{prop}{propsoftmaxupdate}
	\label{prop:oftrl-softmax}
	For $t \ge 1$, let $x_t \in \R^m$ and $y_t \in \R^n$,
	and let $p_t \in \Delta_m$ and $q_t \in\Delta_n$
	be defined as in~\eqref{eq:oftrl-update-components}. 
	Then $p_t = \softmax_m(x_t)$ and $q_t = \softmax_n(y_t)$.
\end{restatable}

This result is standard and is based on the 
first-order optimality conditions of the objective
in~\eqref{eq:oftrl-update-components},
and thus we omit the proof;
See, e.g.,~\cite{shalev2012online}, Section 2.6.

Now, as previously mentioned,
the update rules of~\eqref{eq:oftrl-update-components}
correspond exactly to the~\eqref{eq:omwu-primal} primal iterates.
Due to Proposition~\ref{prop:oftrl-softmax}, this will
follow as an immediate corollary of
Proposition~\ref{prop:omwu-skew-grad} from Section~\ref{sec:omwu}.
We restate and give the proof of this latter result in the next subsection,
which requires using the \textit{concatenated}
primal and dual variables $w_t =(p_t, q_t) \in \relint(\calW)$
and $z_t = (x_t, y_t) \in \R^{m+n}$.
This notation was introduced in
expressions~\eqref{eq:J-def},~\eqref{eq:zt-wt-def},
and~\eqref{eq:omwu-oftrl} from Section~\ref{sec:omwu},
and we briefly review these preliminaries again here:

\paragraph{Concatenated Primal and Dual Iterates.}
Given $A \in \R^{m\times n}$, recall from expression~\eqref{eq:J-def} 
that $J = - J^\top$
is the block skew-symmetric matrix 
$J = ((0, A), (-A^\top, 0)) \in \R^{(m+n)\times (m+n)}$.
Then observe for any $w = (p, q) \in \calW$ that
\begin{equation*}
	\textstyle
	J w = 
	\begin{pmatrix}
	0 & A \\
	- A^\top & 0 
	\end{pmatrix}
	\begin{pmatrix} 
	p \\ q
	\end{pmatrix}
	=
	\begin{pmatrix}
		Aq  \\
		- A^\top p 
	\end{pmatrix} \in \R^{m+n} \;.
\end{equation*}
Thus, given the sequence of primal iterates
$\{w_k\}_{k=1}^t$, where each $w_k = (p_k, q_k) \in \calW$,
it follows for all $t + 1 \ge 2$ that
the dual iterates $x_{t+1} \in \R^m$ and $y_{t+1} \in \R^n$
from~\eqref{eq:omwu-dual-components} can
be written jointly as
\begin{equation*}
	z_{t+1} = (x_{t+1}, y_{t+1})
	= 
	z_0  - \eta \cdot \big(
		\ssum\nolimits_{k=1}^t Jw_k + Jw_t
	\big) \in \R^{m+n} \;,
\end{equation*}
where $z_0 = z_{-1} = (x_0, y_0) \in \R^{m+n}$ is the dual initialization.
It follows recursively that we can write:
\begin{equation}
	\begin{cases}
	z_1 = z_0 - \eta J w_0 &\text{at $t=1$} \\
	z_{t+1} = z_{t} - 2 \eta J w_t + \eta J w_{t-1} 
	&\text{at $t +1 \ge 2$} \;.
	\end{cases}
	\label{eq:zt-wt-def-full}
\end{equation}

Further recall from Section~\ref{sec:omwu} that 
we define the joint negative entropy 
regularizer $R: \calW \to \R$ to be
$R = R_m + R_n = - (\ent_m + \ent_n)$, 
where for $w = (p, q) \in \calW$, we have
\begin{equation}
	R(w) = R_m(p) + R_n(q)
	= - (\ent_m(p) + \ent_n(q)) \;.
	\label{eq:R-def-full}
\end{equation}
The separability of the regularizer $R$ 
allows us to write the OFTRL objective
functions from~\eqref{eq:oftrl-update-components}
using the concatenated dual variable 
$z_{t+1} = (x_{t+1}, y_{t+1}) \in \R^{m+n}$.
Specifically, it holds for $t +1 \ge 1$ that
$w_{t+1} = (p_{t+1}, q_{t+1}) \in \calW$ is given by
\begin{equation}
	w_{t+1}
	= 
	\argmin_{w = (p, q) \in \calW}\,
	\big\{
		\langle w, z_{t+1} \rangle + R(w)
	\big\}
	= 
	\begin{pmatrix}
	\argmin\nolimits_{p \in \Delta_m}\,
	\big\{ \langle p, x_{t+1} \rangle + R_m(p) \big\} \\
	\argmin\nolimits_{q \in \Delta_n}\,
	\big\{ \langle q, y_{t+1} \rangle + R_n(q) \big\} 
 	\end{pmatrix}
 	\in \calW \;.
 	\label{eq:oftrl-primal-joint-full}
\end{equation}

\subsubsection{Energy Function and Gradient Map}

\paragraph{Log-Sum-Exp Energy Function.}
We recall the definition of the energy function 
$F: \R^{m+n} \to \R$ from~\eqref{eq:energy} in Section~\ref{sec:omwu}.
First, we have that $\LSE_m: \R^m \to \R$ and 
$\LSE_n : \R^n \to \R$ are the $m$ and $n$-dimensional
log-sum-exp functions, where:
\begin{equation*}
	\textstyle
	\LSE_m(x) = \log\big(\sum_{i=1}^m x(i)) 
	\;\;\text{and}\;\;
	\LSE_n(y) = \log\big(\sum_{j=1}^n y(j))  \;.
\end{equation*}
Both $\LSE_m$ and $\LSE_n$ are convex and 
continuously differentiable. 
Together, they define the convex and separable energy function
$F = \LSE_m + \LSE_n$,
where for $z = (x, y) \in \R^{m+n}$ we have
\begin{equation*}
	F(z) = \LSE_m(x) + \LSE_n(y)\;.
\end{equation*}

\paragraph{Energy Gradients.}
By a straightforward differentiation,
note that $\nabla \LSE_k(v) = \softmax_k(v)$ for any $k \ge 2$
and $v \in \R^k$. 
Thus, the gradient map $\nabla F: \R^{m+n} \to \R^{m+n}$
is given by
\begin{equation}
	\nabla F(z) 
	= 
	\begin{pmatrix}
	\nabla_x F(z) \\
	\nabla_y F(z)
	\end{pmatrix} 
	= 
	\begin{pmatrix}
		\softmax_m(x) \\
		\softmax_n(y)
	\end{pmatrix}
	\in \relint(\calW)
	\;\;\text{for all $z =(x, y)\in \R^{m+n}$.}
	\label{eq:F-grad-full}
\end{equation}
Note also that the log-sum-exp energy function $F$ 
and the negative entropy regularizer $R$
are dual functions (convex conjugates).
We give more details and implications
of this duality in Section~\ref{app:energy-kl-equiv:duality}.

\subsection{OMWU as Optimistic Skew-Gradient Descent}
\label{app:omwu-prelims:skew-grad}

The joint OMWU dual iterates $\{z_t\}$ can
be written as an optimistic skew-gradient
descent on the energy function $F$.
This key property was stated in Proposition~\ref{prop:omwu-skew-grad}
in Section~\ref{sec:omwu}.
Here, we restate the proposition and give the proof:

\smallskip

\propskewgrad*

\subsubsection{Proof of Proposition~\ref{prop:omwu-skew-grad}}

For $t \ge 0$, let $w_t = (p_t, q_t) \in \relint(\calW)$
be the primal iterates from~\eqref{eq:omwu-primal}
(equivalently, from~\eqref{eq:omwu-primal-components}),
and let $z_t = (x_t, y_t) \in \R^{m+n}$ be
the dual iterates from~\eqref{eq:zt-wt-def}
(equivalently, from~\eqref{eq:zt-wt-def-full}). 
Moreover, recall from~\eqref{eq:F-grad-full}
that $\nabla_x F(z) = \softmax_m(x)$
and $\nabla_x F(z) = \softmax_n(y)$
for any $z = (x, y) \in \R^{m+n}$.
Thus, we first prove inductively that 
$w_t = (p_t, q_t) = (\softmax_m(x_t), \softmax_n(y_t))$
for all $t \ge 0$.

For the base case, we have by assumption of the
proposition that $z_0 = (x_0, y_0) \in\R^{m+n}$ 
is such that $w_0 = (p_0, q_0) = \nabla F(z_0) 
= (\softmax_m(x_0), \softmax_n(y_0))$.
Now suppose through time $t$
that $w_t = (\softmax_m(x_t), \softmax_n(y_t))$.
Then at time $t+1$, 
we have  by definition
of~\eqref{eq:omwu-primal} that 
for all coordinates $i \in [n]$:
\begin{align}
	p_{t+1}(i)
	&=
	\frac{
		p_t(i) \cdot \exp(-\eta (2 Aq_t - Aq_{t-1})(i))
	}{
		\sum_{k=1}^m p_{t}(k) \cdot 
		\exp(- \eta (2Aq_t - Aq_{t-1})(k))
	} \nonumber \\
	&=
	\frac{\frac{\exp(x_t(i))}{\sum_{\ell=1}^m \exp(x_t(\ell))} 
	\cdot 
	\exp(-\eta (2 Aq_t - Aq_{t-1})(i))
	}{
		\sum_{k=1}^m
		\frac{\exp(x_t(k))}{\sum_{\ell=1}^m \exp(x_t(\ell))}
		\cdot 
		\exp(-\eta (2 Aq_t - Aq_{t-1})(k))
	} \nonumber \\
	&= 
	\frac{\exp(x_t(i) - \eta (2 Aq_t - A q_{t-1})(i))}
	{\sum_{k=1}^m 
	\exp(x_t(k) - \eta (2 Aq_t - Aq_{t-1})(k))
	} \nonumber \\
	&= \softmax_m\big(x_t - \eta (2Aq_t - Aq_{t-1})\big)(i)  \;.
	\label{eq:sderive-01} 
\end{align}
Here, the second equality follows by applying
the inductive hypothesis at time $t$.
By definition of $x_{t+1}$ from~\eqref{eq:omwu-dual-components},
we have $x_{t+1} = x_t - \eta (2 Aq_t - Aq_{t-1})$
for $t + 1\ge 2$ (and note that a
similar derivation also holds if $t = 1$).
As~\eqref{eq:sderive-01} holds for all 
coordinates $i \in [m]$, we
conclude that $p_{t+1} = \softmax_m(x_{t+1})$.
By an identical calculation, 
we similarly find $q_{t+1} = \softmax_n(y_{t+1})$.
Thus it holds inductively that 
$w_{t} = (\softmax_m(x_{t}), \softmax_n(y_{t})) 
= \nabla F(z_{t})$ for all $t \ge 0$.

For the second claim, 
we have by definition of $\{z_t\}$
from~\eqref{eq:zt-wt-def-full} that
$z_1 = z_0 -  \eta J \nabla F(z_1)$
at time $t=1$,
and at times $t+1 \ge 2$:
\begin{equation*}
	z_{t+1} = z_t - 2 \eta J\nabla F(z_t) + \eta J\nabla F(z_{t-1})
	= z_t - \eta J \nabla F(z_t) - \eta J(\nabla F(z_t) - F(z_{t-1})) \;.
\end{equation*}
This concludes the proof. 
\hfill ~ $\blacksquare$

\smallskip

\paragraph{Equivalence of~\eqref{eq:omwu-primal}
and the OFTRL Instantiation.}
As a corollary of Proposition~\ref{prop:omwu-skew-grad},
we also have that the OFTRL iterates from~\eqref{eq:omwu-oftrl} 
(equivalently~\eqref{eq:oftrl-primal-joint-full})
are also identical to those of~\eqref{eq:omwu-primal}.
This follows from the fact that
$\nabla F(z_t) = (\softmax_m(x_t), \softmax_n(y_t))$
for all $t \ge 1$
(from the first part of Proposition~\ref{prop:omwu-skew-grad}),
and from the characterization
of the iterates of~\eqref{eq:omwu-oftrl} from
Proposition~\ref{prop:oftrl-softmax}.

\subsubsection{Skew-Gradient Flow and Standard
Skew-Gradient Descent}
\label{app:omwu-prelims:omwu-related}

Here, we briefly review the continuous-time 
\textit{skew-gradient flow} 
and the standard forward discretization
\textit{skew-gradient descent}.

\paragraph{Skew-gradient flow.}
The skew-gradient flow dynamics is the continuous-time 
ODE: 
\begin{equation}
	\dot z(t) = - J \nabla F(z(t)) \;. 
	\label{eq:skew-grad-flow}
	\tag{Skew-gradient flow}
\end{equation}
By the chain rule, observe that energy is 
\textit{conserved} under~\eqref{eq:skew-grad-flow},
as we can compute
\begin{equation*}
	\frac{d}{dt} F(z(t))
	= 
	\langle
	\nabla F(z(t)), \dot z(t)
	\rangle
	= 
	-\langle
	\nabla F(z(t)),
	J \nabla F(z(t))
	\rangle
	= 0 \;,
\end{equation*}
where the final equality is due to the skew-symmetry of
$J = - J^\top$.
Geometrically, this means that the skew-gradient
vector $J \nabla F(z)$ is \textit{orthogonal}
to the gradient $\nabla F(z)$.
Hence, the skew-gradient flow follows directions
that are tangent to the energy levelset at
the current point, thus conserving energy. 

\paragraph{Skew-gradient descent.}
The first-order forward (Euler) discretization
of~\eqref{eq:skew-grad-flow} is
the forward skew-gradient descent.
This yields the 
dual iterates corresponding to the standard MWU algorithm
from~\eqref{eq:standard-mwu-primal}
with stepsize $\eta > 0$:
\begin{equation}
	z_{t+1} = z_t - \eta J \nabla F(z_t)  \;.
	\label{eq:mwu-dual} 
	\tag{MWU Dual}
\end{equation}
By convexity of $F$, 
the one-step change in energy under~\eqref{eq:mwu-dual} is 
always non-decreasing, since
\begin{equation*}
	F(z_{t+1})
	- F(z_t)
	\ge 
	\langle \nabla F(z_t), z_{t+1} - z_t \rangle
	= 
	- \eta \langle \nabla F(z_t), J \nabla F(z_t) \rangle
	= 
	0 \;,
\end{equation*}
where again the final equality is due
to $J = -J^\top$.

In contrast to the forward discretization,
we show for~\eqref{eq:omwu-dual} that
the additional optimistic correction 
term (which can be viewed as an approximation
of a backward discretization of skew-gradient flow),
leads to strictly dissipating energy
(see Lemma~\ref{lem:energy-one-step-full}).

See also~\cite{wibisono2022alternating}
and~\cite{katona2024symplectic}
for additional background on skew-gradient flows
and their discretizations. 


\section{Details on Equivalence of KL Minimization and 
Energy Dissipation}
\label{app:omwu-prelims:kl-energy}

The energy function $F$ and the regularizer $R$
are closely related via duality. 
Here, we specify several key aspects of this relationship in 
more detail. For the OMWU iterates, this ultimately allows
for establishing an equivalence between minimizing 
KL divergence in the primal space, and energy dissipation
in the effective dual space.

\paragraph{Organization of Section.}
This section is organized as follows:
\begin{itemize}[
	leftmargin=1em
]
\item
\textbf{Section~\ref{app:energy-kl-equiv:duality}} 
gives more details on the dual relationship
of the log-sum-exp energy function $F$ and the 
negative-entropy regularizer $R$, 
and their interplay over the effective dual space $\calZ$.
This includes establishing a certain
modified Fenchel-Young identity (Proposition~\ref{prop:restricted-fenchel})
and proving that $\calZ$ is orthogonal 
to the interior NE of $A$ (Proposition~\ref{prop:dual-subspace-property}).

\item
\textbf{Section~\ref{app:energy-kl-equiv:prop-proof}}
gives the proof of Proposition~\ref{prop:change-kl-energy-equiv},
which establishes an equivalence between
differences in $\KL$ divergence
and differences in energy $F$ over the effective dual space $\calZ$.
The proof of the proposition relies on the properties
introduced in Section~\ref{app:energy-kl-equiv:duality}.

\item
\textbf{Section~\ref{app:energy-kl-equiv:energy-suboptimality}}
establishes additional properties of the energy gradient map 
$\nabla F$ over the effective dual space $\calZ$. 
This culminates in establishing, under the assumption of a
unique and interior NE, that the KL divergence
from Nash is exactly the energy suboptimality
gap in the effective dual space (Proposition~\ref{prop:calZ-minimizer}).

\end{itemize}

\subsection{Duality of Regularizer and Energy Functions}
\label{app:energy-kl-equiv:duality}

The negative entropy regularizer $R$
is the \textit{convex (Fenchel) conjugate} 
$F^\star$ of the log-sum-exp energy function $F$.
Here, recall by definition that $F^*: \R^{m+n} \to \R$ is 
the convex function given by
\begin{equation}
	F^*(w) 
	\;:=\;
	\sup\nolimits_{z \in \R^{m+n}}\, 
	\big\{
		\langle w, z \rangle
		- F(z)
	\big\} \;.
	\label{eq:H-conjugate}
\end{equation}
for all $w \in \R^{m+n}$.
As $F$ and $R$ are both separable,
the fact that $F^* = R$
follows from standard relationships between
the log-sum-exp and negative entropy functions. 
Moreover, it is also well known that 
the Bregman divergence $D_R(w', w)$
induced by the negative entropy regularizer $R$
is exactly the $\KL$ divergence $\KL(w', w)$. 
These relationships are stated
in the following proposition:

\smallskip

\begin{restatable}{prop}{propconjugatekl}
	\label{prop:conjugate-energy-KL}
	Let $F\colon \R^{m+n} \to \R$ and $R\colon \calW \to \R$
	be the energy and regularizer functions defined in~\eqref{eq:energy} and~\eqref{eq:R-def-full},
	respectively. Let $F^*$ be the convex conjugate
	of $F$ as in~\eqref{eq:H-conjugate}.
	Then the following two properties hold:
	\begin{enumerate}[
		label={(\roman*)},
		leftmargin=3em,
	]
		\item
		$F^* = R$, 
		meaning $F^*(w) = R(w)$
		for $w \in \calW$, and $F^*(w) = \infty$ otherwise.
		\item
		For $w, w' \in \relint(\calW)$: 
		$\KL(w', w) = 
		D_R(w', w)
		= R(w') - R(w) - \langle \nabla R(w), w' - w\rangle$.
	\end{enumerate}	
\end{restatable}

The proof is standard and thus we omit it here 
(see, e.g.,~\cite{boyd2004convex}, Examples 3.19 and 3.25).
Moreover, by standard properties of convex conjugates,
note also that the energy $F = R^*$ is
the conjugate function of $R$.

\subsubsection{Modified Fenchel-Young Identity}

Due to the dual relationship of $F$ and $R$,
we additionally establish the following useful 
properties that relate $\KL$ divergence to energy. 
Specifically, we prove the following proposition:

\smallskip

\begin{restatable}[Modified Fenchel-Young Identity]
	{prop}{proprestrictedfenchel}
	\label{prop:restricted-fenchel}
	Let $F\colon \R^{m+n} \to \R$ and $R\colon \calW \to \R$
	be the energy and regularizer functions 
	from~\eqref{eq:energy} and~\eqref{eq:R-def-full},
	respectively.
	Fix $z \in \R^{m+n}$, and let $w = \nabla F(z) \in \relint(\calW)$.
	Then the following properties hold: 
	\begin{enumerate}[
		label={(\roman*)}
	]
	\item
	$R(w) + F(z) = \langle w, z \rangle$.
	\item
	$\KL(w', w) = D_R(w', w) = R(w') + F(z) - \langle z, w'\rangle$
	for any $w' \in \relint(\calW)$.
	\end{enumerate}
\end{restatable}

The first statement
is due to the standard Fenchel-Young inequality
(see, e.g.,~\cite{shalev2012online}, Section 2.7).
The second statement requires first introducing
several additional and important properties
of the energy function, the negative entropy regularizer,
and their gradients.
We first introduce these components
and then proceed with the proof of
Proposition~\ref{prop:restricted-fenchel}
further below.

\paragraph{Energy function lacks global strict convexity.}
Recall from expression~\eqref{eq:calZ-calS-def} in
Section~\ref{sec:omwu} that the linear subspace $\calS$
of constant shift directions is defined as
\begin{equation*}
	\textstyle
	\calS =
	\Span\big(
		\big(\begin{smallmatrix}\1_m \\  0\end{smallmatrix}\big),
		\big(\begin{smallmatrix} 0 \\  \1_n \end{smallmatrix}\big)
	\big)\;.
\end{equation*}
While the regularizer $R$ is strictly convex over 
the primal space $\calW = \Delta_m \times \Delta_n$,
the energy function $F$ is \textit{not} globally
strictly convex over $\R^{m+n}$.
In particular, the gradients of $F$ are invariant
to constant shifts (i.e., in directions $s \in \calS$),
meaning that $F$ is affine in those directions.
Specifically, the following properties hold:

\begin{restatable}[Energy function is not strictly convex]
	{prop}{propenergylinear}
	\label{prop:energy-linear}
	Let $F: \R^{m+n} \to \R$ be the energy function
	from~\eqref{eq:energy}.
	Fix any $z = (x, y) \in\R^{m+n}$.
	Then for any $s = (d \1_m, d' \1_n) \in \calS$ 
	for some $d, d' \in \R$:
	\begin{equation*}
		\nabla F(z + s) = 
		\nabla F(z) 
		\;\;\text{and}\;\;
		F(z + s) 
		= 
		F(z) + (d + d') \;.
	\end{equation*}
\end{restatable}

\begin{proof}
	Recall from~\eqref{eq:F-grad-full} that
	$\nabla F(z) = (\softmax_m(x), \softmax_n(y))$.	
	Then using the definition of $\softmax_k(\cdot)$, we have
	\begin{align*}
		\nabla F(z+s) 
		&= (\softmax_m(x + d\1_m), \softmax_n(y + d' \1_n)) \\
		&=(\softmax_m(x), \softmax_n(y))
		= \nabla F(z)\;.
	\end{align*}
	Similarly, using the definition of $F = \LSE_m + \LSE_n$,
	we have 
	\begin{equation*}
		F(z + s) 
		= 
		\LSE_m(x + d\1_m)
		+ \LSE_n(y + d' \1_n)
		= F(z) + (d + d') \;,
	\end{equation*}
	which follows from the definition of log-sum-exp.
\end{proof}

\smallskip

\paragraph{Gradient of negative entropy regularizer.}
Let $\nabla R \colon \relint(\calW) \to \R^n$
denote the gradient map 
$\nabla R(w) = (\nabla_p R(w), \nabla_q R(w))$
for $w = (p, q) \in\relint(\calW)$.
By definition of $R$ and the negative entropy functions
from~\eqref{eq:R-def-full}, it follows that 
\begin{equation}
	\nabla R(w)
 	= 
 	\begin{pmatrix}
 	\log(p) + \1_m \\
 	\log(q) + \1_n
 	\end{pmatrix} \;,
 	\label{eq:R-grad}
\end{equation}
where $\log(p) = (\log(p(1)), \dots, \log(p(m))) \in \R^m$ and 
$\log(q) = (\log(q(1)), \dots, \log(q(n))) \in \R^n$. 

\smallskip 

Although $R = F^*$ is the convex conjugate of $F$, 
the energy function's lack of global strict convexity 
shown in Proposition~\ref{prop:energy-linear}
means that $F$ is not a \textit{Legendre function}
(see, e.g.,~\cite{cesa2006prediction}, Section 11.2 for a definition). 
The main consequence of this fact is that 
the regularizer gradient $\nabla R = \nabla F^*$ \textit{is not} 
the inverse map of $\nabla F$ 
(since, by Proposition~\ref{prop:energy-linear},
$\nabla F \colon \R^{m+n} \to \R$ is not injective).
However, the map $\nabla R \colon \relint(\calW) \to \R^{m+n}$
does serve as an inverse to $\nabla F$ over
the quotient space $\R^{m+n} / \calS$, 
as formalized in the following proposition: 

\smallskip

\begin{restatable}[Gradient map of regularizer]
{prop}{propprimalmap}
	\label{prop:primal-map}
	Let $R \colon \calW \to \R^n$ be the joint regularizer
	from~\eqref{eq:R-def-full}.
	Fix $z = (x, y) \in \R^{m+n}$, and let 
	$w = (p, q) = \nabla F(z)$. 
	Then there exists $s \in \calS$ such that 
	\begin{equation*}
		\nabla R(w)
		= 
		z + s \;.
	\end{equation*} 
	Specifically, $s = (c_x \1_m, c_y \1_n)$, 
	where $c_x = 1 - \LSE_m(x) \in \R$
	and $c_y = 1 - \LSE_n(y) \in \R$.
\end{restatable}

\begin{proof} 
	By definition of $F$ and the assumption that
	$w = (p, q) = \nabla F(z)$,
	we have $p = \softmax_m(x)$
	and $q = \softmax_n(y)$.
	Using the definition of $\softmax_m$
	and $\softmax_n$, it follows that
	\begin{equation}
		\log(p) + \1_m = 
		\begin{pmatrix}
			\log\Big(\tfrac{\exp(x(1))}{\sum_{j=1}^m \exp(x(j))}\Big) + 1 \\
			\vdots 
			\\ 
			\log\Big(\tfrac{\exp(x(m))}{\sum_{j=1}^m \exp(x(j))}\Big) + 1 \\
		\end{pmatrix}
		= x + (1 - \LSE_m(x)) \cdot \1_m\;.
		\label{eq:logp-x-calc}
	\end{equation}
	By a similar calculation, we also have
	$\log(q) + \1_n = y + (1-\LSE_n(y))\cdot \1_n$. 
	Substituting these simplifications into the 
	definition of $\nabla R(w)$ from~\eqref{eq:R-grad} 
	yields the stated claim.
\end{proof}

\smallskip

In contrast to Proposition~\ref{prop:primal-map},
due to the strict convexity of the joint negative 
entropy regularizer $R$, the energy gradient $\nabla F$
\textit{does} serve as an inverse to $\nabla R$.
We state and prove this in the following proposition,
which is used in later subsections.

\begin{restatable}[Energy gradient is inverse of entropy gradient]
	{prop}{propeentropygradientinverse}
	\label{prop:entropy-grad-inverse}
	Let $F$ and $R$ be the functions from~\eqref{eq:energy}
	and~\eqref{eq:R-def-full}, respectively.
	Fix any $w = (p, q) \in \relint(\calW)$.
	Then $\nabla F(\nabla R(w)) = w$.
\end{restatable}

\begin{proof}
	Let $v = (v_x, v_y) = \nabla R(w) \in \R^{m+n}$.
	By definition of $\nabla R$ 
	from~\eqref{eq:R-grad}, this means that
	$v_x = (\log_m(p) + \1_m)$
	and $v_y = (\log_n(q)+ \1_n)$.
	By a straightforward calculation, observe that 
	\begin{equation*}
		\begin{aligned}
		\nabla \LSE_m (v_x)
		&= \softmax_m(v_x)
		= \softmax_m(\log_m(p) + \1_m)
		= p  \\
		\text{and}\;
		\nabla \LSE_n (v_y)
		&= \softmax_n(v_y)
		= \softmax_n(\log_n(q) + \1_n)
		= q \;.
		\end{aligned}
	\end{equation*}
	As $\nabla F(v) = (\softmax_m(v_x), \softmax_n(v_y))$,
	we conclude that $\nabla F(\nabla R(w)) = \nabla F(v) = w$.
\end{proof}

\medskip

We are now equipped to give the proof of the ``modified''
Fenchel-Young identity of Proposition~\ref{prop:restricted-fenchel}:

\smallskip

\paragraph{Proof of Proposition~\ref{prop:restricted-fenchel}.}
For part (i), note that as $F$ is convex and differentiable,
the standard Fenchel-Young inequality 
(see, e.g.,~\cite{boyd2004convex}, Section 3.3.2) says that,
for any $z \in \R^{m+n}$ and $w' \in \calW$:
\begin{equation*}
	F(z) + F^*(w') \ge \langle w', z \rangle \;,
\end{equation*}
where equality holds if and only if $w' = \nabla F(z)$.
Thus part (i) follows from the assumption that 
$w = \nabla F(z)$ and that $R = F^*$ (as established in 
the first claim of Proposition~\ref{prop:conjugate-energy-KL}).

For part (ii), recall from the second claim of 
Proposition~\ref{prop:conjugate-energy-KL} that
$\KL(w', w) = D_R(w', w)$. 
Using the definition of $D_R(w', w)$, we can then write
\begin{align}
	\KL(w', w) 
	= 
	D_R(w', w) 
	&= 
	R(w') - R(w) - \langle \nabla R(w), w' - w\rangle \nonumber \\
	&= 
	R(w') + F(z) - \langle w, z\rangle  - \langle \nabla R(w), w' - w\rangle \;,
	\label{eq:rfy-01}
\end{align}
where the second equality follows from rewriting $R(w)$ using
part (i) of the proposition. 
Moreover, by the assumption that $w = \nabla F(z)$, 
it follows from Proposition~\ref{prop:primal-map}
that $\nabla R(w) = z + s$ for some $s \in \calS$. 
Thus we can further simplify 
the final term of~\eqref{eq:rfy-01} and write
\begin{equation}
	\langle \nabla R(w), w' - w \rangle
	= 
	\langle z + s, w'  - w\rangle 
	= 
	\langle z, w' - w \rangle \;.
	\label{eq:rfy-02}
\end{equation}
Here, observe that the term $\langle s, w' - w\rangle = 0$
by definition of $s \in \calS$ and $w, w' \in \calW = \Delta_m \times \Delta_n$. 
Then substituting~\eqref{eq:rfy-02} into~\eqref{eq:rfy-01}, we find
\begin{align*}
	\KL(w', w) 
	&= 
	R(w') + F(z) - \langle w, z \rangle 
	- \langle z, w' - w\rangle 
	=  
	R(w') + F(z) - \langle z, w' \rangle 
	\;,
\end{align*}
which concludes the proof.
\hfill ~ $\blacksquare$

\subsubsection{Effective Dual Space}

Recall from Section~\ref{sec:omwu-main:primal-dual}
that we define the linear subspace $\calZ \subseteq \R^{m+n}$ by 
\begin{equation}
\textstyle
	\calZ := \Span(J \calW)
	= 
	\big\{
		\sum_{i=1}^k \tau_i \cdot J w_i
		\,\,\text{for}\,\, k \in \mathbb{N}\,,
		\tau_i \in \R\,, 
		w_i \in \calW
	\big\}.
	\label{eq:calZ-def}
\end{equation}

We refer to $\calZ$ as the \textit{effective dual space},
as by~\eqref{eq:zt-wt-def}, 
all dual OMWU iterates $\{z_t\}$ lie in $\calZ$.

\paragraph{Effective Dual Space is Orthogonal to Interior NE.}
The defining structural property of $\calZ$ is that,
when $A$ has an interior Nash equilibrium $w^\star$,
every vector $z \in \calZ$ is orthogonal to $w^\star$.
Formally:

\begin{restatable}[Effective dual space is orthogonal to interior Nash]
{prop}{propdualsubspace}
	\label{prop:dual-subspace-property}
	Let $A \in \R^{m \times n}$ have an interior Nash
	equilibrium $w^\star = (p^\star, q^\star) \in \relint(\calW)$.
	Then $\langle z, w^\star \rangle = 0$ 
	for every $z \in \calZ$.
\end{restatable}

\begin{proof}
	First, let $c = \langle p^\star, A q^\star\rangle$.
	By the definition of $J$ from~\eqref{eq:J-def} and
	using part (i) of Proposition~\ref{prop:interior-NE}, observe 
	that 
	\begin{equation}
		\textstyle
		Jw^\star 
		= 
		\begin{pmatrix}
			0 & A \\
			-A^\top & 0 
		\end{pmatrix}
        \begin{pmatrix}
		p^\star \\ 
		q^\star
		\end{pmatrix}
        =
		\begin{pmatrix}
		A q^\star \\ 
		-A^\top p^\star
		\end{pmatrix}
		= 
		\begin{pmatrix}
			c \1_m \\ 
			-c \1_n 
		\end{pmatrix} \;.
	\label{eq:czp-01}
	\end{equation}
	Now fix $z \in \calZ$, and observe 
	by definition that we can write
 	$z = J \big(\sum_{i=1}^k \tau_i w_i\big)$
	for some $k \ge 1$, where all $\tau_i \in \R$
	and $w_i = (p_i, q_i) \in \calW = \Delta_m \times \Delta_n$.
	Together with the skew-symmetry of $J = -J^\top$, 
	it then follows that
	\begin{equation}
		\textstyle
		\langle z, w^\star \rangle
		= 
		\Big\langle 
		J \big(\sum_{i=1}^k \tau_i w_i\big), w^\star
		\Big\rangle
		= 
		- \Big\langle
		\big(\sum_{i=1}^k \tau_i w_i\big), J w^\star
		\Big\rangle
		= 
		- \sum_{i=1}^k
		\tau_i \langle w_i, J w^\star \rangle \;.
		\label{eq:czp-02}
	\end{equation}
	Using~\eqref{eq:czp-01}, observe for each $i \in [k]$ that
	\begin{equation*}
		\langle w_i, Jw^\star \rangle
		= 
		\Big\langle 
		\begin{pmatrix}
		p_i \\ q_i 
		\end{pmatrix}, 
		\begin{pmatrix}
		c \1_m \\ -c \1_n 
		\end{pmatrix}
		\Big\rangle
		= 
		c \cdot \big(\langle p_i, \1_m \rangle - \langle q_i, \1_n\rangle\big)
		= 
		0 \;,
	\end{equation*}
	where the final equality comes from the fact that
	$p_i \in \Delta_m$ and $q_i \in \Delta_n$
	are probability distributions.
	Thus each term of~\eqref{eq:czp-02} vanishes,
	and it follows that $\langle z, w^\star \rangle = 0$.
\end{proof}

\subsection{Proof of Proposition~\ref{prop:change-kl-energy-equiv}
-- Change in KL is Change in Energy}
\label{app:energy-kl-equiv:prop-proof}

Using the modified Fenchel-Young identity 
of Proposition~\ref{prop:restricted-fenchel}
and the orthogonality of the effective dual space
$\calZ$ to an interior Nash equilibrium
from Proposition~\ref{prop:change-kl-energy-equiv},
we can now prove the the key structural
result of Proposition~\ref{prop:change-kl-energy-equiv}.
Restated here:

\propchangeklenergy*

\begin{proof} 
	As $w = \nabla H(z)$ and $w' = \nabla H(z')$ by assumption, we have
	by Proposition~\ref{prop:restricted-fenchel} that
	\begin{align*}
		\KL(w^\star, w') &= R(w^\star) + F(z') - \langle z', w^\star \rangle
		\;\;\;\;\text{and}\;\;\;\;
		\KL(w^\star, w) = R(w^\star) + F(z) - \langle z, w^\star \rangle \;.
	\end{align*}
	Taking their difference, we find 
	\begin{equation}
		\KL(w^\star, w') - \KL(w^\star, w) 
		= 
		F(z') - F(z) - \langle z' - z, w^\star \rangle \;.
		\label{eq:dkl-01}
	\end{equation}
	As both $z, z' \in \calZ$ and $w^\star \in \relint(\calW)$,  
	we have from 
	Proposition~\ref{prop:dual-subspace-property} that
	both $\langle z', w^\star \rangle = 0$
	and $\langle z, w^\star \rangle = 0$.
	Thus the final term in~\eqref{eq:dkl-01} vanishes,
	which yields the desired statement. 
\end{proof}

\smallskip

\paragraph{Change in KL is Change in Energy for OMWU iterates.} 
By definition in~\eqref{eq:zt-wt-def}, we
have for the OMWU dual iterates $\{z_t\}$
that $z_t \in \calZ$ for all $t \ge 1$.
Thus as an immediate corollary of
Proposition~\ref{prop:change-kl-energy-equiv}, 
we have the following:

\smallskip

\begin{restatable}[Equivalence between
energy and KL differences for OMWU]
{cor}{corchangeklenergyomwu}
	\label{cor:change-kl-energy-equiv-omwu}
	Let $w^\star$ be the unique and interior NE of $A$,
	and let $\{w_t\}$ and $\{z_t\}$ be the primal and dual iterates
	of OMWU on $A$ with stepsize $\eta > 0$,
	and initialized from $w_0 \in \relint(\calW)$
	and $z_0 \in \calZ$ such that $\nabla F(z_0) = w_0$.
	Then
	\begin{equation*}
	\KL(w^\star, w_{t+1}) - \KL(w^\star, w_t) = 
	F(z_{t+1}) - F(z_t)\;\;\text{for all $t \ge 0$.}
	\end{equation*}
\end{restatable}

Note that Corollary~\ref{cor:change-kl-energy-equiv-omwu}
assumes that the dual initialization $z_0$
satisfying $\nabla H(z_0) = w_0$
also lies in $z_0 \in \calZ$.
In the following subsection, we establish
additional properties of the energy gradient map $\nabla F$
and the effective dual space $\calZ$ that imply
such an initialization always exists when $A$ has a
unique and interior NE. 
Thus, this assumption on $z_0 \in \calZ$ holds 
without loss of generality.

\subsection{KL Divergence is Energy Suboptimality Gap 
Over Effective Dual Space}
\label{app:energy-kl-equiv:energy-suboptimality}

In this section, we further establish the following properties
about the energy gradient map
when its domain is restricted to the effective dual space $\calZ$.
Specifically, when $A$ has a unique and interior NE $w^\star$,
we establish in this section the following properties:
\begin{itemize}[
	leftmargin=1em
]
\item
\textbf{In Section~\ref{app:energy-kl-equiv:surjective}},
we prove that $\nabla F: \calZ \to \R$ is \textit{surjective}
(Proposition~\ref{prop:calZ-grad-map})
and also \textit{injective}
when the value of the game is centered at zero
(Proposition~\ref{prop:calZ-grad-map-injectivity}).

\item
\textbf{In Section~\ref{app:energy-kl-equiv:energy-gap}},
we prove for $z \in \calZ$ and $w = \nabla F(z) \in \relint(\calW)$
that $\KL(w^\star, w)$ is
exactly the energy suboptimality gap
$F(z) - \min_{z' \in \calZ} F(z')$ over 
the effective dual space (Proposition~\ref{prop:calZ-minimizer}).
\end{itemize}

\subsubsection{Surjectivity (and Injectivity) 
of Energy Gradient Map Over Effective Dual Space}
\label{app:energy-kl-equiv:surjective}

\paragraph{Helper lemma under unique and interior Nash.}
We start with the following helper lemma,
which establishes that the range of $J$
and the orthogonal complement $\calS^\bot$
have only a trivial intersection under the
assumption of a unique and interior NE:

\begin{restatable}
{lem}{lemuniqueinteriorhelper}
	\label{lem:unique-interior-NE-helper}
	Let $A \in \R^{m \times n}$ have a unique and interior NE 
	$w^\star \in \relint(\calW)$.
	Fix any $v \in \R^{m+n}$.
	Suppose that $Jv \in \calS$ and $v \in \calS^\bot$.
	Then $v = 0$.
\end{restatable}

\begin{proof}
	Since $w^\star$ is a unique and interior 
	Nash equilibrium, we have by part (i) 
	of Proposition~\ref{prop:interior-NE}
	that $Jw^\star \in \calS$.
	Now for $\tau \in \R$, 
	define $w_{\tau} = w^\star + \tau v$.
	In other words, $w_{\tau}$
	is on the line through $w^\star$ in the direction
	of $v$.
 	It follows by linearity that
	\begin{equation*}
		J w_{\tau}
		= J (w^\star + \tau v)
		= J w^\star + \tau J v \in \calS \;.
	\end{equation*}
	Here, the inclusion in $\calS$ 
	follows from the assumption that $Jv \in \calS$.

	Now, as also $v \in \calS^\bot$, we have by definition
	that $\langle (\1_m, \1_n), v \rangle = 0$. 
	Thus, since $w^\star$ is interior, there exists
	a sufficiently small $\eps > 0$ such that, 
	for all $\tau \in (-\eps, \eps)$, 
	the point $w_{\tau} \in \relint(\calW)$ is also
	an interior (component-wise) probability distribution.
	Here, note that the fact that $\langle (\1_m, \1_n), v \rangle = 0$
	is what ensures the components of $w_\tau$
	still have coordinates that sum to 1.

	However, since $w^\star$ is unique, and 
	since $Jw_{\tau} \in \calS$, we have
	by part (ii) of Proposition~\ref{prop:interior-NE} that
	$w_{\tau} = w^\star + \tau v = w^\star$ for all 
	$\tau \in (-\eps, \eps)$.
	Thus by the uniqueness of $w^\star$, we must have $v = 0$.
\end{proof}

\medskip

\paragraph{Surjectivity of energy gradient over effective dual space.}
We now show that, when $A$ has a unique and interior NE,
then over the effective dual space $\calZ$,
the gradient map $\nabla F: \calZ \to \relint(\calW)$ is surjective.
Formally, we prove the following:

\smallskip

\begin{restatable}[Surjectivity of $\nabla F$ over $\calZ$]
{prop}{propcalZgradmap}
	\label{prop:calZ-grad-map}
	Let $A \in \R^{m \times n}$ have a unique and interior 
	NE $w^\star = (p^\star, q^\star) \in \relint(\calW)$.
	Then for all $w \in \relint(\calW)$, 
	there exists $z \in \calZ$ such that
	$\nabla F(z) = w$.
\end{restatable}

\smallskip

\begin{proof} 
We prove the claim via the following three steps:

\smallskip
\noindent
\textbf{1. Energy gradient is globally surjective}:

\noindent
First, we show that $\nabla F: \R^{m+n} \to \relint(\calW)$ 
is surjective over the full domain $\R^{m+n}$. 
For this, fix $w \in \relint(\calW)$. Recalling that $R$ is
the joint negative entropy regularizer, observe that
Proposition~\ref{prop:primal-map} shows that
$\nabla R(w) = v$ for some $v \in \R^{m+n}$. 
Using the definition of $\nabla R$, it follows by 
Proposition~\ref{prop:entropy-grad-inverse}
that $\nabla F(v) = \nabla F(\nabla R(w)) = w$.
Thus $\nabla F$ is globally surjective.

\smallskip

\noindent
\textbf{2. Energy gradient is invariant to constant shifts}:

\noindent
Now, again fix $w \in \relint(\calW)$, and let 
$v = \nabla R(w) \in \R^{m+n}$.
Further recall from Proposition~\ref{prop:energy-linear}
that $\nabla F$ is invariant to constant shifts,
meaning for any $s \in \calS$ that
$\nabla F(v + s) = \nabla F(v) = w$.
Thus to prove $\nabla F$ is surjective over $\calZ$,
it is sufficient to establish $v + s \in \calZ$ for some
$s \in \calS$. As $v \in \R^{m+n}$ can lie in the full
space, it suffices to equivalently show that
$\calZ + \calS = \R^{m+n}$. 

\smallskip

\noindent
\textbf{3. Sufficient property holds under a unique and
interior NE}:

\noindent
We will prove that $\calZ + \calS = \R^{m+n}$
holds under the assumption of a unique and interior NE
$w^\star = (p^\star, q^\star) \in \relint(\calW)$. 
For this, first observe that the property $\calZ + \calS = \R^{m+n}$
is equivalent to establishing equality between the orthogonal 
complements of the two sets. 
Since $\calZ$, $\calS$ are linear subspaces, we have
$(\calZ + \calS)^\bot = \calZ^\bot  \cap \calS^\bot$.
Moreover, $(\R^{m+n})^\bot = \{0\}$. 
Thus our goal is to show $\calZ^\bot \cap \calS^\bot = \{0\}$.

For this, pick any $v = (v_x, v_y) \in \calZ^\bot$.
By definition, this means $\langle v, z \rangle = 0$ for every
$z \in \calZ$.
We will show this implies $Jv \in \calS$. 
In particular, for any $w = (p, q) \in \calW$, 
let $z = Jw$. Thus clearly $z \in \calZ$,
and thus we must have $\langle v, z \rangle = 0$. 
As $J = - J^\top$, we can further
express this condition as 
\begin{align*}
	0 
	= 
	\langle v , z \rangle 
	= 
	\langle v, Jw \rangle 
	= 
	- \langle w, J v \rangle  
	= 
	\big( 
	\langle q, A^\top v_x \rangle
	- \langle p, A v_y \rangle
	\big)
\end{align*}
Since this constraint must be satisfied 
for all $w = (p, q) \in \calW$ simultaneously, 
it follows that we must have
$A^\top v_x = c_1 \1_m$ and $A v_y = c_2 \1_n$
for constants $c_1, c_2 \in \R$.
Thus $Jv \in \calS$ by definition of $\calS$.

We have established that $v \in \calZ^\bot$ implies 
$Jv \in \calS$.
Now suppose further that $v \in \calS^\bot$,
meaning $v \in \calZ^\bot \cap \calS^\bot$.
However, since $w^\star$ is the unique and interior NE,
we have by Lemma~\ref{lem:unique-interior-NE-helper}
that $v = 0$.
Thus $\calZ^\bot \cap  \calS^\bot = \{0\}$,
which establishes $\calZ + \calS = \R^{m+n}$
under the assumptions on $w^\star$.
\end{proof}

\medskip

\paragraph{Conditions for injectivity of energy gradient 
over effective dual space.}
We further establish that $\nabla F$ is
injective over $\calZ$ when the value
of the game $A$ is zero. Formally, we prove the following:

\smallskip 

\begin{restatable}[$\nabla F$ injective when game value is zero]
{prop}{propcalZgradmapinjective}
	\label{prop:calZ-grad-map-injectivity}
	Let $A \in \R^{m \times n}$ have a unique and interior 
	NE $w^\star = (p^\star, q^\star) \in \relint(\calW)$.
	If $d = \langle p^\star, A q^\star \rangle = 0$,
	then $\nabla F(z) \neq \nabla F(z')$
	for all $z \neq z' \in \calZ$.
\end{restatable}

\begin{proof}
First, recall from Proposition~\ref{prop:energy-linear} that $\nabla F$ 
is invariant under constant shifts. Thus it holds for $z, z' \in \calZ$ that
$\nabla F(z) = \nabla F(z')$ if and only if 
$z-z' \in  \calS$.
This means that $\nabla F: \calZ \to \R^{m+n}$ is injective over $\calZ$
if and only if $\calZ \cap \calS = \{0\}$.
We will show this latter condition holds when 
$d = 0$.

For this, pick any $z \in \calZ \cap \calS$. 
As $z \in \calZ$, there exists
some $v = (v_x, v_y) \in \R^{m+n}$
with $z = Jv$,
where $v = \sum_{i=1}^k \tau_i w_i$
for $k \ge 1$ and $\tau_i \in \R$, 
$w_i = (p_i, q_i) \in \calW$.
Since $z \in \calS$, it then further holds
that $z = Jv \in \calS$.
Using the fact that $w^\star$ is unique and interior
and applying a similar argument as in the proof of 
Lemma~\ref{lem:unique-interior-NE-helper},
$Jv \in \calS$ then further implies that $v = \alpha w^\star$ 
for some constant $\alpha \in \R$.
Together, this means that if
$z \in \calZ \cap \calS$, then $z \in \Span(Jw^\star)$.
Moreover, as clearly $Jw^\star \in \calZ$,
and also since $Jw^\star \in \calS$ by Proposition~\ref{prop:interior-NE},
we have that
\begin{equation}
	\calZ \cap \calS = \Span(J w^\star) \;.
	\label{eq:calZ-calS-intersect} 
\end{equation}
Finally, note that part (i) of Proposition~\ref{prop:interior-NE}
more specifically gives
$J w^\star = (d \1_m, -d \1_n) \in \calS$, 
where $d = \langle p^\star, A q^\star\rangle \in \R$
is the value of the game. Thus if $d = 0$,
then $J w^\star = 0$, and $\Span(Jw^\star) = \{0\}$.
Thus $d=0$ means
$\calZ \cap \calS = \{0\}$,
which implies $\nabla F: \calZ \to \R^{m+n}$
is injective.
\end{proof}

\subsubsection{KL from Nash is Energy Suboptimality Gap}
\label{app:energy-kl-equiv:energy-gap}

Finally, we show that when $A$ 
has a unique and interior NE $w^\star$,
the $\KL$ divergence from Nash in the primal space 
is equivalent to the energy subpoptimality gap 
in the effective dual space. Formally:

\smallskip

\begin{restatable}[KL as Energy Suboptimality Gap]
{prop}{propcalZminimizer}
	\label{prop:calZ-minimizer}
	Let $A \in \R^{m \times n}$ have a unique and interior 
	NE $w^\star = (p^\star, q^\star) \in \relint(\calW)$.
	Let $z \in \calZ$ such that
	$w = \nabla F(z)$. 
	Then there exists $z^\star \in\calZ$
	such that $\nabla F(z^\star) = w^\star$, and moreover
	$
		\KL(w^\star, w) = 
		F(z) - \min_{z' \in \calZ} F(z') 
		= F(z) - F(z^\star) 
	$.
\end{restatable}

\begin{proof}	
	We proceed via the following steps:

	\noindent
	\textbf{1. Existence of minimizer in effective dual space}:

    \noindent
	Recall that $F$ is convex and continuously 
	differentiable. Thus over the linear subspace $\calZ \subset \R^{m+n}$,
	the first-order optimality conditions for $F$ give
	(see, e.g.,~\cite{boyd2004convex}, Sec. 4.2.3):
	\begin{equation*}
		\textstyle
		z \in \argmin_{z' \in \calZ} F(z') 
		\;\iff\;
		\langle \nabla F(z), z'' \rangle = 0 
		\;\text{for all $z'' \in \calZ$} \;.
	\end{equation*}
	Now due to Proposition~\ref{prop:calZ-grad-map},
	$\nabla F: \calZ \to \relint(\calW)$ is surjective over
	$\calZ$. Thus, as $w^\star \in \relint(\calW)$,
	there exists $z^\star \in \calZ$ such that
	$\nabla F(z^\star) = w^\star$.
	Due to Proposition~\ref{prop:dual-subspace-property},
	this further implies that
	\begin{equation*}
	\textstyle
		\langle \nabla F(z^\star), z \rangle 
		= 
		\langle w^\star, z \rangle 
		= 0
		\;\text{for all $z \in \calZ$}
		\quad\implies
		z^\star \in \argmin_{z \in \calZ} F(z) \;.
	\end{equation*}

	\smallskip

	\noindent
	\textbf{2. Minimum energy function value attained}:

    \noindent
	Recall from expression~\eqref{eq:calZ-calS-intersect} 
	in the proof of Proposition~\ref{prop:calZ-grad-map-injectivity}
	that $\calZ \cap \calS = \Span(Jw^\star)$.
	Thus fixing $z^\star \in \argmin_{z \in \calZ} F(z)$
	from Step (1), then also 
	$z_{\tau} = z^\star + \tau Jw^\star \in \argmin_{z \in \calZ} F(z)$
	for any $\tau \in \R$. 
	This is because $\nabla F(z^\star + s) = \nabla F(z^\star)$
	for any $s \in \calS$, and since $Jw^\star \in \calS$.
	Thus $z' \in \calZ$ is in the argmin set if and only if 
	$z' = z_{\tau}$ for some $\tau \in \R$.
	However, since $F(z_\tau) = F(z^\star)$ for every $\tau \in \R$,
	$\min_{z \in \calZ} F(z)$ is well-defined, and thus
	$F(z^\star) = \min_{z\in \calZ} F(z)$.

	\smallskip

	\noindent
	\textbf{3. Apply Fenchel-Young identity}:

    \noindent
	Now by Part (ii) of the Fenchel-Young identity 
	from Proposition~\ref{prop:restricted-fenchel}, we have 
	for all $z \in \calZ$ and $w = \nabla F(z)$ that
	\begin{equation}
		\KL(w^\star, w)
		= R(w^\star) + F(z) - \langle w^\star, z \rangle
		= R(w^\star) + F(z) \;,
		\label{eq:fym-01}
	\end{equation}
	where the second equality is due to
	Proposition~\ref{prop:dual-subspace-property}.
	Now let $z^\star \in \calZ$ such that 
	$\nabla F(z^\star) = w^\star$.
	Then it follows from~\eqref{eq:fym-01} that
	$
		0 
		= 
		\KL(w^\star, w^\star) 
		= 
		R(w^\star) +  F(z^\star)
		\implies
	    R(w^\star) = - F(z^\star)
    $.
	By Step (1), we have $z^\star \in \argmin_{z \in \calZ} F(z)$,
	and by Step (2), we further have 
	$F(z^\star) = \min_{z \in \calZ} F(z)$.
	Then substituting $R(w^\star) = - F(z^\star)$ into~\eqref{eq:fym-01},
	we find for all $z \in \calZ$ and $w = \nabla F(z)$ that
	\begin{equation*}
		\textstyle
		\KL(w^\star, w)
		= 
		F(z) - \min_{z' \in \calZ} F(z')\;,
	\end{equation*}
	which concludes the proof.
\end{proof}


\section{Details on Energy Function Hessian Matrix}
\label{app:energy-gsc-prelims}

This section gives additional preliminaries
on the Hessian of the energy function
that are used in the proof of energy dissipation
under OMWU.

\paragraph{Oragnization of Section.}
The section is organized as follows:
\begin{itemize}[
    leftmargin=1em
]
\item
\textbf{Section~\ref{app:energy-gsc-prelims:energy-hessian}}
gives more details on  energy function Hessian that are 
used throughout the sequel.

\item
\textbf{Section~\ref{app:energy-gsc-prelims:gsc}}
gives more details on the \textit{local hessian stability} (LHS)
property introduced in Section~\ref{sec:energy-dissipation}
and provides the proof of Proposition~\ref{prop:energy-gsc}.

\item
\textbf{Section~\ref{app:energy-gsc-prelims:gsc-bounds}}
proves several bounds using the local norms $\|\cdot \|_z$
that hold under the LHS property.

\end{itemize}

\subsection{Properties of Energy Hessian}
\label{app:energy-gsc-prelims:energy-hessian}

Our analysis of OMWU relies on
several properties of the energy function's Hessian $\nabla^2 F$.
Formally, for any $z = (x, y) \in \R^{m+n}$, the Hessian
$\nabla^2 F(z) \in \R^{(m+n) \times (m+n)}$
is the block diagonal symmetric matrix 
\begin{equation}
    \nabla^2 F(z)
    = 
    \begin{pmatrix}
        \nabla^2 \LSE_m(x) & 0 \\
        0 & \nabla^2 \LSE_n(y)
    \end{pmatrix} \;.
    \label{eq:energy-hessian}
\end{equation}
Note by convexity of $F$ that $\nabla^2 F(z) \succeq 0$
for all $z \in \R^{m+n}$.

In the next proposition, we state several
basic properties about the energy Hessian. 
In particular, we note that 
$\nabla^2 F$ has a natural probabilistic 
chacterization: 
for $z = (x, y) \in \R^{m+n}$ and $w = (p, q) = \nabla F(z) \in \relint(\calW)$,
the Hessian blocks $\nabla^2 \LSE_m(x)$
and $\nabla^2 \LSE_n(y)$
are the \textit{covariance matrices} 
of the probability distributions $p$ and $q$, respectively.
Thus, quadratic forms of these Hessian blocks
can be interpreted as variances under the
distributions $p$ and $q$.
Before stating and proving these properties formally,
we first introduce the folllowing notation for discrete variances:

\paragraph{Notation for variances.}
Fix $(v_x, v_y) \in \R^{m+n}$
and $(p, q) \in \calW$. We write $\Var_p(v_x) \in \R$ to denote
the variance of $v_x \in \R^m$ under the distribution $p \in \Delta_m$,
and we similarly write $\Var_q(v_y) \in \R$ to denote
the variance of $v_y \in \R^n$ under the distribution $q \in\Delta_n$.
Formally, we define:
\begin{equation}
    \Var_p(v_x) :=
    \ssum\nolimits_{i=1}^m 
    p(i)\big(v_x(i) - \langle p, v_x \rangle\big)^2 
    \;\;\text{and}\;\;
    \Var_q(v_y) :=
    \ssum\nolimits_{j=1}^n 
    q(j) \big(v_y(j) - \langle q, v_y \rangle\big)^2
    \,.
    \label{eq:variance-def}
\end{equation}

Then for any $z \in \R^{m+n}$, 
the matrix $\nabla^2 F(z)$ has the following properties:

\begin{restatable}[Properies of Energy Hessian]
{prop}{propenergyhess}
    \label{prop:energy-hessian}
    Let $F$ be the energy function 
    from~\eqref{eq:energy}. 
    For any $z \in \R^{m+n}$, 
    let $\nabla^2 F(z)$ be the matrix in~\eqref{eq:energy-hessian}.
    Let $w = (p, q) = \nabla F(z) \in \relint(\calW)$.
    Then:
    \begin{enumerate}[
        label={(\roman*)}
    ]
    \item
    The Hessian $\nabla^2 F(z)$ is given by
    \begin{equation*}
        \nabla^2 F(z)
        = 
        \begin{pmatrix}
            \big(\Diag(p) - p p^\top\big) & 0  \\
            0 & \big(\Diag(q) - qq^\top\big)
        \end{pmatrix} 
        \;.
    \end{equation*} 
    \item
    For any $v = (v_x, v_y) \in \R^{m+n}$, 
    it holds that
    $\big\langle v, \nabla^2 F(z) v \big\rangle 
        =  \Var_{p}(v_x) + \Var_{q}(v_y)$.

    \item
      $\Null\big( \nabla^2 F(z) \big)
        = 
        \calS 
        = 
        \Span\big(
        \big(\begin{smallmatrix}\1_m \\  0\end{smallmatrix}\big),
        \big(\begin{smallmatrix} 0 \\  \1_n \end{smallmatrix}\big)
        \big)$.

    \item  
    $0 \preceq \nabla^2 F(z) \preceq I$.

    \end{enumerate}
\end{restatable}

\begin{proof} 
    We prove each of the statements separately:

    \smallskip

    \noindent
    \textbf{Proof of (i).} 
    This result is standard, and a proof can be found in 
    ~\cite{boyd2004convex}, Section 3.1.5,
    or~\cite{gao-et-al17softmax}, Proposition 2.
   
    \smallskip

    \noindent
    \textbf{Proof of (ii).}
    Using the structure of $\nabla^2 F(z)$ from claim (i),
    for $v = (v_x, v_y) \in \R^{m+n}$ we have 
    \begin{equation*}
        \big\langle v, \nabla^2 F(z) v \big\rangle 
        = 
        \big\langle 
        v_x, (\Diag(p) - pp^\top) v_x 
        \big\rangle
        + 
        \big\langle 
        v_y, (\Diag(q) - qq^\top) v_y
        \big\rangle \;.
    \end{equation*}
    Observe that the matrices $(\Diag(p) - pp^\top)$
    and $(\Diag(q) - qq^\top)$ are exactly the
    covariance matrices of the distributions $p \in \Delta_m$
    and $q \in \Delta_n$. 
    Indeed, recalling the definitions of $\Var_p(v_x)$ 
    and $\Var_q(v_y)$ from~\eqref{eq:variance-def}, 
    it is straightforward to check that
    \begin{align*}
        \langle 
        v_x, (\Diag(p) - pp^\top) v_x 
        \rangle
        &= 
        \langle v_x, \Diag(p) v_x \rangle
        - 
        \langle v_x, pp^\top v_x \rangle  \\
        &= 
        \langle v_x, \Diag(p) v_x \rangle
        - 
        \big(\langle p, v_x \rangle\big)^2 \\
        &= 
        \ssum\nolimits_{i=1}^m p(i) \cdot (v_x(i))^2 
        - 
        \big(\ssum\nolimits_{i=1}^m p(i) \cdot v_x(i)\big)^2
        = 
        \Var_p(v_x)  \;.
    \end{align*}
    Using identical calculations, it follows that also
    $\langle  v_y, (\Diag(q) - qq^\top) v_y \rangle = \Var_q(v_y)$, 
    which completes the proof of claim (ii).

    \smallskip

    \noindent 
    \textbf{Proof of (iii).}
    This property is also proven in
    Proposition of~\cite{gao-et-al17softmax}.

    \paragraph{Proof of (iv).} 
    Observe that 
    $(\Diag(p) - pp^\top) \preceq I$, 
    which follows from the fact that all $0 < p(i) \le 1$
    and $pp^\top \succeq 0$. 
    Similarly, $(\Diag(q) - qq^\top) \preceq I$,
    and thus due to the block structure of $\nabla^2 F(z)$
    from part (i) of the proposition,
    it follows that $\nabla^2 F(z) \preceq I$. 
    By convexity of $F$, also $0 \preceq \nabla^2 F(z)$.
\end{proof}

\subsubsection{Hessian of Regularizer and Inverse Energy Hessian}

We also consider the Hessian matrix of
the joint negative entropy regularizer $R: \calW \to R$
from~\eqref{eq:R-def-full}. 

\paragraph{Negative Entropy Hessian.}
Fix any $w= (p, q) \in \relint(\calW)$.
Then let $\nabla^2 R(w) \in \R^{(m+n)\times(m+n)}$
be the Hessian matrix of $R$ given by
\begin{equation}
    \nabla^2 R(w)
    = 
    \begin{pmatrix}
        \nabla^2 R_m(p) & 0 \\
        0 & \nabla^2 R_n(q) 
    \end{pmatrix}
    = 
    \begin{pmatrix}
    \Diag(1/p) & 0  \\
    0 & \Diag(1/q)
    \end{pmatrix} \;.
    \label{eq:R-hessian-full}
\end{equation}
Here, we write $1/p = (1/p(1), \dots,1/p(m)) \in \R^m$
and $1/q = (1/q(1), \dots, 1/q(n)) \in \R^n$.
Note that the second equality in~\eqref{eq:R-hessian-full}
follows from the definitions
of $\nabla R_m(p)$ and $\nabla R_n (q)$ 
from~\eqref{eq:R-grad}.

\smallskip

\paragraph{Negative Entropy Hessian as
Restricted Inverse of Energy Hessian.}

Due to the duality of the energy function
and the negative entropy regularizer,
$\nabla^2 R$ is the inverse of $\nabla^2 F$
when restricted to the orthogonal complement $\calS^\bot$:

\smallskip

\begin{restatable}[Energy hessian restricted inverse]
    {prop}{propenergyhessianinverse}
    \label{prop:energy-hessian-inverse}
    Fix $z \in \R^{m+n}$, and let $w = (p, q) = \nabla F(z)$.
    Then $(\nabla^2 F(z))^{-1} = \nabla^2 R$ 
    over $\calS^\bot$, meaning
    $\nabla^2 F(z) (\nabla^2 R(w)) v = v$
    for all $v \in \calS^\bot$.
\end{restatable}

\begin{proof}
    Let $z = (x, y)$, and let $v = (v_x, v_y)$. 
    We first prove the claim for the blocks
    $\nabla^2 \LSE_m(x)$ and $\nabla^2 R_m(p)$.
    For this, we have by definition of 
    $\nabla^2 R_m$ and $\nabla^2 \LSE_m$ that 
    \begin{align*}
        \nabla^2 \LSE_m(x) R_m(p) v_x
        &= 
        (\Diag(p) - pp^\top) 
        (\Diag(1/p))v_x \\
        &=
        \Diag(p) \Diag(1/p) v_x
        - 
        p (p^\top \Diag(1/p)) v_x  \\
        &=
        I v_x - p \cdot \langle \1_m, v_x \rangle  
        = v_x \;. 
    \end{align*}
    Here, the penultimate equality is 
    due to $p^\top \Diag(1/p) = \1_m$,
    and the final equality is due to the
    assumption that $v \in \calS^\bot$,
    and thus $\langle v_x, \1_m \rangle = 0$.
    By an identical calculation, we have 
    $\nabla^2 \LSE_n(y) R_n(q) v_y = v_y$,
    which proves the claim 
    by definition of $\nabla^2 F$ and $\nabla^2 R$.
\end{proof}

\subsubsection{Gradient Differences and Bregman Divergence in Integral Form} 

Throughout, we will also use the energy Hessian $\nabla^2 F$ to
express certain remainder terms when taking Taylor expansions of $F$. 
First, we can express the Bregman divergence
$D_F(\cdot, \cdot)$
as follows:

\begin{restatable}[Bregman divergence in integral form]
 {prop}{propbregmanintegral}
    \label{prop:bregman-integral}
    Let $F$ be the function in~\eqref{eq:energy},
    and let $z, z' \in \R^{m+n}$.
    Let $z_s := z - s(z-z')$ for $s \in [0, 1]$.
    Then the Bregman divergence $D_F(z', z)$ is
    \begin{equation*}
        D_F(z', z) 
        =  
        \int_{0}^1 
        (1-s) \big \langle z - z', \nabla^2 F(z_s)(z-z') \big\rangle \, ds \;.
    \end{equation*}
\end{restatable}

\smallskip

For $z, z' \in \R^{m+n}$, we can also express
the difference of energy gradients using the energy Hessian:

\begin{restatable}[Gradient difference and Hessian remainder]
{prop}{propgraddiffgeneral}
    \label{prop:grad-diff-general}
    Let $F$ be the function in~\eqref{eq:energy},
    and let $z, z' \in \R^{m+n}$.
    Let $z_s := z - s(z-z')$ for $s \in [0, 1]$.
    Let $G_F(z, z') \in \R^{m+n}$ be the remainder term
    \begin{equation}
      G_F(z, z') =
      \int_0^1 (\nabla^2 F(z_s) - \nabla^2 F(z))(z-z') \, ds \;.
      \label{eq:hess-remainder}
    \end{equation}
    Then 
    $ \nabla F(z) - \nabla F(z') 
    = 
    \nabla^2 F(z) (z - z') + G_F(z, z')
    $. 
\end{restatable}

\noindent
Propositions~\ref{prop:bregman-integral} 
and~\ref{prop:grad-diff-general} are both standard
consequences of Taylor's theorem (see, 
e.g.,~\cite{nocedal2006numerical}, Theorem 2.1),
and thus we omit the proofs.

\subsection{Local Hessian Stability Property}
\label{app:energy-gsc-prelims:gsc}

In the sequel, our analysis of the energy dissipation 
under OMWU relies on the following \textit{local Hessian stability}
property of the energy function,
introduced originally in Section~\ref{sec:energy-dissipation}:

\smallskip

\propenergygsc*

\begin{proof} 
    Recalling the block structure of the energy Hessian 
    $\nabla^2 F(z)$ from Proposition~\ref{prop:energy-hessian},
    it suffices to show that the property holds 
    for the block components $\nabla^2 \LSE_m(x)$
    and $\nabla^2 \LSE_n(y)$. 
    We will prove the property holds for the former,
    and for the latter it will follow by identical calculations.

    For this, fix any vector $u \in \R^m$, and let $z=(x, y)$ and $z'=(x', y')$.
    Observe that the assumption $\|z-z'\|_\infty \le \alpha$
    implies $\|x-x'\|_{\infty} \le \alpha$.
    Now for every $s \in [0, 1]$, define $x_s = x + s(x' -x)$, 
    and let $h = x'-x$. 
    Moreover, let $p_s  = \nabla \LSE_m(x_s) = \softmax_{m}(x_s)$,
    meaning 
    \begin{equation*}
        p_s = \softmax_m(x_s) 
        = \softmax_m(x + sh) \;.
    \end{equation*}
    Further let $f_u(s)$ be the scalar function given by
    \begin{equation*}
        \textstyle
        f_u(s) 
        =
        \big\langle 
            u, \nabla^2 \LSE_m(x_s) u
        \big\rangle
        = 
        \big\langle
            u, (\Diag(p_s) - p_s p_s^\top) u
        \big\rangle
        = 
        \sum_{i=1}^m 
        p_s(i) \cdot \big(u(i) - \langle p_s, u \rangle\big)^2
        \;,
    \end{equation*}
    where the equalities follow from the structure of 
    $\nabla^2 \LSE_m$ from Proposition~\ref{prop:energy-hessian}.

    Now, differentiating $f_u(s)$ with respect to $s$
    and simplifying, we find 
    \begin{equation*}
        \textstyle
        f_u'(s) 
        =         
        \sum_{i=1}^m 
        p_s(i) \cdot 
        \big( h(i) - \langle p_s, h \rangle \big)  
        \big( u(i) - \langle p_s, u \rangle \big)^2 \;.
    \end{equation*}
    By the assumption $\|x-x'\|_\infty = \|h\|_\infty \le \alpha$, 
    it follows that $|h(i) - \langle p_s, h \rangle | \le 2 \alpha$. 
    Thus we can further bound 
    \begin{equation*}
        |f'_u(s)|
        \le 
        2\alpha \cdot 
        \sum_{i=1}^m 
        p_s(i) \cdot 
        \big(u(i) - \langle p_s, u\rangle \big)^2
        = 
        2\alpha \cdot f_u(s) \;.
    \end{equation*}
    Applying Gr{\"o}nwall's inequality then yields
    \begin{equation*}
        \exp(-2 \alpha)  \cdot f_u(0)
        \le 
        f_u(1)
        \le 
        \exp(2 \alpha) \cdot f_u(0) \;.
    \end{equation*}
    Since $x_0 = x$ and $x_1 = x'$ and
    $u$ was arbitrary, it follows by definition of $f_u$ that 
    \begin{equation*}
        \exp(-2\alpha) \cdot \nabla^2 \LSE_m(x)
        \preceq
        \nabla^2 \LSE_m(x') 
        \preceq
        \exp(2\alpha) \cdot 
        \nabla^2 \LSE_m(x) \;,
    \end{equation*}
    which concludes the proof.
\end{proof}

\paragraph{Local Hessian Stability and Generalized Self-Concordance.} 
Similar versions of this LHS property have also been exploited 
in the optimization literature as a consequence 
of \textit{(generalized) self-concordance}.
See, e.g.,~\citet{bach2010self}, 
\citet{tran2015composite}
(Theorem 4.b, eq. (20)),
\citet{sun2019generalized} (Proposition 8, eq. (16) when $\nu=2$),
or more recently~\citet{freund2026second} 
(Proposition 16, and Lemma 15),
and the references therein.
For the purposes of the present paper, 
our proof of energy dissipation in Lemma~\ref{lem:energy-one-step-full}
uses the the semidefinite ordering of~\eqref{H-gsc}
via the application of several general inequalities
involving local norms under the energy Hessian. 
We present these tools in the next subsection.

\subsection{Consequences of Local Hessian Stability Property}
\label{app:energy-gsc-prelims:gsc-bounds}

\subsubsection{Review of Notation for Local Norms}

We first recall and expand on 
the notation of the local norms introduced
in Section~\ref{sec:energy-dissipation}.
Fix $z \in \R^{m+n}$.
Then for all $u, v \in \R^{m+n}$, we define
\begin{equation*}
    \langle  u, v \rangle_z 
    := 
    \langle u, \nabla^2 F(z) v \rangle
    \;\;\text{and}\;\;
    \|v \|_{z}
    := 
    \sqrt{\langle v, \nabla^2 F(z) v \rangle}  \;.
\end{equation*}
We further define a local \textit{dual} norm $\|\cdot \|_{z, *}$
over the linear subspace $\calS^\bot \subseteq \R^{m+n}$
(as defined in~\eqref{eq:calZ-calS-def}).
For this, when restricted to acting on vectors $v \in \calS^\bot$,
we write $(\nabla^2 F(z))^{-1}$
to denote the inverse energy Hessian, where
$\nabla^2 F(z) (\nabla^2 F(z))^{-1} u = u$
for all $u \in \calS^{\bot}$ and $z \in \R^{m+n}$.
From Proposition~\ref{prop:energy-hessian-inverse},
note in particular that $(\nabla^2 F(z))^{-1} = \nabla^2 R(w)$
as defined in~\eqref{eq:R-hessian-full},
where $w = \nabla F(z)$.
Then over $\calS^\bot$, for any $z \in \R^{m+n}$, 
the local dual norm $\|\cdot \|_{z, *}$ is given by
\begin{equation*}
    \|u \|_{z, *} = \sqrt{\langle u, (\nabla^2 F(z))^{-1} u\rangle}
    \quad
    \text{for $u \in \calS^\bot$.} 
\end{equation*}
For $z \in \R^{m+n}$, the primal-dual pair
$(\|\cdot \|_{z}, \|\cdot \|_{z, *})$ satisfies the 
generalized Cauchy-Schwarz inequality 
\begin{equation*}
    \big| \langle v, u \rangle \big|
    \le 
    \|v\|_{z} \cdot \|u\|_{z, *}
    \quad
    \text{for all $v \in \R^{m+n}$ and $u \in \calS^\bot$}.
\end{equation*}

\subsubsection{Properties of Local Norms under LHS}

We state and prove several useful properties of 
the local norms under the~\eqref{H-gsc} property. 

\smallskip

\paragraph{Relationships between global and local norms.}

The following proposition relates primal 
and local norms to the Euclidean norm:

\begin{restatable}{prop}{propnormrelationships}
    \label{prop:norm-relations}
    For all $z \in \R^{m+n}$,
    $v \in \R^{m+n}$, and 
    $u \in \calS^\bot$:
    $\|v\|_z  \le \|v\|_2$
    and $\|u\|_2 \le \|u\|_{z,*}$.
\end{restatable}

\begin{proof}
    For the first claim, recall by part (iv) of
    Proposition~\ref{prop:energy-hessian} that 
    $\nabla^2 F(z) \preceq I$. Thus for any $v \in \R^{m+n}$, 
    it holds that
    \begin{equation*}
        \|v\|^2_z
        = 
        \langle v, \nabla^2 F(z) v \rangle
        \le 
        \langle v, v \rangle 
        = 
        \|v\|^2_2 \;,
    \end{equation*}
    which implies $\|v\|_z \le \|v\|_2$. 
For the second claim, note that when restricted to 
$\calS^{\bot}$, the ordering $\nabla^2 F(z) \preceq I$
implies $I \preceq (\nabla^2 F(z))^{-1}$. 
Then for any $u \in \calS^\bot$ and $z \in \R^{m+n}$, 
it holds that
\begin{equation*}
    \|u\|^2_2 
    = 
    \langle u, u \rangle 
    \le 
    \langle u, (\nabla^2 F(z))^{-1} u \rangle
    = 
    \|u\|^2_{z, *} \;,
\end{equation*}
which implies $\|u\|_2 \le \|u\|_{z, *}$. 
\end{proof}

\smallskip

\paragraph{Local norm transfer.}
For vectors $z, z', v \in \R^{m+n}$, 
a key consequence of the~\eqref{H-gsc} property
is the ability to relate the
local norm $\|v\|_z$ to the local norm $\|v\|_{z'}$
up to multiplicative factors depending on $\|z-z'\|_{\infty}$.
Formally, we have the following relationships:

\smallskip

\begin{restatable}{lem}{lemnormtransfer}
    \label{lem:norm-transfer-gsc}
    For $\alpha > 0$, let $z, z' \in \R^{m+n}$ 
    such that $\|z-z'\|_\infty \le \alpha$. Then
    for all $v \in \R^{m+n}$: 
    \begin{equation*}
        \exp(-\alpha) \cdot \|v\|_{z} 
        \le
        \|v\|_{z'} 
        \le
        \cdot \exp(\alpha) \|v\|_{z}\,.
    \end{equation*}
\end{restatable}

\begin{proof} 
    Due to the \eqref{H-gsc} property, we have for any $z, z' \in \calZ$ such that $\|z-z'\|_\infty \le \alpha$,
    \begin{equation*}
        \exp(-2\alpha) \nabla^2 F(z)
        \preceq
        \nabla^2 F(z') 
        \preceq 
        \exp(2\alpha) \nabla^2 F(z) \;.
    \end{equation*}
    Using the definition of local norms, this implies that for any $v \in \calZ$, 
    \begin{equation}
      \exp(-2 \alpha) \cdot \|v\|^2_{z}  \le   \|v\|^2_{z'} = \langle v, \nabla^2 F(z') v \rangle \le  \exp(2 \alpha) \cdot \|v\|^2_{z}\,. 
    \end{equation}
    Taking the square root in the above inequality concludes the proof. 
\end{proof}

\smallskip 

\paragraph{Non-expansiveness of local norm.}
The next lemma establishes that the energy Hessian
satisfies a certain non-expansiveness property
under its own induced local norm:

\begin{restatable}{lem}{lemhessianlocalsmooth}
    \label{lem:hessian-local-smooth}
    Let $M \in R^{(m+n) \times (m+n)}$ be a matrix,
    and fix $z \in \R^{m+n}$. 
    Then for any $v \in \R^{m+n}$:
    \begin{equation*}
        \| M \nabla^2 F(z) v \|_z \le \|M\|_2 \cdot \|v\|_z \;.
    \end{equation*}
\end{restatable}

\begin{proof} \;
    Applying Proposition~\ref{prop:norm-relations},
    we can write 
    \begin{equation}
        \|M \nabla^2 F(z) v \|_z
        \le 
        \|M \nabla^2 F(z) v\|_2 
        \le 
         \|M\|_2 \cdot \|\nabla^2 F(z) v\|_2 \;,
        \label{eq:hls-01}
    \end{equation}
    where the final inequality comes from the definition
    of the spectral norm $\|M\|_2$. 
    To further bound the term $\|\nabla^2 F(z) v\|_2$, observe
    by definition that 
    \begin{equation}
        \|\nabla^2 F(z) v\|^2_2
        = 
        \langle v, (\nabla^2 F(z))^2 v \rangle
        \le
        \langle v, \nabla^2 F(z) v \rangle
        = 
        \|v\|^2_z \;.
        \label{eq:hls-02}
    \end{equation}
    Here, the inequality comes from the fact that
    $(\nabla^2 F(z))^2 \preceq \nabla^2 F(z)$ 
    since, by Proposition~\ref{prop:energy-hessian},
    $0 \preceq \nabla^2 F(z) \preceq I$.
    Then taking square roots and substituting~\eqref{eq:hls-02} 
    into~\eqref{eq:hls-01} concludes the proof.
\end{proof}

\subsubsection{Bounds on Bregman Divergence
and Hessian Remainder Terms under LHS}
\label{app:energy-gsc-prelims:gsc:bregman-hess-terms}

\paragraph{Additional notation.}
To start, we define the following two scalar functions
which will be useful in our proofs for controlling certain constant 
multiplicative factors.
Specifically, let $\mu: \R \to \R$ and $\nu: \R \to \R$
be the functions such that, for any $x \neq 0$:
\begin{equation}
    \mu(x) = \frac{\exp(x) - x -1}{x^2}\,, \quad 
    \nu(x) = \frac{\exp(x) - x - 1}{x} \;.
    \label{eq:def-mu-nu}
\end{equation}
In particular we will be using the following straightforward bounds:

\begin{restatable}{lem}{lemexpbounds}
    \label{lem:exp-bounds}
    For any $\delta > 0$ and $|x| \le \delta$,
    it holds that
    \begin{equation*}
        \big|\mu(x) - \tfrac{1}{2}\big|
        \le 
        \frac{\delta \exp(\delta)}{6}
        \quad\text{and}\quad
        \big|\nu(x) - \tfrac{x}{2}\big|
        \le
        \frac{\delta^2 \exp(\delta)}{6} \;.
    \end{equation*}
\end{restatable}

\begin{proof}
    Taking a second-order Taylor expansion with Lagrange remainder, 
    we have: 
    \begin{equation}
        \exp(x) = 1 + x + \frac{x^2}{2} + \frac{\exp(\xi) x^3}{6} \,
        \label{eq:exp-taylor}
    \end{equation}
    for some $\xi \in (0,x)$.
    Rearranging, we find
    $\mu(x) -\frac{1}{2} = \frac{\exp(\xi) x}{6}$.
    Noting that $|\exp(\xi) x| \le \delta \exp(\delta)$ for $|x| \le \delta$
    then gives $|\mu(x) - \frac{1}{2}| \le \frac{\delta \exp(\delta)}{6}$,
    which yields the first claim of the lemma.
    
    For the second claim, we can again rearrange~\eqref{eq:exp-taylor}
    to find
    $\nu(x) - \frac{x}{2} 
    = 
    \frac{\exp(\xi) x^2}{6}$.
    Taking absolute values and again using
    $|\exp(\xi) x^2| \le \delta^2 \exp(\delta)$ for $|x| \le \delta$
    then gives $|\nu(x) - \frac{x}{2}| \le \frac{\delta^2 \exp(\delta)}{6}$.
\end{proof}

\smallskip

\paragraph{Bounds on Bregman divergence.}

\smallskip

\begin{restatable}{lem}{lembregmangsc}
  \label{lem:bregman-gsc}
  Fix $z, z' \in \R^{m+n}$
  such that $\|z-z'\|_\infty \le \alpha$
  for $\alpha > 0$.
  Then the Bregman divergence $D_F(z', z)$
  is bounded above and below by 
  \begin{equation*}
    \mu(-2\alpha) \cdot \|z'-z\|^2_z
    \;\le\;
    D_F(z', z) 
    \;\le \;
    \mu(2\alpha) \cdot \|z' - z\|^2_z \;.
  \end{equation*}
\end{restatable}

\begin{proof} 
    Recall from Proposition~\ref{prop:bregman-integral}
    that the Bregman divergence $D_F(z', z)$
    can be written as
    \begin{equation}
        D_F(z', z) 
        = 
        \int_{0}^1 (1-s) 
        \big\langle z-z', \nabla^2 F(z_s) (z-z') \big\rangle \, ds \;,
        \label{eq:bregman-bound-01}
    \end{equation}
    where $z_s := z - s(z-z') \in \calZ$ for $s \in [0, 1]$.
    Fixing $s$, 
    observe that $\|z-z_s\|_\infty = s \|z-z'\|_\infty \le s\alpha$, 
    which follows from the assumption that $\|z-z'\|_\infty \le \alpha$.
    Thus we have by~\eqref{H-gsc} that
    \begin{equation}
        \exp(-2s\alpha) \nabla^2 F(z)
        \preceq
        \nabla^2 F(z_s) 
        \preceq 
        \exp(2s\alpha) \nabla^2 F(z) \;.
        \label{eq:bregman-bound-gsc}
    \end{equation}
    For the upper bound on $D_F(z', z)$, note that
    the relationship in~\eqref{eq:bregman-bound-gsc} implies
    \begin{align*}    
        \big\langle z- z', \nabla^2 F(z_s) (z-z') \big\rangle
        &\le
        \exp(2s\alpha) \cdot 
        \langle z-z', \nabla^2 F(z)(z-z')\rangle \\
        &= 
        \exp(2s\alpha) \cdot 
        \|z-z'\|_{z}^2 \;.
    \end{align*}
    Substituting back into~\eqref{eq:bregman-bound-01} 
    and integrating, we conclude 
    (using the function $\mu$ from~\eqref{eq:def-mu-nu}):
    \begin{align*}
        D_F(z', z) 
        &\le
            \int_{0}^1 (1-s) \exp(2s\alpha)  \|z-z'\|_z^2 \, ds  \\
        &= 
            \Big(\int_0^1 (1-s) \exp(2s \alpha) \, ds \Big)
            \cdot  \|z-z'\|_z^2 
        = 
            \mu(2\alpha)
            \cdot 
            \|z' - z\|^2_z  \;.
    \end{align*}
    For the lower bound on $D_F(z', z)$, we similarly have
    from~\eqref{eq:bregman-bound-gsc} that 
    \begin{align*}    
        \big\langle z- z', \nabla^2 F(z_s) (z-z') \big\rangle
        &\ge
        \exp(-2s\alpha) \cdot 
        \langle z-z', \nabla^2 F(z)(z-z')\rangle \\
        &= 
        \exp(-2s\alpha) \cdot 
        \|z-z'\|_{z}^2 \;.
    \end{align*}
    Thus we can again substitute back
    into~\eqref{eq:bregman-bound-01} 
    and integrate to find 
    \begin{align*}
        D_F(z', z) 
        &\ge 
            \Big(\int_0^1 (1-s) \exp(-2s \alpha) \, ds \Big)
            \cdot  \|z-z'\|_z^2 
        = 
            \mu(-2\alpha)
            \cdot 
            \|z' - z\|^2_z  \;,
    \end{align*}
    which concludes the proof.
\end{proof}

\smallskip

\paragraph{Bound on Hessian remainder term.}

\smallskip

\begin{restatable}{lem}{lemhessianremainderbound}
    \label{lem:hessian-remainder-bound}
    Fix $z, z' \in \R^{m+n}$.
    Then $G_F(z, z') \in \calS^\bot$. 
    Moreover, if $\|z-z'\|_\infty \le \alpha$
    for some $\alpha > 0$, then
    \begin{equation*}
    \|G_F(z, z')\|_{z, *} \le \nu(2\alpha)
    \cdot \|z - z'\|_{z} \,.
    \end{equation*}
\end{restatable}

\begin{proof} 
    Recall from \eqref{eq:hess-remainder} that we define 
    the remainder term $G_F(z, z')$ in integral form by
    \begin{equation}
         G_F(z, z') = \int_0^1 (\nabla^2 F(z-sv) - \nabla^2 F(z)) v \, ds
        \;,
         \label{eq:RH-repeat}
    \end{equation}  
    where $v = z-z'$. 
    Now for the first claim of the lemma,
    recall from part (iii) of Proposition~\ref{prop:energy-hessian} that
    $\Null(\nabla^2 F(z'')) = \calS$,
    and thus $\nabla^2 F(z'') b = 0$ for every $z'' \in \R^{m+n}$ and $b \in \calS$.
    Then for all $u \in \R^{m+n}$, by symmetry of $\nabla^2 F$, 
    it follows for every $b \in \calS$ that 
    \begin{equation*}
        \langle b, \nabla^2 F(z'') u \rangle
        = 
        \langle u, \nabla^2 F(z'') b \rangle
        = 
        0 \;.
    \end{equation*}
    By definition, this means $\nabla^2 F(z'') u \in \calS^\bot$
    for every $z'', u \in \R^{m+n}$, 
    and thus also $G_F(z, z') \in \calS^\bot$.

    For the second claim of the lemma, 
    taking local (dual) norms in~\eqref{eq:RH-repeat}
    and using the triangle inequality:
     \begin{equation}
        \| G_F(z, z') \|_{z, *}
        \le 
        \int_0^1 
        \| (\nabla^2 F(z-sv) - \nabla^2 F(z)) v \|_{z, *} \; ds \;.
        \label{eq:hrb-00} 
    \end{equation}
    Now repeating the previous argument,
    observe that each term within the norm $\|\cdot\|_{z, *}$
    in the right-hand-side of
    \eqref{eq:hrb-00} belongs to $\calS^\bot$.
    Similar to the proof of Lemma~\ref{lem:bregman-gsc}, 
    let $s \in [0, 1]$ and write $z_s := z - sv$.
    For readability, we further write $P_z = \nabla^2 F(z)$, 
    $P_{z_s} := \nabla^2 F(z_s)$, and $B_s := P_{z_s} - P_z$ 
    for $s \in [0, 1]$.
    Moreover, as we are working over $\calS^\bot$, 
    we write $P_z^{-1} = (\nabla^2 F(z))^{-1}$
    to denote the energy Hessian inverse 
    over $\calS^\bot$.
    With this notation, our goal is thus to control the terms 
    $\|B_s v\|_{z, *}$ in the integrand of~\eqref{eq:hrb-00}. 
    For this, define the matrix $M := P^{-1/2}_{z} P_{z_s} P^{-1/2}_z$, 
    for which it holds that 
    \begin{equation}
        P^{-1/2}_z B_{s} = (M-I) P^{1/2}_z \;.
        \label{eq:hrb-03}
    \end{equation}
    Then taking local (dual) norms, we have 
    \begin{align}
        \|B_s v \|^2_{z,*}
        = 
        \langle 
        B_s v, P^{-1}_z B_s v
        \rangle 
        &= 
        \langle 
        B_s v, P^{-1/2}_z P^{-1/2}_z B_s v 
        \rangle \nonumber \\
        &= 
        \|P^{-1/2}_z B_s v \|^2_2  \nonumber \\
        &= 
        \|(M-I) P^{1/2}_z v \|^2_2
        &&\text{(using~\eqref{eq:hrb-03})} \nonumber \\
        &\le 
        \|M-I\|^2_2 \cdot 
        \| P^{1/2}_z v \|^2_2 \;.
        \label{eq:hrb-04}   
    \end{align}
    Now to bound the two terms of~\eqref{eq:hrb-04}, 
    observe first by definition of the local norm that
    \begin{equation}
        \| P^{1/2}_z v \|^2_2
        = 
        \langle P^{1/2}_z v, P^{1/2}_z v \rangle
        = 
        \langle v, P^{1/2}_z  P^{1/2}_z  v \rangle
        = 
        \langle v, P_z v \rangle
        = 
        \|v \|_{z}^2\;.
        \label{eq:hrb-05}
    \end{equation}

    To bound $\|M-I\|^2_2$, recall from the~\eqref{H-gsc} property that 
    \begin{equation}
        \exp(-2s\alpha) P_z
        \preceq
        P_{z_s}
        \preceq
        \exp(2 s \alpha) P_z \;,
        \label{eq:hrb-gsc}
    \end{equation}
    which follows from the fact that 
    $\|z-z_s\|_\infty = s\|z-z'\|_\infty \le s \alpha$ by assumption.
    Now due to the orderings of~\eqref{eq:hrb-gsc} and using 
    the definition of $M$, it follows by congruence that 
    \begin{equation}
        M  
        = 
        P^{-1/2}_{z} P_{z_s} P^{-1/2}_z
        \preceq 
        \exp(2 s\alpha) \cdot P_z^{-1/2} P_z P^{-1/2}_z 
        = 
        \exp(2 s\alpha) I 
        \label{eq:hrb-gsc-M}  
    \end{equation}
    and similarly
    \begin{equation}
        M  \succeq \exp(-2s\alpha) I
        \;.
        \label{eq:hrb-gsc-M2}  
    \end{equation}
    Thus rearranging~\eqref{eq:hrb-gsc-M} and~\eqref{eq:hrb-gsc-M2}
    and using $\exp(2s\alpha) - 1 \ge 1- \exp(-2s\alpha)$, 
    it follows that
    \begin{equation*}
        -(\exp(2s\alpha) - 1)I
        \preceq
        M-I
        \preceq
        (\exp(2s\alpha) - 1)I \;.
    \end{equation*}
    Thus, we have $\|M-I\|_2 \le (\exp(2s\alpha) - 1)$
    by definition of spectral norm. 
    Combining this with~\eqref{eq:hrb-05} and substituting
    back into~\eqref{eq:hrb-04} then yields
    \begin{equation*}
         \|B_s v \|_{z,*} 
         \le 
         (\exp(2s\alpha) - 1) \cdot \|v\|_z \;.
    \end{equation*}
    Finally, putting this bound back into~\eqref{eq:hrb-00}
    and integrating, we find
    \begin{equation*}
        \| G_F(z, z') \|_{z, *}
        \le 
        \int_0^1 
        \|B_s v\|_{z, *} \; ds
        \le 
        \Big(
        \int_0^1             
        (\exp(2s\alpha) - 1)
        \, ds
        \Big)
        \cdot 
        \|v\|_z
        = 
       \nu(2\alpha)
        \cdot 
        \|v\|_z \;,
    \end{equation*}
    which concludes the proof. 
\end{proof}


\section{Details on Energy Dissipation Under OMWU}
\label{app:energy-dissipation}

In this section, we develop the proof of 
Lemma~\ref{lem:energy-one-step-full}, 
which gives matching upper and lower bounds on 
the one-step change in energy under~\eqref{eq:omwu-dual}.
We first restate the lemma, before giving 
a more detailed overview of the proof:

\smallskip

\lemenergyonestep*

\smallskip

\begin{restatable}[Constant Stepsize]
{rem}{remstepsize}
	\label{remark:constant-stepsize}
	Note that, for the sake of presentation,
	we do not attempt to optimize the constraint 
	$0 < \eta \le \frac{1}{4(54\sigma_{\max} + 9)}$
	from Assumption~\ref{ass:stepsize}.
	Up to the leading constants,
	the conclusion of 
	Lemma~\ref{lem:energy-one-step-full}
	should hold under even larger absolute
	constant stepsizes.
\end{restatable}

\smallskip
\noindent
\textbf{Simplifying Notation.}
For readability, we write $\sigma = \sigma_{\max} = \|J\|_2$
to denote the spectral norm of $J$.
By Proposition~\ref{prop:spectral-norm-J-A}, 
recall that $\|J\|_2 = \|A\|_2$.

\subsection{Proof of Lemma~\ref{lem:energy-one-step-full}}
\label{app:energy-dissipation:proof}

The proof follows two steps, both using several intermediate results. 
Here, we will state the conclusions of these results and
derive their consequences for proving Lemma~\ref{lem:energy-one-step-full}.
We defer their full statements and proofs 
to the subsequent subsections.

\smallskip
\noindent
\textbf{1. Exact expansion of change in energy}:

\noindent
The starting point is to take a first-order Taylor expansion
of $F(z_{t+1})$ around $z_t$.
Using the structure of the~\eqref{eq:omwu-dual} update,
as well as the Hessian-based approximation  
of gradient differences from Proposition~\ref{prop:grad-diff-general},
we derive an exact expansion 
of $\Delta F(z_t) = F(z_{t+1}) - F(z_t)$,
the one-step change in energy over the dual OMWU iterates.
Specifically, we prove the following:

\noindent
\textbf{Proposition~\ref{prop:energy-one-step-expand}}
(One-step expansion).
\textit{For every $t \ge 1$, there exists
$\calE_{t} \in \R^{m+n}$ such that}
\begin{equation}
	\Delta F(z_t)
	= 
	- \eta^2 \|J \nabla F(z_t)\|^2_{z_t} 
	+ D_F(z_{t+1}, z_t)
	+ \eta \langle J \nabla F(z_t), G_F(z_t, z_{t-1}) \rangle
	+ 
	\eta^2 \calE_t \;.
	\label{eq:diss-01}
\end{equation}

In particular, the error term $\calE_t$ is comprised
of four subterms that are defined explicitly 
in~\eqref{eq:calE-def}. 
The proof of Proposition~\ref{prop:energy-one-step-expand}
is in Section~\ref{app:energy-dissipation:expansion}.

\smallskip
\noindent
\textbf{2. Controlling the error terms}:

\noindent
Notice that the first term of~\eqref{eq:diss-01} is exactly the
desired dissipation term $-\eta^2 \|J \nabla F(z_t)\|^2_{z_t}$ 
(up to a constant factor) appearing
in the statement of the lemma. 
Thus, the remaining technical challenge is to ensure the
final three terms of~\eqref{eq:diss-01} also scale
at most like $O(\eta^2 \|J \nabla F(z_t)\|^2_{z_t})$,
with a constant factor that can be made less than 1 via
a sufficiently small stepsize $\eta$.

For this task, we distinguish between
the middle two terms of~\eqref{eq:diss-01},
which we refer to as the \textit{quadratic error terms},
and the final term of~\eqref{eq:diss-01},
which we refer to as the \textit{cubic error terms}.
We derive the requisite bounds
on these two sets of terms using 
the local norm machinery developed
in Section~\ref{app:energy-gsc-prelims:gsc-bounds},
which hold due to the~\eqref{H-gsc} property
of the energy Hessian. 
In more details:

\noindent
\textit{Bounds on quadratic error terms}.
Under a sufficiently small stepsize, we establish:

\noindent
\textbf{Proposition~\ref{prop:DH-error}}
(Bound on Bregman divergence term). 
\textit{%
	Fix $L \ge 21$. For $\eta \le \tfrac{1}{L\sigma}$
	and all $t \ge 1$:
}
\begin{equation}
	D_F(z_{t+1},  z_t)
	\le 
	\big(\tfrac{18}{25} + \phi(L) \big) 
	\cdot \eta^2 \|J \nabla F(z_t)\|^2_{z_t} \;,
	\label{eq:bregman-bound}
\end{equation}
\textit{where $\phi(L) = \tfrac{36\exp(6/L)}{25L}$
 is a strictly decreasing function of $L$
that can be made arbitrarily small.
In particular, it holds that $\tfrac{18}{25} + \phi(L) \le \tfrac{149}{200}$, 
when $L \ge 72$.}

\smallskip 

We obtain a similar bound for 
the inner product involving the Hessian remainder term:

\noindent
\textbf{Proposition~\ref{prop:RH-inner-bound}}
(Bound on inner product term).
\textit{%
	Fix $L \ge 21$. For $\eta \le \frac{1}{L\sigma}$ and all $t \ge 1$:
}
\begin{equation}
	\big| 
	\eta \langle 
	J \nabla F(z_t), G_F(z_t, z_{t-1})
	\rangle
	\big|
	\le 
	\psi(L) \cdot 
	\eta^2 \|J \nabla F(z_t)\|^2_{z_t} \;,
	\label{eq:bound-RH-inner}
\end{equation}
\textit{where
$\psi(L) = \tfrac{54}{5L} + \tfrac{108\exp(6/L)}{5L^2}$
is a strictly decreasing function of $L$ that can be made
arbitrarily small. In particular, it holds that
$\psi(L) \le \tfrac{31}{200}$ for all $L \ge 72$.}

\noindent
The proofs of Propositions~\ref{prop:DH-error}
and~\ref{prop:RH-inner-bound} are given in
Section~\ref{app:energy-dissipation:error-control:quadratic}.

\smallskip 

\noindent
\textit{Bounds on cubic error terms}.
For the cubic error terms $\eta^2 \calE_t$,
with $\calE_t$ as defined in~\eqref{eq:calE-def}, 
we prove:

\noindent
\textbf{Proposition~\ref{prop:error-third-order}}
(Bound on cubic eror terms).
\textit{Let $B = \tfrac{1}{5}(54 \sigma + 9)$.
Then for $\eta \le \frac{1}{21 \sigma}$ and $t \ge 1$}:
\begin{equation}
\eta^2 |\calE_t|  \le 
B \cdot \eta^3 \|J \nabla F(z_t) \|^2_{z_t} \;.
\label{eq:cubic-bound}
\end{equation}
The proof of Proposition~\ref{prop:error-third-order}
is in Section~\ref{app:energy-dissipation:error-control:cubic}.

\smallskip

\noindent
\textbf{Combining the pieces for the upper bound.}

\noindent
For the upper bound of Lemma~\ref{lem:energy-one-step-full},
assume the stepsize satisfies
$\eta \le \min\{\frac{1}{72\sigma}, \frac{1}{20B}\} = \tfrac{1}{4(54\sigma + 9)}$,
which is the exact setting of Assumption~\ref{ass:stepsize}.
Then applying Propositions~\ref{prop:DH-error},~\ref{prop:RH-inner-bound}
and~\ref{prop:error-third-order}, 
the error terms of~\eqref{eq:diss-01} can be collectively bounded as
\begin{equation}
	D_F(z_{t+1}, z_t) + 
	\eta \langle J \nabla F(z_t), G_F(z_t, z_{t-1}) \rangle
	+ \eta^2 \calE_t 
	\le 
	\big(\tfrac{9}{10}  + \eta \tfrac{1}{20} \big)
	 \cdot \eta^2 \|J\nabla F(z_t)\|^2_{z_t} \;.
	\label{eq:diss-upper-proof}
\end{equation}
Then together with the expansion of~\eqref{eq:diss-01}, we conclude
for $\Delta F(z_t) = F(z_{t+1}) - F(z_t)$ that
\begin{align*}	
	\Delta F(z_t)
	&\le 
	- \eta^2 \|J \nabla F(z_t)\|^2_{z_t} 
	+ 
	\big(\tfrac{9}{10}  + \eta \tfrac{1}{20}\big) \cdot 
	\eta^2 \|J \nabla F(z_t)\|^2_{z_t}  \\
	&\le 
	- \tfrac{1}{20} \eta^2 \|J \nabla F(z_t)\|^2_{z_t} \;,
\end{align*}
where the final inequality is due to $\eta \le 1$.
This yields the upper bound statement of the lemma.

\smallskip

\noindent
\textbf{Combining the pieces for the lower bound.}

\noindent
For the lower bound on $\Delta F(z_t)$, 
note that $D_F(z_{t+1}, z_t) \ge 0$
by convexity of $F$. Thus again under the setting 
$\eta \le \min\{\frac{1}{72\sigma}, \frac{1}{20B}\} = \tfrac{1}{4(54\sigma + 9)}$,
we apply Propositions~\ref{prop:RH-inner-bound}
and~\ref{prop:error-third-order} to bound
\begin{equation*}
	- \eta | \langle J \nabla F(z_t), G_F(z_t, z_{t-1}) \rangle|
	- \eta^2 | \calE_t | 
	\ge
	- \big( 
		\tfrac{31}{200} + \eta \tfrac{1}{20}
	\big)
	\cdot \eta^2 \|J \nabla F(z_t) \|^2_{z_t}  \;.
\end{equation*}
Then again substituting this bound into~\eqref{eq:diss-01},
we obtain the matching lower bound
\begin{align*}
	\Delta F(z_t)
	&\ge 
	- \eta \|J \nabla F(z_t)\|^2_{z_t} 
	- \eta | \langle J \nabla F(z_t), G_F(z_t, z_{t-1}) \rangle|
	- \eta^2 | \calE_t |  \\
	&\ge
	- 
	\big(1 + \tfrac{31}{200} + \eta \tfrac{1}{20}\big)
	\cdot
	\eta^2 \|J \nabla F(z_t) \|^2_{z_t}\\	
	&\ge 
	- \big(1 + \tfrac{1}{4} \big)
	\cdot
	\eta^2 \|J \nabla F(z_t) \|^2_{z_t}
	= 
	- \tfrac{5}{4} \cdot \eta^2 \|J \nabla F(z_t)\|^2_{z_t} \;,
\end{align*}
which completes the proof of the lemma.
\hfill  ~ $\blacksquare$

\smallskip

\begin{restatable}[Generality of Lemma~\ref{lem:energy-one-step-full}]
{rem}{remarklemmagenerality}
	\label{remark:energy-dissipation-proof}
	We note that the proof of Lemma~\ref{lem:energy-one-step-full}
 	only requires that (1) $J = - J^\top$ is a skew-symmetric 
 	matrix and (2) the function $F$ satisfies
 	the local Hessian stability property~\eqref{H-gsc}. 
	While this latter property is established specifically for
	the log-sum-exp energy $F$ in Proposition~\ref{prop:energy-gsc},
	Lemma~\ref{lem:energy-one-step-full} holds
	for any other function satisfying this property. 
\end{restatable}

\smallskip

\paragraph{Organization of remaining subsections.}
The remainder of this section is organized as follows:
\begin{itemize}[
	leftmargin=1em
]
\item
\textbf{Section~\ref{app:energy-dissipation:expansion}}
gives the full proof of the exact one-step expansion
of $F(z_{t+1}) - F(z_t)$.

\item
\textbf{Section~\ref{app:energy-dissipation:error-control}}
gives the proofs of Propositions~\ref{prop:DH-error}, 
\ref{prop:RH-inner-bound}, and~\ref{prop:error-third-order}
for bounding the error terms.  
This also involves first introducing several technical lemmas
related to the stability of the OMWU iterates 
(see Section~\ref{app:energy-dissipation:error-control:stability-helpers})

\item
\textbf{Section~\ref{app:energy-dissipation:stability}}
gives the proof of the core
stability property of Proposition~\ref{prop:omwu-local-stability}.

\item
\textbf{Section~\ref{app:energy-dissipation:initial}}
gives a useful bound on the initial change in energy 
$F(z_1) - F(z_0)$ during the first step of the algorithm
(which contains no optimistic correction term).
\end{itemize}

\subsection{Expanding the One-Step Change in Energy}
\label{app:energy-dissipation:expansion}

\begin{restatable}{prop}{propenergyonestepexpand}
\label{prop:energy-one-step-expand}
Let $\{z_t\}$ denote the iterates of~\eqref{eq:omwu-dual}.
Then for any $t \ge 2$: 
\begin{align*}
	F(z_{t+1}) - F(z_t)
	= 
	- \eta^2 \|J \nabla F(z_t) \|^2_{z_t}
	+ D_F(z_{t+1}, z_t)  
	+ \eta \langle J \nabla F(z_t), G_F(z_t, z_{t-1}) \rangle 
	+ \eta^2 \calE_t \;,
\end{align*}
where $\calE_t := \calE_{t, 1} + \calE_{t,2} + \calE_{t, 3} + \calE_{t, 4}$
for 
\begin{equation}
	\begin{cases}
		\calE_{t,1}
		:=  \langle J \nabla F(z_t), J \nabla^2 F(z_t) (z_t - z_{t-1}) \rangle_{z_t} \\
		\calE_{t,2}
		:= 
		 \langle J \nabla F(z_t), J G_F(z_t, z_{t-1})\rangle_{z_t} \\
		\calE_{t,3}
		:= 
		- \langle J\nabla F(z_t), J \nabla^2 F(z_{t-1}) (z_{t-1} - z_{t-2})\rangle_{z_t} \\
		\calE_{t,4}
		:= 
		-  \langle J \nabla F(z_t), J G_F(z_{t-1}, z_{t-2}) \rangle_{z_t} \;.
	\end{cases}
	\label{eq:calE-def}
\end{equation} 
\end{restatable}

\begin{proof}
	Let $\Delta F(z_t) := F(z_{t+1}) - F(z_t)$.
  	By definition of the Bregman divergence $D_F$, 
  	recall that:
  	\begin{equation}
  		\Delta F(z_t) 
  		= 
  		z_{t+1} - z_t
  		= 
  		\langle \nabla F(z_t), z_{t+1} - z_t \rangle
  		+ 
  		D_F(z_{t+1}, z_t) \;.
  		\label{eq:H-00}
  	\end{equation}
  	Thus, to obtain the statement of the proposition,
  	our task is to expand the first-order term
  	$\langle \nabla F(z_t), z_{t+1} -z_t \rangle$ 
  	in~\eqref{eq:H-00}. 
  	For this, by the~\eqref{eq:omwu-dual} update rule, 
  	observe that we can write
  	\begin{equation}
  		z_{t+1} - z_t 
  		= - \eta J \nabla F(z_t)
  		- \eta J \big(  \nabla F(z_t) - \nabla F(z_{t-1}) \big) \;.
  		\label{eq:H-01}
  	\end{equation}
  	Now for $k \ge 2$, 
  	let us write $\calM_k := \nabla F(z_k) - \nabla F(z_{k-1})$. 
  	Then observe by substituting~\eqref{eq:H-01}
  	into the first-order term of~\eqref{eq:H-00} that
  	\begin{align}
  		\big \langle \nabla F(z_t), z_{t+1} - z_t  \big\rangle
  		&= 
  		- \eta \big \langle
  		\nabla F(z_t), 
  		J \nabla F(z_t)
  		\big\rangle
  		-
  		\eta 
  		\big\langle
  		\nabla F(z_t), J \calM_t
  		\big\rangle  \nonumber \\
  		&= 
  		- \eta 
  		\big\langle
  		\nabla F(z_t), J \calM_t
  		\big\rangle \nonumber \\
  		&= 
  		\eta 
  		\big\langle 
  		J \nabla F(z_t), \calM_t
  		\big\rangle	\;,
  		\label{eq:H-02}
  	\end{align}
  	where the second and third equalities are due to the
  	skew-symmetry of $J = -J^\top$. 
  	We can then further expand the term $\calM_t$ using
  	the Hessian-based approximation of gradient differences from 
	Proposition~\ref{prop:grad-diff-general}.
	Specifically, applying the proposition yields
	\begin{equation}
		\calM_t 
		= 
		\nabla F(z_t) - \nabla F(z_{t-1})
		= 
		 \nabla^2 F(z_t)(z_t - z_{t-1})
		+ G_F(z_{t}, z_{t-1}) \;.
		\label{eq:H-03}
	\end{equation}
	We can further use the update rule of~\eqref{eq:omwu-dual}
	at time $t$ to write 
	\begin{align}
		z_t - z_{t-1} 
		&= - 2 \eta J \nabla F(z_{t-1}) - \nabla F(z_{t-2}) 
		\nonumber \\
		&= 
		- \eta J \nabla F(z_t)
		+ \eta J ( \nabla F(z_t) - \nabla F(z_{t-1}))
		- \eta J ( \nabla F(z_{t-1}) - \nabla F(z_{t-2}) )
		\nonumber \\
		&= 
		- \eta J \nabla F(z_t)
		+ \eta J \calM_{t} 
		- \eta J \calM_{t-1} \;.
		\label{eq:H-04}
	\end{align}
	Similar to $\calM_t$, we can apply
	Proposition~\ref{prop:grad-diff-general} at time $t-1$ to obtain
	\begin{equation}
		\calM_{t-1} 
		= 
		\nabla F(z_{t-1}) - \nabla F(z_{t-2})
		= 
		 \nabla^2 F(z_{t-1})(z_{t-1} - z_{t-2})
		+ G_F(z_{t-1}, z_{t-2}) \;.
		\label{eq:H-05}
	\end{equation}
	Then combining~\eqref{eq:H-03},~\eqref{eq:H-04},
	and~\eqref{eq:H-05}, we can write
	\begin{align*}
		z_{t} - z_{t-1} 
		= - \eta J \nabla F(z_t) 
		&+ \eta J \nabla^2 F(z_t) (z_t - z_{t-1}) \\
		&+ \eta J G_F(z_t, z_{t-1})  \\
		&- \eta J \nabla^2 F(z_{t-1})(z_{t-1} - z_{t-2})  \\
		&- \eta J G_F(z_{t-1}, z_{t-2}) \;,
	\end{align*}
	from which it follows by~\eqref{eq:H-03} that 
	\begin{equation}
		\begin{aligned}
		\calM_t
		= 
		- \eta \nabla^2 F(z_t) J \nabla F(z_t)
		&+ \eta \nabla^2 F(z_t) J \nabla^2 F(z_t) (z_t - z_{t-1}) \\
		&+ \eta \nabla^2 F(z_t) J G_F(z_t, z_{t-1}) \\
		&- \eta \nabla^2 F(z_t) J \nabla^2 F(z_{t-1}) (z_{t-1} - z_{t-2}) \\
		&- \eta \nabla^2 F(z_t) J G_F(z_{t-1}, z_{t-2}) \\ 
		& + G_F(z_t, z_{t-1}) \;.
		\end{aligned}
		\label{eq:H-mt-expand}
	\end{equation}
	Now recall from~\eqref{eq:H-00} and~\eqref{eq:H-02} that 
	$\Delta F(z_t) = D_F(z_{t+1}, z_t) + 
	\eta \langle J \nabla F(z_t), \calM_t \rangle$.
	Then plugging in the expansion of $\calM_t$ from~\eqref{eq:H-mt-expand}
	and using the local norm and local inner product notation 
	$\|\cdot \|_{z_t}$ and $\langle \cdot, \cdot \rangle_{z_t}$,
	we conclude:
	\begin{equation}
		\begin{aligned}
		\Delta F(z_t) 
		= 
		- \eta^2 \|J \nabla F(z_t) \|^2_{z_t}
		+ D_F(z_{t+1}, z_t)  
		&+ \eta \langle J \nabla F(z_t), G_F(z_t, z_{t-1}) \rangle \\
		&+ \eta^2 \langle J \nabla F(z_t), J \nabla^2 F(z_t) (z_t - z_{t-1}) \rangle_{z_t} \\
		&+ \eta^2 \langle J \nabla F(z_t), J G_F(z_t, z_{t-1})\rangle_{z_t} \\
		&- \eta^2\langle J\nabla F(z_t), J \nabla^2 F(z_{t-1}) (z_{t-1} - z_{t-2})\rangle_{z_t} \\
		&- \eta^2 \langle J \nabla F(z_t), J G_F(z_{t-1}, z_{t-2}) \rangle_{z_t}
		\;.
		\end{aligned}
		\label{eq:H-06}
	\end{equation}
	Here, up to the $\eta^2$ factor, the final four terms of~\eqref{eq:H-06} 
	are the quantities $\calE_{t, 1}$, $\calE_{t, 2}$, $\calE_{t, 3}$, 
	and $\calE_{t, 4}$ defined in the statement of the proposition,
	which concludes the proof.
\end{proof}
  
\subsection{Controlling the Error Terms}
\label{app:energy-dissipation:error-control}

\subsubsection{Helper Propositions on Stability of OMWU Iterates}
\label{app:energy-dissipation:error-control:stability-helpers}

Here, we establish several key properties related 
to the stability of the OMWU dual iterates:
First, the following proposition gives a bound in $\ell_{\infty}$ 
norm on the distance between consecutive iterates:

\smallskip

\begin{restatable}{prop}{propdualdiff}
	\label{prop:omwu-dual-diff}
	Let $\{z_t\}$ be iterates of~\eqref{eq:omwu-dual}.
	Then for all $t \ge 0$:  
	\begin{equation*}
	  \|z_{t+1} - z_{t}\|_\infty \le 3 \eta \sigma_{\max} 
	  \;\;\text{and}\;\;
	  \|z_{t+2} - z_{t}\|_\infty \le 6 \eta \sigma_{\max} \;.
	\end{equation*}
\end{restatable}

\begin{proof}
	Recall from the definition of~\eqref{eq:omwu-dual} that 
	for $t \ge 0$:
	\begin{equation*}
		z_{t+1} - z_{t}
		= 
		- 2\eta J \nabla F(z_{t}) 
		+ \eta J \nabla F(z_{t-1}) 
		\;.
	\end{equation*}
	As $\nabla F(z) \in \relint(\calW)$ for 
	all $z \in \R^{m+n}$, 
	we have by definition of $J$ that
	$\|J \nabla F(z) \|_{\infty} \le a_{\max}$, 
	where $a_{\max}= \max_{(i, j) \in [m] \times [n]} | A(i, j)|$. 
	By Proposition~\ref{prop:amax-sigmamax}, $a_{\max} \le \sigma = \|J \|_2$,
	and thus by the triangle inequality
	$\|z_{t+1} - z_{t}\|_{\infty} \le 3 \eta \sigma$.
	The second claim for $\|z_{t+2} - z_{t}\|_{\infty}$
	then follows similarly. 
\end{proof}

\smallskip

The next proposition relates the 
local norm of the dual increments $z_{t+1} - z_t$
to the magnitude of the payoff vector $J \nabla F(z_t)$
in local norm: 

\begin{restatable}{prop}{proplocalnormdiff}
  \label{prop:local-norm-diff}
  Let $\{z_t\}$ be the iterates of~\eqref{eq:omwu-dual}
  with $\eta \le \tfrac{1}{21 \sigma}$.
  Then for $t \ge 1$: 
  \begin{equation*}
    \| z_{t+1} - z_{t} \|_{z_t} 
    \le 
    \tfrac{6}{5} \cdot \eta \|J \nabla F(z_t) \|_{z_t} \;.
  \end{equation*}
\end{restatable}

\smallskip
The proof of Proposition~\ref{prop:local-norm-diff},
which is central to establishing the bounds on the error terms,
relies further on the following stability property 
of the OMWU iterates: 

\smallskip

\begin{restatable}{prop}{proplocalstability}
	\label{prop:omwu-local-stability}
	Let $\{z_t\}$ be the iterates of~\eqref{eq:omwu-dual}
    with $\eta \le \tfrac{1}{21\sigma}$.
	Then for all $t \ge 2$, the following inequalities hold:
    \begin{equation*}
	    \text{(i)}
		\;\;
		\|z_{t-1} - z_{t-2}\|_{z_t}
		\le
		2 \cdot  \|z_t - z_{t-1} \|_{z_t}
		\;\;\text{and}\;\;
		\text{(ii)}
		\;\;
		\|z_t - z_{t-1} \|_{z_t}
		\le
		3 \cdot \|z_{t+1} - z_t\|_{z_t} \;.
    \end{equation*}
\end{restatable}

We defer the proof of Proposition~\ref{prop:omwu-local-stability}
to Section~\ref{app:energy-dissipation:stability},
and we assume it is true for now.
We then proceed with the proof of
Proposition~\ref{prop:local-norm-diff}:

\paragraph{Proof of Proposition~\ref{prop:local-norm-diff}}
First, using the OMWU update rule at time $t+1$
and taking local norms, we have 
\begin{align}
    \|z_{t+1} - z_t \|_{z_t} 
    &= 
    \eta 
    \| J \nabla F(z_t) + J (\nabla F(z_t) - \nabla F(z_{t-1})) \|_{z_t} 
    \nonumber \\
    &\le
    \eta \|J \nabla F(z_t)\|_{z_t}
    + 
    \eta \|J (\nabla F(z_t) - \nabla F(z_{t-1})) \|_{z_t}  \;.
    \label{eq:lnd-01}
\end{align}
Thus to prove the proposition, it suffices to control
the second term of~\eqref{eq:lnd-01} in terms of $\|J \nabla F(z_t)\|_{z_t}$.

\paragraph{Controlling the second term of~\eqref{eq:lnd-01}.}
Recall from Proposition~\ref{prop:grad-diff-general} that 
\begin{equation*}
    \nabla F(z_t) - \nabla F(z_{t-1})
    = \nabla^2 F(z_{t}) (z_t - z_{t-1})
    + G_F(z_t, z_{t-1}) \;.
\end{equation*}
Thus it follows by taking norms and applying the triangle inequality that
\begin{align}
    \|J (\nabla F(z_t) - \nabla F(z_{t-1})) \|_{z_t}
    &=
    \|J \nabla^2 F(z_{t}) (z_t - z_{t-1})
    + J G_F(z_t, z_{t-1})\|_{z_t} \nonumber  \\
    &\le
    \|J \nabla^2 F(z_t)(z_t - z_{t-1})\|_{z_t} 
    + 
    \|J G_F(z_t, z_{t-1})\|_{z_t} \;.
    \label{eq:lnd-02}
\end{align}
To control the two terms of~\eqref{eq:lnd-02},
the key is to upper bound both quantities in terms of 
$\|z_{t+1} - z_t \|_{z_t}$. 
For this, we rely on the stability property established
in Proposition~\ref{prop:omwu-local-stability}. 
Specifically, for the first term of~\eqref{eq:lnd-02} we have 
\begin{align}
    \|J \nabla^2 F(z_t)(z_t - z_{t-1})\|_{z_t} 
    &\le 
    \|J\|_2 \cdot \|z_t - z_{t-1} \|_{z_t} 
    &&\text{(by Lemma~\ref{lem:hessian-local-smooth})} \nonumber \\
    &\le 
    3 \sigma \cdot \|z_{t+1} - z_{t} \|_{z_t} 
    &&\text{(by Proposition~\ref{prop:omwu-local-stability})}
    \label{eq:lnd-03} \;.
\end{align}
For the second term of~\eqref{eq:lnd-02}, recall by 
Proposition~\ref{prop:omwu-dual-diff} 
that $\|z_t - z_{t-1}\|_\infty \le 3\eta \sigma$.  
Then using similar calculations as in the proof of 
Part \ref{item:stability-helper-c} of Lemma~\ref{lem:helper-inequalities}, 
and also applying Proposition~\ref{prop:omwu-local-stability}, 
we have 
\begin{align}
    \|J G_F(z_t, z_{t-1})\|_{z_t} 
    &\le 
    \|J G_F(z_t, z_{t-1})\|_2 
    &&\text{(by Proposition~\ref{prop:norm-relations})} \nonumber \\
    &\le 
    \sigma
    \cdot 
    \|G_F(z_t, z_{t-1})\|_2 
    &&\text{(by definition of $\|J\|_2$)} \nonumber \\
    &\le 
    \sigma
    \cdot 
    \|G_F(z_t, z_{t-1})\|_{z_t, *} 
    &&\text{(by Proposition~\ref{prop:norm-relations})} \nonumber \\
    &\le 
    \sigma r_1 \cdot \|z_t - z_{t-1}\|_{z_t}  
    &&\text{(by Lemma~\ref{lem:hessian-remainder-bound})} \nonumber \\
     &\le 
    3 \sigma r_1 \cdot \|z_{t+1} - z_{t}\|_{z_t}  
    &&\text{(by Proposition~\ref{prop:omwu-local-stability})}
    \label{eq:lnd-04} \;,
\end{align} 
where for readability we write
$r_1 := \nu(6 \eta \sigma)$
(for the function $\nu$ as defined in~\eqref{eq:def-mu-nu}).
Then combining~\eqref{eq:lnd-03} and~\eqref{eq:lnd-04}
and substituting back into~\eqref{eq:lnd-02} yields
\begin{equation}
    \|J (\nabla F(z_t) - \nabla F(z_{t-1})) \|_{z_t}   
    \le 
    3\sigma (1 + r_1) \cdot \|z_{t+1} - z_{t}\|_{z_t} \;.
    \label{eq:lnd-05}
\end{equation}

\paragraph{Combining the pieces.}
Combining~\eqref{eq:lnd-05} with~\eqref{eq:lnd-01}, 
we can thus write
\begin{equation}
    \|z_{t+1}-z_t \|_{z_t}
    \le
    \eta \|J \nabla F(z_t)\|_{z_t}
    + 
    3\eta \sigma (1+r_1) \cdot \|z_{t+1} -z_t\|_{z_t} \;.
    \label{eq:lnd-06}
\end{equation}
Using the definition of $r_1 := 6\eta \sigma$ and 
the upper bound on $\nu(\cdot)$ from
Lemma~\ref{lem:exp-bounds}, it is then 
straightforward to check that 
$3\eta \sigma (1+  r_1) \le \tfrac{1}{6}$ when
$\eta \le \tfrac{1}{21\sigma}$.
Thus rearranging~\eqref{eq:lnd-06}, we find
\begin{equation}
    \| z_{t+1} - z_t \|_{z_t} 
    \le 
    \Big(
        \frac{1}{1-3\eta(1+r_1)\sigma}
    \Big)
    \cdot \eta 
    \|J \nabla F(z_t)\|_{z_t} 
    \le
    \frac{6}{5} \cdot \eta \|J \nabla F(z_t)\|_{z_t}  \;,
\end{equation}
which concludes the proof. 
\hfill $\blacksquare$

\subsubsection{Bounds on Quadratic Error Terms}
\label{app:energy-dissipation:error-control:quadratic}

We bound the Bregman divergence term
in local norm as follows:

\smallskip

\begin{restatable}{prop}{propDHerror}
	\label{prop:DH-error}
	Let $\{z_t\}$ be the iterates of~\eqref{eq:omwu-dual}
	with $\eta \le \tfrac{1}{21 \sigma}$.
	Then for $t \ge 1$:
	\begin{equation*}
	D_F(z_{t+1}, z_t)
	\le 
	\tfrac{36}{25} \cdot 
	\mu(6\eta \sigma) \cdot
	\eta^2 \|J \nabla F(z_t)\|^2_{z_t} \;.
	\end{equation*}
	Moreover, if $\eta \le \tfrac{1}{L\sigma}$
	for $L \ge 72$, then 
	$\tfrac{36}{25} \cdot 
	\mu(6\eta \sigma) \le \tfrac{18}{25} + \tfrac{36\exp(6/L)}{25L}
	\le \tfrac{149}{200}$.
\end{restatable}

\smallskip

\begin{proof}
	Recall from Proposition~\ref{prop:omwu-dual-diff}
	that $\|z_{t+1} - z_t \|_{\infty} \le 3 \eta \sigma$. 
	Then applying the upper bound on $D_F(z_{t+1}, z_t)$
	from Lemma~\ref{lem:bregman-gsc}, we have 
	\begin{equation}
		D_F(z_{t+1}, z_t)
		\le 
		\mu(6\eta\sigma) 
		\cdot \|z_{t+1} - z_t\|^2_{z_t} \;.
		\label{eq:dhb-01} 
	\end{equation}
	Further using the bound
	$\|z_{t+1} - z_t\|_{z_t} \le \frac{6}{5} \cdot \|J \nabla F(z_t)\|_{z_t}$
	from Proposition~\ref{prop:local-norm-diff} (which holds under
	the constraint $\eta \le \frac{1}{21\sigma}$), we find 
	\begin{equation*}
		D_F(z_{t+1}, z_t)
		\le 
		\tfrac{36}{25} \cdot 
		\mu(6\eta\sigma)
		\cdot \eta^2 \|J\nabla F(z_t)\|^2_{z_t} \;.
	\end{equation*}
	Now using the definition of $\mu(\cdot)$	
	from~\eqref{eq:def-mu-nu} and the bounds of 
	Lemma~\ref{lem:exp-bounds}, observe that if $\eta \le 1/(L\sigma)$
	for some positive $L > 0$, then 
	\begin{equation*}
		\tfrac{36}{25} \cdot 
		\mu(6\eta\sigma)
		\le 
		\tfrac{18}{25} + \tfrac{36 \exp(6/L)}{25L} \;.
	\end{equation*}
	The right hand side of this term is decreasing with $L$.
	It is then straightforward to check that, for $L \ge 72$, 
	then $\tfrac{36\exp(6/L)}{25L} \le \tfrac{1}{40}$.
	Thus for all such $L \ge 72$, 
	it holds that 
	$\tfrac{18}{25} + \tfrac{36 \exp(6/L)}{25L} \le \tfrac{149}{200}$.
\end{proof}

\smallskip

For the inner product involving $G_F(z_t, z_{t-1})$, 
we derive a similar bound: 

\begin{restatable}{prop}{proprhinnerbound}
	\label{prop:RH-inner-bound}
	Let $\{z_t\}$ be the iterates of~\eqref{eq:omwu-dual}
	with $\eta \le \frac{1}{21\sigma}$.
	Then for $t \ge 1$: 
	\begin{equation*}
		\big| \eta \langle J \nabla F(z_t), G_F(z_t, z_{t-1}) \rangle \big|
		\le 
		\tfrac{18}{5} \cdot \nu(6 \eta \sigma) \cdot
		\eta^2 \|J \nabla F(z_t)\|^2_{z_t} \;.
	\end{equation*}
	Moreover, if $\eta \le \tfrac{1}{L \sigma}$
	for $L \ge 72$, then 
	$\tfrac{18}{25} \cdot \nu(6 \eta \sigma) \le 
	\tfrac{54}{5L} + \tfrac{108 \exp(6/L)}{5L^2}
	\le
	\tfrac{31}{200}$.
\end{restatable}

\begin{proof}
	Applying the generalized Cauchy-Schwarz inequality
	with the pair $(\|\cdot\|_{z_t}, \|\cdot\|_{z_t, *})$, 
	we have 
	\begin{equation}
		\big| \eta \langle J \nabla F(z_t), G_F(z_t, z_{t-1}) \rangle\big|
		\le 
		\eta \|J \nabla F(z_t) \|_{z_t} 
		\cdot 
		\|G_F(z_t, z_{t-1})\|_{z_t, *} \;.
		\label{eq:rhb-01}
	\end{equation}
	Since by Proposition~\ref{prop:omwu-dual-diff} 
	$\|z_{t} - z_{t-1}\|_{\infty} \le 3 \eta \sigma$, 
	we apply the dual norm bound 
	of Lemma~\ref{lem:hessian-remainder-bound}
	on $G_F(z_t, z_{t-1})$ to further write 
	\begin{align}
		\|G_F(z_t, z_{t-1})\|_{z_t, *}
		\le
		\nu(6\eta\sigma) \cdot \|z_t - z_{t-1}\|_{z_t} 
		&\le 
		3 \cdot  \nu(6\eta \sigma) 
		\cdot \|z_{t+1} - z_{t}\|_{z_t}  \nonumber \\
		&\le
		\tfrac{18}{5} \cdot \nu(6 \eta \sigma) 
		\cdot \eta \|J\nabla F(z_t)\|_{z_t} \;.
		\label{eq:rhb-02}
	\end{align}
	Here, the second inequality comes from applying
	the stability bound of Proposition~\ref{prop:omwu-local-stability},
	and the final inequality comes from 
	Proposition~\ref{prop:local-norm-diff}
	(which holds under the constraint $\eta \le 1/(21\sigma)$).
	Combining expressions~\eqref{eq:rhb-02}
	and~\eqref{eq:rhb-01} then yields the 
	first claim.

	For the second claim, we have from the definition of $\nu(\cdot)$
	from~\eqref{eq:def-mu-nu} and the bounds of 
	Lemma~\ref{lem:exp-bounds}, that, if $\eta \le 1/(L\sigma)$
	for some positive $L > 0$, then 
	\begin{equation*}
		\tfrac{18}{5} \cdot \nu(6 \eta \sigma)
		\le 
		\tfrac{54}{5 L} + \tfrac{108 \exp(6/L)}{5L^2} \;.
	\end{equation*}
	The right hand side of this term is decreasing with $L$,
	and it is straightforward to check that 
	this term is at most $\tfrac{31}{200}$
	for all $L \ge 72$.
\end{proof}

\subsubsection{Bounds on Cubic Error Terms}
\label{app:energy-dissipation:error-control:cubic}

\begin{restatable}{prop}{properrorthirdorder}
    \label{prop:error-third-order}
    Let $\{z_t\}$ be the iterates of~\eqref{eq:omwu-dual}
    with $\eta \le \tfrac{1}{21\sigma}$, 
    where $\sigma = \|J\|_2$. 
    Recall the terms $\calE_{t, 1}$, 
    $\calE_{t,  2}$, $\calE_{t, 3}$, and $\calE_{t,4}$
    from~\eqref{eq:calE-def} in Proposition~\ref{prop:energy-one-step-expand} 
    are given by:
    \begin{equation*}
	\begin{cases}
		\calE_{t, 1} 
		:=  \langle J \nabla F(z_t), J \nabla^2 F(z_t) (z_t - z_{t-1}) \rangle_{z_t} \\
		\calE_{t, 2}
		:= 
		 \langle J \nabla F(z_t), J G_F(z_t, z_{t-1})\rangle_{z_t} \\
		\calE_{t, 3}
		:= 
		- \langle J\nabla F(z_t), J \nabla^2 F(z_{t-1}) (z_{t-1} - z_{t-2})\rangle_{z_t} \\
		\calE_{t, 4}
		:= 
		-  \langle J \nabla F(z_t), J G_F(z_{t-1}, z_{t-2}) \rangle_{z_t} \;.
	\end{cases}
    \end{equation*}
    Then the following bounds hold:
    \begin{align*}
        |\calE_{t, 1}|
        &\le 
        \frac{18}{5} \sigma \cdot \eta \|J \nabla F(z_t)\|^2_{z_t}\,,
        \quad
        |\calE_{t, 2}|
        \le
        \frac{3}{5} \cdot \eta \|J \nabla F(z_t)\|^2_{z_t} \,, \\
        |\calE_{t,3}|
        &\le 
        \frac{36}{5} \sigma \cdot \eta \|J \nabla F(z_t)\|^2_{z_t} \,,
        \quad
        |\calE_{t, 4}|
        \le
        \frac{6}{5} \cdot \eta \|J \nabla F(z_t)\|^2_{z_t} \;.
    \end{align*}
    Moreover, recalling that 
    $\calE_t := \calE_{t,1} + \calE_{t,2} + \calE_{t,3} + \calE_{t, 4}$, 
    then 
    $
        \eta^2 |\calE_t|
        \le         
        \Big(\frac{54}{5}\sigma + \frac{9}{5}\Big)
        \cdot \eta^3
        \|J \nabla F(z_t)\|^2_{z_t}
    $.
\end{restatable}

\begin{proof}
	For each term, we use a combination of 
    (i) (generalized) Cauchy-Schwarz
    (ii) the bounds from 
    Lemma~\ref{lem:hessian-local-smooth} and
    Lemma~\ref{lem:hessian-remainder-bound} stemming
    from the~\eqref{H-gsc} property
    and (iii) the local stability property of OMWU 
    from Proposition~\ref{prop:omwu-local-stability}
    to show the magnitude of the term 
    scales at most like a constant multiple of
    $\|J \nabla F(z_t)\|_{z_t} \cdot \| z_{t+1} - z_t\|_{z_t}$. 
    From there, we apply to each case Proposition~\ref{prop:local-norm-diff},
    which establishes $ \| z_{t+1} - z_t\|_{z_t} = O(\eta \|J\nabla F(z_t)\|_{z_t})$.
    In more details:
    \begin{itemize}[
        leftmargin=1em,
        itemsep=0em,
    ]
        \item
        \textbf{Term $\calE_{t, 1}$}.
        Using Cauchy-Schwarz, we have
        \begin{align*}
        |\calE_{t, 1}|
        &= 
            \big|\big\langle
            J \nabla F(z_t), J \nabla^2 F(z_t)(z_t - z_{t-1})
            \big\rangle_{z_t}\big| \\
        &\le
            \|J \nabla F(z_t)\|_{z_t}
            \cdot 
            \|J \nabla^2 F(z_t)(z_t-z_{t-1}) \|_{z_t}\\
        &\le
            \|J \nabla HF(z_t)\|_{z_t}
            \cdot 
            \|J\|_2 \cdot \| z_t  - z_{t-1} \|_{z_t} 
            &&(\text{by Lemma~\ref{lem:hessian-local-smooth}})\\
        &\le 
            \|J \nabla F(z_t)\|_{z_t}
            \cdot 
            \sigma  \cdot 
            3 \cdot \| z_{t+1} - z_t \|_{z_t} 
            &&(\text{by Proposition~\ref{prop:omwu-local-stability}})\\
        &\le
            \|J \nabla F(z_t)\|_{z_t}
            \cdot 
            \sigma \cdot 3 \cdot \tfrac{6}{5}  
            \eta 
            \cdot
            \|J \nabla F(z_t)\|_{z_t} 
            &&(\text{by Proposition~\ref{prop:local-norm-diff}})\\
        &=
            \tfrac{18}{5} \sigma \cdot \eta \|J \nabla F(z_t)\|^2_{z_t} \;.
        \end{align*}

        \item
        \textbf{Term $\calE_{t, 2}$}.
        Recall by Proposition~\ref{prop:omwu-dual-diff}
        that $\|z_t - z_{t-1}\|_\infty \le 3\eta \sigma$.
        Let $r_1 := \nu(6 \eta \sigma)$ 
        for $\nu(\cdot)$ as defined in~\eqref{eq:def-mu-nu}.
        Using the upper bound on $\nu(\cdot)$ from
        Lemma~\ref{lem:exp-bounds}, it is straightforward
        to check that $r_1 \le \frac{1}{6}$ when $\eta \le \frac{1}{21\sigma}$. 
        Then using generalized Cauchy-Schwarz, we can write 
        \begin{align*}
            |\calE_{t, 2} |
            &=
                \big|
                \big\langle
                J \nabla F(z_t), G_F(z_t, z_{t-1})
                \big\rangle_{z_t}\big| \\
            &\le 
                \|J \nabla F(z_t)\|_{z_t}
                \cdot 
                \|G_F(z_t, z_{t-1})\|_{z_t, *} \\
            &\le
                \|J \nabla F(z_t)\|_{z_t}
                \cdot 
                r_1 \cdot \|z_t - z_{t-1}\|_{z_t} 
                &&(\text{by Lemma~\ref{lem:hessian-remainder-bound}}) \\
            &\le 
                \|J \nabla F(z_t)\|_{z_t}
                \cdot 
                r_1
                \cdot 3 \|z_{t+1} - z_t \|_{z_t}
                &&(\text{by Proposition~\ref{prop:omwu-local-stability}}) \\
            &\le
                \|J \nabla F(z_t)\|_{z_t}
                \cdot  r_1
                \cdot 3
                \cdot \tfrac{6}{5} 
                \cdot \eta \|J \nabla F(z_t)\|_{z_t} 
                &&(\text{by Proposition~\ref{prop:local-norm-diff}}) \\
            &\le
                \tfrac{3}{5}
                \cdot \eta \|J \nabla F(z_t)\|^2_{z_t}
                &&(\text{using $r_1 \le \tfrac{1}{6}$ when $\eta \le \tfrac{1}{28\sigma}$}) 
                \;.
        \end{align*}
        \item
        \textbf{Term $\calE_{t, 3}$}.
        We repeat nearly identitical steps as for 
        $\calE_{t, 1}$, but for the third inequality we get 
        from the application of Proposition~\ref{prop:omwu-local-stability}
        that $\|z_{t-1} - z_{t-2}\|_{z_t} \le 6 \|z_{t+1} - z_t\|_{z_t}$. 
        Then it follows that:
        \begin{align*}
            |\calE_{t, 3}| 
            = 
            \big|
            \big\langle
                J \nabla F(z_t), J \nabla^2 F(z_{t-1})(z_{t-1} - z_{t-2})
                \big\rangle_{z_t}
            \big| 
            \le
            2 \cdot |\calE_{t, 1} |  
            \le 
            \tfrac{36}{5} \sigma \cdot \eta \|J \nabla F(z_t)\|^2_{z_t} \;.
        \end{align*}
    
        \item
        \textbf{Term $\calE_{t, 4}$}.
        We repeat nearly identical steps as for 
        $\calE_{t, 2}$, and from the application of 
        Proposition~\ref{prop:omwu-local-stability}
        we get $\|z_{t-1} - z_{t-2}\|_{z_t} \le 6 \|z_{t+1} - z_{t}\|_{z_t}$. 
        Thus it follows that
        \begin{align*}
            |\calE_{t,4} |
            \le
            \big|
            \big\langle
            J \nabla F(z_t), G_F(z_{t-1}, z_{t-2})
            \big\rangle_{z_t} 
            \big|
            \le 
            2 \cdot |\calE_{t, 2}| 
            \le 
            \tfrac{6}{5} \cdot \eta \|J \nabla F(z_t)\|^2_{z_t} \;.
        \end{align*}
    \end{itemize}
    Summing up the constant multiplicative factors, we find
    \begin{equation*}
        |\calE_t|
        \le 
        |\calE_{t, 1}| + |\calE_{t, 2}| + |\calE_{t, 3}| + |\calE_{t, 4}|
        \le         
        \big(\tfrac{54}{5}\sigma + \tfrac{9}{5}\big)
        \cdot \eta
        \|J \nabla F(z_t)\|^2_{z_t} \;,
    \end{equation*}
    which concludes the proof. 
\end{proof}

\subsection{Stability of Iterates Under Local Norm}
\label{app:energy-dissipation:stability}

In this section, we give the proof of
Proposition~\eqref{prop:omwu-local-stability}, 
which establishes the stability of 
the dual OMWU iterates in local norm. 
Restated here:

\smallskip

\proplocalstability*

\smallskip

\subsubsection{Proof of Proposition~\ref{prop:omwu-local-stability}}

We start by proving inequality (i). 
Inequality (ii) will follow from (i). 

\paragraph{Shorthand notation.}
First, for readability, we define several pieces of shorthand notation
that will be used throughout: 
\begin{equation}
	\begin{cases}
	\, \Delta_{t} := z_t - z_{t-1}  \\
	\, g_t := \nabla F(z_t) - \nabla F(z_{t-1}) 
	\end{cases} 
  \;\text{for all $t \ge 1$}.
	\label{eq:stability-notation}
\end{equation}
Moreover, for readability, we write $\sigma = \sigma_{\max} := \|J \|_2$
throughout to denote the spectral norm of $J$.

\paragraph{Definition of constants.}
Recall by Proposition~\ref{prop:omwu-dual-diff} that
the dual iterates of OMWU satisfy the $\ell_\infty$ bounds 
of $\|z_{t} - z_{t-1}\|_\infty \le 3\eta \sigma$
and $\|z_{t} - z_{t-2}\|_\infty \le 6\eta \sigma$
for all $t \ge 2$.
Thus, to streamline the proof of the proposition,
we also define and use the following constants. 
These capture the multiplicative dependencies in applying
the norm transfer inequality of Lemma~\ref{lem:norm-transfer-gsc}
and the Hessian remainder bound 
of Lemma~\ref{lem:hessian-remainder-bound}
with either $\alpha_1 := 3\eta \sigma$ or
$\alpha_2 := 6\eta \sigma$:
\begin{equation}
    \begin{aligned}
        \kappa_1
        &:=  \exp(3\eta \sigma)
        \quad
        &&\text{(for applying Lemma~\ref{lem:norm-transfer-gsc} with $\alpha_1$)}
        \\
        \kappa_2 
        &:=  \exp(6\eta \sigma)
        \quad
        &&\text{(for applying Lemma~\ref{lem:norm-transfer-gsc} with $\alpha_2$)}
        \\
        r_1 
        &:= 
        \nu(6\eta \sigma)
        &&\text{(for applying Lemma~\ref{lem:hessian-remainder-bound} with $\alpha_1$)}
        \\
        r_2
        &:= \nu(12\eta \sigma)
        &&\text{(for applying Lemma~\ref{lem:hessian-remainder-bound} with $\alpha_2$)} 
        \\
        \sigma 
        &:= \|J\|_2 
        &&\text{(spectral norm of $J$)} 
    \end{aligned}
    \;\;.
    \label{eq:stability-constants}
\end{equation}
Here, recall that $\nu(x) = \frac{\exp(x)-x -1}{x}$ is the
scalar function defined in~\eqref{eq:def-mu-nu}.
We then further define the constants $B_1$, $B_2$, and $B$ by:
\begin{equation}
	B_1 :=  \sigma \kappa_1^2 + \sigma r_1 \kappa_1  
	\qquad
	B_2 := \sigma \kappa_2^2 + \sigma r_2 \kappa_2 
	\qquad
    B := 2 (B_1 + B_2) 
    \;\;.
    \label{eq:stability-constants-B}
\end{equation}
We will make use of the following properties
which are straightforward to verify: 
\begin{restatable}{prop}{propstabilityconstantbounds}
	\label{prop:stability-constants-bounds}
	Consider the constants
	defined in expressions~\eqref{eq:stability-constants}
	and~\eqref{eq:stability-constants-B}. 
	Then the following inequalities hold 
	for all $0 < \eta \le \frac{1}{21\sigma}$:
	\begin{equation*}
		\textstyle
		\text{(a)}\;\;  0 \le B_1 \le B_2 
		\qquad
		\text{(b)}\;\; \eta \le \tfrac{1}{4\kappa_1^2(B_1+B_2)}
		\qquad
		\text{(c)}\;\; \tfrac{1}{2\kappa_1^2} \le \tfrac{1}{2}  
		\qquad
		\text{(d)}\;\;\kappa_1 \le \tfrac{3}{2} \;.
	\end{equation*}
\end{restatable}

\paragraph{Proof of Proposition~\ref{prop:stability-constants-bounds}.}
We prove each of relationships separately:

\begin{itemize}[
  label={},
  leftmargin=*,
]

\item\textit{Part (a)}. 
Observe first that each term in the definition of $B_1$ and $B_2$, 
is non-negative, and thus $B_1, B_2 \ge 0$. 
Moreover, as $r_1 := \nu(6\eta \sigma)$
and $r_2 = \nu(12 \eta \sigma)$, 
it follows by definition of $\nu$ (from~\eqref{eq:def-mu-nu})
that $r_1 \le r_2$. Similarly, 
$\kappa_1 = \exp(3\eta \sigma) \le \exp(6\eta \sigma) = \kappa_2$.
Thus $B_1 \le B_2$.

\item\textit{Part (b)}.
Suppose $\eta \le \frac{1}{L\sigma}$ for some $L \ge 1$.
Then to verify $\eta \le 1/(4 \kappa^2_1 (B_1 + B_2))$
under such a constraint on $\eta$, it is sufficent to check that 
$4 \kappa^2_1 (B_1 + B_2) \le L \sigma$. 
For this, we start by simplifying $4 \kappa^2_1 (B_1 + B_2)$.
Using the definition of $B_1$ and $B_2$ we have:
\begin{equation}
	4 \kappa_1^2 (B_1 + B_2) 
	= 
	4 \sigma \kappa^2_1
	(\kappa^2_1 + r_1 \kappa_1
	+ \kappa^2_2 + r_2 \kappa_2 
	) 
	= 
	4 \sigma 
	(
		\kappa^4_1 + \kappa^3_1 r_1 
		+ 
		\kappa^2_1 \kappa^2_2 +  r_2 \kappa^2_1 \kappa_2
	) \;.
	\label{eq:sc-01}
\end{equation}
Note further that for $\eta \le \frac{1}{L\sigma}$, 
we have $6\eta\sigma \le \frac{6}{L}$
and $12 \eta \sigma \le \frac{12}{L}$. 
Thus by definition of $r_1$ and $r_2$, 
and using the bounds on $\nu(\cdot)$
from Lemma~\ref{lem:exp-bounds}, we can further write:
\begin{equation}
	\begin{aligned}
	r_1 
	&= 
		\nu(6\eta \sigma)
		\le 
		\tfrac{6\eta\sigma}{2}
		+ 
		\tfrac{1}{6} \cdot (\tfrac{6}{L})^2 \exp(\tfrac{6}{L})
		\le 
		\tfrac{3}{L} + \tfrac{6}{L^2} \exp(\tfrac{6}{L})  \\
	\text{and}\quad
	r_2
	&=
		\nu(12 \eta \sigma)
		\le
		\tfrac{12\eta\sigma}{2}
		+ 
		\tfrac{1}{6} \cdot (\tfrac{12}{L})^2 \exp(\frac{12}{L})
		\le 
		\tfrac{6}{L} + \tfrac{24}{L^2} \exp(\tfrac{12}{L}) \;.
	\end{aligned}
	\label{eq:sc-r1r2}
\end{equation}
Combining expresions~\eqref{eq:sc-r1r2} and~\eqref{eq:sc-01}
and also using the definitions from~\eqref{eq:stability-constants} 
of $\kappa_1 = \exp(3\eta\sigma)$ and 
$\kappa_2 = \exp(6\eta \sigma)$, 
we can then write:
\begin{equation}
	4 \kappa^2_1 (B_1 + B_2)
	\le 
	4 \sigma \cdot 
	\big( 
	\exp(\tfrac{12}{L})
	+ 
	\tfrac{3}{L} \exp(\tfrac{9}{L}) + \tfrac{6}{L^2} \exp(\tfrac{15}{L})
	+ 
	\exp(\tfrac{18}{L})
	+ \tfrac{6}{L} \exp(\tfrac{12}{L}) + \tfrac{24}{L^2}\exp(\tfrac{24}{L}) 
	\big)
	\label{eq:sc-02} \;.
\end{equation}
Note that each term in the right hand side of~\eqref{eq:sc-02} is decreasing
in $L$. 
It is then straightforward to verify that
the right hand side is at most $L \sigma$ for all $L \ge 21$,
which corresponds to the constraint on $L$ in the assumptions of
the proposition and concludes the proof of part (b).

\item\textit{Part (c)}.
Observe that $\frac{1}{2\kappa^2_1} \le \frac{1}{2} 
\iff \kappa^2_1 \ge 1$. 
By definition, we have $\kappa^2_1 = \exp(6\eta \sigma) \ge 1$
for all $\eta > 0$, which proves the claim. 

\item\textit{Part (d)}.
By definition $\kappa_1 = \exp(6\eta \sigma)$.
Thus for $\eta \le \frac{1}{21\sigma}$, 
it follows that $\kappa_1 \le \exp(\tfrac{6}{21}) \le \frac{3}{2}$.
\hfill $\square$

\end{itemize}

\medskip

Using this new notation and relationships between constants, 
we now give the overall template for proving claim (i)
of the proposition:

\paragraph{General template for proving (i).}
Our goal in proving (i) is to establish that
\begin{equation}
  \| \Delta_{t-1}\|_{z_t} \le C \cdot \|\Delta_t\|_{z_t}\;
  \label{eq:stability-i}
\end{equation}
for $C = 2$.
For this, observe by the triangle inequality that 
\begin{equation}
    \|\Delta_{t-1}\|_{z_t} 
    =  \|\Delta_{t} - (\Delta_t - \Delta_{t-1})\|_{z_t} 
    \le
    \|\Delta_{t}\|_{z_t} + 
    \|\Delta_{t} - \Delta_{t-1}\|_{z_t} \;.
    \label{eq:stability-i-01}
\end{equation}
Now suppose there exists an absolute constant
$B < 1/\eta$ such that 
\begin{equation}
  \|\Delta_{t} - \Delta_{t-1}\|_{z_t} 
  \le 
  \eta B \|\Delta_{t-1} \|_{z_t} \;.
  \tag{$\clubsuit$}
  \label{eq:stability-i-goal}
\end{equation}
Then it will follow from~\eqref{eq:stability-i-01} that
\begin{equation}
  \|\Delta_{t-1}\|_{z_t} 
  \le 
  \Big(\frac{1}{1-\eta B}\Big)
  \|\Delta_t\|_{z_t} \;,
  \label{eq:stability-i-final}
\end{equation}
which establishes the inequality (i) so long as $\eta B < 1/2$.
Thus, our goal is to prove the inequality~\eqref{eq:stability-i-goal}.

\paragraph{Helper inequalities.}
To prove the inequality~\eqref{eq:stability-i-goal}, 
we first establish the following helper inequalities 
that hold at general time indices. 
These follow primarily via application of
the norm transfer inequality of Lemma~\ref{lem:norm-transfer-gsc}
and the dual norm bound on the Hessian remainder
terms of Lemma~\ref{lem:hessian-remainder-bound}. 
We state and prove these inequalities in the following lemma: 

\begin{restatable}{lem}{lemhelperinequalities}
\label{lem:helper-inequalities}
  Let $\{z_t\}$ be the dual iterates of OMWU. 
  Then the following inequalities hold:

  \begin{enumerate}[
    label={({\alph*})},
    leftmargin=3.6em,
  ]

  \item\label{item:stability-helper-a}
  For all $k \ge 2$: \, 
  $\|J \nabla^2 F(z_{k-1}) \Delta_{k-1}\|_{z_k} 
  \le \sigma \kappa_1^2  \cdot \|\Delta_{k-1}\|_{z_k}$. 

  \item\label{item:stability-helper-b}
  For all $k \ge 3$: \, 
  $\|J \nabla^2 F(z_{k-2}) \Delta_{k-2}\|_{z_k} 
  \le \sigma \kappa_2^2  \cdot \|\Delta_{k-2}\|_{z_k}$.

  \item\label{item:stability-helper-c}
  For all $k \ge 3$: \, 
  $\|J G_F(z_{k-1}, z_{k-2})\|_{z_k} \le 
  \sigma r_1 \kappa_1 \cdot \|\Delta_{k-1}\|_{z_k}$.

  \item\label{item:stability-helper-d}
  For all $k \ge 4$: \, 
  $\|J G_F(z_{k-2}, z_{k-3})\|_{z_k} \le 
  \sigma r_2 \kappa_2 \cdot \|\Delta_{k-2}\|_{z_k}$.
  \end{enumerate}
  
\end{restatable}

\medskip

\noindent\textbf{Proof of Lemma~\ref{lem:helper-inequalities}}\;
We prove each of the inequalities separately:

\begin{itemize}[
  label={},
  leftmargin=*,
]
  \item\textit{Part}~\ref{item:stability-helper-a}.
  By Proposition~\ref{prop:omwu-dual-diff}, we have
  $\|z_k - z_{k-1}\|_\infty \le 3\eta \sigma$, and thus we can write
  \begin{align*}  
    \|J \nabla^2 F(z_{k-1}) \Delta_{k-1}\|_{z_k}
    &\le  
    \kappa_1 \cdot \|J \nabla^2 F(z_{k-1}) \Delta_{k-1}\|_{z_{k-1}} 
    &&(\text{by Lemma~\ref{lem:norm-transfer-gsc}}) \\
    &\le 
    \sigma \kappa_1  \cdot \|\Delta_{k-1}\|_{z_{k-1}} 
    &&(\text{by Lemma~\ref{lem:hessian-local-smooth}}) \\
    &\le 
    \sigma \kappa_1^2  \cdot \|\Delta_{k-1}\|_{z_{k}} 
    &&(\text{by Lemma~\ref{lem:norm-transfer-gsc}}) \,.
  \end{align*}

\item\textit{Part}~\ref{item:stability-helper-b}.
Follows identically to that of \ref{item:stability-helper-a}, 
but now using $\|z_k - z_{k-2}\|_\infty \le 6\eta \sigma$.

\item\textit{Part}~\ref{item:stability-helper-c}.
Using Proposition~\ref{prop:norm-relations}, we have 
\begin{align*}
  \|J G_F(z_{k-1}, z_{k-2})\|_{z_k} 
  &\le 
  \|J G_F(z_{k-1}, z_{k-2})\|_{2} \\
  &\le 
  \|J\|_2 \cdot  \|G_F(z_{k-1}, z_{k-2})\|_{2}  
  \le 
  \sigma  \cdot \|G_F(z_{k-1}, z_{k-2})\|_{z_{k-1}, *} \;.
\end{align*}

Moreover, as $\|z_{k-1}-z_{k-2}\| \le 3 \eta a_{\max}$, then
then applying Lemma~\ref{lem:hessian-remainder-bound} yields 
\begin{align*}
    \|G_F(z_{k-1}, z_{k-2})\|_{z_{k-1}, *}
    &\le 
    r_1  \|z_{k-1} - z_{k-2}\|_{z_{k-1}} 
    &&(\text{by Lemma~\ref{lem:hessian-remainder-bound}}) \\
    &\le 
    r_1 \kappa_1 \|\Delta_{k-1}\|_{z_k} 
    &&(\text{by Lemma~\ref{lem:norm-transfer-gsc}}) \;.
\end{align*}

\item\textit{Part}~\ref{item:stability-helper-d}.
Follows identically to that of \ref{item:stability-helper-c}, 
but now using $\|z_k - z_{k-2}\|_\infty \le 6\eta \sigma$.
\hfill $\square$
\end{itemize}

\smallskip 

\paragraph{Proof of Inequality~\eqref{eq:stability-i-goal}.}
We prove this inequality by induction on $t$
and with $B := 2(B_1 + B_2)$. 

\smallskip 
\noindent
\textbf{Base case}: At time $t=2$, recall fist by the initialization
of the dual OMWU iterates that we have 
$z_1 = z_0 - \eta J \nabla F(z_0)$. It follows that 
\begin{equation*}
  \Delta_2 - \Delta_{1}
  = 
  - \eta J(2\nabla F(z_1) - \nabla F(z_0))
  + \eta J \nabla F(z_0)
  = 
  - 2 \eta J g_1 \;,
\end{equation*}
which means that $\|\Delta_2 - \Delta_1\|_{z_2}  \le 2 \eta \|J g_1\|_{z_2}$.
Now by the definitions of $g_1$ and $\Delta_1$ and using the 
gradient difference expression of Proposition~\ref{prop:grad-diff-general}, 
we have 
\begin{equation}
  g_1 = \nabla F(z_1) - \nabla F(z_0)
  = \nabla^2 F(z_1)\Delta_1 + G_F(z_1, z_0)  \;.
\end{equation}
Thus we can write 
\begin{align}
  \|J g_1 \|_{z_2} 
  &\le 
    \|J \nabla^2 F(z_1) \Delta_1 \|_{z_2} 
    + 
    \|J G_F(z_1, z_0) \|_{z_2}  
    \label{eq:stability-i-base-01}  \\
  &\le
    \sigma \kappa_1^2 \|\Delta_1\|_{z_2} 
    + 
    \sigma r_1 \kappa_1 \|\Delta_1\|_{z_2} \;,
\end{align}
where the final inequality comes from 
applying Part~\ref{item:stability-helper-a} and 
Part~\ref{item:stability-helper-c} of 
Lemma~\ref{lem:helper-inequalities} to the two 
left-hand terms of~\eqref{eq:stability-i-base-01}, respectively.
Recalling the definition of $B_1$ from~\eqref{eq:stability-constants-B},
we thus conclude
\begin{equation*}
  \|\Delta_2 - \Delta_1\|_{z_2}  
  \le 
 \eta  \cdot 2 \|J g_1\|_{z_2}   
  \le
  \eta \cdot 2 B_1 \|\Delta_1\|_{z_2}
  \le 
\eta \cdot 2(B_1+B_2) \|\Delta_1\|_{z_2} \;.
\end{equation*}
Here, the inequality $2B_1 \le 2(B_1 + B_2) = B$
comes from Part (a)
of Proposition~\ref{prop:stability-constants-bounds}.
Thus, \eqref{eq:stability-i-goal} holds for the
base case at time $t=2$. 

\smallskip
\noindent
\textbf{Inductive step:}
Now suppose that~\eqref{eq:stability-i-goal} holds 
for all $2 \le k \le t-1$. 
To show the inequality holds also at time $t$, 
we start by repeating similar calculations as in 
proof of the base case. First, we decompose the
difference $\Delta_{t} - \Delta_{t-1}$ by 
\begin{align*}
  \Delta_t - \Delta_{t-1}
  &= 
  - \eta J \big(
  2\nabla F(z_{t-1}) - 2\nabla F(z_{t-2})
  - (\nabla F(z_{t-2}) - \nabla F(z_{t-3}))
  \big)    \\
  &= 
  - \eta J (2 g_{t-1} - g_{t-2})  \;.
\end{align*}
Taking local norms and using the triangle inequality,
we can then further write
\begin{equation}
  \|\Delta_t - \Delta_{t-1}\|_{z_t} 
  \le 
  2 \eta \|J g_{t-1}\|_{z_t} + \eta \|J g_{t-2}\|_{z_t} \;.
  \label{eq:stability-i-inductive-01}
\end{equation}
To bound $\|J g_{t-1}\|_{z_t}$, using the definitions 
of $g_{t-1}$ and $\Delta_{t-1}$, and using 
Proposition~\ref{prop:grad-diff-general} to express 
the difference of gradients, we have
\begin{equation*}
  g_{t-1} = \nabla F(z_{t-1}) - \nabla F(z_{t-2}) 
  = \nabla^2 F(z_{t-1}) \Delta_{t-1} + G_F(z_{t-1}, z_{t-2}) \;.
\end{equation*}
Thus it follows that 
\begin{align}
  \|J g_{t-1}\|_{z_t} 
  &\le 
  \|J \nabla^2 F(z_{t-1}) \Delta_{t-1}\|_{z_t}
  + 
  \|J G_F(z_{t-1}, z_{t-2})\|_{z_t}  
  \label{eq:stability-i-inductive-02} \\
  &\le 
  (\sigma \kappa^2_1 + \sigma r_1 \kappa_1 ) 
  \cdot \|\Delta_{t-1}\|_{z_t}  \\
  &= 
  B_1  \cdot \|\Delta_{t-1}\|_{z_t} \;,
  \label{eq:stability-i-inductive-02.5}
\end{align}
where the final inequality comes from applying 
Part~\ref{item:stability-helper-a} and Part~\ref{item:stability-helper-c} 
of Lemma~\ref{lem:helper-inequalities} to the two terms 
in~\eqref{eq:stability-i-inductive-02}, respectively. 
Now using an analogous setup for $\|J g_{t-2}\|_{z_t}$ 
and applying Part~\ref{item:stability-helper-b}
and Part~\ref{item:stability-helper-d} of 
Lemma~\ref{lem:helper-inequalities}, we similarly find
\begin{equation}
  \|J g_{t-2}\|_{z_t} 
  \le 
  (\sigma \kappa^2_2 + \sigma r_2 \kappa_2) \|\Delta_{t-2}\|_{z_t}  \;.
  \label{eq:stability-i-inductive-03}  
\end{equation}
Our task is now to further derive an upper bound on $\|J g_{t-2}\|_{z_t}$
in terms of the $\|\Delta_{t-1}\|_{z_t}$. 
For this, as we assume the inductive hypothesis holds at time $t-1$, 
we know
\begin{equation}
  \|\Delta_{t-1} - \Delta_{t-2}\|_{z_{t-1}} \le \eta B \|\Delta_{t-2}\|_{z_{t-1}} \;.
  \label{eq:stability-i-inductive-04}
\end{equation}
Thus by two applications of the local-norm transfer
of Lemma~\ref{lem:norm-transfer-gsc} and using 
the inductive hypothesis from~\eqref{eq:stability-i-inductive-04}, 
we can write
\begin{equation*}
	\|\Delta_{t-1} - \Delta_{t-2}\|_{z_t} 
	\le 
	\kappa_1 \cdot \| \Delta_{t-1} - \Delta_{t-2} \|_{z_{t-1}}
	\le 
	\kappa_1 \eta B \cdot \| \Delta_{t-2} \|_{z_{t-1}}
	\le 
	\kappa_1^2 \eta B \cdot \| \Delta_{t-2} \|_{z_t}   \;.
\end{equation*}
Then using the reverse triangle inequality, we have 
\begin{equation}
  \|\Delta_{t-1}\|_{z_t} 
  \ge 
  \|\Delta_{t-2}\|_{z_t} 
  - \|\Delta_{t-1} - \Delta_{t-2}\|_{z_t} 
  \ge 
  \big(1-\kappa^2_1 \eta B\big) \cdot \|\Delta_{t-2}\|_{z_t} \;.
\end{equation}
Now recalling $B := 2(B_1 + B_2)$, it follows
from Parts(b) and (c) in 
Proposition~\ref{prop:stability-constants-bounds}
that $\eta \kappa_1^2 B \le 1/2 < 1$ 
under the constraint that $\eta \le 1/(21\sigma)$.
Thus it follows that 
\begin{equation}
  \|\Delta_{t-2}\|_{z_t} 
  \le 
  \Big(\frac{1}{1-\kappa^2_1 \eta B} \Big)
  \cdot 
  \|\Delta_{t-1}\|_{z_t}  \;.
  \label{eq:stability-i-inductive-05}
\end{equation}
Ultimately, susbtituting this bound on $\|\Delta_{t-2}\|_{z_t}$
back into~\eqref{eq:stability-i-inductive-03} 
and recalling $B_2 := \sigma \kappa^2_2 + \sigma r_2 \kappa_2$,
we find 
\begin{equation}
   \|J g_{t-2}\|_{z_t} 
  \le 
  \Big(
    \frac{B_2}
    {1-\kappa^2_1 \eta B}    
  \Big) 
  \cdot 
  \|\Delta_{t-2}\|_{z_t} \;.
  \label{eq:stability-i-inductive-06}
\end{equation}
Together with the bound on $\|J g_{t-1}\|_{z_t}$ 
from~\eqref{eq:stability-i-inductive-02.5}, 
we conclude from~\eqref{eq:stability-i-inductive-01} that
\begin{align*}
	\|\Delta_t - \Delta_{t-1}\|_{z_t} 
	&\le 
	2 \eta \|J g_{t-1}\|_{z_t} + \eta \|J g_{t-2}\|_{z_t}   \\
	&\le   
	\eta \Big(
	2B_1 + \frac{B_2}{1-\kappa^2_1 \eta B}
	\Big) \cdot 
	\|\Delta_{z_{t-1}}\|_{z_t} 
	\le   
	\eta \cdot 2(B_1 + B_2) \cdot 
	\|\Delta_{z_{t-1}}\|_{z_t} \;.
\end{align*}
Here, the final inequality comes from the assumption 
that $\eta < 1/(21\sigma)$ and applying Part (b) 
of Proposition~\ref{prop:stability-constants-bounds}.
Thus we conclude that the inductive step of~\eqref{eq:stability-i-goal}
also holds for $B := 2(B_1+B_2)$, which proves
the inequality for all $t \ge 2$. 
Finally, substituting this setting of $B$ into~\eqref{eq:stability-i-final}
yields 
\begin{equation*}
    \| \Delta_{t-1}\|_{z_t} 
    \le 
    \big(
    \tfrac{1}{1-\eta \cdot 2 (B_1+B_2)}
    \big)
    \|\Delta_t\|_{z_t}
    \le 
    2 \|\Delta_t \|_{z_t} \;,
\end{equation*}
where the final inequality is due to $\eta B \le 1/2$, which holds 
by Parts (b) and (c)
of Proposition~\ref{prop:stability-constants-bounds}.
This completes the proof of part (i) of the proposition. 

\paragraph{Proof of claim (ii) of main proposition.}
To prove the second inequality of the proposition,
we apply the previous inequality (i) at time $t+1$. 
This yields
\begin{equation}
    \|z_{t} - z_{t-1}\|_{z_{t+1}}
    \le 
    2 \|z_{t+1} - z_t\|_{z_{t+1}} \;.
    \label{eq:stability-ii-01}
\end{equation}
As $\|z_{t+1}-z_t\|_\infty \le 3 \eta a_{\max}$, applying 
two norm transfers via Lemma~\ref{lem:norm-transfer-gsc} 
to both sides of~\eqref{eq:stability-ii-01} yields:
\begin{equation}
    \|z_{t+1} - z_t\|_{z_{t+1}}
    \le 
    \kappa_1 \|z_{t+1} -z_t \|_{z_t} 
    \;\;\text{and}\;\;
    \|z_t - z_{t-1}\|_{z_{t+1}} 
    \ge 
    \frac{1}{\kappa_1} \|z_t - z_{t-1} \|_{z_t}
    \;.
    \label{eq:stability-ii-02}  
\end{equation}
Combining~\eqref{eq:stability-ii-01} 
and~\eqref{eq:stability-ii-02} then yields 
\begin{equation*}
    \|z_t - z_{t-1} \|_{z_t}
    \le 
    2 (\kappa_1)^2 \cdot 
    \|z_{t+1} -z_t \|_{z_t}
    \le
    3  \|z_{t+1} -z_t \|_{z_t} \;,
\end{equation*}
where the final inequality is due to 
$\kappa_1 \le 3/2$, which comes from 
Part~(d) of 
Proposition~\ref{prop:stability-constants-bounds}
under the assumption that $\eta \le 1/(21\sigma)$.
This completes the proof of part (ii) of the proposition. 
\hfill $\blacksquare$

\subsection{Bound on Initial Change in Energy}
\label{app:energy-dissipation:initial}

By definition of~\eqref{eq:omwu-dual}, 
we have at time $t=1$ that $z_1 = z_0 - \eta J \nabla F(z_0)$.
In this initial step, the iterate $z_1$
is obtained only with a single skew-gradient
step, with no optimistic correction term.
Thus, the initial change in energy $F(z_1) - F(z_0)$
will be non-decreasing 
(although by Lemma~\ref{lem:energy-one-step-full},
$F(z_{t}) - F(z_{t-1}) < 0$
for all subsequent $t \ge 1$). 
In the following lemma, we thus derive 
a worst-case upper bound on this initial change 
in energy at the first iteration:

\smallskip

\begin{restatable}[Initial Change in Energy]
{lem}{leminitialenergychange}
\label{lem:initial-change-energy}
Let $\{z_t\}$ be the iterates of~\eqref{eq:omwu-dual}
with $\eta$ satisfying Assumption~\ref{ass:stepsize}.
Then at time $t=1$:
\begin{equation*}
	F(z_1) - F(z_0)
	\le 
	\tfrac{5}{8} \cdot \eta^2 \|J \nabla F(z_0)\|^2_{z} \;.
\end{equation*}
\end{restatable}

\begin{proof}
	By a first-order Taylor expansion of 
	$F(z_1)$ around $z_0$
	and using the dual OMWU update rule
	$z_1 = z_0 - \eta J \nabla F(z_0)$, 
	we have 
	\begin{align*}
		F(z_1) - F(z_0)
		&= 
		\langle \nabla F(z_0), z_1 - z_0 \rangle
		+ 
		D_F(z_1, z_0) \\
		&= 
		- \eta \langle
			\nabla F(z_0), J \nabla F(z_0)
		\rangle
		+ D_F(z_1, z_0) \\
		&= 
		D_F(z_1, z_0) \;.
	\end{align*}
	Here, the final equality comes from the skew-symmetry
	of $J = - J^\top$. 
	Now recall from Proposition~\ref{prop:omwu-dual-diff}
	that $\|z_1 - z_0\|_{\infty} \le 3 \eta \sigma$. 
	Then applying the bound on $D_F$ in local norm
	from Lemma~\ref{lem:bregman-gsc} and recalling
	the definition of $\mu(\cdot)$ from~\eqref{eq:def-mu-nu},
	we find
	\begin{equation*}	
		D_F(z_1, z_0)
		\le 
		\mu(6\eta \sigma)
		\cdot \|z_1 - z_0 \|^2_{z_0}
		= 
		\mu(6\eta \sigma)
		\cdot \eta^2 \|J \nabla F(z_0)\|^2_{z_0} \;.
	\end{equation*}
	Now if $\eta \le \tfrac{1}{L\sigma}$ for any $L > 0$, 
	the bound on $\mu(\cdot)$ from 
	Lemma~\ref{lem:exp-bounds} further gives
	\begin{equation*}
		\mu(6 \eta \sigma)
		\le 
		\tfrac{1}{2} + \tfrac{\exp(6/L)}{L} \;.
	\end{equation*}
	The second term in the above is strictly
	decreasing for all $L > 0$, and it is straightforward
	to check that 
	$\tfrac{1}{2} + \frac{\exp(6/L)}{L} \le \tfrac{1}{2} + \tfrac{1}{8} =\tfrac{5}{8}$
	for all $L \ge 72$.
	Thus for $\eta \le \tfrac{1}{72\sigma}$
	(which holds under Assumption~\ref{ass:stepsize}),
	we conclude that 
	\begin{equation*}
		F(z_1)
		- 
		F(z_0)
		\le 
		\tfrac{5}{8} \cdot \eta^2 \|J \nabla F(z_0)\|^2_{z_0} \;,
	\end{equation*}
	which concludes the proof.
\end{proof}


\section{Details on Non-Uniform Skew-Gradient Domination}
\label{app:dissipation-kl}

In this section, we develop the proof of 
Proposition~\ref{prop:dissipation-kl-bound},
which relates the \textit{dissipation term} $\|J \nabla F(z)\|^2_{z}$
from Lemma~\ref{lem:energy-one-step-full}
to the KL divergence $\KL(w^\star, w)$
via a \textit{skew-gradient domination} inequality.
To restate the proposition:

\smallskip 

\propdissipationklbound*

\smallskip

\paragraph{Oragnization of Section.}
The proof of Proposition~\ref{prop:dissipation-kl-bound}
first requires introducing several intermediate results, 
and thus this section is organized as follows:
\begin{itemize}[
    leftmargin=1em
]
\item
\textbf{Section~\ref{sec:prelims-spectral-params}}
introduces preliminaries on the 
the spectrum of the energy function hessian $\nabla^2 F$
and on the spectrum of the payoff operator $J$
when restricted to the subspace $\calS^\bot$.

\item
\textbf{Section~\ref{app:dissipation-kl:energy-gap}}
establishes a dual characterization 
of the $\KL \le \chi^2$ relationship
from Proposition~\ref{prop:kl-chisq}
that bounds the energy suboptimality gap 
in the effective dual space $\calZ$.

\item
\textbf{Section~\ref{sec:gradient-domination-proof}}
then gives the proof of Proposition~\ref{prop:dissipation-kl-bound}
using the preliminaries of Sections~\ref{sec:prelims-spectral-params}
and~\ref{app:dissipation-kl:energy-gap}.

\item
\textbf{Section~\ref{sec:primal-gradient-domination}}
presents a second proof of
the skew-gradient domination inequality, 
but using a primal perpsective (see 
Proposition~\ref{prop:gradient-dom-primal}). 
This section also derives a general (but crude) upper 
bound on the dissipation term $\|J \nabla F(z)\|^2_{z}$.

\item
\textbf{Section~\ref{sec:sigma-wmin-relationship}}
gives examples that demonstrate the independence
of the magnitude of $\sigma_{\min}$
with respect to the location of the interior NE $w^\star$.
\end{itemize}

\subsection{Preliminaries on Restricted Spectral Parameters}
\label{sec:prelims-spectral-params}

\subsubsection{Non-uniform local strong convexity in primal norm.}

As mentioned in Section~\ref{sec:omwu} 
(and explained further in 
Section~\ref{app:energy-kl-equiv:duality}), the 
energy function $F$ lacks global strict convexity
over $\R^{m+n}$. In particular, 
Proposition~\ref{prop:energy-linear} implies that
$F$ is affine over the subspace
of constant shifts $\calS$, which reflects the fact that 
$\Null(\nabla^2 F(z)) = \calS$ for all $z \in \R^{m+n}$.
However, when restricted to the orthogonal complement 
$\calS^\bot$, the energy Hessian $\nabla^2 F(z)$
is positive definite, and $F$ is strictly convex,
although the degree of its curvature is highly dependent on 
the local point $z$.
To quantify this dependency, we consider a local and 
non-uniform notion of 
restricted strong convexity that can be defined in terms of 
the minimum restricted eigenvalue $M(z)$ 
of $\nabla^2 F(z)$:

\smallskip

\begin{restatable}[Minimum restricted eigenvalue]
{define}{defminrestrictedeigenvalue}
	\label{def:local-norm-strong-convex}
	For $z \in \R^{m+n}$, let $M(z)$ be the scalar
	\begin{equation*}
		M(z) 
		:= 
		\inf_{v \in \calS^{\bot} \setminus \{0\}} \,
		\frac{\langle v, \nabla^2 F(z) v \rangle}{\|v\|^2_2} \;.
	\end{equation*}
\end{restatable}

By definition of the local norm $\|\cdot \|_z$ 
and the subspaces $\calS$ and $\calS^\bot$,
we then have the following:

\begin{restatable}{prop}{proplocalnormstrongconvex}
	\label{prop:local-norm-strong-convex}
	Fix $z \in \R^{m+n}$.
	Then for any $v \in \R^{m+n}$,
	it holds that
	\begin{equation*}
		\|v\|^2_z \ge M(z)
	 	\cdot \|\Pi_{\calS^{\bot}}(v)\|^2_2 \;.
	\end{equation*}
\end{restatable}

\begin{proof}
	Recall that every $v \in \R^{m+n}$ 
	has a unique orthogonal decomposition 
	with respect to $\calS$ and $\calS^\bot$, 
	where $v = \Pi_{\calS}(v) + \Pi_{\calS^\bot}(v)$,
	where $\Pi_{\calS}(\cdot)$ and $\Pi_{\calS^\bot}(\cdot)$ 
	denote the orthogonal	projections onto 
	$\calS$ and $\calS^\bot$, respectively.
	Then using this decomposition and also the 
	definition of the local norm $\|\cdot \|_z$,
	we have
	\begin{align*}
		\|v \|^2_{z}
		= 
		\langle v, \nabla^2 F(z) v \rangle
		&= 
		\big\langle (\Pi_{\calS}(v) + \Pi_{\calS^\bot}(v)),
		\nabla^2 F(z)
		( \Pi_{\calS}(v) + \Pi_{\calS^\bot}(v))
		\big\rangle \\
		&= 
		\big\langle
			\Pi_{\calS^\bot}(v),
			\nabla^2 F(z) 
			\Pi_{\calS^\bot}(v)
		\big\rangle \\
		&\ge 
		M(z) 
		\cdot 
		\|\Pi_{\calS^\bot}(v)\|^2_2 \;.
	\end{align*}
	Here, the second equality comes from 
	the fact that $\Null(\nabla^2 F(z)) = \calS$
	(part (iii) of Proposition~\ref{prop:energy-hessian}),
	and the symmetry of the $\nabla^2 F(z)$,
	and the final inequality comes from the definition
	of $M(z)$ in Definition~\ref{def:local-norm-strong-convex},
	which yields the claim.
\end{proof}

Moreover, using the characterization of the energy Hessian
from Proposition~\ref{prop:energy-hessian}, 
the minimum restricted eigenvalue $M(z)$ is exactly the minimum
coordinate of the primal variable $w = \nabla F(z)$ 
corresponding to the dual variable $z$:

\begin{restatable}[Minimum restricted eigenvalue 
is minimum primal coordinate]{prop}{propenergymz}
	\label{prop:energy-mz}
	Fix $z \in \R^{m+n}$,
	and let $w = (p, q) = \nabla F(z) \in \relint(\calW)$.
	Let $p_{\min} = \min_{i \in [m]} p(i)$,
	$q_{\min} = \min_{j \in [n]} q(j)$,
	and $w_{\min} := \min\{p_{\min}, q_{\min}\}$.
	Then $M(z) = w_{\min}$.
\end{restatable}

\begin{proof}
To prove the proposition, it suffices to show that,
when restricted to $\calS^\bot$, 
the minimum eigenvalue of $\nabla^2 F(z)$
is $w_{\min}$. Now recall by part (i) of 
Proposition~\ref{prop:energy-hessian} that 
\begin{equation*}
	\nabla^2 F(z) = 
	\begin{pmatrix}
		\Diag(p) - pp^\top & 0 \\
		0 & \Diag(q) - qq^\top 
	\end{pmatrix}\;.
\end{equation*}
Due to its separable block structure,
we will prove the corresponding block-wise inequalities:
e.g.,  when restricted to $\1_m^\bot$, 
the minimum eigenvalue of $\nabla^2 F_x(z) = \Diag(p) - pp^\top$
is $p_{\min}$, and when restricted to $\1_n^\bot$,
the minimum eigenvalue of $\nabla^2 F_y(z) = \Diag(q) - qq^\top$
is $q_{\min}$.

We start with the claim for $\nabla^2 F_x(z)$.
For this, let $\{\lambda^p_i\}$ denote the eigenvalues of
$\nabla^2 F_x(z)$, and observe by part (iv)
of Proposition~\ref{prop:energy-hessian} that
$0 \le \lambda^p_i \le 1$ for all $i \in [m]$.
Moreover, $\nabla^2 F(z)$ indeed has minimum eigenvalue 0
corresponding to eigenvector $\1_m$, since:
\begin{equation*}
	\nabla^2 F_x(z) \1_m 
	= 
	\Diag(p) \1_m - p \langle p, \1_m \rangle
	= 
	p - p
	=
	0 \;.
\end{equation*}
Thus, because $\nabla^2 F_x(z)$ is symmetric,
all other non-zero eigenvalues correspond
to eigenvectors orthogonal to $\1_m$
(i.e., belonging to $\1_m^\bot$).
Therefore, it suffices to prove that
all positive eigenvalues are bounded below by $p_{\min}$.
For this, let $v \in \1^\bot_{m} \setminus \{0\}$ and $\lambda > 0$
be such an eigenvector-eigenvalue pair, 
which means
$\Diag(p) v - p \langle p, v \rangle = \lambda v$.
By rearranging, we must have at each coordinate
\begin{equation}
	(p(i) - \lambda) \cdot v(i) = p(i)\cdot \langle p, v\rangle
	\quad\text{for all $i \in [m]$.}
	\label{eq:sp-01}
\end{equation}
Now assume by way of contradiction that 
$\lambda < p_{\min}$. This implies that 
$p(i) - \lambda > 0$ for all coordinates $i \in [m]$.
Then by rearranging~\eqref{eq:sp-01} and summing
over all $i \in [m]$, we find
\begin{equation}
	\textstyle
	\langle \1_m, v \rangle 
	= 
	\sum_{i=1}^m v(i)
	= 
	\langle p, v \rangle \cdot 
 	\sum_{i=1}^m \frac{p(i)}{p(i) - \lambda} \;.
 	\label{eq:sp-02}
\end{equation}
As $v \in \1^{\bot}_m$ by assumption,
we then must have $\langle \1_m, v \rangle = 0$,
meaning expression~\eqref{eq:sp-02} evaluates to 0.
However, the sum $\sum_{i=1}^m \frac{p(i)}{p(i) - \lambda} > 0$
by the assumption that $\lambda < p_{\min}$
(meaning $p(i) - \lambda > 0$ for all $i$)
and the fact that $p \in \relint(\Delta_m)$. 
Moreover, if $\langle p, v \rangle = 0$, 
the equality in expression~\eqref{eq:sp-01} further forces
$v(i) = 0$ at every coordinate $i \in [m]$,
which contradicts the assumption that $v \neq 0$.
Thus the equality in~\eqref{eq:sp-02} cannot be satisfied
when $\lambda < p_{\min}$,
and we conclude $\lambda \ge p_{\min}$.

Repeating an identical argument for $\nabla^2 F_y(z)$
similarly implies all positive eigenvalues
corresponding to eigenvectors in $\1_n^\bot$
are bounded below by $q_{\min}$.
By the block structure of $\nabla^2 F(z)$,
it follows that all positive eigenvalues of $\nabla^2 F(z)$
are bounded below by $w_{\min} = \min(p_{\min}, q_{\min})$.
Further using that $\calS^\bot = \1_m^\bot \times \1_n^\bot$,
it then follows by definition that $M(z) \ge w_{\min}$, 
which completes the proof.
\end{proof}

\subsubsection{Non-uniform local smoothness in dual norm.}

We also establish a local notion of smoothness
in the dual norm $\|\cdot \|_{z, *}$ when restricted to $\calS^\bot$. 
For this, we define maximum restricted eigenvalue
$L(z)$ of the inverse energy Hessian as follows:

\begin{restatable}[Maximum restricted eigenvalue of inverse energy Hessian]
{define}{defmaxrestrictedeigenvalueinverse}
	\label{def:local-dual-norm-smoothness}
	For $z \in \R^{m+n}$, let $L(z)$ be the scalar
	\begin{equation*}
		L(z) 
		:= 
		\sup_{v \in \calS^{\bot} \setminus \{0\}} \,
		\frac{\langle v, (\nabla^2 F(z))^{-1} \, v \rangle}{\|v\|^2_2} \;.
	\end{equation*}
\end{restatable}

Using the definition of the local dual norm $\|\cdot \|_{z, *}$,
the following proposition is then immediate:
\begin{restatable}{prop}{proplocaldualsmoothness}
	\label{prop:local-dual-smoothness}
	Fix $z \in  \R^{m+n}$.
	Then for any $v \in \calS^\bot$, it holds that
	$\|v\|^2_{z, *} \le L(z) \cdot \|v\|^2_2$.
\end{restatable}

Moreover, it follows similarly to Proposition~\ref{prop:energy-mz}
that $L(z)$ is characterized exactly in terms of the minimum 
coordinate of the primal variable $w$ corresponding to $z$:

\begin{restatable}{prop}{propenergyLz}
	\label{prop:energy-Lz}
	Fix $z \in \R^{m+n}$,
	and let $w = (p, q) = \nabla F(z) \in \relint(\calW)$.
	Let $p_{\min} = \min_{i \in [m]} p(i)$ and
	$q_{\min} = \min_{j \in [n]} q(j)$,
	and define $w_{\min} := \min(p_{\min}, q_{\min})$.
	Then $L(z) = \frac{1}{w_{\min}}$.
\end{restatable}

\begin{proof}
	By Proposition~\ref{prop:energy-hessian-inverse},
	we have that
	$(\nabla^2 F(z))^{-1} = \nabla^2 R(w) = \Diag(1/w)$
	when restricted to $\calS^\bot$.	
	Thus the maximum restricted eigenvalue of
	$L(z)$ is exactly the maximum eigenvalue of
	$\Diag(1/w)$, which is $1/w_{\min}$. 
	Alternatively, note in general that if
	$\nabla^2 F(z) \succeq c I$ (over the restricted subspace)
	for some $c > 0$,
	then $(\nabla^2 F(z))^{-1} \preceq \frac{1}{c} I$.
	Recall that in Proposition~\ref{prop:energy-mz}, 
	we established $\nabla^2 F(z) \succeq w_{\min} I$
	over $\calS^\bot$, and thus we have 
	$(\nabla^2 F(z))^{-1} \preceq \frac{1}{w_{\min}} I$
	over this subspace. 
\end{proof}

\subsubsection{Minimum restricted singular value of J.}

For the game's block skew-symmetric matrix payoff matrix 
$J \in \R^{(m+n)\times(m+n)}$ (as defined in~\eqref{eq:J-def}), 
recall that we define $\sigma_{\min} := \sigma_{\min}(J, \calS^\bot)$ 
to be the minimum singular value 
of $J$ when restricted to $\calS^\bot$.
This definition uses the projection operator 
$\Pi_{\calS^\bot}$ onto $\calS^\bot$,
which we first explicitly define:

\smallskip

\begin{restatable}[Projection onto $\calS^\bot$]
	{define}{defSbotprojection}
	\label{def:S-bot-projection}
	Let $\Pi_{\calS^\bot}: \R^{m+n} \to \calS^\bot$ 
	be the orthogonal projection operator onto $\calS^\bot$.
	Specifically, $\Pi_{\calS^\bot}$ is given by the matrix
	\begin{equation*}
		\Pi_{\calS^\bot} 
		:= 
		\begin{pmatrix}
			\Pi_{\1_m^\bot} & 0 \\
			0 & \Pi_{\1_n^{\bot}}
		\end{pmatrix} 
		:=
		\begin{pmatrix}
			 I_m - \frac{1}{m} \1_m \1^\top_m & 0 \\
			 0  & I_n - \frac{1}{n} \1_n \1^\top_n
		\end{pmatrix} \;.
	\end{equation*}
	Here, $\Pi_{\1_m^{\bot}}$ and $\Pi_{\1_n^{\bot}}$	
	are the blockwise projection matrices onto 
	$\1_m^\bot$ and $\1_n^\bot$, respectively. 
	By slight abuse of notation, for $v \in \R^{m+n}$,
	we write $\Pi_{\calS^\bot}(v)$
	to denote the matrix-vector product $\Pi_{\calS^\bot} v$.
\end{restatable}

Then recall from Section~\ref{sec:prelims}
that we define $\sigma_{\min}$ as follows:

\begin{restatable}[Minimum restricted
singular value of J]{define}{defrestrictedsvJ}
	\label{def:restricted-sv-J}
	Let $J \in \R^{(m+n) \times (m+n)}$ 
	be the matrix from~\eqref{eq:J-def}.
	Then let $\sigma_{\min}$ denote the minimum singular
	value of $J$ restricted to $\calS^\bot$:
	\begin{equation*}
		\sigma_{\min} := 
		\sigma_{\min}(J, \calS^\bot)
		= 
		\inf_{v \in \mathcal{S}^\bot \setminus \{0\}}\,
		\frac{\|\Pi_{\calS^\bot}(Jv)\|_2}{\|v\|_2} \;.
	\end{equation*}
\end{restatable}

When $A$ has a unique and interior Nash
equilibrium $w^\star$, then $\sigma_{\min} > 0$ 
is strictly positive:

\begin{restatable}{prop}{propgammapositive}
	\label{prop:gammapositive}
	Let $A \in \R^{m \times n}$ have a unique interior 
	Nash equilibrium $w^\star \in \relint(\calW)$,
	let $J$ be the corresponding
	block skew-symmetric matrix from~\eqref{eq:J-def},
	and let $\sigma_{\min}$ be the minimum restricted singular
	value of $J$ from Definition~\ref{def:restricted-sv-J}.
	Then $\sigma_{\min} > 0$.
\end{restatable}

\begin{proof}
	First, note by definition of $\sigma_{\min}$ 
	that we can equivalently write
	\begin{equation*}
		\textstyle
		\sigma_{\min} =
		\inf_{v \in \mathcal{S}^\bot \setminus \{0\}, \|v\|_2 =1} 
		\|\Pi_{\calS^\bot}(Jv)\|_2 \;.
	\end{equation*}
	It is straightforward to see that 
	$\sigma_{\min} = 0$ if and only if 
	$\Null(\Pi_{\calS^\bot} J)$ and $\calS^\bot$
	have a non-trivial intersection.
	Thus our goal is to establish that, when $A$ 
	has a unique interior Nash equilibrium, then 
	\begin{equation}
		\Null(\Pi_{\calS^\bot} J) \cap \calS^\bot = \{0\} \;.
		\label{eq:g-goal}
	\end{equation}
	For this, suppose $v \in \Null(\Pi_{\calS^\bot} J) \cap \calS^\bot$.
	This implies that the following two properties hold:
	\begin{enumerate}[
		label={(\roman*)},
		leftmargin=2.5em,
		topsep=0em,
	]
	\item
	As $v \in \calS^\bot$, then $\langle v, s \rangle = 0$
	for all $s \in \mathcal{S}$.
	\item
	As $v \in \Null(\Pi_{\calS^\bot} J)$,
	then by linearity 
	$0 = \Pi_{\calS^\bot} J v =  \Pi_{\calS^\bot}(Jv)$,
	which means $Jv \in \calS$.
	\end{enumerate}
	Recall that we assume $A$ has a unique
	and interior NE. 
 	Then applying Lemma~\ref{lem:unique-interior-NE-helper},
	we have that $v \in \calS^\bot$ and $Jv \in \calS$
	implies $v = 0$.
	Thus we conclude that if $v \in \Null(\Pi_{\calS^\bot} J) \cap \calS^\bot$,
	then $v = 0$. This establishes~\eqref{eq:g-goal}
	and concludes the proof. 
\end{proof}

\subsection{Properties of Energy Suboptimality Gap}
\label{app:dissipation-kl:energy-gap}

Recall that Proposition~\ref{prop:calZ-minimizer}
established that, over the effective dual space $\calZ$,
and under the assumption of a unique and interior NE, then 
for $z \in \calZ$ and $w = \nabla F(z)$:
\begin{equation}
	\KL(w^\star, w) 
	= 
	F(z) - \min\nolimits_{z' \in\calZ} F(z')
	= 
	F(z) - F(z^\star)  \;,
	\label{eq:energy-gap}
\end{equation}
where $z^\star \in \calZ$ is a point such that
$\nabla F(z^\star) = w^\star$
(and such a point $z^\star$ exists due to
Proposition~\ref{prop:calZ-grad-map}).
In the following proposition, using the primal relationship 
$\KL(w^\star, w) \le \chi^2(w^\star, w)$,  we then further
bound the energy suboptimality gap $F(z) - F(z^\star)$
in terms of the dual analogue of $\chi^2(w^\star, w)$.
Formally, we have the following relationship:

\begin{restatable}[Dual KL vs. $\chi^2$ relationship]
	{lem}{lemdualklchisquare}
	\label{lem:dual-kl-chisq}
	Let $A \in \R^{m \times n}$ have unique interior 
	Nash equilibrium $w^\star \in \relint(\calW)$, and let
	$z^\star \in \calZ$ such that $\nabla F(z^\star) = w^\star$.
	Fix $z \in \calZ$, and let $w = \nabla F(z)$. Then:
	\begin{equation*}
		F(z) - F(z^\star)
		\le
		\|\nabla F(z) - \nabla F(z^\star)\|^2_{z, *} \;.
	\end{equation*}
\end{restatable}

\begin{proof} 
	First, we have from~\eqref{eq:energy-gap} that
	$F(z) - F(z^\star) = \KL(w^\star, w)$. 
	Then by Proposition~\ref{prop:kl-chisq}:
	\begin{equation}
		F(z) - F(z^\star)
		= 
		\KL(w^\star, w) \le \chi^2(w^\star, w) 
		= 
		\sum\nolimits_{i=1}^{m+n}
		\frac{(w(i) - w^\star(i))^2}{w(i)} \;.
		\label{eq:dk-01}
	\end{equation}
	By definition of the dual local norm $\| \cdot \|^2_{z,*}$,
	we for $v \in \calS^\bot$ that 
	\begin{equation}
		\|v\|^2_{z, *}
		= 
		\langle 
		v, (\nabla^2 F(z))^{-1} v \rangle
		= 
		\sum\nolimits_{i=1}^{m+n}
		\frac{(v(i))^2}{(\nabla F(z))(i)} \;.
		\label{eq:dk-02}
	\end{equation}
	In the final equality, we use the characterization
	of the inverse $(\nabla^2 F(z))^{-1} = \Diag(1/(\nabla F(z)))$
	over $\calS^\bot$ from Proposition~\ref{prop:energy-hessian-inverse}.
	As $\nabla F(z) - \nabla F(z^\star) = w - w^\star \in \calS^\bot$,
	substituting $v = \nabla F(z) - \nabla F(z^\star)$ 
	in~\eqref{eq:dk-02} then exactly yields 
	the right hand side of~\eqref{eq:dk-01}.
\end{proof}

\subsection{Proof of Proposition~\ref{prop:dissipation-kl-bound}}
\label{sec:gradient-domination-proof}

We now give the proof of Proposition~\ref{prop:dissipation-kl-bound}
through the lens of non-uniform \textit{skew-gradient domination}.
Recall that we have $z \in \calZ$ and $w = \nabla F(z) \in \relint(\calW)$.
The proof uses the following four steps:

\smallskip

\noindent
\textbf{1. Non-uniform restricted strong convexity.}

\noindent
Using the definition of the minimum restricted
eigenvalue $M(z)$ from Definition~\ref{def:local-norm-strong-convex}, 
the inequality of Proposition~\ref{prop:local-norm-strong-convex},
and the characterization of $M(z)$ from 
Proposition~\ref{prop:energy-mz}, we have
\begin{equation*}
	\|J \nabla F(z)\|^2_z 
	\ge
	M(z) \cdot 
	\|\Pi_{\calS^\bot}\big(J \nabla F(z)\big)\|^2_2 \
	\ge
	w_{\min} \cdot \|\Pi_{\calS^\bot}\big(J \nabla F(z)\big)\|^2_2 \;.
\end{equation*}

\smallskip

\noindent
\textbf{2. Invariance to Nash and 
restricted spectrum of J.}

\noindent
Let $z^\star \in \calZ$ such that $\nabla F(z^\star) = w^\star$
(recall that such a $z^\star$ exists due to
Proposition~\ref{prop:calZ-grad-map}).
Now by Proposition~\ref{prop:interior-NE},
we have $J \nabla F(z^\star) = Jw^\star \in \calS$.
Thus by definition of the projection $\Pi_{\calS^\bot}$, 
we have $\Pi_{\calS^{\bot}}(J \nabla F(z^\star)) = 0$,
and therefore by linearity
\begin{equation*}
	\Pi_{\calS^\bot}(J \nabla F(z))
	= 
	\Pi_{\calS^\bot}(J (\nabla F(z) -\nabla F(z^\star)))\;.
\end{equation*}
Moreover, as $w, w^\star \in \calW$, it follows that
$\langle w - w^\star, s \rangle = 0$ for all $s \in \calS$,
and thus $w - w^\star = \nabla F(z)  - \nabla F(z^\star) \in \calS^\bot$.
Combining these pieces and using the definition
of the minimum restricted singular value $\sigma_{\min}$ from
Definition~\ref{def:restricted-sv-J}, we can then further write:
\begin{align*}
	\|\Pi_{\calS^\bot}\big(J \nabla F(z)\big)\|^2_2
	= 
	\|\Pi_{\calS^\bot}\big(J (\nabla F(z) - \nabla F(z^\star)\big)\|^2_2  
	\ge
	\sigma_{\min}^2 \cdot 
	\|\nabla F(z) - \nabla F(z^\star) \|^2_2 \;.
\end{align*}
Together with the inequality of Step 1, this means
\begin{equation}
	\|J \nabla F(z)\|^2_z \ge 
	\sigma_{\min}^2 \cdot w_{\min} \cdot 
	\|\nabla F(z) - \nabla F(z^\star) \|^2_2 \;.
	\label{eq:gd-01}
\end{equation}

\smallskip

\noindent
\textbf{3. Bounding the dual suboptimality gap.}

\noindent
We now derive the following lower bound on
$\| \nabla F(z) - \nabla F(z^\star)\|^2_2$
in terms of the dual gap $F(z) - F(z^\star)$.
We proceed in two steps. 
First, we have from Lemma~\ref{lem:dual-kl-chisq}
the dual relationship between 
$\KL(w^\star, w)$ and $\chi^2(w^\star, w)$,
which gives
\begin{equation*}
	F(z) - F(z^\star) 
	\le 
	\| \nabla F(z) - \nabla F(z^\star) \|^2_{z, *} \;.
\end{equation*}
Next, by the restricted local smoothness properties
of Proposition~\ref{prop:local-dual-smoothness} and
Proposition~\ref{prop:energy-Lz}, we further have
\begin{equation*} 
	F(z) - F(z^\star)
	\le 
	\|\nabla F(z) - \nabla F(z^\star)\|^2_{z, *}
	\le 
	\frac{1}{w_{\min}} \cdot
	\|\nabla F(z) - \nabla F(z^\star)\|^2_2 \;.
\end{equation*}
Rearranging  then yields
\begin{equation}
	\|\nabla F(z) - \nabla F(z^\star)\|^2_2 
	\ge 
	w_{\min} \cdot \big(F(z) - F(z^\star)) \;.
	\label{eq:gd-02}
\end{equation}

\smallskip

\noindent
\textbf{4. Equivalence between primal and
dual suboptimality gaps.}

\noindent
Combining~\eqref{eq:gd-01} from Step 2
and~\eqref{eq:gd-02} from Step 3 yields
\begin{equation}
	\|J \nabla F(z) \|^2_{z} 
	\ge 
	\sigma_{\min}^2  \cdot w_{\min}^2  \cdot
	(F(z) - F(z^\star)) \;.
	\label{eq:gd-03}
\end{equation}
This is exactly a non-uniform skew-gradient-domination property for 
the energy function $F$.
To relate $\|J \nabla F(z)\|^2_z$ back to the primal space,
we apply the equivalence of Proposition~\ref{prop:calZ-minimizer}
(also restated in~\eqref{eq:energy-gap}), which establishes 
for $z \in \calZ$ and $w = \nabla F(z)$ that
\begin{equation}
	F(z) - F(z^\star) = 
	F(z) -\min\nolimits_{z' \in \calZ}  F(z')
	= 
	\KL(w^\star, w) \;.
	\label{eq:gd-04}
\end{equation}
Combining~\eqref{eq:gd-03} with~\eqref{eq:gd-04}
yields the statement of the lemma
and concludes the proof. 
\hfill $\blacksquare$

\subsection{Primal Viewpoint of Non-Uniform Skew-Gradient Domination}
\label{sec:primal-gradient-domination}

The proof of Proposition~\ref{prop:dissipation-kl-bound}
in Section~\ref{sec:gradient-domination-proof} is written
purely in the dual perspective. In this section, we provide
an alternative proof of the inequality from a primal perspective.
This primal viewpoint stems from the structure of the 
energy Hessian $\nabla^2 F(z)$ from part (ii) of
Proposition~\ref{prop:energy-hessian}. This implies 
that the dissipation term $\|J \nabla F(z)\|^2_z$
is the \textit{variance} of the payoff vector 
$J \nabla F(z)$ under the component-wise pair of distributions
$w = (p, q) = \nabla F(z)$. 

The proof and slightly refined result 
via this second perspective is 
useful in later establishing the
uniform best-iterate duality gap 
bound of Theorem~\ref{thm:dg-best-2x2}.
We state and prove the result
further below in Proposition~\ref{prop:gradient-dom-primal}.
However, we first introduce several additional preliminaries
used in the proof:

\paragraph{Component-wise restricted
spectrum of A.}
First, we define the following component-wise analogues
of the minimum restricted singular value 
$\sigma_{\min}(J, \calS^\bot)$ from Definition~\ref{def:restricted-sv-J}.
Specifically, we define $\sigma_{\min, m}$ and $\sigma_{\min, n}$
as follows:
\begin{equation}
	\begin{aligned}
	\sigma_{\min, n}
	&:= 
	\sigma_{\min}(A, \1^{\bot}_n)
	= 
	\inf_{v \in \1^\bot_{n} \setminus \{0\}} \,
	\frac{\|\Pi_{\1_m^{\bot}}(A v)\|_2}{\|v\|_2} \\
	\text{and}\quad
	\sigma_{\min, m}
	&:= 
	\sigma_{\min}(A^\top, \1^{\bot}_m)
	= 
	\inf_{u \in \1^\bot_{m} \setminus \{0\}} \,
	\frac{\|\Pi_{\1_n^{\bot}}(A^\top u)\|_2}{\|u\|_2}  \;.
	\end{aligned}
	\label{eq:component-rsv-A}
\end{equation}

\smallskip 

We have the following relationship between
$\sigma_{\min}(J, \calS^\bot)$
and the component-wise analogues in~\eqref{eq:component-rsv-A}:

\smallskip

\begin{restatable}{prop}{propcomponentminrestrictedsv}
	\label{prop:component-restricted-sv}
	Let $A \in \R^{m \times n}$,
	and consider the quantities $\sigma_{\min}$
	from Definition~\ref{def:restricted-sv-J}
	and $\sigma_{\min, m}$ and $\sigma_{m, n}$
	from~\eqref{eq:component-rsv-A}.
	Then
	$
	\sigma_{\min}
	= 
	\min\big\{
	\sigma_{\min,m}, \, \sigma_{\min, n}
	\}$ .
\end{restatable}

\begin{proof}
	Fix $x = (u, v) \in \R^{m+n} \in \calS^\bot \setminus \{0\}$.
	By definition of $\calS^\bot$, 
	this means $u \in \1_{m}^{\bot}$
	and $v \in \1_n^{\bot}$. Moreover, by the structure of
	of $\Pi_{\calS^\bot}$ from Definition~\ref{def:S-bot-projection},
	we have:
	\begin{equation*}
		\Pi_{\calS^\bot}(Jx) 
		= 
		\begin{pmatrix}
			\Pi_{\1_m^\bot} & 0 \\
			0 & \Pi_{\1_n^{\bot}}
		\end{pmatrix} 
		\begin{pmatrix}
			A v \\ 
			- A^\top u
		\end{pmatrix} 
		= 
		\begin{pmatrix}
			\Pi_{\1_m^\bot}(A v) \\
			\Pi_{\1_m^\bot}(-A^\top u) 
		\end{pmatrix} \;.
	\end{equation*}
	Taking squared Euclidean norms, this means
	that both
	\begin{equation*}
			\|x\|^2_2 = \|u\|^2_2 + \|v\|^2_2 \\
			\quad\text{and}\quad
			\|\Pi_{\calS^\bot}(Jx) \|^2_2
			= 
			\|\Pi_{\1_m^\bot}(A v)\|^2_2
			+ 
			\|\Pi_{\1_n^\bot}(A^\top u)\|^2_2 \;.
	\end{equation*}
	Moreover, by definition of $\sigma_{\min}$,
	we have
	\begin{equation}
		\sigma^2_{\min}
		= 
		\inf_{x = (u, v) \in \calS^\bot \setminus \{0\}}\,
		\frac{\|\Pi_{\1_m^\bot}(A v)\|^2_2
			+ 
			\|\Pi_{\1_n^\bot}(A^\top u)\|^2_2
		}{\|u\|^2_2 + \|v\|^2_2} \;.
		\label{eq:sv-01}
	\end{equation}

	Now let $\sigma_{*} = \min\{\sigma_{\min, m}, \sigma_{\min_n}\}$.
	To prove the claim, we show in two parts that
	both $\sigma_{\min} \ge \sigma_{*}$ and $\sigma_{\min} \le \sigma_{*}$.
	For the first direction, observe by the definitions of
	$\sigma_{\min, m}$ and $\sigma_{\min, n}$ in
	\eqref{eq:component-rsv-A} that both
	\begin{equation*}
		\begin{aligned}
			\|\Pi_{\1^\bot_m}(Av)\|^2_2 
			&\ge 
			\sigma^2_{\min,n} \cdot \|v\|^2_2 
			\ge \sigma^2_{*} \cdot \|v\|^2_2 
			\\
			\text{and}
			\quad
			\|\Pi_{\1^\bot_n}(A^\top u)\|^2_2 
			&\ge 
			\sigma^2_{\min,m} \cdot \|u\|^2_2 
			\ge \sigma^2_{*} \cdot \|u\|^2_2 
			\;,
		\end{aligned}
	\end{equation*}
	which implies that 
	\begin{equation*}
	\|\Pi_{\calS^\bot}(Jx) \|^2_2
	\ge \sigma^2_{*} \cdot 
	\big(
	\|u\|^2_2 + \|v\|^2_2
	\big)
	= 
	\sigma^2_{*}\cdot \|x\|^2_2 \;.
	\end{equation*}
	Dividing by $\|x\|^2$, taking an infimum
	over all $\calS^\bot \setminus \{0\}$, 
	and taking square roots then yield
	$\sigma_{\min} \ge \sigma_{*}$.
 
	For the reverse direction, observe from
	expression~\eqref{eq:sv-01} that by 
	further restricting the domain, 
	the infimum is non-decreasing, and thus
	\begin{equation*}
		\sigma^2_{\min}
		\le
		\inf_{x = (u, v) \in (0, \1_{n}^\bot) \setminus \{0\}}\,
		\frac{\|\Pi_{\1_m^\bot}(A v)\|^2_2
			+ 
			\|\Pi_{\1_n^\bot}(0)\|^2_2
		}{\|0\|^2_2 + \|v\|^2_2} 
		= 
		\sigma^2_{\min, n}
		\;,
	\end{equation*}
	where the final equality follows by definition of 
	$\sigma_{\min, n}$ in~\eqref{eq:component-rsv-A}.
	Taking square roots, we thus have 
	$\sigma_{\min} \le \sigma_{\min, n}$.
	By an identical calculation, we also have
	$\sigma_{\min} \le \sigma_{\min, m}$.
	Together, this yields 
	$\sigma_{\min} \le \min\{\sigma_{\min, n}, \sigma_{\min, m}\}$,
	which concludes the proof.
\end{proof}

\smallskip 

\paragraph{Variational characterization of variances.}
We also use the following 
variational characterization of variance terms 
$\Var_p(\cdot)$ and $\Var_q(\cdot)$
from~\eqref{eq:variance-def}, 
which we state and prove here:

\begin{restatable}{prop}{propvarmin}
	\label{prop:variance-variational}
    For any $v = (v_x, v_y) \in \R^{m+n}$,
    the following hold:
    \begin{equation*}
    	\begin{aligned}
        \Var_{p}(v_x)  
        &= 
        {\textstyle 
        \min_{c \in \R} \, \sum_{i=1}^m p(i) \cdot \big(v_x(i) - c\big)^2 
        }\\
        \quad\text{and}\quad
        \Var_{q}(v_y)
        &= 
        {\textstyle
        \min_{c \in \R} \, \sum_{j=1}^n q(j) \cdot \big(v_y(j) - c\big)^2 
        } \;.
        \end{aligned}
    \end{equation*}
\end{restatable}

\begin{proof}
    We prove the claim for $\Var_p (v_x)$, and the claim for
    $\Var_q(v_y)$ will follow by identical calculations.
    For this, define the scalar function 
    $F(c) = \sum_{i=1}^m p(i) \cdot \big(v_x(i) - c \big)^2$. 
    Using the fact that $\sum_{i=1}^m p(i) = 1$, expanding
    the definition of $F(c)$ yields
    \begin{equation*}
    	\textstyle
        F(c)
        = 
        \sum_{i=1}^m 
        p(i) \cdot \big((v_x(i))^2 - 2c v_x(i) + c^2\big) 
        = 
        c^2 - 2 c \langle p, v_x \rangle  + 
        \sum_{i=1}^m p(i) \cdot (v_x(i))^2 \;.
    \end{equation*} 
    Thus $F$ is convex in $c$, and differentiating
    shows $F$ is minimized at 
    $c^* = \langle p, v_x\rangle$. It follows that
    \begin{equation*}
    	\textstyle
        \min_{c \in \R}
        F(c) = 
        F(c^*)
        =
        \sum_{i=1}^m p(i) \cdot (v_x(i))^2 
        - \big(\langle p, v_x\rangle \big)^2
        = 
        \Var_p(v_x) \;,
    \end{equation*}
    which concludes the proof.
\end{proof}

\medskip

We now state following proposition, which establishes
the non-uniform skew-gradient domination result of
Proposition~\ref{prop:dissipation-kl-bound}
from a primal viewpoint:

\medskip

\begin{restatable}[
	Primal view of non-uniform
	skew-gradient domination]
{prop}{lemgradientdomprimal}
	\label{prop:gradient-dom-primal}
	Let $A \in \R^{m \times n}$ have a unique
	and interior Nash equilibrium 
	$w^\star = (p^\star, q^\star) \in \relint(\calW)$. 
	Fix $z \in \R^{m+n}$, and let 
	$w = (p, q) = \nabla F(z)$. 
	Let $p_{\min} = \min_{i \in [m]}p(i)$
	and $q_{\min} = \min_{j \in [n]} q(j)$.
	Let $\sigma_{\min, m}$ 
	and $\sigma_{\min, n}$ be defined
	as in~\eqref{eq:component-rsv-A},
	and let $\sigma_{\min}$ be
	the value from Definition~\ref{def:restricted-sv-J}.
	Then the following hold:
	\begin{enumerate}[
		label={(\roman*)},
		topsep=0em,
		leftmargin=3em,
	]
	\item
	$\|J \nabla F(z)\|^2_{z} = \Var_{p}(Aq) + \Var_q(A^\top p)$.

	\item
	$\Var_p(Aq) 
	\ge \sigma^2_{\min, n} \cdot p_{\min} 
	\cdot \|q - q^\star \|^2_2$ 
	and 
	$\Var_q(A^\top p ) 
	\ge \sigma^2_{\min, m} \cdot q_{\min} 
	\cdot \|p - p^\star \|^2_2$ 

	\item
	$\|p- p^\star\|^2_2 \ge p_{\min} \cdot \KL(p^\star, p)$
	and 
	$\|q- q^\star\|^2_2 \ge q_{\min} \cdot \KL(q^\star, q)$.
	\end{enumerate}
	Furthermore, it then holds that:
	\begin{equation*}
		\|J \nabla F(z)\|^2_{z} 
		\ge 
		\sigma^2_{\min} \cdot p_{\min} \cdot q_{\min} 
		\cdot \KL(w^\star, w) \;.
	\end{equation*}
\end{restatable}

\subsubsection{Proof of Proposition~\ref{prop:gradient-dom-primal}}

Granting parts (i), (ii), and (iii) of the lemma 
as true for now, the proof of the final statement 
follows by combining these inequalities. Specifically, we find
\begin{align*}
	\|J \nabla F(z) \|^2_{z}
	&\ge 
	\min\{\sigma^2_{\min, m}, \sigma^2_{\min ,n}\} 
	\cdot p_{\min} \cdot q_{\min}
	\cdot 
	\big(
		\KL(p^\star, p) + \KL(q^\star, q)
	\big) \\
	&= 
	\sigma^2_{\min} 
	\cdot p_{\min} \cdot q_{\min}
	\cdot 
	\KL(w^\star, w) \;.
\end{align*}
Here, the final equality comes from the definition
of $\KL(w^\star, w)$, and from 
Proposition~\ref{prop:component-restricted-sv},
which established that 
$\sigma_{\min} = \min\{\sigma_{\min, m}, \sigma_{\min, n}\}$.
Thus it remains to prove the first three statements of the lemma,
which we do in independent steps here:

\smallskip
\noindent
\textbf{Proof of part (i).} 

\noindent
Using the definition of $J$, the relationship 
$w = (p, q) = \nabla F(z)$, the local norm $\|\cdot \|_z$,
and the characterization of part (ii) in 
Proposition~\ref{prop:energy-hessian}, we have
\begin{align*}
	\|J \nabla F(z) \|^2_{z}
	&= 
	\langle J \nabla F(z), \nabla^2 F(z) J \nabla F(z)\rangle \\
	&= 
	\Var_{p}(Aq) + \Var_q(-A^\top p)  
	=  
	\Var_{p}(Aq) + \Var_q(A^\top p)  \;.
\end{align*}
Here, the final equality uses the fact that
$\Var_q(v) = \Var_q(-v)$ for any $v \in \R^n$. 

\smallskip

\noindent
\textbf{Proof of part (ii).} 

\noindent
Using the variational characterization of variance 
from Proposition~\ref{prop:variance-variational}, 
we have
\begin{align}
	\Var_p(Aq)
	&=
	\min\nolimits_{c \in \R} \,
	\ssum\nolimits_{i =1}^m p(i) \cdot 
	\big(
		(Aq - c \1_m)(i)
	\big)^2 \nonumber \\
	&\ge 
	p_{\min} \cdot 
	\min\nolimits_{c \in \R} \,
	\ssum\nolimits_{i =1}^m
	\big(
		(Aq - c \1_m)(i)
	\big)^2
	= 
	p_{\min} \cdot 
	\|\Pi_{\1^{\bot}_m}(Aq) \|^2_2 \;.
	\label{eq:p2-01}
\end{align}
Here, the final equality comes from the
definition of the length of the orthogonal projection
onto $\1_m^{\bot}$.
Moreover, by Proposition~\ref{prop:interior-NE},
we have that $Aq^\star = c \cdot \1_m$ for some
constant $c \in \R$, and thus $\Pi_{\1_m^\bot}(Aq^\star) = 0$.
Thus, it follows from the linearity  of $\Pi_{\1_m^\bot}$ that
we can write
\begin{equation}
	\Pi_{\1_m^\bot}(Aq)
	= 
	\Pi_{\1_m^\bot}(Aq)  - 
	\Pi_{\1_m^\bot}(Aq^\star) 
	= 
	\Pi_{\1_m^\bot}(A(q - q^\star)) \;.
	\label{eq:p2-02}
\end{equation}
Observe also that since $q, q^\star \in \Delta_n$,
we have $\langle q - q^\star, s \rangle = 0$
for all $s \in \Span(\1_n)$, which means by 
definition that  $q-q^\star \in \1_n^\bot$.
Using the definition of $\sigma_{\min, n}$
from~\eqref{eq:component-rsv-A},
it then follows from~\eqref{eq:p2-02} that
\begin{equation}
	\| \Pi_{\1_m^\bot}(Aq) \|^2_2
	\ge 
	\sigma^2_{\min, n} \cdot \|q - q^\star \|^2_2 \;.
	\label{eq:p2-03}
\end{equation}
Combining~\eqref{eq:p2-03} with~\eqref{eq:p2-01}
then yields 
$\Var_p(Aq) \ge \sigma^2_{\min, n} 
\cdot p_{\min}
\cdot \|q - q^\star \|^2_2$.
Following identical calculations, we similarly conclude 
$\Var_q(A^\top p) \ge \sigma^2_{\min, m} 
\cdot q_{\min}
\cdot \|p - p^\star \|^2_2$.

\smallskip

\noindent
\textbf{Proof of part (iii).} ~

\noindent
Using the relationship $\KL(p^\star, p) \le \chi^2(p^\star, p)$,
we can write and further bound
\begin{equation*}
	\KL(p^\star, p) \le \chi^2(p^\star, p)
	= 
	\sum_{i \in [m]} 
	\frac{(p(i) - p^\star(i))^2}{p(i)} 
	\le	
	\frac{1}{p_{\min}} \cdot 
	\sum_{i \in [m]} 
	(p(i) - p^\star(i))^2
	= 
	\frac{1}{p_{\min}} \cdot \| p - p^\star \|^2_2 \;.
\end{equation*}
Rearranging yields 
$\|p - p^\star\|^2_2 \ge p_{\min} \cdot \KL(p^\star, p)$.
Similarly, following identical calculations,
we also conclude that
$\|q - q^\star\|^2_2 \ge q_{\min} \cdot \KL(q^\star, q)$,
which concludes the proof.
\hfill $\blacksquare$

\subsubsection{Uniform Upper Bound on Dissipation Term}
\label{app:skew-grad-primal:uniform-upper}

Here, we additionally establish the following uniform
\textit{upper bound} on the dissipation term $\|J \nabla F(z)\|^2_{z}$
in terms of the KL divergence to Nash:

\smallskip 

\begin{restatable}[Uniform Upper Bound on Dissipation Term]
    {prop}{propskewgraduniformupper}
    \label{prop:skew-grad-uniform-upper}
	Fix $A \in \R^{m \times n}$ with unique and interior 
	NE $w^\star =  (p^\star, q^\star) \in \relint(\calW)$.
	Let $z \in \R^{m+n}$ and $w = (p, q) = \nabla F(z) \in \relint(\calW)$.
	Recall that $\sigma_{\max} = \|A \|_2$.
	Then:
	\begin{equation*}
		\|J \nabla F(z)\|^2_{z}
		\le 
		2 \cdot \sigma_{\max} \cdot \KL(w^\star, w) \;.
	\end{equation*}
\end{restatable}

\begin{proof}
	The proof of the upper bound uses
	the primal perspective established by 
	Proposition~\ref{prop:gradient-dom-primal}.
	For this, recall by Part (i) of that proposition that
	\begin{equation*}
		\|J \nabla F(z)\|^2_{z}
		=
		\Var_p(Aq) + \Var_q(A^\top p) \;.
	\end{equation*}
	Now by Part (i) of Proposition~\ref{prop:interior-NE},
	recall that $Aq^\star = d \1_m$
	and $A^\top p^\star = d \1_n$ for some
	constant $d \in \R$.
	Then by the variational characterization of variance from
	Proposition~\ref{prop:variance-variational}, 
	observe that 
	\begin{align}
		\Var_p(Aq)
		&=
		\min\nolimits_{c \in \R} \,
		\ssum\nolimits_{i =1}^m p(i) \cdot 
		\big(
			(Aq - c \1_m)(i)
		\big)^2 \nonumber \\
		&\le 
		\ssum\nolimits_{i=1}^m
		p(i) \cdot 
		\big((A(q-q^\star))(i)\big)^2  \nonumber \\
		&\le
		\|A(q-q^\star)\|^2_2  \le 
		\sigma^2_{\max} 
		\cdot \|q -q^\star\|^2_2 
		\label{eq:var-upperb}
	\end{align}
	Here, the final inequality follows by definition of $\sigma_{\max} = \|A\|_2$.
	As $\|u\|^2_2 \le \|u\|^2_1$ for any $u \in \R^m$, 
	and recalling that $\TV(q^\star,q) = \tfrac{1}{2} \|q^\star - q\|_1$,
	we then further have by Pinsker's inequality
	\begin{equation}
		\|q-q^\star\|^2_2
		\le 
		\|q-q^\star\|^2_1
		= 
		4 \TV(q^\star, q)^2 
		\le 
		2\KL(q^\star, q) \;.
		\label{eq:var-upper-pinsker}
	\end{equation}
	Combining expressions~\eqref{eq:var-upperb}
	and~\eqref{eq:var-upper-pinsker} and repeating
	an identical calculation for $\Var_q(A^\top p)$,
	we thus conclude that 
	\begin{equation*}
		\|J \nabla F(z)\|^2_{z}
		\le
		2 \sigma^2_{\max}
		\cdot 
		\big(
		\KL(q^\star, q)		
		+ \KL(p^\star, p)
		\big)
		= 
		2 \cdot \sigma^2_{\max} \cdot \KL(w^\star, w) \;,
	\end{equation*}
	which yields the desired claim.
\end{proof}

\smallskip

Note that this upper bound holds uniformly
over the dual and primal spaces and does 
not contain multiplicative factors depending on the local 
spectrum of $\nabla^2 F(z)$.
However, we will later prove in
Proposition~\ref{prop:kl-dissipation-upper-2x2}
a much tighter \textit{non-uniform}
upper bound on $\|J \nabla F(z)\|^2_{z}$ for the $2\times 2$ setting.
In particular, this latter bound \textit{does} contain
factors depending on the
minimum coordinates $p_{\min}, q_{\min}$
that essentially matches the lower
bound on $\|J \nabla F(z)\|^2_{z}$
from Proposition~\ref{prop:gradient-dom-primal}.
See Section~\ref{app:kl-last-lower}.

\subsection{Independence Between
Restricted Spectrum of $J$ and Location of Nash Equilibrium}
\label{sec:sigma-wmin-relationship}

In this section, we show that 
the magnitude of the
minimum restricted value $\sigma_{\min}$of $J$
is independent from the 
location of the game's Nash equilibrium.
In particular, we show examples 
where (i) the interior NE is uniform
but the $\sigma_{\min}$
can be arbitrarily small,
and (ii) where the interior NE of the
game is arbitrarily close to a vertex,
but $\sigma_{\min}$ is constant.

In the context of Proposition~\ref{prop:dissipation-kl-bound},
these examples also imply a certain independence
between $\sigma_{\min}$
and the local strong convexity parameter
$w_{\min}$ of the energy function: even
if a primal iterate $w$ is 
very close to the simplex boundary
(which one would expect for the OMWU iterates 
if the NE $w^\star$ is near the boundary),
the value of $\sigma_{\min}$ can be still be 
an absolute constant. 
For this reason, the result of
Proposition~\ref{prop:dissipation-kl-bound}
indicates that the more salient geometric bottleneck
in the skew-gradient domination inequality 
comes from the restricted strong convexity 
parameter of $F$, and not from $\sigma_{\min}$.

We proceed with the two examples:

\subsubsection{Example: Uniform Nash, 
but Arbitrarily Small $\sigma_{\min}$.}

Consider the scaled Matching Pennies game
with payoff matrix:
\begin{equation}
	A = 
	\begin{pmatrix}
		\eps & -\eps \\
		-\eps & \eps
	\end{pmatrix}
	\quad\text{for $\eps \in [-1, 1]$.}
	\label{eq:scaled-mp}
	\tag{$\eps$-Scaled Matching Pennies}
\end{equation}

\begin{restatable}[Minimum Restricted Singular
Value of $\eps$-Scaled MP]
{prop}{propscaledmprsv}	
	\label{prop:scaled-mp-rvs-example}
	Fix any $\eps \in [-1,1]$, and let $A \in \R^{2\times 2}$
	be the~\eqref{eq:scaled-mp}
	payoff matrix. Let
	$\sigma_{\min} = \min\{\sigma_{\min, n}, \sigma_{\min,n}\}$
	be the minimum restricted singular value 
	from Definition~\ref{def:restricted-sv-J}
	with $\sigma_{\min, n}$ and $\sigma_{\min,n}$
	as in~\eqref{eq:component-rsv-A}.
	Then the following hold:
	\begin{enumerate}[
		label={(\roman*)},
		itemsep=0em,
		topsep=0em,
		leftmargin=3em
	]
		\item
		The unique Nash equilibrium $w^\star = (p^\star, q^\star)$
		of $A$ is given by 
		$p^\star = q^\star = (\frac{1}{2}, \frac{1}{2})$.
		\item
		$\sigma_{\min} = 2\eps$ \;.
	\end{enumerate}
\end{restatable}

\begin{proof}
We prove the two claims of the proposition separately:

\paragraph{Proof of part (i).} ~

\noindent
Observe by definition of $A$ that
\begin{equation*}
A \begin{pmatrix}
	\tfrac{1}{2} \\ 
	\tfrac{1}{2}
\end{pmatrix}
= 
A^\top 
\begin{pmatrix}
	\frac{1}{2} \\ 
	\frac{1}{2}
\end{pmatrix}
= 
\begin{pmatrix}
		\eps & -\eps \\
		-\eps & \eps
\end{pmatrix}
 \begin{pmatrix}
	\frac{1}{2} \\ 
	\frac{1}{2}
\end{pmatrix}
= 
\begin{pmatrix} 
	\frac{1}{2} (\eps - \eps) \\
	\frac{1}{2} (- \eps + \eps)
\end{pmatrix}
= 
\begin{pmatrix}
	0 \\ 0
\end{pmatrix}\;.
\end{equation*}
By Part (ii) of Proposition~\ref{prop:interior-NE}, this means 
$w^\star = ((\frac{1}{2}, \frac{1}{2}), (\frac{1}{2}, \frac{1}{2}))$
is a Nash equilibrium of $A$. 
Since $A \in \R^{2 \times 2}$
and $w^\star$ is interior, we further have by
Part (ii) of Proposition~\ref{prop:2x2-interior-unique} that
$w^\star$ is unique.

\paragraph{Proof of part (ii).} ~ 

\noindent
Let $\1 = (1, 1)$.
Observe that in this two-dimensional setting,
the orthogonal subspace $\1^\bot$
is given by vectors $c \cdot (1, -1)$ for $c \in \R$.
Thus, fix any $c \in \R$, and let $v = c (1, -1)$.
We can compute
\begin{equation*}
	Av = 
	\begin{pmatrix}
		\eps & -\eps \\
		-\eps & \eps
	\end{pmatrix}
	\begin{pmatrix}
		c \\ -c
	\end{pmatrix} 
	= 
	\begin{pmatrix}
		2c \eps \\
		- 2 c \eps
	\end{pmatrix}
	= 
	(2c \eps) \cdot 
	\begin{pmatrix}
		1 \\ -1 
	\end{pmatrix} \;.
\end{equation*}
Observe by definition that $Av \in \1^\bot$.
Thus taking norms, we have
\begin{equation}
	\frac{\|\Pi_{\1^\bot}(Av) \|_2}{\|v\|_2}
	= 
	\frac{\|(Av) \|_2}{\|v\|_2}
	= 
	\frac{(2c\eps)\cdot \sqrt{2}}{c \sqrt{2}}
	= 2 \eps \;.
	\label{eq:e-01}
\end{equation}
As~\eqref{eq:e-01} holds for any $c \in \R$,
and since $A = A^\top$, 
it follows that $\sigma_{\min, n} = \sigma_{\min, m} = 2 \eps$.
\end{proof}

\subsubsection{Example: 
	Vanishing Minimum Nash Coordinate,
	but Constant $\sigma_{\min}$.}

Consider the payoff matrix $A_\delta \in \R^{2 \times 2}$
given by
\begin{equation}
	A_\delta = 
	\begin{pmatrix}
		\delta^2 & -\delta(1-\delta) \\
		- \delta(1-\delta) & (1-\delta)^2
	\end{pmatrix}
	\quad
	\text{for $\delta \in (0, 0.5]$.}
	\label{eq:A2x2-sym}
\end{equation}

\begin{restatable}[Minimum Restricted Singular Value 
of $A_{\delta}$]
	{prop}{propsmallnashlargersv}
	\label{prop:small-nash-large-rsv}
	Fix $\delta \in (0, 0.5]$,
	and let $A_\delta \in \R^{2\times 2}$ be the 
	payoff matrix from~\eqref{eq:A2x2-sym}. Let
	$\sigma_{\min} = \min\{\sigma_{\min, n}, \sigma_{\min,n}\}$
	be the minimum restricted singular value 
	from Definition~\ref{def:restricted-sv-J}
	with $\sigma_{\min, n}$ and $\sigma_{\min,n}$
	as in~\eqref{eq:component-rsv-A}.
	Then the following hold:
	\begin{enumerate}[
			label={(\roman*)},
			itemsep=0em,
			topsep=0em,
			leftmargin=3em
	]
		\item
		The unique Nash equilibrium $w^\star = (p^\star, q^\star)$
		of $A$ is given by $p^\star = q^\star = (1-\delta, \delta)$.
		\item
		$\sigma_{\min} = \frac{1}{2}.$

	\end{enumerate}
\end{restatable}

\begin{proof} 
We again prove the two claims of the proposition separately.

\paragraph{Proof of part (i).} ~

\noindent
Observe by definition of $A_\delta = A^\top_{\delta}$ that
\begin{align*}
	A_\delta 
	\begin{pmatrix}
		1-\delta \\
		\delta
	\end{pmatrix}
	= 
	A^\top_\delta 
	\begin{pmatrix}
		1-\delta \\
		\delta
	\end{pmatrix}
	&= 
	\begin{pmatrix}
		\delta^2 & -\delta(1-\delta) \\
		- \delta(1-\delta) & (1-\delta)^2
	\end{pmatrix}
	\begin{pmatrix}
		1-\delta \\
		\delta
	\end{pmatrix} \\
	&= 
	\begin{pmatrix}	
		\delta^2 (1-\delta) - \delta^2 (1-\delta) \\
		-\delta (1-\delta)^2 + \delta (1-\delta)^2 
	\end{pmatrix}
	= 
	\begin{pmatrix}
		0 \\ 0
	\end{pmatrix}
	\;.
\end{align*}
By Part (ii) of Proposition~\ref{prop:interior-NE}
and Part (ii) of Proposition~\ref{prop:2x2-interior-unique},
this means $w^\star = ((1-\delta, \delta), (1-\delta, \delta))$
is the unique interior Nash equilibrium of $A_\delta$.

\paragraph{Proof of part (ii).} ~

\noindent
Similarly to the proof of part (ii) of 
Proposition~\ref{prop:scaled-mp-rvs-example}, 
let $\1 = (1, 1)$, and thus $v \in \1^\bot$
when $v = c (-1, 1)$ for some $c \in \R$.
For such $v$, we can compute
using the definition of $A_{\delta}$ that
\begin{equation*}
	A_\delta v = A_\delta^\top v
	= 
	\begin{pmatrix}
	\delta^2 & -\delta (1-\delta) \\
	-\delta(1-\delta) & (1-\delta)^2
	\end{pmatrix}
	\begin{pmatrix}
		c \\ -c
	\end{pmatrix}
	= 
	c 
	\begin{pmatrix}
		\delta \\
		\delta-1
	\end{pmatrix} \;.
\end{equation*}
Moreover, recalling from Definition~\ref{def:S-bot-projection}
that $\Pi_{\1^\bot} = I - \frac{1}{2} \1 \1^\top$, 
it follows that 
\begin{equation*}
	\Pi_{\1^\bot}(A_\delta v)
	= 
	c \cdot 
	\begin{pmatrix}
	\delta - \frac{1}{2} (2\delta - 1) \\
	\delta -1 - \frac{1}{2} (2\delta - 1)
	\end{pmatrix}
	= 
	c \begin{pmatrix}
		\frac{1}{2} \\ - \frac{1}{2}
	\end{pmatrix} \;.
\end{equation*}
Taking norms, we find for any $c \in \R \setminus \{0\}$ that
\begin{equation*}
	\frac{\|\Pi_{\1^\bot}(A_\delta v)\|_2}{\|v\|_2}
	= 
	\frac{|c|/\sqrt{2}}{|c|\sqrt{2}}
	= 
	\frac{1}{2} \;.
\end{equation*}
It then follows by definition that
$\sigma_{\min} = \frac{1}{2}$.
\end{proof}


\section{Details on Universal Last-Iterate Convergence in KL}
\label{app:kl-last-iterate}

This section gives the proofs of 
Theorem~\ref{thm:asymptotic-convergence}
and Theorem~\ref{thm:kl-last-unified},
which establish asymptotic last-iterate convergence
and a linear last-iterate convergence rate
in $\KL$ divergence, respectively.

\paragraph{Oranization of section.}
The section is organized as follows:
\begin{itemize}[
    leftmargin=1em
]
\item
\textbf{Section~\ref{app:kl-last-iterate-initial}}
first proves a helper lemma that 
establishes an upper bound on the
initial change in $\KL(w^\star, \cdot)$
over the first step of the algorithm.
\item
\textbf{Section~\ref{app:kl-last-iterate:rate}}
gives the proof of Theorem~\ref{thm:asymptotic-convergence}
(asymptotic last-iterate convergence to NE).
\item
\textbf{Section~\ref{app:kl-last-iterate:rate}}
gives the proof of the 
Theorem~\ref{thm:kl-last-unified}
(new, linear last-iterate convergence rate in KL).
\item
\textbf{Section~\ref{app:wei-comparisons}}
gives a brief comparison with the prior
result of~\cite{wei2021linear}.
\item
\textbf{Section~\ref{app:kl-last:numerical}}
presents several numerical simulations.

\end{itemize}

\subsection{Initial Change in KL Divergence}
\label{app:kl-last-iterate-initial}

We start by proving the following lemma:

\smallskip

\begin{restatable}[Bound on Initial Change in KL]
{lem}{leminitialchangekl}
    \label{lem:kl-init-upperbound}
    Fix $A \in \R^{m \times n}$ with 
    unique and interior NE $w^\star \in \relint(\calW)$.
    Let $\{w_t\}$ denote the iterates
    of~\eqref{eq:omwu-primal}
    and let $\{z_t\}$ denote the iterates
    of~\eqref{eq:omwu-primal}
    with $\eta$ satisfying Assumption~\ref{ass:stepsize}.
    Recall $\sigma_{\max} = \|A\|_2$.
    Then
    $\KL(w^\star, w_1) \le 
    \tfrac{5}{4} \cdot \KL(w^\star, w_0)$.
\end{restatable}

\smallskip

\begin{proof}
    First, recall by Corollary~\ref{cor:change-kl-energy-equiv-omwu} 
    the equivalence
    \begin{equation}
        \KL(w^\star, w_{1}) - \KL(w^\star, w_0)
        = 
        F(z_1) - F(z_0) \;.
        \label{eq:kl-init-01}
    \end{equation}
    Then applying the upper bound on $F(z_1) - F(z_0)$
    from Lemma~\ref{lem:initial-change-energy}
    and the uniform upper bound of 
    Proposition~\ref{prop:skew-grad-uniform-upper},
    we find
    \begin{equation}
        F(z_1) - F(z_0)
        \le \tfrac{5}{8} \cdot \eta^2 \|J\nabla F(z_0)\|^2_{z}
        \le \tfrac{10}{8} \cdot \sigma^2_{\max} \cdot \eta^2 \cdot \KL(w^\star, w_0) 
        \le \tfrac{1}{4} \cdot \KL(w^\star, w_0) \;.
        \label{eq:kl-init-02}
    \end{equation}
    Here, the final inequality is due to 
    $\eta \le \tfrac{1}{4(54 \sigma_{\max} + 9)}
    \le 
    \tfrac{1}{4(54 \sigma_{\max})}$
    by Assumption~\ref{ass:stepsize},
    which implies crudely that 
    $\frac{10}{8} \sigma^2_{\max} \eta^2 \le \tfrac{1}{4}$.
    Then substituting~\eqref{eq:kl-init-02} into~\eqref{eq:kl-init-01}
    and rearranging further gives
    \begin{equation*}
        \KL(w^\star, w_1)
        \le 
        \KL(w^\star, w_0) \cdot \big(1 + \tfrac{1}{4}\big) \;,
    \end{equation*}
    which completes the proof.
\end{proof}

\subsection{Proof of Theorem~\ref{thm:asymptotic-convergence} -- 
Asymptotic Last-Iterate Convergence}
\label{app:kl-last-iterate:asymptotic}

We now give a proof of asymptotic
last-iterate convergence to a unique and interior NE.
We first restate the theorem:

\smallskip

\thmasymptotic*

\smallskip

\begin{proof}
Let $\{z_t\}$ denote the dual iterates of OMWU. 
The proof follows via the following steps:

\smallskip
\noindent
\textbf{1. Primal iterates lie in compact sublevelset}:

\noindent
Combining Corollary~\ref{cor:change-kl-energy-equiv-omwu}
and the upper bound of Lemma~\ref{lem:energy-one-step-full}, 
we have for all $t \ge 1$:
\begin{equation*}
    \KL(w^\star, w_{t+1}) - \KL(w^\star, w_{t})
    = 
    F(z_{t+1}) - F(z_t) 
    \le 
    - \tfrac{1}{20} \eta^2 \|J \nabla F(z_t)\|^2_{z_t} \le 0 \;.
\end{equation*}
Moreover, by Lemma~\ref{lem:kl-init-upperbound},
we have $\KL(w^\star,w_1) \le (5/4) \cdot \KL(w^\star, w_0)$.
Together, this means that 
$\KL(w^\star, w_t) \le (5/4) \cdot \KL(w^\star, w_0)$
for all $t \ge 0$.
Letting $\kappa = (5/4) \cdot \KL(w^\star, w_0)$,
this implies all iterates $w_t$
lie in the compact sublevelset 
\begin{equation*}
\calU = \{w \in \relint(\calW) \;:\; 
\KL(w^\star, w) \le \kappa\} \subset \relint(\calW) \;.
\end{equation*}

\smallskip
\noindent
\textbf{2. Dissipation term is zero only at Nash}:

\noindent
For $z \in \R^{m+n}$, let $w = (p, q) = \nabla F(z) \in \relint(\calW)$.
Observe by Part (i) of Proposition~\ref{prop:gradient-dom-primal}
that  $\|J \nabla F(z)\|^2_{z} = \Var_p(Aq) + \Var_q(A^\top p)$. 
By definition of variance (see expression~\eqref{eq:variance-def}),
note that $\Var_p(v) = 0$ and  $\Var_q(u) = 0$ if and only if 
$v = c \1_m$ and $u = d \1_n$ for some $c, d \in \R$. 
Together, this implies that 
$\|J \nabla F(z)\|^2_{z} = 0$ if and only if both 
$Aq$ and $A^\top p$ are constant vectors. 
Due to the uniqueness of $w^\star$,
this means by Part (ii) of Proposition~\ref{prop:interior-NE}
that $\|J \nabla F(z) \|^2_{z} = 0$ if and only if $\nabla F(z) = w = w^\star$.
Now for $w=(p,q)$ let $D(w) = (\eta^2/20)\cdot (\Var_p(Aq) + \Var_q(A^\top p))$.
Then for all $t \ge 1$:
\begin{equation}
    \KL(w^\star, w_{t+1}) - \KL(w^\star, w_t) \le - D(w_t) 
    \quad\;\text{and}\;
    \begin{cases}
    D(w_t)
    \ge 0 \;\text{for all $w_t$} \\
    D(w_t) = 0 \;\text{iff $w_t = w^\star$} \;.
    \end{cases}
    \label{eq:klasym-01}
\end{equation}

\noindent
\textbf{3. Subsequence and global convergence}:

\noindent
As all $w_t \in \calU$ for the compact set $\calU$,
we have by the Bolzano-Weierstrass theorem
that every infinite subsequence of $\{w_t\}$
has at least one limit point $w_{\infty} \in \calU$.
Moreover, by the continuity of $D(\cdot)$, 
if $w_{t_k} \to w_{\infty}$ for some subsequence 
$\{w_{t_k}\}$, 
then also $D(w_{t_k}) \to D(w_{\infty})$. 
Now as $\calU$ is compact,~\eqref{eq:klasym-01}
implies that $\KL(w^\star, w_{t_k})$ converges
to a finite limit, which means that $D(w_{t_k}) \to 0$.
Thus,we must also have $D(w_{\infty}) = 0$,
which further implies that $w_{\infty} = w^\star$.
Because every infinite subsequence 
converges to the same limit point $w_{\infty} = w^\star$,
then the entire sequence $\{w_t\}$ converges to $w^\star$.
\end{proof}

\subsection{Proof of Theorem~\ref{thm:kl-last-unified} -- 
Universal Last-Iterate Converence Rate}
\label{app:kl-last-iterate:rate}

We first restate the theorem:

\smallskip

\thmkllastunified*
 
\medskip

\begin{proof}
The proof follows from the
two main previously established properties:
energy dissipation (Lemma~\ref{lem:energy-one-step-full})
and skew-gradient domination (Proposition~\ref{prop:dissipation-kl-bound}).
We proceed to outline these steps.

\smallskip
\noindent
\textbf{1. Change in KL is change in energy}:

\noindent
First, let $\{z_t\}$ denote the dual OMWU iterates. 
Recall from Corollary~\ref{cor:change-kl-energy-equiv-omwu}
that for all $t \ge 0$, we have
\begin{equation}
    \KL(w^\star, w_{t+1})
    - \KL(w^\star, w_t)
    = 
    F(z_{t+1}) - F(z_t) \;.
    \label{eq:kllast-01}
\end{equation}

\smallskip
\noindent
\textbf{2. Upper bound on energy dissipation}:

\noindent
From Lemma~\ref{lem:energy-one-step-full}, we have 
for all $t \ge 1$ under the constraint on $\eta$ that
\begin{equation}
    F(z_{t+1}) - F(z_t)
    \le 
    - \tfrac{1}{20} \cdot \eta^2 \|J \nabla F(z_t)\|^2_{z_t} \;.
    \label{eq:kllast-02}
\end{equation}

\smallskip
\noindent
\textbf{3. Non-uniform skew-gradient domination}:

\noindent
By Proposition~\ref{prop:dissipation-kl-bound}, 
we further have the lower bound
\begin{equation}
    \|J \nabla F(z_t) \|^2_{z_t}
    \ge 
    \sigma^2_{\min} \cdot w^2_{t, \min} 
    \cdot \KL(w^\star, w_t) \;,
    \label{eq:kllast-03}
\end{equation}
where $\sigma_{\min} > 0$ under the
unique and interior Nash equilibrium assumption. 

Combining expressions~\eqref{eq:kllast-01},~\eqref{eq:kllast-02},
and~\eqref{eq:kllast-03}, we find for all $t \ge 1$
\begin{equation*}
    \KL(w^\star, w_{t+1})
    - 
    \KL(w^\star, w_t)
    \le 
    - \tfrac{1}{20} \eta^2 \cdot \sigma^2_{\min} w^2_{t, \min} 
    \cdot
    \KL(w^\star, w_t) \;.
\end{equation*}
Rearranging then yields
\begin{equation}
    \KL(w^\star, w_{t+1})
    \le 
    \KL(w^\star, w_t)
    \cdot \big(
    1
    - \tfrac{1}{20} \eta^2 \cdot \sigma^2_{\min} w^2_{t, \min}  
    \big) \;,
    \label{eq:kllast-p1}
\end{equation}
which yields statement (1) of the theorem.

\smallskip

For the second statement, 
we require establishing a lower bound on 
the coordinates $w_{t, \min}$
that holds \textit{uniformly} for all iterates $t \ge 0$.
For this, we state and prove the following 
lemma, which establishes a
worst-case uniform bound of $w_{t, \min} \ge \exp(-2\Lambda/\delta)$,
where $\Lambda = \KL(w^\star, w_0) - R(w^\star)$
is the cross-entropy between $w^\star$
and the initialization $w_0 \in \relint(\calW)$.

\begin{restatable}[Uniform lower bound on minimum coordinates]
{lem}{lemuniformlbcoordinates}
    \label{lem:wtmin-bound-uniform}
    Assume the setting of Theorem~\ref{thm:kl-last-unified}.
    Then for all $t \in [T]$ it holds that
    $w_{t, \min} \ge \exp(\frac{-2\Lambda}{\delta})$.
\end{restatable}

We note that the statement and proof of the lemma 
follows similarly to that of Lemma 19 in~\cite{wei2021linear},
but generalized to hold under arbitrary initializations.

\paragraph{Proof of Lemma~\ref{lem:wtmin-bound-uniform}}
First, we have by
Corollary~\ref{cor:change-kl-energy-equiv-omwu},
Lemma~\ref{lem:energy-one-step-full},
and Lemma~\ref{lem:initial-change-energy}
that, for all $t \ge 1$:
\begin{equation}
    \KL(w^\star, w_t) 
    \le 
    \KL(w^\star, w_{t-1})
    \le 
    \dots 
    \le 
    \KL(w^\star, w_1)
    \le
    \big(\tfrac{5}{4}\big)\cdot
    \KL(w^\star, w_0) \;.
    \label{eq:kl-decrease}
\end{equation}
Now for $w = (p, q)$, let $\ent(w) = \ent_m(p) + \ent_n(q)$
be the sum of component-wise entropies.
Moreover, recall that $R(w) = - \ent(w)$,
and let $\crossent(w^\star, w)$ denote the cross
entropy between $w^\star$ and $w$ defined by
\begin{equation*}
    \crossent(w^\star, w) 
    = 
    \KL(w^\star, w) - R(w^\star) \;.
\end{equation*}
Then using~\eqref{eq:kl-decrease}, we can write
\begin{align}
\crossent(w^\star, w_t) 
    &= 
    \KL(w^\star, w_t) - R(w^\star) \nonumber \\
    &\le  
    \big(\tfrac{5}{4}\big) \cdot \KL(w^\star, w_0) - R(w^\star) \\
    &\le
    \big(\tfrac{5}{4}\big) \cdot \big(\KL(w^\star, w_0) - R(w^\star)\big)
    = 
    \big(\tfrac{5}{4}\big) \cdot \crossent(w^\star, w_0) \;.
    \label{eq:ckl-01}
\end{align}
Here, the second inequality comes from the fact
that $-R(w^\star) \ge 0$. 
Now, let $i_{\min} \in \argmin_{i \in [m]} p_t(i)$, 
and $j_{\min} \in \argmin_{j \in [n]} q_t(j)$
denote the indices of $p_t$ and $q_t$ with
smallest mass, respectively.
Thus $p_{t, \min} = p_t(i_{\min})$ and $q_{t, \min} = q_{t}(j_{\min})$.
Then by definition of $\crossent(w^\star, w_t)$, we further have
\begin{align}
    \crossent(w^\star, w_t)
    &= 
    \sum_{i=1}^m p^\star(i)\cdot \log\Big(\frac{1}{p_t(i)}\Big)
    + 
    \sum_{j=1}^n q^\star(j) \cdot \log\Big(\frac{1}{q_t(j)}\Big) \nonumber \\
    &\ge 
    p^\star(i_{\min}) \cdot \log \bigg(\frac{1}{p_{t, \min}}\bigg)
    + 
    q^\star(j_{\min}) \cdot \log \bigg(\frac{1}{q_{t, \min}}\bigg) \;.
    \label{eq:ckl-02}
\end{align}
Observe that both terms in~\eqref{eq:ckl-02} are positive. 
Now let $\Lambda := \crossent(w^\star, w_0)$ 
and $c = \tfrac{5}{4}$.
Then combining~\eqref{eq:ckl-01} and~\eqref{eq:ckl-02}, 
we have
\begin{equation*}
    p^\star(i_{\min}) \log \left(\frac{1}{p_{t, \min}}\right) 
    \le 
    c \Lambda
    \implies 
    \log \left(\frac{1}{p_{t, \min}}\right) 
    \le \frac{c \Lambda}{p^\star(i_{\min})}
    \le 
    \frac{c \Lambda}{\delta_p}
    \le
    \frac{c \Lambda}{\delta}
     \;,
\end{equation*}
where the final two inequalities are due to 
$p^\star(i_{\min}) \ge \delta_p \ge \delta$.
Rearranging, we find 
\begin{equation}
    \textstyle
    p_{t, \min} \ge
    \exp\left(\frac{- c \Lambda}{\delta} \right)   \;.
\end{equation}
By an identical calculation, we also have
$q_{t, \min} \ge \exp(\frac{-c \Lambda}{\delta})$.
As $w_{t, \min} = \min\{p_{t, \min}, q_{t,\min}\}$,
it follows that 
$w_{t, \min} \ge \exp(\frac{- c \Lambda}{\delta})$,
which holds for all $t \ge 1$.
Noting the crude bound of $c=\tfrac{5}{4} \le 2$
completes the proof.
~\hfill $\square$

\smallskip

We now return to the main proof of 
Theorem~\ref{thm:kl-last-unified}.
Picking up from expression~\eqref{eq:kllast-p1},
we can apply the bound of Lemma~\ref{lem:wtmin-bound-uniform}
to write for all $t \ge 0$ that
\begin{align*}
    \KL(w^\star, w_{t+1})
    &\le 
    \KL(w^\star, w_{t})
    \cdot \big(
    1
    - \tfrac{1}{20} \eta^2 \cdot \sigma^2_{\min} w^2_{t, \min}  
    \big) \\
    &\le 
    \KL(w^\star, w_{t})
    \cdot \big(
    1
    - \tfrac{1}{20} \eta^2 \cdot \sigma^2_{\min} \exp(-2\Lambda/\delta)
    \big)  \\
    &\le 
    \KL(w^\star, w_{t})
    \cdot \exp\big(
    - \tfrac{1}{20} \eta^2 \cdot \sigma^2_{\min} \exp(-2\Lambda/\delta)
    \big) \;,
\end{align*}
where the final inequality is due to $1-u \le \exp(-u)$ for all $u \in \R$.
Inductively, for $t \ge 1$ we have:
\begin{equation}
    \KL(w^\star, w_{t+1})
    \le
    \KL(w^\star, w_{1})
    \cdot \exp\big(
    - \tfrac{1}{20} \eta^2 \cdot \sigma^2_{\min} \exp(-2\Lambda/\delta) \cdot t
    \big) 
    \;.
    \label{eq:kl-rate-last}
\end{equation}
Finally, recall from Lemma~\ref{lem:initial-change-energy} that 
$\KL(w^\star, w_1) \le 
\big(\tfrac{5}{4}\big) \cdot \KL(w^\star, w_0)$.
Substituting this bound into~\eqref{eq:kl-rate-last}
and using the crude bound $\tfrac{5}{4} \le 2$
then establishes for all $t+1 \ge 1$ that
\begin{equation*}
    \KL(w^\star, w_{t+1})
    \le 
    2 \KL(w^\star, w_0)
    \cdot 
      \exp\big(
    - \tfrac{1}{20} \eta^2 \cdot \sigma^2_{\min} \exp(-2\Lambda/\delta) \cdot t
    \big)  \;.
\end{equation*}
This yields statement (2) of the theorem 
and concludes the proof.
\end{proof}

\smallskip

\paragraph{Improved dependence
    on local state in Part (1) of Theorem~\ref{thm:kl-last-unified}.}
We note that, as a corollary of 
the ``primal'' variant proof of non-uniform
skew-gradient domination from 
Proposition~\ref{prop:gradient-dom-primal},
we can obtain a slightly improved bound 
on the multiplicative one-step change in $\KL$. 
In particular, assuming the setting of 
Theorem~\ref{thm:kl-last-unified} and 
comparing the skew-gradient domination bounds
of Proposition~\ref{prop:dissipation-kl-bound} 
(proven using primarily the ``dual" perspective)
and Proposition~\ref{prop:gradient-dom-primal}
(proven using a slightly more involved ``primal'' perspective),
we have:
\begin{equation*}
\begin{aligned}
\text{Proposition~\ref{prop:dissipation-kl-bound}}:
\quad
&\|J \nabla F(z)\|^2_{z} 
\ge \sigma^2_{\min} \cdot  w^2_{\min} \cdot \KL(w^\star, w) \\
\text{Proposition~\ref{prop:gradient-dom-primal}}:
\quad
&\|J \nabla F(z)\|^2_{z} 
\ge \sigma^2_{\min} \cdot p_{\min} \cdot q_{\min} \cdot \KL(w^\star, w)  \;.
\end{aligned}
\end{equation*}
Thus the bound of Proposition~\ref{prop:gradient-dom-primal}
can be larger than
that of Proposition~\ref{prop:dissipation-kl-bound}
when only one of $p_{\min}, q_{\min}$ is small.
Repeating an identical argument as 
in the first three steps of the 
proof of Theorem~\ref{thm:kl-last-unified}, 
we then have the following corollary:

\smallskip

\begin{restatable}[Improved one-step multiplicative
decrease in KL]
{cor}{corimprovedklonestep}
\label{cor:kl-last-improved-one-step}
Assume the setting of Theorem~\ref{thm:kl-last-unified}.
Then for every $t \ge 0$, it holds that
$
\KL(w^\star, w_{t+1}) 
\le 
\KL(w^\star, w_t)
\cdot 
\exp\big(- \tfrac{1}{20} \eta^2 \sigma^2_{\min}
\cdot p_{t,\min} \cdot  q_{t, \min}
\big) 
$.
\end{restatable}

\subsection{Comparisons with Analysis of Wei et al. (2021)}
\label{app:wei-comparisons}

In this section, we give a brief comparison of
our result with that  of \cite{wei2021linear}.

\paragraph{Overview of Rate in Wei et al.}
Theorem 3 of~\cite{wei2021linear} proves
the following result: 
for a zero-sum game $A \in [-1, 1]^{m \times n}$
with a unique Nash equilibrium $w^\star \in \calW$,
it holds for the iterates $\{w_t\}$ of OMWU 
with stepsize $\eta \le \tfrac{1}{8}$ and
initialized from the uniform distribution that
\begin{equation}
    \KL(w^\star, w_t)
    \le 
    C_3 (1+C_4)^{-t}
    \;\;\text{for all $t \ge 1$} \;,
    \label{eq:wei21-01}
\end{equation}
where $C_3, C_4 > 0$ are 
constants depending on the matrix $A$.
Here, note that the constants $C_3, C_4$ 
depend on the minimum non-zero mass
$\delta$ in the Nash equilibrium $w^\star$ of $A$,
and one must explicitly track their dependence on
$\delta$ throughout the proof to make a comparison.

Note also that Theorem 3 of Wei et al. does not necessarily
assume that $w^\star \in \relint(\calW)$ is interior,
as we require in Theorem~\ref{thm:kl-last-unified}.
Thus, while their proof technique involves a two-phase analysis,
under the assumption that the $w^\star$ is interior, 
it suffices to apply only the second phase of their analysis
(using the second case of their Lemma 2).
For this, tracking the explicit dependence on $\delta$
in the constants $C_3$ and $C_4$ 
(and, for simplicity, hiding other absolute constant 
and instance-specific dependencies) yields that both
\begin{equation}
    C_3 = \exp\big(\tfrac{\log{mn}}{\delta}\big)
    \;\;\text{and}\;\;
    C_4 = \eta^2 \cdot \exp\big(\tfrac{- 3\log{mn}}{\delta}\big) \;.
    \label{eq:wei21-02}
\end{equation}
Here, we note that the $\log(mn)$ term is exactly
the cross-entropy $\Lambda = \crossent(w^\star, w_0)$ 
under the uniform initialization $p_0 = \tfrac{\1_m}{m}$
and $q_0 = \tfrac{\1_n}{n}$.
Combining~\eqref{eq:wei21-01} and~\eqref{eq:wei21-02} 
then yields for Wei et al. a last-iterate convergence rate of
\begin{equation}
    \KL(w^\star, w_t)
    \le
    \exp\big(\tfrac{\Lambda}{\delta}\big)
    \cdot 
    \big(1 + \eta^2 \exp(\tfrac{-3\Lambda}{\delta})\big)^{-t}
    \;.
    \label{eq:wei21-03}
\end{equation}

\paragraph{Comparison with rate in Theorem~\ref{thm:kl-last-unified}.}
Comparing~\eqref{eq:wei21-03} to Part (2) of Theorem~\ref{thm:kl-last-unified}, 
we note the following: in the geometrically decaying term,
both bounds depend on $\exp(-\Lambda/\delta)$.
This stems from the \textit{uniform} lower bound on
the minimum mass $w_{t, \min}$ of the iterates over
time, a bound we establish in Lemma~\ref{lem:wtmin-bound-uniform},
and which is similarly established by Wei et al.
in their Lemma 19
(note also this uniform bound on $w_{t, \min}$ is tight in general;
see the proof our lower bound of Theorem~\ref{thm:kl-last-lower}
in Section~\ref{app:kl-last-lower}).
On the other hand, 
the rate of Wei et al. from~\eqref{eq:wei21-03} also 
contains the leading multiplicative factor of $\exp(\Lambda/\delta)$,
a term that grows unbounded as $\delta \to 0$.
In comparison, the leading multiplicative factor
in Part (2) of our Theorem~\ref{thm:kl-last-unified}
is the initial $\KL(w^\star, w_0)$.
Thus, compared to Wei et al., our new rate in 
Theorem~\ref{thm:kl-last-unified} is sharper by least this
leading term of $\exp(\Lambda/\delta)$.
Moreover, we prove in Theorem~\ref{thm:kl-last-lower}
a linear last-iterate convergence rate lower bound 
showing that the dependence on $\delta$
in Theorem~\ref{thm:kl-last-unified} is optimal.

\paragraph{Comparison of proof techniques.}
Our technique for proving Theorem~\ref{thm:kl-last-unified}
is novel and distinct from the techniques
of Wei et al. Note that the analysis of Wei et al. 
does not directly establish a bound on the one-step change 
$\KL(w^\star, w_{t+1}) - \KL(w^\star, w_t)$
over the iterates of the algorithm
(in contrast to our proof, which does).
Instead, the proof of Wei et al. establishes 
a decreasing one step change of the potential function
\begin{equation*}
    \Theta_t = \KL(w^\star, \widehat w_t) + \KL(\widehat w_t, w_{t-1}) \;,
\end{equation*}
where $\{\widehat w_t\}$ is the sequence of  half-step iterates 
when defining OMWU via 
its Optimistic Mirror Descent formulation
(for the details of this formulation, 
    see, e.g.,~\cite{wei2021linear}, expressions (1) and (2),
    Section 3.1.1 of~\cite{syrgkanis2015fast},
    or Section 2 of~\cite{FGKLLZ24}).
Establishing a linear convergence rate 
for the sequence $\{\Theta_t\}$ then translates
into the linear convergence rate for $\KL(w^\star, w_t)$.

In contrast, recall that our proof of Theorem~\ref{thm:kl-last-unified}
is based on the dual perspective of \textit{energy dissipation}
(via the one-step bounds of Lemma~\ref{lem:energy-one-step-full}),
which we then relate back to the primal space via
the skew-gradient domination inequality of 
Proposition~\ref{prop:dissipation-kl-bound}. 
While the calculations required to establish the tight energy dissipation
bounds of Lemma~\ref{lem:energy-one-step-full}
are still somewhat involved, the approach is conceptually simple,
as outlined in the proof in Section~\ref{app:energy-dissipation:proof}.
Together, our analysis also gives a more precise and quantitative
understanding of the geometric bottlenecks that cause 
the energy dissipation (and thus primal convergence in KL)
to be slow over the iterates.

\paragraph{On the assumption of an interior Nash equilibrium.}
While our analysis does assume both a unique and interior Nash
(while the results of~\cite{DaskalakisP19} and \cite{wei2021linear}
only assume a unique Nash equilibrium), 
we believe our energy-based analysis can likely be adapted
to also hold only under the uniqueness assumption.
We leave this as a direction for future work. 

\subsection{Numerical Simulations}
\label{app:kl-last:numerical}

In this section, we present plots of several numerical simulations
that higlight the dependence on $\delta$ 
in the linear rate of Theorem~\ref{thm:kl-last-unified}.

\paragraph{Setup.}
We run OMWU on two classes of payoff matrices: 
the $2\times 2$ canonical matrix $A_{\delta_p, \delta_q}$
from Definition~\ref{def:A-2x2}, and a diagonal 
payoff matrix $A \in [-1, 1]^{10 \times 10}$. 
\begin{itemize}[
    leftmargin=1em
]
\item
    \textbf{On canonical 2x2 matrix}: 
    We instantiate $A_{\delta_p, \delta_q} \in [-1, -1]^{2\times 2}$
    with $\delta = \delta_p = \delta_q$,
    for a range of $\delta$ between 0 and 0.5. 
    In this instantiation, $w^\star = (p^\star, q^\star)$
    is given by $p^\star = q^\star = (1-\delta, \delta)$.
    Thus, smaller $\delta$ corresponds to 
    a Nash equilibrium closer to a vertex of $\calW$.

\item
    \textbf{On diagonal 10x10 matrix:}
    We instantiate a diagonal $A \in [-1, 1]^{10 \times 10}$
    as follows:
    we construct a vector $v_{\delta} \in (0, 1)^{10}$
    where the first 5 coordinates 
    are $\delta/5$, and the last 5 coordinates are $(1-\delta)/5$. 
    Then, we set $A= \Diag(d_\delta)$, where
    $d_{\delta} \propto 1/v_{\delta}$ (normalized so that 
    the maximum entry of the vector $d_{\delta}$ is 1).
    It follows that the unique and interior NE
    $w^\star = (p^\star, q^\star)$ of $A$ is given by 
    $p^\star = q^\star = v_{\delta}$.
    Thus for smaller values of $\delta$,
    the NE components $p^\star$ and $q^\star$ are roughly uniform among 
    the latter 5 coordinates, and close to boundary face of the simplex 
    corresponding to the support of these coordinates. 
\end{itemize}

\paragraph{Figure~\ref{fig:sim-compare-broad}: 
Dependence of $\delta$ in linear rate.}
In Figure~\ref{fig:sim-compare-broad}, 
we run OMWU on instantiations of the 2x2 and 10x10 diagonal
matrices previously described with several values of $\delta$.
In each instance, we set $\eta = 0.2$, and we initialize the
iterates at the uniform distributions.
We plot the value of $\KL(w^\star, w_t)$ over the iterates
in log scale. 
In both cases, the figure highlights the dependence on
$\delta$ seen in the last-iterate convergence rate of 
Part (2) of
Theorem~\ref{thm:kl-last-unified}.
For smaller $\delta$, the slope of the linear convergence
in the figure is much flatter, as suggested by the analytic bound.

\captionsetup[figure]{labelfont=small}
\begin{figure}[htb!]
    \centering
    \includegraphics[width=0.8\textwidth]{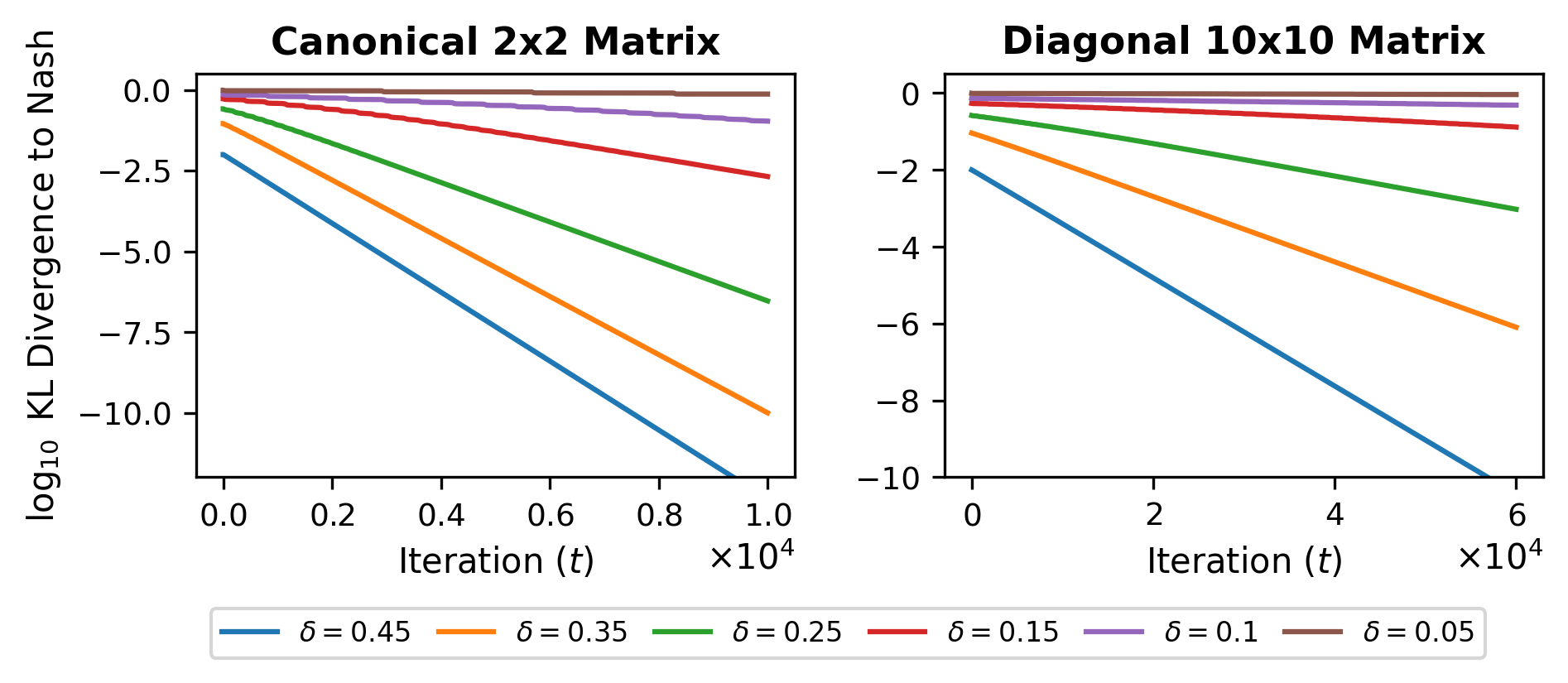}
    \vspace*{-0.5em}
    \caption{%
        \small
        Plot of $\KL(w^\star, w_t)$ of OMWU in log-scale over time
        on instantiations of 2x2 and 10x10 
        payoff matrices 
        with varying values of $\delta$.
        The minimum NE coordinate is
        $\delta$ for the 2x2 instances and 
        $\delta/5$ for the 10x10 instances.
    }    
    \label{fig:sim-compare-broad}
\end{figure}

\paragraph{Figure~\ref{fig:sim-compare-small}: 
Non-uniform one-step change.}
In Figure~\ref{fig:sim-compare-small}, 
we run OMWU on instantiations of the 2x2 and 10x10 diagonal
matrices previously described with varying values of small $\delta$.
In each instance, we set $\eta = 0.2$, and we initialize the
iterates at the uniform distributions.
We plot the value of $\KL(w^\star, w_t)$ over the iterates
in log scale. 
In both cases, the figure highlights the fact that 
the one-step change in $\KL(w^\star, w_t)$ can
decay in a highly non-uniform manner, especially 
when $\delta$ is small, and when
$\KL(w^\star, w_t)$ is still relatively large.
This property is exactly indicated by the non-uniform dependence
on $w_{t, \min}$ in Part (1) of Theorem~\ref{thm:kl-last-unified}.
In this regime, the iterates will initially spend longer periods of
time close to the simplex boundary 
(where $w_{t, \min}$ is smaller), which leads to
smaller magnitudes of dual energy dissipation,
and consequently a slower overall last-iterate convergence to NE.

\captionsetup[figure]{labelfont=small}
\begin{figure}[htb!]
    \centering
    \includegraphics[width=0.8\textwidth]{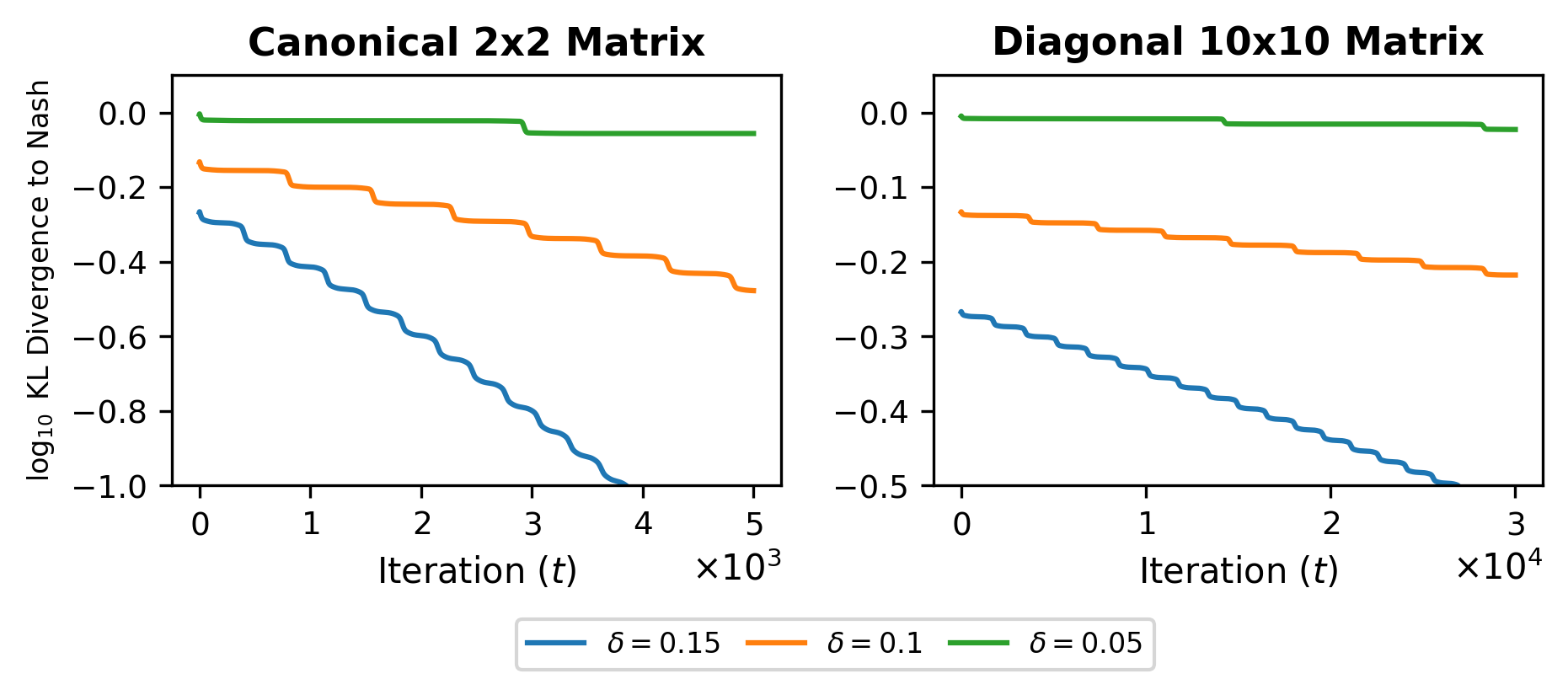}
    \vspace*{-0.5em}
    \caption{%
        \small
        Plot of $\KL(w^\star, w_t)$ of OMWU in log-scale over time
        on instantiations of  2x2 and 10x10 
        payoff matrices with varying values of $\delta$.
        The minimum NE coordinate is
        $\delta$ for the 2x2 instances and 
        $\delta/5$ for the 10x10 instances.
        Observe that the one-step change in $\KL$
        decays non-uniformly over time, 
        as suggested by Part (1) of Theorem~\ref{thm:kl-last-unified}.
    }    
    \label{fig:sim-compare-small}
\end{figure}


\section{Details on OMWU Dynamics in 2x2 Setting}
\label{app:omwu-2x2}

This section gives additional preliminaries and details on OMWU 
in the $2 \times 2$ setting. These preliminaries are used in the proofs
of the lower bounds of Theorem~\ref{thm:kl-last-lower} and
Theorem~\ref{thm:uniform-lb-main}, and the upper bound of 
Theorem~\ref{thm:dg-best-2x2}. In particular, we establish
properties of the following canonical matrix:

\begin{restatable}{define}{defAmatrix}
	\label{def:A-2x2}
	Let $\delta_p, \delta_q \in (0, 1)$.
	Then $A_{\delta_p, \delta_q} \in [-1, 1]^{2\times 2}$ 
	is the payoff matrix given by
	\begin{equation*}
		A_{\delta_p, \delta_q} = 
		\begin{pmatrix}
			\delta_p \delta_q & -\delta_p (1-\delta_q) \\
			-(1-\delta_p)\delta_q & (1-\delta_p)(1-\delta_q)
		\end{pmatrix} \;.
	\end{equation*}
\end{restatable}

\paragraph{Significance of Canonical Matrix.}
For the $2\times 2$ setting, the matrix $A_{\delta_p, \delta_q}$
serves as a \textit{canonical} payoff matrix. 
In particular, the unique and interior Nash equilbrium $w^\star = (p^\star, q^\star)$
of  $A_{\delta_p, \delta_q}$ is given by
$p^\star = (1-\delta_p, \delta_p)$ and $q^\star = (1-\delta_q, \delta_q)$,
and every other $A \in \R^{2\times 2}$
with $w^\star$ as an interior NE  can be written as an 
affine transformation  of $A_{\delta_p, \delta_q}$. 
We state and prove these structural properties in
Proposition~\ref{prop:A-properties} in the following subsection.
Note also that running OMWU on the matrix $A_{\delta_p, \delta_q}$ 
with $\delta_p = \delta \in (0, 1/2)$ and $\delta_q = 1/2$
captures the behavior identified in the
hard $2\times 2$ instance of~\cite{FGKLLZ24}.

\paragraph{Organization of Section.}
The remainder of this section is organized as follows:
\begin{itemize}[
    leftmargin=1em
]
\item
\textbf{Sections~\ref{app:2x2-setup:primal-dual},
\ref{app:2x2-setup:energy-distances},
and~\ref{app:2x2-setup:monotonicity}} give 
details on the OMWU primal and dual iterates 
when run on $A_{\delta_p, \delta_q}$.
These preliminaries are used 
in the lower bound proofs of Theorem~\ref{thm:kl-last-lower}
and Theorem~\ref{thm:uniform-lb-main},
which are based on instantiations of $A_{\delta_p, \delta_q}$
(see Sections~\ref{app:kl-last-lower} and~\ref{app:lowerbound-details},
respectively),
and the upper bound proof of Theorem~\ref{thm:dg-best-2x2},
which relies on the affine transformation of $A_{\delta_p, \delta_q}$
(see Section~\ref{app:dg-upper}).

\item
\textbf{Section~\ref{app:2x2-setup:energy-distances}}
gives simplified expressions for duality gap, 
TV distance, and KL divergence in the $2 \times 2$ setting,
as well as a visual  comparison of their levelsets,
in Section~\ref{sec:2x2-levelsets-compare}.

\end{itemize}

\paragraph{Simplifications on notation.}
Throughout, we will drop the subscript and write
$\1 = (1, 1) \in \R^2$ to denote the two-dimensional
all-ones vector, and $\softmax = \softmax_2$ to denote the
two-dimensional softmax map. 

\subsection{Structural Properties of Canonical Payoff Matrix}

We prove the following properties of the 
matrix $A_{\delta_p, \delta_q}$ from Definition~\ref{def:A-2x2}:

\medskip

\begin{restatable}{prop}{propAproperties}
	\label{prop:A-properties}
	Fix $\delta_p, \delta_q \in (0, 1)$.
	Let $A_{\delta_p, \delta_q} \in [-1, 1]^{2\times 2}$
	be the matrix from Definition~\ref{def:A-2x2}.
	Let $w^\star = (p^\star, q^\star)$
	with $p^\star = (1-\delta_p, \delta_p)$
	and $q^\star = (1-\delta_q, \delta_q)$.
	Then the following properties hold:
	\begin{enumerate}[
		label={(\roman*)},
		itemsep=0em,
		topsep=0em,
		leftmargin=2.5em,
	]
	\item
	The unique and interior NE of
	$A_{\delta_p, \delta_q}$ is $w^\star$,
	and $\langle p^\star, A_{\delta_p, \delta_q} q^\star \rangle = 0$.
	\item
	Fix $A \in [-1, 1]^{2 \times 2}$,
	and suppose that $w^\star$ is an interior NE of $A$.
	Let $v = \langle p^\star, A q^\star \rangle \in \R$,
	and let the entries of $A$ be given by
	\begin{equation*}
		A =
		 \begin{pmatrix}
			a & b \\
			c & d 
		\end{pmatrix} \;.
	\end{equation*}
	Define $\gamma =  \frac{|a- v|}{\delta_p \delta_q}$,
	and assume further that $A$ is not a constant matrix 
	of the form $A = u \1 \1^\top$
	for some $u \in \R$.
	Then $w^\star$ is the unique Nash 
	equilibrium of $A$, $0 < \gamma \le 4$, and 
	\begin{equation*}
		A = \gamma \cdot A_{\delta_p, \delta_q} + v \1 \1^\top \;.
	\end{equation*}
	\end{enumerate}
\end{restatable}

\smallskip

Before giving the proof of Proposition~\ref{prop:A-properties},
we first state and prove the following useful
facts about $2\times 2$ zero-sum games with
an interior equilibrium:

\smallskip

\begin{restatable}[Interior NE implies uniqueness in $2\times2$ games]
{prop}{propinteriortwounique}
	\label{prop:2x2-interior-unique}
	Fix $A \in \R^{2\times 2}$ with entries
	\begin{equation*}
		A = 
		\begin{pmatrix}
			a & b \\ 
			c & d
		\end{pmatrix} \;.
	\end{equation*}
	Suppose $A$ has an interior NE
	$w^\star = (p^\star, q^\star)$
	where $p^\star = (1-\mu_p, \mu_p)$
	and $q^\star = (1-\mu_q, \mu_q)$ for $\mu_p, \mu_q \in (0,1)$.
	Let $v = \langle p^\star, A q^\star \rangle \in \R$,
	and define $\phi = a-b - (c-d)$. 
	Then the following properties hold:
	\begin{enumerate}[
		label={(\roman*)},
		itemsep=0em,
		topsep=0em,
		leftmargin=2.5em,
	]	
	\item
		$\phi = 0$ if and only if $A = u \1 \1^\top$ is a constant
		matrix for some $u \in \R$.
	\item
		If $\phi \neq 0$, then $w^\star$ is the unique Nash 
		equilibrium of $A$.
	\item 
		$\phi = \frac{a-v}{\mu_p \mu_q}$.
	\end{enumerate}
\end{restatable}

\begin{proof}
	As $w^\star = (p^\star, q^\star)$ is an interior NE of $A$,
	by Part (i) of Proposition~\ref{prop:interior-NE}, we must have
	$A q^\star = v \1$ and $A^\top p^\star = v\1$ for 
	a constant $v \in \R$. Using the definition of $A$ and 
	$p^\star = (1-\mu_p, \mu_p)$ and 
	$q^\star = (1-\mu_q, \mu_q)$, 
	this yields the constraints 
	\begin{equation}
		\begin{cases}
			a - (a-b) \cdot \mu_q  = v \\
			c - (c-d) \cdot \mu_q  = v
		\end{cases}
		\quad\text{and}\quad
		\begin{cases}
			a - (a-c) \cdot \mu_p   = v \\
			b - (b-d) \cdot \mu_p  = v \;.
		\end{cases} 
		\label{eq:2x2-00}
	\end{equation}
	Rearranging further yields the constraints 
	\begin{equation}
		\begin{cases}
			\mu_q \cdot (a-b - (c-d)) = a-c \\
			\mu_p \cdot (a-b - (c-d)) = a-b \;.
		\end{cases}
		\label{eq:2x2-01}
	\end{equation}
	For Part (i) of the proposition, recall that $\phi = a-b - (c-d)$.
	Now observe that if all $a=b=c=d$,
	meaning $A$ is a constant matrix, then trivially $\phi =0$.
	In the other direction, if
	$\phi = 0$, then~\eqref{eq:2x2-01} implies 
	$a=c$ and $a=b$.
	By definition of $\phi$, this further implies $a = d$.
	Thus if $\phi = 0$, then $A$ is a constant matrix.
	In this degenerate case, every
	$(p, q) \in \relint(\Delta_2 \times \Delta_2)$
	is an NE.	

	For Part (ii), suppose that $\phi \neq 0$. 
	Then~\eqref{eq:2x2-01} implies that
	$\mu_q = (a-c) / \phi$ and 
	$\mu_p = (a-b) / \phi$
	have unique solutions that 
	also uniquely define $p^\star = (1-\mu_p, \mu_p)$
	and $q^\star = (1-\mu_q, \mu_q)$.
	This also further precludes $A$ from simultaneously having 
	a non-interior Nash equlibrium: if such an equilibrium $w'$ exists,
	then every interior $\w w$ on the line segment between 
	$w'$ and $w^\star = (p^\star, q^\star)$ 
	would also be a Nash equilibrium (which follows
	by Part (ii) of Proposition~\ref{prop:interior-NE}).
	However, this contradicts the uniqueness of the 
	solution of~\eqref{eq:2x2-01}.
	Thus, in the case that $\phi \neq 0$, 
	an interior NE of $A$
	is the unique NE of $A$.

	For Part (iii), observe from~\eqref{eq:2x2-00} 
	that	$a-v = (a-c) \cdot \mu_p$. 
	Substituting the first equality of~\eqref{eq:2x2-01}
	further gives $a-v = \phi \mu_q \mu_p$. 
	Rearranging for $\phi$ yields the desired statement.
\end{proof}

\smallskip

We now proceed with the proof of Proposition~\ref{prop:A-properties}:

\paragraph{Proof of Proposition~\ref{prop:A-properties}.}
We prove the two parts of the proposition separately.

\noindent
\textbf{Proof of Part (i).} 
Using the definition of $A_{\delta_p, \delta_q}$ and
$p^\star = (1-\delta_p, \delta_p)$ and $q^\star =(1-\delta_q, \delta_q)$,
observe by a direct calculation that 
\begin{equation*}
	\textstyle
	A_{\delta_p, \delta_q} 
	\begin{pmatrix}
		1-\delta_q \\
		\delta_q
	\end{pmatrix} 
	= 
	\begin{pmatrix}
		\delta_p\delta_q (1-\delta_q)
		- \delta_p (1-\delta_q)  \delta_q \\
		-(1-\delta_p) \delta_q (1-\delta_q)
		+ 
		(1-\delta_p)(1-\delta_q) \delta_q
	\end{pmatrix} 
	= 
	\begin{pmatrix} 
		0 \\ 0
	\end{pmatrix} \;.
\end{equation*}
By a similar calculation, we find 
$A^\top_{\delta_p, \delta_q} p^\star = 0$.
Thus $A_{\delta_p, \delta_q} q^\star$ and 
$A^\top_{\delta_p, \delta_q} p^\star$
are both constant vectors. Thus by Part (ii) of
Proposition~\ref{prop:interior-NE} it follows
that $w^\star = (p^\star, q^\star)$
is an NE of $A_{\delta_p, \delta_q}$,
and that 
$\langle p^\star, A_{\delta_p, \delta_q} q^\star \rangle = 0$.
Now observe that since $\delta_p, \delta_q \in (0, 1)$, the 
equilbrium $w^\star$ is interior by definition of $p^\star$
and $q^\star$.
Moreover, by Definition~\ref{def:A-2x2}$, 
A_{\delta_p, \delta_q}$ is not a constant matrix. 
Thus, by Parts (i) and (ii) of Proposition~\ref{prop:2x2-interior-unique},
it follows that $w^\star$ is also the unique NE of $A_{\delta_p, \delta_q}$.

\smallskip

\noindent
\textbf{Proof of Part (ii).} 
Recalling the entries of $A \in \R^{2\times 2}$ from the
statement of the proposition, 
let $\phi = a-b - (c-d)$. 
Recall that we assume $A$ has an interior 
Nash equilibrium $w^\star = (p^\star, q^\star)$,
and that $A$ is not a constant matrix.
Thus, by Parts (i) and (ii) of Proposition~\ref{prop:2x2-interior-unique}, 
$\phi \neq 0$, and $w^\star$ is the unique Nash equilibium of $A$.
Now, let  $v = \langle p^\star, A q^\star \rangle$,
and suppose that $\phi > 0$.
In this case, define 
\begin{equation*}
\textstyle
A' = A - v\1 \1^\top \in \R^{2\times 2} \;.
\end{equation*}
We will show that $A' = \gamma \cdot A_{\delta_p, \delta_q}$,
for some $0 < \gamma \le 4$,
which will imply that $A = \gamma A_{\delta_p, \delta_q} + v \1\1^\top$.

For this, observe by Proposition~\ref{prop:interior-NE}
that since $w^\star = (p^\star, q^\star)$
is an interior NE of $A$ with $v = \langle p^\star, A q^\star \rangle$, 
then $Aq^\star = v\1$ and $A^\top p^\star = v \1$. 
Thus by definition of $A' = A - v\1\1^\top$,
it follows that both $A' q^\star = 0$ and $A'^\top p^\star = 0$.
Using the definition of the entries of $A'$ (which depend on
the entries of $A$) it is straightforward to show that this further implies 
\begin{equation}
	b - v =  \tfrac{- (a-v) \cdot (1-\delta_q)}{\delta_q},\;
	d - v = \tfrac{(a-v) \cdot(1-\delta_p)(1-\delta_q)}{\delta_p\delta_q},\;
	\text{and}\,
	c - v = \tfrac{-(a-v) \cdot (1-\delta_p)}{\delta_p} \;.
	\label{eq:aprime-01}
\end{equation}
By Part (iii) of Proposition~\ref{prop:2x2-interior-unique},
recall that $\phi =  (a-v)/(\delta_p \delta_q)$,
and recall that we assumed $\phi > 0$.
Then factoring out $\phi$ from each term
of~\eqref{eq:aprime-01}, we find exactly that
$A' = \phi \cdot A_{\delta_p, \delta_q}$.
By definition of $A' = A - v \1\1^\top$,
we conclude $A = \phi \cdot A_{\delta_p, \delta_q} + v \1\1^\top$.

In the case that $\phi < 0$, 
we instead define $A' = v \1 \1^\top - A$. 
Repeating an identical set of calculations
again gives $A' = \phi A_{\delta_p, \delta_q}$,
and thus $A = - \phi \cdot A _{\delta_p, \delta_q} + v \1 \1^\top $.
In either case, set $\gamma := |\phi|$.
By definition of $\phi = a-b- (c-d)$
and $A \in \R^{2\times 2}$, it follows
that $|\phi| \le 4$.
Thus we conclude
$A = \gamma A_{\delta_p, \delta_q} + v \1 \1^\top$
for $0 < \gamma \le 4$, which yields the statement of
Part (ii) of the proposition.
\;~ \hfill $\blacksquare$

\medskip

\begin{restatable}{rem}{remcanonicalmatrix}
	\label{rem:canonical-2x2-matrix}
	We note that the Part (ii) of
	Proposition~\ref{prop:A-properties}
	is similar to Lemma 5 of~\cite{C25}.
	The key difference is that the ``base'' matrix
	in their lemma is not $A_{\delta_p, \delta_q}$,
	but a different $2\times 2$ payoff matrix. 
	The key property of $A_{\delta_p, \delta_q}$ is
	that the value of the game is zero
	(i.e., $\langle p^\star, A_{\delta_p, \delta_q} q^\star \rangle = 0$).
	This allows for a 
	more convenient primal-dual relationship under the 
	OMWU iterates
	(c.f., Proposition~\ref{prop:calZ-grad-map-injectivity})
	that we introduce in the sequel.
\end{restatable}

Finally, we also show that $A_{\delta_p,\delta_q}$
is well-conditioned in the sense that
its minimum restricted singular value 
$\sigma_{\min}$ is always an absolute constant.

\begin{restatable}
[Minimum restricted singular value of Canonical Matrix]
{prop}{propcanonicalrsv}
	\label{prop:rsv-canonical-2x2}
	Fix any $\delta_p, \delta_q \in (0, 1)$.
	Let $A := A_{\delta_p, \delta_q} \in [-1, 1]^{2\times 2}$ 
	be the matrix from Definition~\ref{def:A-2x2}.
	Let $\sigma_{\min}$ be the minimum restricted singular value 
	from Definition~\ref{def:restricted-sv-J}.
	Then $\sigma_{min} = \frac{1}{2}$.
\end{restatable}

We note that the proof is the generalization of
and follows similarly to that of
Part (2) of Proposition~\ref{prop:small-nash-large-rsv}.
There, we established for the special case
$\delta = \delta_p =\delta_q$ that $\sigma_{\min} = \frac{1}{2}$.
Proposition~\ref{prop:rsv-canonical-2x2} extends
this result to hold for general $\delta_p, \delta_q \in (0, 1)$.

\begin{proof}
First, let $\sigma_{\min, n}$ and $\sigma_{\min, m}$
be the component-wise restricted singular
values from~\eqref{eq:component-rsv-A},
and recall from Proposition~\ref{prop:component-restricted-sv}
that $\sigma_{\min} = \min\{\sigma_{\min}(A, \1^\bot), 
\sigma_{\min}(A^\top, \1^\bot)\}$.
Using calculations similar to the proof of Part (2) of 
Proposition~\ref{prop:small-nash-large-rsv}, 
we will show that
$\sigma_{\min}(A, \1^\bot)= \sigma_{\min}(A^\top, \1^\bot)
= \frac{1}{2}$.
For this, recall by definition that
\begin{equation*}
	\textstyle
	\sigma_{\min}(A, \1^{\bot})
	= 
	\inf_{v \in \1^\bot \setminus \{0\}} \,
	\frac{\|\Pi_{\1^{\bot}}(A v)\|_2}{\|v\|_2}  \;.
\end{equation*}
Moreover, for any $v \in \1^\bot \subset \R^2$, 
it follows that $v = c \cdot (-1, 1)$ for some $c \in \R$.
Then using the definition of $A = A_{\delta_p, \delta_q}$,
and recalling that 
$\Pi_{\1^\bot} = I - \frac{1}{2} \1\1^\bot \in \R^{2 \times 2}$,
it follows by a direct calculation that
\begin{equation*}
	\Pi_{\1^\bot} (A v)
	= 
	c \cdot 
	\begin{pmatrix}
		\delta_p - (\delta_p - (1-\delta_p))/2 \\ 
		-(1-\delta_p) - (\delta_p - (1-\delta_p))/2
	\end{pmatrix}
	= 
	\frac{c}{2} \cdot \begin{pmatrix}1 \\ -1 \end{pmatrix} \;.
\end{equation*}
It follows that $\|\Pi_{\1^\bot} (A v)\|_2 = |c| \sqrt{2}/2$.
Moreover, we have by definition of 
$v = c \cdot (1, -1)$ that $\|v\|_2 = |c| \sqrt{2}$.
Thus we conclude $\sigma_{\min}(A, \1^\bot) = 1/2$. 
By an identical calculation, we also find
$\sigma_{\min}(A^\top, \1^\bot) = 1/2$,
which completes the proof. 
\end{proof}

\subsection{Primal and Dual OMWU Iterates}
\label{app:2x2-setup:primal-dual}

In this section, we give details on the OMWU iterates
when run on the canonical
$2\times 2$ game $A_{\delta_p, \delta_q}$
from Definition~\ref{def:A-2x2}.
These derivations are used to prove
the lower bounds of Theorems~\ref{thm:kl-last-lower}
and Theorems~\ref{thm:uniform-lb-main},
and also the best-iterate upper bound of Theorem~\ref{thm:dg-best-2x2}.

In this setting, note that while the dual iterates 
$\{z_t\}$ all belong to $\R^4$,
the effective dual space $\calZ$ (as introduced
in Section~\ref{sec:omwu-main:primal-dual})
is a two-dimensional subspace.
This allows for a convenient two-dimensional
representation of both the primal and dual iterates 
that facilitates obtaining precise
control of the OMWU trajectory as
a function of the Nash equilibrium coordinates
$\delta_p$ and $\delta_q$.

\paragraph{Roadmap of subsection.}
Section~\ref{app:2x2-setup:primal-dual:effective-dual}
presents structural properties of 
the lower-dimensional effective dual space. 
Section~\ref{app:2x2-setup:primal-dual:mapping}
then gives simplified expressions for the 
energy function and primal-dual mapping
over the effective dual space.
Section~\ref{app:2x2-setup:primal-dual:omwu-iterates}
states the simplified primal and dual OMWU 
updates when run on $A_{\delta_p, \delta_q}$.
Throughout, we assume the following setting:


\begin{restatable}[$2\times 2$ Setting Using Canonical Matrix]
{setting}{settwotwo}
	\label{setting:2x2-canonical}
	For fixed $\delta_p, \delta_q \in (0, 1)$,
	assume that $A := A_{\delta_p, \delta_q}$
	is the canonical $2 \times 2$ payoff matrix
	from Definition~\ref{def:A-2x2}. 
\end{restatable}


In this setting, we will also make use of the 
following simplified notation for readability:

\paragraph{Simplified notation and definitions.}
Under Setting~\ref{setting:2x2-canonical},
let $\LSE \colon \R^2 \to \R$ denote the two-dimensional
log-sum-exp function, 
and let $R: \Delta_2 \to \R$ denote the
two-dimensional negative-entropy function
(as defined in Section~\ref{sec:omwu}).
Moreover, we write $\sigmoid \colon \R \to \R$
to denote the function given by 
$\sigmoid(u) = 1/(1+\exp(-u))$ for $u \in \R$.

\subsubsection{Effective Dual Space}
\label{app:2x2-setup:primal-dual:effective-dual}

Recall from Section~\ref{sec:omwu} that
$\calZ = \Span(J \calW) \subseteq \R^{2+2}$,
where $J = ((0, A), (-A^\top, 0)) \in \R^{4 \times 4}$ 
and $\calW = \Delta_2 \times \Delta_2$.
In order to introduce the lower-dimensional 
representation of $\calZ$, it will be more
convenient to further work with the individual
components of the dual variables $z = (x, y) \in \R^{2 + 2}$.
Specifically, we first introduce the component-wise
effective dual spaces $\calX, \calY \subset \R^2$:

\begin{restatable}[Component-wise effective dual spaces]
	{define}{defcomponenteffectivedual}
	\label{def:calX-calY}
	Fix $\delta_p, \delta_q \in (0, 1)$ and assume 
	Setting~\ref{setting:2x2-canonical}.
	Let $\calX, \calY \subset \R^2$ be the linear subspaces
	given by 
	\begin{equation*}
		\calX := \Span(A \Delta_2) 
		\;\text{and}\;
		\calY := \Span(A^\top \Delta_2) \;.
	\end{equation*}
\end{restatable}

Given the product structure of $\calW = \Delta_2 \times \Delta_2$,
and using the definition of $J$, the following characterization
of $\calZ = \Span(J \calW)$ is immediate:

\smallskip

\begin{restatable}{prop}{propeffectivedualequiv}
	\label{prop:2x2-calZ-equiv}
	Fix $\delta_p, \delta_q \in (0, 1)$,
	and assume Setting~\ref{setting:2x2-canonical}.
	Then $\calZ = \Span(J\calW) = \calX \times \calY$.
\end{restatable}

Now define $p^\star = (1-\delta_p, \delta_p)$
and $q^\star = (1-\delta_q, \delta_q)$,
and recall from Part (i) of Proposition~\ref{prop:A-properties}
that $(p^\star, q^\star)$ is the unique and interior NE
of $A := A_{\delta_p, \delta_q}$. 
The key property of $\calX$ and $\calY$
is that these spaces are orthogonal 
to $p^\star$ and $q^\star$, respectively. 
Formally, we have the following refinement of
Proposition~\ref{prop:dual-subspace-property}
for the present $2\times 2$ setting:

\smallskip

\begin{restatable}[Component-wise orthogonality to Nash]
{prop}{proptwotwoeffectivedual}
	\label{prop:2x2-effective-dual-orthogonal}
	Fix $\delta_p, \delta_q \in (0, 1)$,
	and assume Setting~\ref{setting:2x2-canonical}.
	Then for any $x = (x(1), x(2)) \in \calX$
	and $y = (y(1), y(2)) \in \calY$, it holds that:
	\begin{equation*}
		\begin{aligned}
		\langle x, p^\star \rangle = 0 
		&\;\implies\;
		x(2) = x(1) \cdot (1-(1/\delta_p)) \\ 
		\text{and}\quad
		\langle y, q^\star \rangle = 0 
		&\;\implies\;
		y(2) = y(1) \cdot (1-(1/\delta_q)) \;.
		\end{aligned}
	\end{equation*}
\end{restatable}

\begin{proof}
	First, fix $x \in \calX$. By definition, this means that
	there exists $k \ge 1$ such that 
	$x = \sum_{i=1}^k \tau_i A p_i$, where all $\tau_i \in \R$
	and $p_i \in \Delta_2$.
	Then for each $i \in [k]$, observe that 
	\begin{equation*}
		\langle \tau_i A p_i, p^\star \rangle 
		= 
		\tau_i 
		\langle
			p_i, A^\top p^\star
		\rangle
		= 
		\tau_i 
		\langle
			p_i, 0
		\rangle 
		= 0 \;.
	\end{equation*}
	Here, the final equality is due to the fact
	that $(p^\star, q^\star)$ is an interior Nash equilibrium
	of $A = A_{\delta_p, \delta_q}$ 
	with $\langle p^\star, A q^\star \rangle = 0$
	(Part (i) of Proposition~\ref{prop:A-properties}),
	which implies that
	$A^\top p_{\star} = 0$ is the zero vector
	(Part (i) of Proposition~\ref{prop:interior-NE}).
	Thus $\langle x, p^\star \rangle = 0$
	by definition of $x$.
	Now, as $p^\star = (1-\delta_p, \delta_p)$,
	a direct calculation and rearrangement then finds 
	\begin{equation*}
		\langle x, p^\star \rangle
		= x(1)\cdot (1-\delta_p) + x(2) \cdot \delta_p
		= 0 
		\quad\implies\quad
		x(2) = 
		x(1) \cdot (1 - (1/\delta_p))\;.
	\end{equation*}
	This proves the first statement of the proposition.
	The second statement follows using
	an identical argument for $y \in \calY$.
\end{proof}

Before proceeding, we make the following two remarks:

\begin{restatable}[Comparison with Proposition~\ref{prop:dual-subspace-property}]
{rem}{remorthogonalcompare}
	\label{remark:calZ-orthogonal-compare}
	For a matrix $A$ with interior NE 
	$w^\star = (p^\star, q^\star)$, 
	Proposition~\ref{prop:dual-subspace-property}
	establishes for the general case that
	$\langle z, w^\star \rangle = 0$ for all $z = (x, y) \in \calZ$.
	In particular, this implies by definition of $J$ that
	$\langle x, p^\star \rangle + \langle y, q^\star \rangle = 0$.
	Proposition~\ref{prop:2x2-effective-dual-orthogonal}
	establishes under Setting~\ref{setting:2x2-canonical}
	the slightly stronger \textit{component-wise} property that
	both $\langle x, p^\star \rangle = 0$
	and $\langle y, q^\star \rangle = 0$.
	As seen in the proof
	of Proposition~\ref{prop:2x2-effective-dual-orthogonal},
	this property is a consequence of the fact that 
	$\langle p^\star, A_{\delta_p, \delta_q} q^\star \rangle = 0$,
	as proven in Part (i) of Proposition~\ref{prop:A-properties}.
	Indeed, it is then straightforward to see that this 
	stronger property also holds for any other (general-dimension)
	$A$ whose value of the game is zero. 
\end{restatable}

\begin{restatable}[Main consequence is lower-dimensional 
effective dual space]
{rem}{remcalZlowerdim}
	\label{remark:calZ-lower-dimensional}
	The key consequence of 
	Proposition~\ref{prop:2x2-effective-dual-orthogonal} 
	is that, under Setting~\ref{setting:2x2-canonical},
	the dual variable components $x \in \calX \subset \R^2$ 
	and $y \in \calY \subset \R^2$ each lie on one-dimensional
	subspaces. Thus, a dual variable $z = (x, y) \in \calZ$
	can be fully specified by the first coordinates
	$(x(1), y(1)) \in \R^2$ of its components.
	This means that the effective dual space $\calZ \subset \R^{2+2}$
	is a \textit{two-dimensional subspace}.
	This will allow for additional, convenient 
	simplifications of the energy function $F$
	and the primal-dual relationship over $\calZ$.
	We note that similar approaches were used in 
	prior works of~\cite[Section 3.2]{bailey2019fast} 
	and~\cite[Proposition 4.2]{lazarsfeld2025optimism},
	albeit for different base matrices.
\end{restatable}

\subsubsection{Primal-Dual Mapping Over Effective Dual Space}
\label{app:2x2-setup:primal-dual:mapping}

Under Setting~\ref{setting:2x2-canonical}, 
working over the lower-dimensional structure of $\calZ$ 
established by Proposition~\ref{prop:2x2-effective-dual-orthogonal}
allows for simplified chacterizations of 
the energy function $F$ and its gradient map.
Formally, we establish the following properties:

\smallskip

\begin{restatable}[Energy and primal-dual map over $\calZ$]
{prop}{propenergysimplified}
	\label{prop:2x2-energy-map-simplified}
	Fix $\delta_p, \delta_q \in (0, 1)$ and 
	assume Setting~\ref{setting:2x2-canonical}.
	Let $z = (x, y) \in \calZ = \calX \times \calY$, 
	and let $w = (p, q) = \nabla F(z) \in \relint(\calW)
	= \relint(\Delta_2) \times \relint(\Delta_2)$.
	Then the following relationships hold:
	\begin{enumerate}[
		label={(\roman*)},
		itemsep=0em,	
		topsep=0em,
		leftmargin=2.5em,
	]
	\item
	$\LSE(x) 
		= 
		x(1) + \log\big(
			1 + \exp\big(\frac{-x(1)}{\delta_p}\big)
		\big)
	$
	and 
	$\LSE(y)
		= 
		y(1) + \log\big(
			1 + \exp\big(\frac{-y(1)}{\delta_q}\big)
		\big)
	$.

	\item
	$p(1) = \sigmoid\big(\frac{x(1)}{\delta_p}\big)$
	and
	$q(1) = \sigmoid\big(\frac{y(1)}{\delta_q}\big)$.

 	\item
 	$x(1) = \delta_p \cdot \log\big(\frac{p(1)}{1-p(1)}\big)$
 	and 
 	$y(1) = \delta_q \cdot \log\big(\frac{q(1)}{1-q(1)}\big)$.
	\end{enumerate}
\end{restatable}

\begin{proof}
	Note that as $x \in\calX$ and $y \in \calY$, we have
	by Proposition~\ref{prop:2x2-effective-dual-orthogonal}
	that $x(2) = x(1) \cdot (1- (1/\delta_p))$
	and $y(2) = y(1) \cdot (1 - (1/\delta_q))$.
	We use this relationship to prove of
	each of the three parts as follows:

	\smallskip

	\noindent
	\textbf{Proof of Part (i).}
	By definition of $\LSE(\cdot)$ and using 
	Proposition~\ref{prop:2x2-effective-dual-orthogonal},
	we can compute
	\begin{align*}
		\LSE(x) &= 
		\log\big(
			\exp(x(1)) + \exp(x(2))
		\big) \\
		&=  
		\log\big(
			\exp(x(1)) + \exp\big(x(1) \cdot (1- \tfrac{1}{\delta_p})\big)
		\big) \\
		&= 
		\log\big(
			\exp(x(1)) \cdot (1 + \exp\big(\tfrac{-x(1)}{\delta_p})\big)
		\big)
		= 
		x(1) + \log\big(1 + \exp\big(\tfrac{-x(1)}{\delta_p} \big)\big) \;.
	\end{align*}
	The statement for $\LSE(y)$ follows by an identical calculation. 

	\smallskip

	\noindent
	\textbf{Proof of Part (ii).}
	Recall by definition of $F$ and the assumptions
	of the proposition that
	$\nabla F(z) = (p, q) = (\softmax(x), \softmax(y))$. 
	Thus be definition of $\softmax$ and using 
	Proposition~\ref{prop:2x2-effective-dual-orthogonal},
	we can write
	\begin{align*}
		p(1) 
		= 
		\frac{\exp\big(x(1)\big)}{\exp\big(x(1)\big) + \exp\big(x(2)\big)} 
		&= 
		\frac{\exp\big(x(1)\big)}{\exp\big(x(1)\big)
		+ \exp\big(x(1) \cdot (1-\frac{1}{\delta_p})\big)} \\
		&= 
		\frac{\exp\big(x(1)\big)}
		{\exp\big(x(1)\big) 
		\cdot \big(1 + \exp\big(\frac{-x(1)}{\delta_p}\big)\big)}
		= 
		\sigmoid\Big(\frac{x(1)}{\delta_p}\Big) \;.
	\end{align*}	
	The statement for $q(1)$ follows by an identical calculation.
	\smallskip

	\noindent
	\textbf{Proof of Part (iii).}
	Recall for $u \in \R$ that 
	$\sigmoid^{-1}(u) = \log\big(\frac{u}{1-u}\big)$.
	Thus by Part (ii), we have 
	\begin{equation*}
		\frac{x(1)}{\delta_p}
		= 
		\sigmoid^{-1}(p(1))
		= \log\Big(\frac{p(1)}{1-p(1)}\Big) \;,
	\end{equation*}
	Rearranging gives 
	$x(1) = \delta_p \cdot \log\big(\frac{p(1)}{1-p(1)}\big)$.
	The claim for $y(1)$ follows identically.
\end{proof}

\smallskip

\begin{restatable}[Bijectivity of energy gradient map
	over $\calZ$.]
	{rem}{remarktwotwoprimaldual}
	\label{remark:2x2-primal-dual}
	Note that Parts (ii) and (iii) of 
	Proposition~\ref{prop:2x2-energy-map-simplified}
	establish that, under Setting~\ref{setting:2x2-canonical},
	$\nabla F: \calZ \to \relint(\calW)$ is bijective.
	Indeed, note that as the zero-sum game value under
	$A := A_{\delta_p, \delta_q}$
	is zero (Part (i) of Proposition~\ref{prop:A-properties}),
	the more general statements of Propositions~\ref{prop:calZ-grad-map} 
	and~\ref{prop:calZ-grad-map-injectivity} also together imply 
	that $\nabla F: \calZ \to \relint(\calW)$ is bijective.
	For the $2\times 2$ setting, the closed-form expression of  
	Part (iii) of Proposition~\ref{prop:2x2-energy-map-simplified}
	thus allows for recovering the coordinates of the 
	dual variable $z = (x, y) \in \calZ$ solely using the 
	coordinates of the corresponding primal variable
	$w \in \relint(\calW)$.
\end{restatable}

Part (iii) of Proposition~\ref{prop:2x2-energy-map-simplified}
also leads to several useful corollaries. 
First, it allows for defining the coordinates of the unique 
``dual'' Nash variable $z^\star \in \calZ$ that maps
to the primal NE of $A := A_{\delta_p, \delta_q}$
under $\nabla F$. More formally:


\begin{restatable}[Dual Nash coordinate in $\calZ$]
	{cor}{cordualnashtwotwo}
	\label{cor:2x2-dual-nash}
	Fix $\delta_p, \delta_q \in (0, 1)$
	and assume Setting~\ref{setting:2x2-canonical}.
	Let $w^\star = (p^\star, q^\star) \in \relint(\calW)$.
	Let $z^\star = (x^\star, y^\star) \in \calZ$ be such that
	\begin{equation*}
		\textstyle
		x^\star(1) = \delta_p \cdot \log \Big(\frac{1-\delta_p}{\delta_p}\Big)
		\;\;\text{and}\;\;
		y^\star(1) = \delta_q \cdot \log \Big(\frac{1-\delta_q}{\delta_q}\Big) \;.
	\end{equation*}
	Then $\nabla F(z^\star) = w^\star$.
\end{restatable}

As a second corollary, note that
the expressions $\delta_p \log(p(1)/(1-p(1)))$
and $\delta_q \log(q(1)/(1-q(1)))$ are increasing
in $p(1)$ and $q(1)$, respectively.
Similarly, $\sigmoid(x(1)/\delta_p)$
and $\sigmoid(y(1)/\delta_q)$
are increasing in $x(1)$ and $y(1)$, respectively.
Together this implies the following further
relationships between $(p, q) \in \relint(\calW)$
and $(x, y) \in \calZ$:

\smallskip

\begin{restatable}{cor}{corrprimaldualthreshold}
	\label{corr:2x2-pd-threshold}
	Fix $\delta_p, \delta_q \in (0, 1)$
	and assume Setting~\ref{setting:2x2-canonical}.
	Let $z = (x, y) \in \calZ$, and
	let $w = (p, q) = \nabla F(z) \in \relint(\calW)$.
	Fix any $\rho \in (0, 1)$. 
	Then the following relationships hold:
	\begin{equation*}
			p(1) \ge \rho 
			\,\iff\,
			x(1) \ge \delta_p \cdot 
			\log\Big(\frac{\rho}{1-\rho}\Big)
			\;\;\text{and}\;\;
			q(1) \ge \rho 
			\,\iff\,
			y(1) \ge \delta_q \cdot 
			\log\Big(\frac{\rho}{1-\rho}\Big) \;.
	\end{equation*}
\end{restatable}

\subsubsection{OMWU Dynamics}
\label{app:2x2-setup:primal-dual:omwu-iterates}

Under Setting~\ref{setting:2x2-canonical},
the primal-dual relationships over the effective
dualspace $\calZ$ lead to simplified expressions
for the primal and dual OMWU updates. 
Formally, we have the following: 

\smallskip

\begin{restatable}[OMWU update rules
on canonical $2 \times 2$ matrix]
{prop}{propomwudualtwo}
	\label{prop:omwu-pd-2x2}
	Fix $\delta_p, \delta_q \in (0, 1)$
	and assume Setting~\ref{setting:2x2-canonical}.
	Let $\{w_t\}$ and $\{z_t\}$ denote the
	primal and dual iterates of OMWU on 
	$A := A_{\delta_p, \delta_q}$ with 
	stepsize $\eta > 0$
	initialized from $w_0 = (p_0, q_0) \in \relint(\calW)$,
	where each $w_t = (p_t, q_t) \in \relint(\calW)$
	and $z_t = (x_t, y_t) \in \calZ$.
	Then for all $t \ge 1$:
	\begin{equation*}
		\begin{cases}
		\, x_{t+1}(1)
		= x_t(1) - \eta \cdot \delta_p
		\left(
			q_t(1) - (1-\delta_q)
			+ \big(q_t(1) - q_{t-1}(1)\big)
		\right) \\
		\, y_{t+1}(1)
		= y_t(1) + \eta \cdot \delta_q
		\left(
			p_t(1) - (1-\delta_p)
			+ \big(p_t(1) - p_{t-1}(1)\big)		
		\right) \;,
		\end{cases} 
	\end{equation*}
	and moreover
	\begin{equation*}
		\begin{cases}
			\, p_{t+1}(1) = \sigmoid
			\left(\frac{x_{t+1}(1)}{\delta_p}\right) \\
			\, q_{t+1}(1) = \sigmoid
			\left(\frac{y_{t+1}(1)}{\delta_q}\right) \;.
		\end{cases}
	\end{equation*}
\end{restatable}

\begin{proof}
	First, recalling the structure of $A = A_{\delta_p, \delta_q}$
	from Definition~\ref{def:A-2x2}, a straightforward
	calculation shows that, for
	any $p, q \in \Delta_2$:
	\begin{equation}
		(Aq)(1) = \big(q(1) - (1-\delta_q)\big) \cdot \delta_p
		\;\;\text{and}\;\;
		(A^\top p)(1) 
		= 
		\big(p(1) - (1-\delta_p)
		\big) \cdot \delta_q \;.   
		\label{eq:A2-01}
	\end{equation}
	Now recall by definition of the OMWU dual iterates$\{z_t\}$ 
	from~\eqref{eq:zt-wt-def} that, for all $t \ge 1$: 
	\begin{equation*}
		z_{t+1} = 
		(x_{t+1}, y_{t+1}) = z_{t} - 
		\eta (J w_t + J (w_t - w_{t-1})) \;,
	\end{equation*}
	which means component-wise that
	\begin{equation}
		\begin{aligned}
		x_{t+1} 
		&= 
		x_t
		- \eta
		(A q_t + A(q_{t} - q_{t-1})) \\
		\text{and}\quad
		y_{t+1}
		&= 
		y_t
		+ \eta 
		(A^\top p_t + A^\top(p_{t} - p_{t-1})) \;.
		\end{aligned}
		\label{eq:A2-02}
	\end{equation}
	Considering only the first coordinates
	of the components in~\eqref{eq:A2-02}
	and applying the simplifications from~\eqref{eq:A2-01}
	yields the first statement of the proposition. 
	
	For the second statement, observe that
	all $z_t = (x_t, y_t) \in \calZ = \calX \times \calY$ by definition. 
	Moreover, by Proposition~\ref{prop:omwu-skew-grad},
	we have $w_t = \nabla F(z_t)$ for all $t \ge 0$.
	Then applying Part (ii) of 
	Proposition~\ref{prop:2x2-energy-map-simplified}
	yields the second statement and concludes the proof. 
\end{proof}

\smallskip

\begin{restatable}[Leading coordinates suffice
for tracking OMWU iterates]{rem}{remfirstcoords}
	\label{remark:2x2-omwu-first-coords}
	Due to Proposition~\ref{prop:2x2-effective-dual-orthogonal},
	for $z = (x, y) \in \calZ$,
	observe that 
	the leading coordinates $x(1) \in \R$ and $y(1) \in \R$
	fully specify $z$.
	Similarly, for $w = (p, q) \in \calW$,
	the leading coordinates $p(1) \in [0, 1]$ and $q(1) \in [0,1]$
	fully specify $w$.
	Thus, the simplified expressions 
	in Proposition~\ref{prop:omwu-pd-2x2}
	are sufficient for fully specifying the 
	the OMWU iterates $\{z_t\}$ and $\{w_t\}$
	under Setting~\ref{setting:2x2-canonical}.
\end{restatable}

\subsection{Montonicity and Cycling Properties of OMWU Iterates}
\label{app:2x2-setup:monotonicity}

Using the simplified OMWU update expressions
from Proposition~\ref{prop:omwu-pd-2x2},
we prove in this section an additional set of 
monotonicty properties that the OMWU iterates satisfy 
under Setting~\ref{setting:2x2-canonical}. 
In the statement of the lemma, we assume without loss
of generality that $\delta_p, \delta_q < \frac{1}{2}$,
and that the stepsize $\eta$ is bounded by an absolute constant.

\smallskip

\begin{restatable}{lem}{lemdriftbounds}
	\label{lemma:2x2-drift-bounds-unified}
	Fix $0 < \delta_p \le \delta_q < \frac{1}{2}$
	and assume Setting~\ref{setting:2x2-canonical}.
	Let $\{w_t\}$ and $\{z_t\}$ be the primal and dual
	iterates of OMWU with stepsize $0 < \eta < \frac{1}{8}$,
	initialized from some $w_0 \in \relint(\calW)$.
	For each $t \ge 1$,
	let $\w x_t = x_t(1), \w y_t = y_t(1)$
	and $\w p_t = p_t(1), \w q_t = q_t(1)$.
	Fix any $k \ge 1$.
	Then the following hold: 
	\begin{enumerate}[
		label={(\roman*)},
		itemsep=0em,
		topsep=0em,
		leftmargin=3em,
	]
		\item
		Suppose $\w q_{k-1} \le 1 - 3 \delta_q$.
		Then:
		$\w x_{k+1} - \w x_k \ge \frac{\eta}{2} \cdot \delta_p \delta_q$.

		\item
		Suppose $\w q_{k-1} \ge 1 - \frac{\delta_q}{3}$.
		Then: 
		$\w x_{k+1} - \w x_k \le - \frac{\eta}{3} \cdot \delta_p \delta_q$.

		\item
		Suppose $\w p_{k-1} \le 1 - 3 \delta_p$. 
		Then:
		$\w y_{k+1} - \w y_k \le -  \frac{\eta}{2} \cdot \delta_p \delta_q $.

		\item
		Suppose $\w p_{k-1} \ge 1- \frac{\delta_p}{3}$.
		Then: 
		$\w y_{k+1} - \w y_k \ge \frac{\eta}{3} \delta_p \delta_q$. 
	\end{enumerate}

\end{restatable}

\smallskip

\begin{proof}
We start by proving Part (i) of the lemma,
and we note that Parts (ii), (iii), and (iv) will follow
by nearly identical calculations.
For this, using the update rule of 
Proposition~\ref{prop:omwu-pd-2x2},
it follows that 
\begin{equation}
	\textstyle
	\w x_{k+1} - \w x_k \ge \frac{\eta}{2} \delta_p \delta_q
	\iff
	1 - \delta_q - 2 \w q_k + \w q_{k-1} \ge \frac{\delta_q}{2} \;.
	\label{eq:qd-01}
\end{equation}
For this, let $\w q_{k-1} = 1-\alpha_{k-1}$
for $\alpha_{k-1} \in (0,1)$. 
Our goal is to derive a lower bound on $\w q_k$
in terms of $\w q_{k-1}$ that will facilitate
establishing the inequality in the right hand side 
of~\eqref{eq:qd-01}.
For this, observe again by the dual update rule
of Proposition~\ref{prop:omwu-pd-2x2} that
\begin{equation}
	\w y_k - \w y_{k-1}
	= 
	\eta \delta_q 
	\big(
		\w p_k - (1-\delta_p) + \w p_k - \w p_{k-1}
	\big) 
	\le
	2 \eta \delta_q \,
	\label{eq:qd-02}
\end{equation}
where the inequality follows from the fact that
$\w p_k \le 1$, that $\w p_k - \w p_{k-1} \le 1$,
and that $\delta_p > 0$.
Then by definition of $\w q_k = \sigmoid(\w y_k / \delta_q)$,
and using the fact that $1 - \sigmoid(u) = 1/(1+ \exp(u))$
for any $u \in \R$, it follows that
\begin{align*}
	1 - \w q_k 
	= 
	\frac{1}{1 + \exp(\w y_k / \delta_q)} 
	&\ge
	\frac{1}{1+ \exp(\w y_{k-1}/\delta_q)\cdot \exp(2 \eta)} \\
	&\ge 
	\frac{1}{(1+ \exp(\w y_{k-1}/\delta_q))\cdot \exp(2 \eta)} 
	= 
	\exp(- 2 \eta) (1- \w q_{k-1}) \;.
\end{align*}
Here, the first inequality comes from applying the bound 
$\w y_{k} \le \w y_{k-1} + 2 \eta \delta_q$
from~\eqref{eq:qd-02}, 
the second inequality comes from
$\exp(2\eta) \ge 1$,
and the final equality comes from 
$\w q_{k-1} = \sigmoid(\w y_{k-1}/\delta_q)$.
As $\exp(-u) \ge 1-u$ for all $u \in \R$, we then further have
$1-\w q_k \ge (1-2 \eta) (1-\w q_{k-1})$.
Using the definition $1- \w q_{k-1} = \alpha_{k-1}$,
this means that $\w q_{k} \le 1 - (1-2\eta) \alpha_{k-1}$,
and thus we have 
\begin{align*}
	1-\delta_q - 2 \w q_k + \w q_{k-1}
	&\ge 
	1-\delta_q - 2(1 - (1-2\eta)) \alpha_{k-1} + 1 - \alpha_{k-1} 
	= 
	(1- 4\eta) \alpha_{k-1}
	- \delta_q \;.
\end{align*}
Now under the assumption that 
$\w q_{k-1} = 1-\alpha_{k-1} \le 1 - 3 \delta_q$,
we have $\alpha_{k-1} \ge 3 \delta_q$.
Moreover as we assume $\eta \le 1/8$, we have
$1-4\eta \ge 1/2$. This allows us to further write
\begin{equation*}
	\textstyle
	1- \delta_q - 2 \w q_k + \w q_{k-1}
	\ge 
	\frac{3}{2}\delta_q - \delta_q
	= 
	\frac{1}{2} \delta_q\;,
\end{equation*}
which exactly establishes the right-hand side
of~\eqref{eq:qd-01} and thus the statment of Part (i)
of the lemma.

The proof of Part (ii) follows via nearly identical calculations
as the poof of Part (i).
In particular, we further use the identity 
$\exp(u) = \sigmoid(u)/(1-\sigmoid(u))$
for $u \in \R$ and the fact that $\exp(2\eta) \le 1 + 4\eta$
when $\eta \le 1/8$.
The proofs of Part (iii) and Part (iv) then follow
identically to Parts (i) and (ii), respectively.
This is due to the symmetry of the dual update rules 
for $\w x_{k+1}$ and $\w y_{k+1}$ in 
Proposition~\eqref{prop:omwu-pd-2x2},
and thus we omit these calculations. 
\end{proof}

\subsubsection{Maximum Width of Dual Coordinates}

In the following proposition, we establish 
a bound on the coordinates of a dual variable
in terms of its energy value.
This is used in the sequel to establish an upper
bound on the maximum distance between any two dual 
(and also primal) OMWU iterates. Formally:

\smallskip

\begin{restatable}[Energy bound implies dual coordinate bound]
{prop}{propmaxdualwidth}
	\label{prop:2x2-max-dual-width}
	Fix $\delta_p, \delta_q \in (0, 1)$
	and assume Setting~\ref{setting:2x2-canonical}.
	Fix $z = (x, y) \in \calZ$, and let $\alpha > 0$.
	If $F(z) \le \alpha$, then the following hold:
	\begin{equation*}
		\textstyle
		- \big(\frac{\delta_p}{1-\delta_p} \big) \cdot \alpha 
		\le x(1) \le \alpha
		\quad\text{and}\;\;
		- \big( \frac{\delta_q}{1-\delta_q}\big) \cdot \alpha
		\le y(1) \le \alpha \;.
	\end{equation*}
\end{restatable}

\smallskip

\begin{proof}
Fix $z = (x, y) \in \calZ$, and let $\w x = x(1) \in \R$ and 
$\w y = y(1) \in \R$.
Recall by definition of the energy function $F$
and its simplification over $\calZ$ from 
Proposition~\ref{prop:2x2-energy-map-simplified} that
\begin{equation*}
	F(z) 
	= \LSE(x) + \LSE(y) 
	= \big(\w x + \log(1+\exp(-\w x/\delta_p))\big)
	+ \big(\w y + \log (1 + \exp(-\w y / \delta_q))\big)
	\;.
\end{equation*}
Note that over $\calZ$, both terms
$\w x + \log(1+\exp(-\w x/\delta_p)$
and 
$\w y + \log(1+\exp(-\w y/\delta_q)$
are convex (in particular, minimized at the coordinates 
$x^\star(1)$ and $y^\star(1)$ from
Corollary~\eqref{cor:2x2-dual-nash})
and always positive for all $\delta_p, \delta_q \in (0, 1)$.
Thus, for any $z = (x, y) \in \calZ$,
both  $\LSE(x) \ge 0$ and $\LSE(y) \ge 0$.
Therefore, $F(z) \le \alpha$ implies that both
\begin{equation*}
	\LSE(x) \le \w x + \log\big(1 + \exp(-\w x / \delta_p))  \le \alpha
	\;\;\text{and}\;\;
	\LSE(y) \le \w y + \log\big(1 + \exp(-\w y / \delta_q)) \le \alpha\;.
\end{equation*}
For the $\w x $ variable, observe further that
$\log(1 + \exp(-\w x / \delta_p)) \ge 0$, and thus
\begin{equation*}
	\w x \le \alpha - \log(1+\exp(-\w x /\delta_p)) \le \alpha \;.
\end{equation*}
This proves the upper bound on $x(1) = \w x$.
For the lower bound, recall that
$\log(1+\exp(-u)) \ge -u$ for any $u \in \R$.
Thus if $\LSE(x) \le \alpha$, then
\begin{equation*}
\alpha 
\ge 
\w x + \log(1+\exp(-\w x / \delta_p))
\ge \w x - ( \w x/\delta_p) = \w x((\delta_p - 1)/\delta_p)
= - \w x((1-\delta_p)/\delta_p) \;.
\end{equation*}
Then rearranging gives
\begin{equation*}
	\w x \ge - (\delta_p/ (1-\delta_p)) \cdot \alpha \;,
\end{equation*}
which proves the lower bound on $x(1)$.
The bounds on $y(1) = \w y$ follow by identical calculations. 
\end{proof}

\subsection{Structural Properties of Distances to Nash in 2x2 Setting}
\label{app:2x2-setup:energy-distances}

In this section, we give simplified expressions
for duality gap, TV distance, and KL divergence
under Setting~\ref{setting:2x2-canonical}
(i.e., when $A := A_{\delta_p, \delta_q}$
is the canonical matrix from Definition~\ref{def:A-2x2}).
For this, we consider throughout the following subsections
$w = (p, q) \in \calW = \Delta_2 \times \Delta_2$, and we use the 
notation $\w p = p(1) \in [0, 1]$ and $\w q = q(1) \in [0, 1]$ for readability. 
Moreover, recall by Proposition~\ref{prop:A-properties}
that the unique and interior NE $w^\star = (p^\star, q^\star)$ of 
$A = A_{\delta_p, \delta_q}$ 
is given by $p^\star = (1-\delta_p, \delta_p)$
and $q^\star = (1-\delta_q, \delta_q)$.

\subsubsection{Duality Gap under Setting~\ref{setting:2x2-canonical}}
\label{sec:2x2-dg}

For $w = (p, q) \in \calW$,
recall that duality gap is defined as
\begin{equation*}
	\textstyle
	\DG(w) = \max_{q' \in \Delta_2}
	\langle q', A^\top p \rangle
	- \min_{p' \in \Delta_2} \langle p', A q\rangle .
\end{equation*}
Then in terms of the leading coordinates $\w p = p(1)$ and $\w q = q(1)$,
we have by definition of $A_{\delta_p, \delta_q}$:
\begin{align*}
	\DG(w) 
	&= 
	\max\big\{ (A^\top p)(1), (A^\top p)(2) \big\}
	+ 
	\max\big\{ (-Aq)(1), (-Aq)(2) \big\} \\
	&= 
	\max \big\{  
		(\w p - (1-\delta_p))\delta_q, 
		(1-\delta_p - \w p) (1-\delta_q)
	\big\}  \\
	&\qquad\qquad + 
		\max \big\{
			(1- \delta_q - \w q)\delta_p, 
			(\w q - (1-\delta_q)) (1-\delta_p)
		\big\} \;.
\end{align*}
Here, we also used the fact that $p(2) = 1 - \w p$
and $q(2) = 1- \w q$.
Thus $\DG(w)$ has a piecewise linear structure, 
which we can summarize depending on the 
location of $(\w p, \w q)$ relative to 
$(1-\delta_p, 1-\delta_q)$
(i.e., the leading coordinates of the unique and interior NE of 
$A_{\delta_p, \delta_q}$).
Then $\DG(w)$ is given by:
\begin{equation}
	\DG(w) = 
	\begin{cases}
		(\w p - (1-\delta_p)) \delta_q + (\w q - (1-\delta_q)) (1-\delta_p) 
        & \text{if $\w p \ge 1-\delta_p$, \, $\w q \ge 1- \delta_q$} \\
		(1-\delta_p - \w p)(1-\delta_q) + (\w q - (1-\delta_q)) (1-\delta_p) 
        & \text{if $\w p \le 1-\delta_p$, \, $\w q \ge 1- \delta_q$} \\
		(1-\delta_p - \w p )(1-\delta_q) + (1-\delta_q - \w q)\delta_p 
        & \text{if $\w p \le 1-\delta_p$, \, $\w q \le 1-\delta_q$} \\
		(\w p - (1-\delta_p))\delta_q + (1-\delta_q - \w q)\delta_p 
        & \text{if $\w p \ge 1-\delta_p$, \, $\w q \le 1-\delta_q$}\;.
	\end{cases}
	\label{eq:dg-2x2-full}
\end{equation}

\subsubsection{TV Distance under Setting~\ref{setting:2x2-canonical}}
\label{sec:2x2-tv}

Recall by definition that the total variation distance 
$\TV(w^\star, w)$ is given by 
\begin{equation*}
	\textstyle
	\TV(w^\star, w) = 
	\TV(p^\star, p)
	+ \TV(q^\star, q)
	= 
	\frac{1}{2}\|p^\star - p\|_1
	+ 
	\frac{1}{2}\|q^\star - q\|_1 \;.
\end{equation*}
As $p = (\w p, 1-\w p)$ and $p^\star = (1-\delta_p, \delta_p)$,
observe that
\begin{equation*}
	\textstyle
 	\|p^\star - p\|_1
	= 
	| 1- \delta_p - \w p | 
	+ 
	| \delta_p - (1 - \w p) |
	= 
	2 \cdot | 1 - \delta_p - \w p | \;.
\end{equation*}
Similarly, $\|q^\star - q\|_1 = 2 \cdot | 1- \delta_q - \w q|$.
Thus we have that 
\begin{equation}
	\TV(w^\star, w)
	= \TV(p^\star, p)
	+ \TV(q^\star, q)
	= 
	|1-\delta_p - \w p|
	+ 
	| 1- \delta_q - \w q |\;.
	\label{eq:tv-2x2-full}
\end{equation}

\subsubsection{KL Divergence under 
Setting~\ref{setting:2x2-canonical}}
\label{sec:2x2-kl}

Recall by definition that the KL divergence 
$\KL(w^\star, w)$ is given by
\begin{equation*}
	\KL(w^\star, w)
	= 
	\KL(p^\star, p)
	+ \KL(q^\star, q) \;.
\end{equation*}
Then under this setting, the components $\KL(p^\star, p)$
and $\KL(q^\star, q)$ can be simplified as:
\begin{equation}
	\begin{aligned}
	\KL(p^\star, p) 
	&= 
	\sum\nolimits_{i=1}^2 p^\star(i) \log\Big(\frac{p^\star(i)}{p(i)}\Big)
	= 
	(1-\delta_p) \cdot \log\Big(\frac{1-\delta_p}{\w p}\Big)
	+ \delta_p \cdot \log\Big(\frac{\delta_p}{1-\w p}\Big)\\
	\KL(q^\star, q) 
	&= 
	\sum\nolimits_{j=1}^2 q^\star(j) \log\Big(\frac{q^\star(j)}{q(j)}\Big)
	= 
	(1-\delta_q) \cdot \log\Big(\frac{1-\delta_q}{\w q}\Big)
	+ \delta_p \cdot \log\Big(\frac{\delta_q}{1-\w q}\Big)\;.
	\end{aligned}
	\label{eq:kl-2x2-full}
\end{equation}

\subsubsection{Comparison of Levelsets Between Distances}
\label{sec:2x2-levelsets-compare}

In the following set of figures, we plot the levelsets 
of $\DG(w)$, $\TV(w^\star, w)$, and $\KL(w^\star, w)$
for the $2\times 2$ setting using the simplified
expressions from Section~\ref{app:2x2-setup:energy-distances}.
We plot the levelsets over the leading coordinates
$(p(1), q(1))$ 
for three different settings  of $w^\star = (p^\star ,q^\star)$.
In each subplot, $w^\star$ is denoted by the yellow star. 
The figures highlight a fundamental distinction between 
the geometries of these distance functions over the simplex.

\paragraph{Figure~\ref{fig:levelset-uniform}: Uniform NE.}
Here, $p^\star = q^\star = (0.5, 0.5)$.
\captionsetup[figure]{labelfont=small}
\begin{figure}[htb!]
	\centering
	\includegraphics[width=0.85\textwidth]{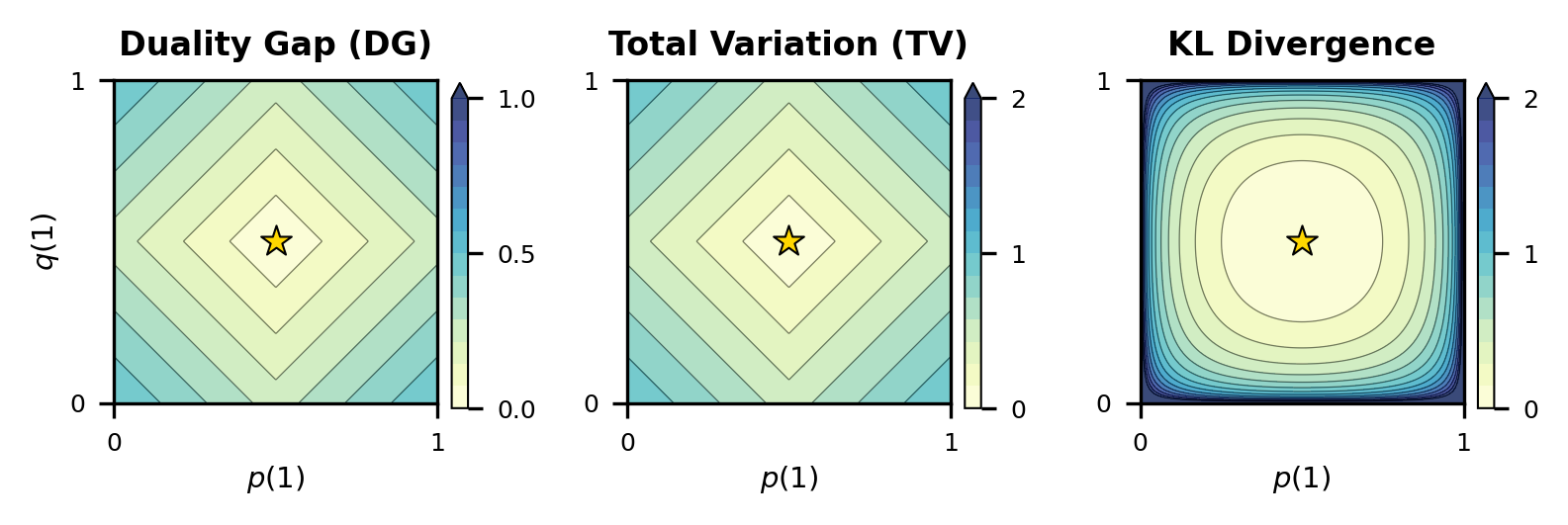}
	\vspace*{-1em}	
	\caption{%
	    \footnotesize
	    $\DG(w)$, $\TV(w^\star, w)$, and $\KL(w^\star, w)$ levelsets
	    in the $2\times 2$ setting 
	    for $p^\star = q^\star = (0.5, 0.5)$.
	}   
	\label{fig:levelset-uniform}
\end{figure}

\paragraph{Figure~\ref{fig:levelset-boundary-1}: 
Boundary NE in one component.}
Here, $p^\star = (0.9, 0.1)$
and $q^\star =(0.5, 0.5)$.

\captionsetup[figure]{labelfont=small}
\begin{figure}[htb!]
	\centering
	\includegraphics[width=0.85\textwidth]{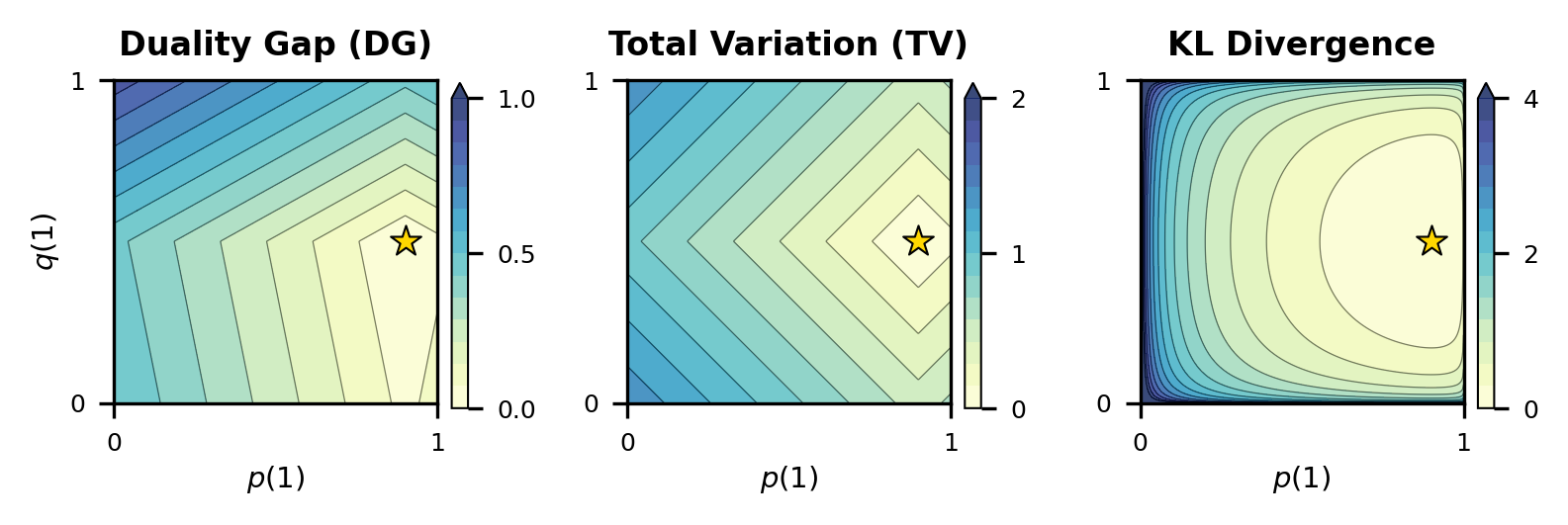}
	\vspace*{-1em}	
	\caption{%
	    \footnotesize
	    $\DG(w)$, $\TV(w^\star, w)$, and $\KL(w^\star, w)$ levelsets
	    in the $2\times 2$ setting 
	    for $p^\star = (0.9, 0.1)$, and $q^\star = (0.5, 0.5)$.
	}    
	\label{fig:levelset-boundary-1}
\end{figure}

\paragraph{Figure~\ref{fig:levelset-boundary-2}: 
Boundary NE in both components.}
Here, $p^\star = q^\star = (0.9, 0.1)$.

\captionsetup[figure]{labelfont=small}
\begin{figure}[htb!]
	\centering
	\includegraphics[width=0.9\textwidth]{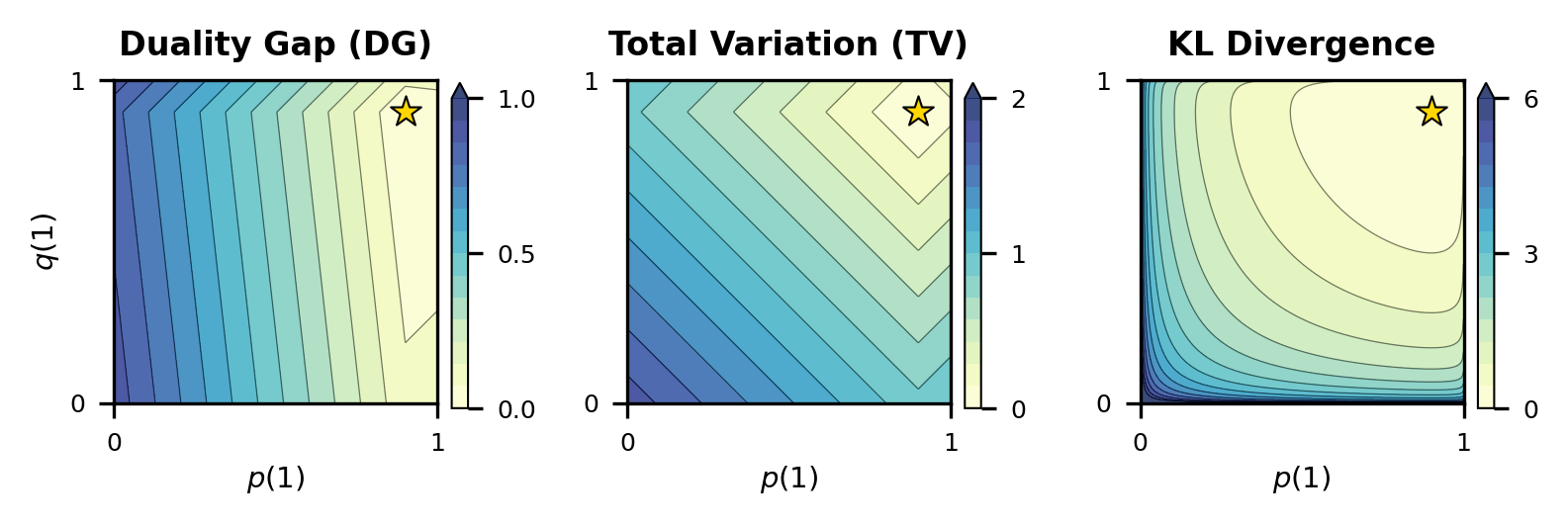}
	\vspace*{-1em}
	\caption{%
	    \footnotesize
	    $\DG(w)$, $\TV(w^\star, w)$, and $\KL(w^\star, w)$ levelsets
	    in the $2\times 2$ setting 
	    for $p^\star = q^\star = (0.9, 0.1)$.
	}    
	\label{fig:levelset-boundary-2}
\end{figure}

\bigskip

\paragraph{Discussion.}
Observe that the geometry of the levelsets of each distance function
is highly dependent on the location of the Nash equilibrium $w^\star$.
In particular, when $w^\star$ is closer to the simplex boundary 
(Figure~\ref{fig:levelset-boundary-1} and Figure~\ref{fig:levelset-boundary-2}),
the $\KL(w^\star, w)$ levelsets become 
``squished'' near the boundary closest to $w^\star$.
Note in these regions that a primal variable $w = (p, q)$  can be 
simultaneously close to $w^\star$ in $\TV$ distance
(which is symmetric), yet still far away in $\KL(w^\star, w)$. 
This property is captured in our uniform best-iterate lower
bound of Theorem~\ref{thm:uniform-lb-secondary}, Part (1).
Additionally, observe in Figure~\ref{fig:levelset-boundary-2} 
that for $w^\star$ close to a vertex, a primal variable
$w = (p, q)$ with, for example, $p(1)= 0.95$
and $q(1) = 0.5$, can simultaneously be close in $\DG(w)$,
but far in $\TV(w^\star, w)$ and $\KL(w^\star, w)$.
This property is captured by the separation in uniform 
best-iterate guarantees between
$\TV$ and $\KL$ (via the constant lower bound of 
Theorem~\ref{thm:uniform-lb-main})
and $\DG$ (via the $\widetilde O(T^{-1/2})$
best-iterate rate of Theorem~\ref{thm:dg-best-2x2}).
Finally, note that the location of the NE $w^\star$
in the example of Figure~\ref{fig:levelset-boundary-1}
roughly corresponds to that of the hard $2\times 2$
instance used in the results of~\cite{FGKLLZ24}.


\section{Details on Lower Bound for Universal Last-Iterate Convergence in KL}
\label{app:kl-last-lower}

This section gives the proof of
Theorem ~\ref{thm:kl-last-lower},
which gives a universal lower bound on the 
convergence rate of OMWU in $\KL$ divergence.
We first restate the theorem:

\thmkllastlower*

\smallskip

Observe that in Theorem~\ref{thm:kl-last-lower},
the dependence on $\delta$ in the lower bound
matches that in the upper bound of Theorem~\ref{thm:kl-last-unified},
indicating that our analysis is tight.
The proof of the theorem follows two steps that we outline here:

\paragraph{High-level overview of proof.}
Let $\{w_t\}$ and $\{z_t\}$ denote the
primal and dual iterates of the OMWU algorithm.
We have by Corollary~\ref{cor:change-kl-energy-equiv-omwu}
and the lower bound on energy dissipation 
from Lemma~\ref{lem:energy-one-step-full} that, for all $t \ge 0$:
\begin{equation*}
    \KL(w^\star, w_{t+1})
    - \KL(w^\star, w_t)
    = F(z_{t+1}) - F(z_t)
    \ge - \tfrac{5}{4} \eta^2 \|J \nabla F(z_t)\|^2_{z_t} \;.
\end{equation*}
Recall in Proposition~\ref{prop:dissipation-kl-bound}
that we established a lower bound on the dissipation
term $\|J \nabla F(z_t)\|^2_{z_t}$. This translates
into the upper bound on last-iterate convergence rate in $\KL$
in Theorem~\ref{thm:kl-last-unified}.
For the present lower bound of Theorem~\ref{thm:kl-last-lower},
we first establish in Proposition~\ref{prop:kl-dissipation-upper-2x2}
an analogous \textit{non-uniform upper bound} on the
dissipation term $\|J \nabla F(z)\|^2_{z}$.

In particular, this upper bound is a structural and general
property that we prove for the $2\times 2$ setting, 
and which holds over a region of the primal space $\calW$
near the simplex boundary. This leads to a much tighter
bound compared to the crude uniform upper bound of 
Proposition~\ref{prop:skew-grad-uniform-upper}.
Similar to the lower bound on $\|J \nabla F(z)\|^2_{z}$,
the upper bound of Proposition~\ref{prop:kl-dissipation-upper-2x2}
contains a non-uniform multiplicative factor 
that depends on the minimum coordinates $p_{\min}$ 
and $q_{\min}$ of the corresponding primal variable $w = (p, q) = \nabla F(z)$.
This further indicates that the local state-dependence
in the lower bound of Proposition~\ref{prop:dissipation-kl-bound}
is necessary. 

Then, the last-iterate lower bound of Theorem~\ref{thm:kl-last-lower}
follows by constructing a $2\times 2$ payoff matrix and a set of 
initializations where we establish
an upper bound on $p_{t, \min}$ and $q_{t, \min}$
that is exponentially small in the minimum Nash coordinate 
$\delta$, and that holds for all iterates $t$ over a fixed time horizon.
This relies on the preliminaries for OMWU 
in the $2\times 2$ setting that we developed in Section~\ref{app:omwu-2x2}.

\paragraph{Roadmap of the section.}
In Section~\ref{app:dissipation-kl:upper}, 
we state and prove the intermediate result of 
Proposition~\ref{prop:kl-dissipation-upper-2x2}.
The full proof of Theorem~\ref{thm:kl-last-lower}
then follows in Section~\ref{app:kl-last-iterate:lower}.

\subsection{Non-Uniform Upper Bound on Dissipation Term}
\label{app:dissipation-kl:upper}

In this section, we establish a non-uniform 
\textit{upper bound} on the dissipation term
$\|J \nabla F(z)\|^2_z$. 
For simplicity, we focus on the low-dimensional 
setting where $A \in \R^{2 \times 2}$.

\begin{restatable}[Non-uniform upper bound on dissipation term]
{prop}{propkldissipationuppertwotwo}
    \label{prop:kl-dissipation-upper-2x2}
    Let $A \in \R^{2\times 2}$
    have a unique and interior Nash equilibrium
    $w^\star = (p^\star, q^\star)$ such that
    $p^\star = q^\star = (1-\delta, \delta)$
    for $\delta \in (0, 0.5)$.
    Let $\sigma_{\max} = \|A\|_2$.
    Consider $z \in \R^{2+2}$ and 
    $w = (p,q) = \nabla F(z)$
    such that $p_{\min} = p(2) \le \delta$ 
    and $q_{\min} = q(2) \le \delta$. 
    Define $\widehat w_{\min} = \max(p_{\min},q_{\min})$.
    Then: 
    \begin{equation*}
        \|J \nabla F(z)\|^2_z
        \le 
        8 \cdot \sigma^2_{\max} \cdot \widehat w_{\min} \cdot \delta 
        \cdot \KL(w^\star, w) \;.
    \end{equation*}
\end{restatable}

\begin{proof}
To prove the lemma, we leverage the primal
perspective of the non-uniform gradient domination 
result from Proposition~\ref{prop:gradient-dom-primal}. 
Specifically, using part (i) of Proposition~\ref{prop:gradient-dom-primal},
we have
\begin{equation}
    \|J \nabla F(z)\|^2_z 
    = \Var_{p}(Aq) + \Var_q(A^\top p)\;.
    \label{eq:gdu-01}
\end{equation}
We will thus derive upper bounds on the individual
terms in~\eqref{eq:gdu-01} via the following steps:

\noindent
\textbf{1. Exact Characterization of Variance in $2 \times 2$ Setting.} 

\noindent
Let $p = (p(1) , p(2))$ and $q = (q(1), q(2))$. 
We first prove that 
\begin{align}
    \Var_p(Aq) &\le 2 \cdot p(1) \cdot p(2) \cdot \|A\|^2_2 \cdot \|q - q^\star\|^2_2 
    \label{eq:v2g-01} \\
    \text{and}\quad
    \Var_q(A^\top p) &\le 2 \cdot q(1) \cdot q(2) \cdot \|A\|^2_2 \cdot \|p - p^\star\|^2_2 \;.
    \label{eq:v2g-02}
\end{align}
To prove~\eqref{eq:v2g-01},
let $v = Aq \in \R^{2}$.
By definition of $\Var_p(v)$ from~\eqref{eq:variance-def},
it is straightforward to compute in the two-dimensional case 
that $\Var_p(v) = p(1)\cdot p(2) \cdot (v(1)-v(2))^2$.
Now, let let $\1 \in \R^2$ denote the all-ones vector. 
By Proposition~\ref{prop:RH-inner-bound}, 
we have $Aq^\star = c\1$ for some $c \in \R$.
Thus, we can write:
\begin{equation*}
    v 
    = Aq 
    = A(q^\star + (q - q^\star))
    = c \1 + A(q - q^\star) \;.
\end{equation*}
Further letting $d = (1, -1)$, it follows that we can write
and simplify
\begin{align*}
    v(1) - v(2) 
    = 
    \langle d, v \rangle 
    &= 
    \langle 
    d, c\1 + A(q - q^\star)
    \rangle \\
    &= \langle d, A(q-q^\star) \rangle  \\
    &\le 
    \|d\|_2 \cdot \|A (q-q^\star) \|_2
    \le
    \sqrt{2} \cdot \|A\|_2 \cdot \|q - q^\star \|_2 \;.
\end{align*}
Substituting this bound on $v(1)- v(2)$ into the
definition of $\Var_p(v)$ then exactly yields
the desired inequality of~\eqref{eq:v2g-01}.
By an identical calculation, we also 
find $\Var_q(A^\top p) \le 2 q(1) q(2) \|A\|^2_2 \|p-p^\star\|^2_2$, 
where we additionally use the fact that
$\|A\|_2 = \|A^\top\|_2$.

\smallskip 

\noindent
\textbf{2. Relating KL to Euclidean Distance to Nash.}

\noindent
Next, we prove the following relationships between
$\KL$ and euclidean distance to Nash:
\begin{align}
    \|p-p^\star\|^2_2 
    &\le 4 \cdot \max(p_{\min}, \delta_p)\cdot \KL(p^\star, p) 
    \label{eq:p2-003} \\
    \text{and}\quad
    \|q-q^\star\|^2_2 
    &\le 
    4 \cdot \max(q_{\min}, \delta_q) \cdot \KL(q^\star, q) \;.
    \label{eq:p2-004}
\end{align}
We start by proving~\eqref{eq:p2-003}. 
For this, recall from Section~\ref{sec:omwu}
that for $w = (p, q) \in \calW$, we write 
$R(w) = R_m(q) + R_n(q)$
to denote the (separable) negative entropy function,
where $R_m = - \ent_m$ and $R_n = - \ent_n$.
Moreover, $\KL(w^\star, w) = D_R(w^\star, w)$ is the 
Bregman divergence of $R$,
and more specifically $\KL(p^\star, p) = D_{R_m}(p^\star, p)$
and $\KL(q^\star, q) = D_{R_n}(q^\star, q)$.
Thus, by using the integral form of Bregman divergences
(see, e.g.,~\cite{nocedal2006numerical}, Theorem 2.1), 
it follows that 
\begin{equation}
    \KL(p^\star, p)
    = D_{R_m}(p^\star, p)
    = 
    \int_{0}^1 (1-s) 
    \big\langle
    (p-p^\star), \nabla^2 R_m(p_s) (p-p^\star)
    \big\rangle\, ds  \;,
    \label{eq:kl-int}
\end{equation}
where we write $p_s = p - s(p-p^\star)$ for $s \in [0, 1]$.

We simplify the integrand in~\eqref{eq:kl-int} as follows:
first, by definition in~\eqref{eq:R-hessian-full},
we have for all $s\in [0, 1]$ 
\begin{equation*}
    \textstyle
    \nabla^2 R_m(p_s) = \Diag\left(\big(\frac{1}{p_s(1)}, \dots, \frac{1}{p_s(n)}\big) \right) \;.
\end{equation*}
Thus for each $s$, using the facts that $m=2$ and
$p^\star = (1-\delta_p, \delta_q)$
we can simplify
\begin{align*}
    \big\langle
        (p-p^\star), \nabla^2 R_m(p_s) (p-p^\star)
    \big\rangle
    &= 
    \ssum\nolimits_{i=1}^2 \frac{(p(i) - p^\star(i))^2}{p_s(i)} \\
    &= 
    \frac{(p(1) - (1-\delta_p))^2}{p_s(1)}
    + 
    \frac{(1- p(1) - \delta_p)^2}{p_s(2)} \\
    &= 
    \big(p(1) - (1-\delta_p)\big)^2 
    \cdot 
    \bigg(
        \frac{1}{p_s(1)}
        + 
        \frac{1}{p_s(2)}
    \bigg) \\
    &=
    \frac{1}{2} \cdot \|p- p^\star\|^2_2 \cdot \frac{1}{p_s(1) \cdot p_s(2)} \;.
\end{align*}
Here, the last line follows from 
the fact that $\|p - p^\star\|^2_2 = 2 \cdot (p(1)-(1-\delta_p))^2$ 
in this two-dimensional setting. 
Thus, we can further rewrite~\eqref{eq:kl-int} as
\begin{equation}
    \KL(p^\star, p) = 
    \frac{1}{2} \cdot \|p - p^\star\|^2_2  \cdot \int_0^1 \frac{(1-s)}{p_s(1) \cdot p_s(2)} \, ds \;.
    \label{eq:kl-int-02}
\end{equation}
To further derive a lower bound on~\eqref{eq:kl-int-02}, 
we will show a uniform upper bound on 
$p_s(1)\cdot p_s(2) = p_s(2) \cdot (1-p_s(2))$ for all $s \in [0, 1]$. 
For this, we recall by the assumptions of the lemma
that $p^\star(2) = \delta = \min(p^\star(1), p^\star(2)) \le \frac{1}{2}$,
and also $p(2) = p_{\min} =  \min(p(1), p(2)) \le \frac{1}{2}$.
In particular, this also implies that $p_s(2) = \min(p_s(1), p_s(2)) \le \frac{1}{2}$
for all $s \in [0, 1]$. 
Then by concavity, for all $s \in [0, 1]$:
\begin{equation*}
    p_s(2) \cdot (1-p_s(2)) \le p_s(2) \le \max(p(2), \delta_p)
    = \max(p_{\min}, \delta_p)\;.
\end{equation*}
Substituting this uniform bound into~\eqref{eq:kl-int-02},
and using the fact that $\int_{0}^1 (1-s) ds = \frac{1}{2}$,
we find
\begin{equation*}
    \KL(p^\star, p) 
    \ge 
    \frac{1}{4} \cdot \|p - p^\star \|^2_2 
    \cdot \frac{1}{\max(p_{\min}, \delta_p)} \;.
\end{equation*}
Rearranging then yields the desired inequality from~\eqref{eq:p2-003}.
Using identical calculations similarly yields the inequality
from~\eqref{eq:p2-004}.

\noindent
\textbf{3. Combining the Pieces.}

\noindent
Combining the inequalities of~\eqref{eq:v2g-01} and~\eqref{eq:v2g-02}
from Step 1 and
expressions~\eqref{eq:p2-003} and~\eqref{eq:p2-004} 
from Step 2, we find
\begin{equation}
    \begin{aligned}
    \Var_p(Aq)
    &\le 
    8 \cdot \|A\|^2_{2} \cdot p_{\min} \cdot \max(q_{\min}, \delta)
    \cdot \KL(q^\star, q) \\
    \text{and}\quad
    \Var_q(A^\top p)
    &\le
    8 \cdot \|A\|^2_{2} \cdot q_{\min} \cdot \max(p_{\min}, \delta) 
    \cdot \KL(p^\star, p) \;.
    \end{aligned}
    \label{eq:p2-05}
\end{equation}
Here, we additionally used the fact that
$p(1)p(2) = (1-p(2)) p(2) \le p(2) = p_{\min}$
and similarly $q(1) q(2) \le q_{\min}$.
Then under the assumptions that $q_{\min} < \delta$ 
and $p_{\min} < \delta$, and using the 
notation $\widehat w_{\min} = \max(p_{\min}, q_{\min})$,
substituting the bounds of~\eqref{eq:p2-05}
into~\eqref{eq:gdu-01} and simplifying yields
\begin{align*}
    \|J \nabla H(z)\|^2_z
   &= 
    \Var_p(Aq) + \Var_q(A^\top p) \\
    &\le 
    8 \cdot \|A\|^2_2 \cdot 
    \big(
        \widehat w_{\min} \cdot \delta \cdot \KL(p^\star, p)
        + 
        \widehat w_{\min} \cdot \delta \cdot \KL(q^\star, q)
    \big) \\
    &= 
    8 \cdot \sigma^2_{\max} \cdot
    \widehat w_{\min} \cdot \delta 
    \cdot \KL(w^\star, w_t) \;.
\end{align*}
This yields the statement of the lemma.
\end{proof}

\subsection{Proof of Theorem~\ref{thm:kl-last-lower}}
\label{app:kl-last-iterate:lower}

We begin the proof by giving a more detailed overview of the main steps:
\begin{enumerate}[
    label={(\arabic*)},
    leftmargin=2em
]
    \item
    We consider the payoff matrix $A \in [-1, 1]^{2 \times 2}$ 
    that is the instantiation of the canonical $2 \times 2$
    matrix $A_{\delta_p, \delta_q}$
    from Definition~\ref{def:A-2x2}, 
    with $\delta_p = \delta_q = \delta \in (0, 1/2)$.
    Given the time horizon $T$, we set $\delta = 1/T$,
    which means the interior NE $w^\star = (p^\star, q^\star)$ of $A$ 
    will be near a vertex of $\calW$.
    We also consider initializations $w_0 \in \calW$ 
    that are near the same vertex of $\calW$.

    \item
    Using the control of the OMWU trajectory on
    the matrix $A_{\delta_p, \delta_q}$ established
    in Section~\ref{app:2x2-setup:primal-dual:omwu-iterates}, 
    our setting of $\delta$ and the set of initializations
    allows for proving that the minimum coordinates
    $p_{t, \min}$ and $q_{t, \min}$ of the OMWU iterates
    have an \textit{upper bound}
    of $\exp(-1/\delta)$ that holds uniformly
    over all iterates $t \in [T]$.
    Note that this matches (in its dependence on $\delta$) 
    the general uniform lower bound on 
    $p_{t, \min}$ and $q_{t, \min}$ from Lemma~\ref{lem:wtmin-bound-uniform}.

    \item
    Then, using the \textit{lower bound} on energy dissipation
    from Lemma~\ref{lem:energy-one-step-full}
    and the \textit{upper bound} on the dissipation term 
    $\|J \nabla F(z)\|^2_{z}$ from
    Proposition~\ref{prop:kl-dissipation-upper-2x2},
    we obtain the final rate.
\end{enumerate}

We proceed to give the details of these steps below:

\smallskip

\noindent
\textbf{Setup of payoff matrix and set of initializations.}

\noindent
Given $T \ge 3$, 
let $\delta = 1/T < 1/2$.
Let $A := A_{\delta, \delta} \in [-1, 1]^{2\times 2}$
be the instantiation of the matrix $A_{\delta_p, \delta_q}$
from Definition~\ref{def:A-2x2}
with $\delta_p = \delta_q = \delta$. 
Thus $A$ is the matrix 
\begin{equation*}
    A = A_{\delta,\delta}
    = 
    \begin{pmatrix}
        \delta^2 & -\delta(1-\delta) \\
        -\delta(1-\delta) & (1-\delta)^2  
    \end{pmatrix} \;.
\end{equation*}
Moreover, recall from Proposition~\ref{prop:A-properties}
that $A$ has the unique and interior NE 
$w^\star = (p^\star,q^\star)$ 
with $p^\star = q^\star = (1-\delta, \delta) \in \relint(\Delta_2)$.
By Proposition~\ref{prop:amax-sigmamax}, 
we further have $\sigma_{\max} := \|A\|_2 \le 2$.

Now, for $w = (p, q) \in \calW$ with $p = (p(1), p(2))$
and $q = (q(1), q(2))$, we define the following
regions $\calC_{\delta}, \calP_{\delta,0}, \calQ_{\delta, 0}
\subset \relint(\Delta_2)$:
\begin{equation}
    \begin{aligned}
        \calW_{\delta, 0}
        &:= 
        \big\{
            (p, q) \in \calW \;:\;
            \exp(-3/\delta) \le p(2), q(2) \le \exp(-2/\delta)
        \big\} \\
        \calC_{\delta} 
        &:= 
        \big\{
            (p, q) \in \calW \;:\;
            p(2), q(2) \le \exp(-1/\delta)
        \big\} \;.
    \end{aligned}
    \label{eq:kl-lower-inits}
\end{equation}
Observe by definition that $\calW_{\delta, 0} \subset \calC_{\delta}$,
and that $\calW_{\delta, 0}$ has positive Lebesgue measure 
for any $\delta > 0$.
Now let $\{w_t\}$ be the primal iterates of OMWU initialized
from $w_0 \in \calW_{\delta, 0}$, 
and let $\{z_t\}$ denote the corresponding dual OMWU iterates.
For all $t \ge 0$, 
let $w_t = (p_t, q_t) \in \relint(\calW)$ and $z_t = (x_t, y_t) \in \calZ$.

\smallskip

\noindent
\textbf{Tracking the OMWU trajectory.}

\noindent
We show for all $t \in [T]$ that 
$w_t \in \calC_{\delta}$. 
By definition of $\calC_{\delta}$, 
this implies that the minimum coordinate of each primal iterate
$p_t$ and $q_t$ remains bounded by $\exp(-1/\delta)$
for all times $t \in [T]$.

First, using the bounds on the trajectory of the dual iterates
from Lemma~\ref{lemma:2x2-drift-bounds-unified}, 
observe that if ever $w_{t-1}, w_t \in \calC_\delta$,
then $x_{t+1}(1) - x_t(1) < 0$ and $y_{t+1}(1) - y_t(1) > 0$.
By the primal-dual relationships of Proposition~\ref{prop:omwu-pd-2x2},
this further means that 
$p_{t+1}(1) - p_t(1) < 0$ and $q_{t+1}(1) - q_t(1) > 0$.
By definition of $\calW_{\delta, 0}$, this latter bound implies 
that $q_{t}(1) \ge 1- \exp(-2/\delta) \ge 1 - \exp(-1/\delta)$ for all $t \in [T]$.
Thus, in order to show that $w_t \in \calC_{\delta}$ for all $t \in [T]$,
it suffices to establish for all iterates that $p_t(1) \ge 1- \exp(-1/\delta)$.

For this, using the primal-dual relationship of 
Corollary~\ref{corr:2x2-pd-threshold}, observe 
for any $t \ge 0$ that 
\begin{equation}
    \textstyle
    p_t(1) \ge 1- \exp(-1/\delta)
    \;\iff \;
    x_t(1) \ge \delta \cdot \log\big(\frac{1-\exp(-1/\delta)}{\exp(-1/\delta)}) \;.    
    \label{eq:kltr-goal}
\end{equation}
Moreover, as $w_0 \in \calW_{\delta,0}$, it additionally holds 
by definition of $p_0(1) = 1-p_0(2)$ and
by Corollary~\ref{corr:2x2-pd-threshold} that 
\begin{equation}
    \textstyle
    p_0(2) \le \exp(-2/\delta) 
    \;\implies\;
    p_0(1) \ge 1- \exp(-2/\delta)
    \;\implies\; 
    x_0(1) \ge \delta \cdot \log\big(\frac{1-\exp(-2/\delta)}{\exp(-2/\delta)}) \;.
    \label{eq:kltr-01}
\end{equation}
Now, observe further from the dual update rule of
Proposition~\ref{prop:omwu-pd-2x2} that we can crudely bound
\begin{equation}
    x_{t+1}(1) - x_t(1)
    = 
    - \eta \delta \cdot 
    \big(
      q_t(1) - (1-\delta)  
    \big)
    - 
    \eta \delta \big(
        q_t(1) - q_{t-1}(1)
    \big)
    \ge
    - 2 \eta \delta \;.
    \label{eq:kltr-02}
\end{equation}
Here, the inequality is due to the fact
that $\delta, q_t(1), q_{t-1}(1) \in (0, 1)$.
Thus, for all $t \in [T]$ it holds
by~\eqref{eq:kltr-01}
and~\eqref{eq:kltr-02} that
$x_t(1) \ge x_0(1) - t \cdot 2 \eta \delta \ge x_0(1) - t \delta$, 
where the last inequality is due to the assumption $\eta \le 1/2$.
Now recalling that $\delta = 1/T$, 
this further gives for all $t \in [T]$ that
\begin{align*}
    x_t(1)
    \ge
    x_0(1)
    - T \delta  
    &\ge
    \delta \cdot \log\big(\tfrac{1-\exp(-2/\delta)}{\exp(-2/\delta)})
    - 1 \\
    &=
    \delta \cdot \log\big(\tfrac{1-\exp(-2/\delta)}{\exp(-2/\delta)})
    - \delta \cdot \log\big(\exp(\tfrac{1}{\delta})) \\
    &=
    \delta \cdot 
    \log
    \big(
        \tfrac{1- \exp(-2/\delta)}{\exp(1/\delta)\cdot \exp(-2/\delta)}
    \big) \\
    &= 
    \delta \cdot 
    \log
    \big(
        \tfrac{1- \exp(-2/\delta)}{\exp(-1/\delta)}
    \big)
    \ge
     \delta \cdot 
    \log
    \big(
        \tfrac{1- \exp(-1/\delta)}{\exp(-1/\delta)}
    \big) \;.
\end{align*}
Thus $x_t(1)$ satisfies the bound 
in~\eqref{eq:kltr-goal} for all $t \in [T]$,
which implies that 
also $p_t(1) \ge 1-\exp(-1/\delta)$
for all $t \in [T]$.
Together with the previously established fact that 
all $q_{t}(1) \ge 1- \exp(-1/\delta)$,
we then have by definition
that $w_t \in \calC_{\delta}$ for all $t \in [T]$.

\smallskip
\noindent
\textbf{Bound on change in KL divergence}.

\noindent
Since $w_t \in C_\delta$ for all $t \in [T]$, 
it follows by definition of $C_\delta$ for all $t \ge 0$ that
\begin{equation*}
    p_{t, \min} = \min(p_t(1), p_t(2)) \le \exp(-1/\delta)
    \;\;\text{and}\;\;
    q_{t, \min} = \min(q_t(1), q_t(2)) \le \exp(-1/\delta) \;.
\end{equation*}
Letting $\widehat w_{t, \min} = \max(p_{t, \min}, q_{t, \min})$,
we then also have $\widehat w_{t, \min} \le \exp(-1/\delta)$
for all $t \ge 0$.
It follows that each iterate $w_t$
satisfies the preconditions of Proposition~\ref{prop:kl-dissipation-upper-2x2}.
This allows us to bound
\begin{equation}
    \|J \nabla F(z_t) \|^2_{z_t}
    \le 
    8 \cdot \sigma^2_{\max} \cdot {\widehat w_{t, \min}}
    \cdot \delta \cdot \KL(w^\star, w_t) 
    \le
    16 \cdot \exp(-1/\delta) \cdot \KL(w^\star, w_t) \;.
    \label{eq:klb-01}
\end{equation}
Here, we use the fact that $\sigma_{\max} \le 2$ and $\delta \le 1/2$
by definition of $A$.
Moreover, recall from Corollary~\ref{cor:change-kl-energy-equiv-omwu}
that the primal and dual OMWU iterates satisfy
$\KL(w^\star, w_{t+1}) - \KL(w^\star, w_t) = F(z_{t+1}) - F(z_t)$
for all times $t \ge 0$. Then applying the 
\textit{lower bound} on energy dissipation from 
Lemma~\ref{lem:energy-one-step-full}, we can further bound
\begin{align*}
    \KL(w^\star, w_{t+1})
    - 
    \KL(w^\star, w_t)
    &= 
    F(z_{t+1}) - F(z_t) \\
    &\ge 
    - \frac{5}{4} \cdot \eta^2 \|J \nabla F(z_t)\|^2_{z_t} \\
    &\ge
    - 20 \cdot \eta^2 \exp(-1/\delta) \cdot \KL(w^\star, w_t) \;.
\end{align*}
Here, the final inequality comes from the 
upper bound on $\|J \nabla F(z_t)\|^2_{z_t}$
from expression~\eqref{eq:klb-01}.
Rearranging, we find
\begin{align}
    \KL(w^\star, w_{t+1})
    &\ge 
    \KL(w^\star,  w_t)
    \cdot 
    (1- 20 \eta^2 \cdot \exp(-1/\delta)) \nonumber \\
    &\ge 
    \KL(w^\star, w_t) 
    \cdot 
    \exp(
        - 40 \eta^2 \cdot \exp(-1/\delta)
    ) \;.
    \label{eq:klb-02}
\end{align}
The final inequality follows from the 
fact that $20 \eta^2 \cdot \exp(-1/\delta) \le 1/2$
under the constraint on $\eta$ in Assumption~\ref{ass:stepsize},
and from the bound $1-u \ge \exp(-2u)$ for all $u \in [0,0.5]$.
As the bound in~\eqref{eq:klb-02} holds uniformly 
over all $t \in [T]$, we conclude for all $t \ge 0$ that
\begin{align}
    \KL(w^\star, w_t)
    &\ge 
    \KL(w^\star, w_0)
    \cdot
    \exp(- 40 \eta^2 \cdot \exp(-1/\delta) \cdot t) \;.
    \label{eq:kllower-conclusion}
\end{align}

\smallskip

\noindent
\textbf{Initial bound on KL divergence}.

\noindent
It remains to show a constant upper bound on $\KL(w^\star, w_0)$
under initializations $w_0 = (p_0, q_0) \in \calW_{\delta, 0}$.
For this, using the definition of 
$\KL(w^\star, w_0) = \KL(p^\star, p_0) + \KL(q^\star, q_0)$
from~\eqref{eq:kl-2x2-full}, we have
\begin{equation*}
    \KL(p^\star, p_0)
    = 
    (1-\delta) \cdot \log\big(
        \tfrac{1-\delta}{p_0(1)}
    \big)
    + 
    \delta \cdot \log\big(
        \tfrac{\delta}{p_0(2)}
    \big)
    \le 
    \delta \cdot \log\big(
        \tfrac{\delta}{p_0(2)}
    \big) \;.
\end{equation*}
Here, the inequality comes from the fact that
$p_0(1) > 1-\delta$
for all $(p_0, q_0) \in \calW_{\delta, 0}$.
Additionally using the fact that $p_0(2) > \exp(\frac{-3}{\delta})$
for all $(p_0, q_0) \in \calW_{\delta, 0}$,
we can further derive the bound
\begin{equation*}
    \KL(p^\star, p_0) 
    \le 
    \delta \cdot \log\big(
        \tfrac{\delta}{\exp(-3/\delta)}
    \big)
    = 
    \delta \log \delta 
    - \delta \big(\tfrac{-3}{\delta}\big)
    \le 
    3 \;.
\end{equation*}
The final inequality is due 
to $\delta \log \delta \le 0$ for $\delta \le 1$.

By an identical calculation,
we also find $\KL(q^\star, q_0) \le 3$
for all $w_0 = (p_0, q_0) \in \calW_{\delta,0}$.
Thus it follows for all such $w_0$ that
$\KL(w^\star, w_0) \le 6$, which concludes the proof.
~\hfill $\blacksquare$


\section{Details on Uniform Convergence Rate Lower Bounds}
\label{app:lowerbound-details}

In this section, we develop the proof of the
uniform best-iterate convergence rate lower bounds
of Theorem~\ref{thm:uniform-lb-main} that
were introduced in Section~\ref{sec:uniform-bounds}.
To restate the theorem:

\thmuniformlbmain*

Theorem~\ref{thm:uniform-lb-main}
establishes a separation between 
$\TV$ distance and $\KL$ divergence 
and duality gap in terms of
uniform best-iterate convergence guarantees for OMWU:
while a quantitative uniform best-iterate convergence 
rate is attainable in duality gap
under the uniform initialization 
(due to ~\cite{C25},
and also Theorem~\ref{thm:dg-best-2x2}),
no such rate is possible under the stronger
notions of $\TV$ and $\KL$ divergence.

Using the intuition for the geometric bottlenecks
developed in Section~\ref{sec:energy-dissipation},
the hard instances in the proof of the theorem have a 
(symmetric) NE close to a simplex vertex, and with
$\delta$ scaling inversely with~$T$. In this regime,
we rely on the primal-dual relationship of the OMWU iterates
to exactly track the OMWU dynamics in the $2 \times 2$ setting
(see Section~\ref{app:omwu-2x2} for these preliminaries). 
Subsequently, this allows for establishing on these instances 
that \textit{all $T$ iterates} of OMWU remain a bounded
distance away from equilibrium 
in $\TV$ and $\KL$.

\smallskip

\paragraph{Additional Lower Bounds in KL and DG
under Non-Uniform Initializations.}
In addition to Theorem~\ref{thm:uniform-lb-main},
we also prove a second set of uniform best-iterate convergence
rate lower bounds that further highlight the 
different geometries of the distance functions
and their interplay with the OMWU iterates.
Specifically, when initialized from some positive measure 
subsets of non-uniform distributions
(that may be near the boundary of the simplex),
we prove two additional lower bounds:
(i) all iterates of OMWU can be a constant distance away
from Nash  in $\KL$, but convergent in $\TV$ and $\DG$,
and (ii) all iterates of OMWU can also be a constant
distance away from Nash in $\DG$.
Formally:

\begin{restatable}[Additional Uniform Best-Iterate Lower Bounds]
{thm}{thmuniformkl}
	\label{thm:uniform-lb-secondary}
	For every $T \ge 3$, it holds that:
	\begin{enumerate}[
		label={(\arabic*)}
	]
	\item
	There exists $A \in [-1, 1]^{2 \times 2}$
	with interior NE $w^\star$ and a positive 
	measure set of initializations, not including the 
	joint uniform distributions, such that,
	for the iterates $\{w_t\}$ of~\eqref{eq:omwu-primal} 
	on $A$ with stepsize $\eta$ satisfying Assumption~\ref{ass:stepsize}:
	\begin{equation*}
	\text{for all $t \in [T]$, both}
	\begin{cases}
		\TV(w^\star, w_t) \le O(T^{-1}) \\
		\DG(w_t) \le O(T^{-1}) 
	\end{cases}
	\text{while also}\;
	\min\nolimits_{t \in [T]} \KL(w^\star, w_t) \ge \tfrac{1}{6}\;.
	\end{equation*}
	\item
	There exists $A \in [-1, 1]^{2 \times 2}$
	with interior NE $w^\star$ and a positive 
	measure set of initializations, not including the 
	joint uniform distributions, such that,
	for the iterates $\{w_t\}$ of~\eqref{eq:omwu-primal} 
	on $A$ with stepsize $\eta$ satisfying Assumption~\ref{ass:stepsize}:
	\begin{equation*}
		\min\nolimits_{t \in [T]} \DG(w_t) \ge \tfrac{1}{6} \;.
	\end{equation*}
	\end{enumerate}
\end{restatable}

Part (1) of Theorem~\ref{thm:uniform-lb-secondary} 
is based on the \textit{universal} last-iterate lower bound in 
$\KL$ from Theorem~\ref{thm:kl-last-lower}.
Note that this result gives a separation in uniform
best-iterate convergence guarantees between 
$\TV$ and $\KL$
(which is not established by Theorem~\ref{thm:uniform-lb-main}).

Moreover, Part (2) of Theorem~\ref{thm:uniform-lb-secondary} 
further shows that, when initialized near the simplex boundary,
no best-iterate convergence guarantee is possible
in duality gap. 
Note that this does not contradict the uniform best-iterate
convergence rate upper bounds of
Theorem~\ref{thm:dg-best-2x2}
and of~\cite{C25}, both of which hold under 
the assumption of joint uniform initializations.

In the following subsections, we proceed
to give the proofs of Theorem~\ref{thm:uniform-lb-main}
and Theorem~\ref{thm:uniform-lb-secondary},
both of which use the preliminaries for OMWU 
on the canonical $2\times 2$ instance $A_{\delta_p,\delta_q}$
from Section~\ref{app:omwu-2x2}.

\subsection{Proof of Theorem~\ref{thm:uniform-lb-main}}

To prove the theorem, we instantiate the payoff matrix $A$ 
and the set of intializations as follows:
\begin{itemize}[
	topsep=0em,
	itemsep=0em,
	leftmargin=2em,
]
\item
Set $A := A_{\delta_p, \delta_q} \in [-1, 1]^{2\times 2}$
from Definition~\ref{def:A-2x2} 
with $\delta := \delta_p = \delta_q = \frac{1}{6 \exp(T)}$. 

\item
Initialize $w_0 = (p_0, q_0) \in \calP_0 \times \calQ_0$,
where:
\begin{equation*}
	\calP_0 
	= 
	\{ p \in \relint(\Delta_2) \,:\, p(1) \le 1/2 \}
	\;\text{and}\;
	\calQ_0 
	= 
	\{q \in \relint(\Delta_2) \,:\, q(1) \le 1/2 \} \;.
\end{equation*}
\end{itemize}
Let $\calW_0 = \calP_0 \times \calQ_0 \subset \relint(\calW)$,
and observe that $\calW_0$ has positive Lebesgue measure,
and that $\calP_0$ and $\calQ_0$
both contain the uniform distributions 
$p_0 = (1/2, 1/2)$ and $q_0 = (1/2, 1/2)$.
Moreover, by the setting of $A$, we have
by Corollary~\ref{corr:2x2-pd-threshold}
that the corresponding dual initializations
$(x_0, y_0) \in \calZ$
under OMWU satisfy $x_0(1) \le 0$ 
and $y_0(1) \le 0$ for all $(p_0, q_0) \in \calW_0$.

The proof then proceeds by (1) tracking the trajectory 
of the dual OMWU iterates over the time horizon $T$,
which (2) leads to a lower bound on the 
TV distance and KL divergence from Nash
that holds uniformly over all iterates.
We proceed to give the details of these two steps:

\smallskip

\noindent
\textbf{Bounds on Trajectory.}\;
For every $t$, let $w_t = (p_t, q_t) \in \relint(\Delta_2)$
and $z_t = (x_t, y_t) \in \calZ$ denote the primal 
and dual iterate of OMWU. 
Moreover, let $\w p_t = p_t(1) \in (0, 1)$, $\w q_t = q_t(1) \in (0, 1)$
and $\w x_t = x_t(1) \in \R$, $\w y_t = y_t(1) \in \R$
denote the leading coordinate of each variable. 
Under $A := A_{\delta, \delta}$, Proposition~\ref{prop:omwu-pd-2x2}
then gives for $t+1 \ge 1$ the simplified OMWU dual update rule of
\begin{equation*}
	\textstyle
	\w x_{t+1} - \w x_t 
	= 
	\eta \delta
	\cdot \big(
		1-\delta - \w q_t + \w q_{t-1} - \w q_t
	\big)
	<
	2 \eta \delta 
	\le \delta
	\;.
\end{equation*}
Here, the inequality comes from $0 < \w q_t, \w q_{t-1}< 1$ and
$\delta > 0$, and the assumption that $\eta < 1/2$.
It follows that, under the initialization $\w p_0 \in \calP_0$:
\begin{equation*}
	\textstyle
	\w x_t 
	\le \w x_0 + \delta \cdot  t
	\le \delta   T
	\le \delta \log \big(\frac{1}{6\delta}\big)
	\le \delta \log \big(\frac{1-3\delta}{3\delta}\big) 
	\;\text{for all $t \in [T]$}
	\;.
\end{equation*}
Here, the first inequality comes from $\w x_0 \le 0$, the second inequality comes from the 
setting of $\delta = 1/(6\exp(T))$,
and the third inequality holds when $\delta \le 1/6$, 
which is true for all $T \ge 2$.
By Corollary~\ref{corr:2x2-pd-threshold}, it then follows that
\begin{equation}
	\textstyle
	\w x_t \le \delta \log \big(\frac{1-3\delta}{3\delta}\big) 
	\;\iff\; 
	\w p_t \le 1- 3 \delta 
	\;\;\text{for all $t \in [T]$.}
	\label{eq:tvlb-01}
\end{equation}
By Part (iii) of Lemma~\ref{lemma:2x2-drift-bounds-unified},
this further implies that 
$\w y_{t+1} - \w y_{t} \le 0$ for all $t \in [T]$. 
Under the initializations $\w q_0 \in \calQ_0$ 
and again using Corollary~\ref{corr:2x2-pd-threshold},
it then holds that
\begin{equation}
	\textstyle
	\w y_t \le \w y_0 \le 0 
	\;\iff\; 
	\w q_t \le \frac{1}{2}
	\;\;\text{for all $t \in [T]$.}
	\label{eq:tvlb-02}
\end{equation}

\smallskip 

\noindent
\textbf{Bounds on Distance to Nash.}\;
Recall for $A = A_{\delta, \delta}$ that
the unique NE $w^\star = (p^\star, q^\star)$ of $A$ is given 
by $p^\star = q^\star =(1-\delta, \delta)$.
Then by the definition of $\TV(w^\star, w_t)$
from~\eqref{eq:tv-2x2-full},
and using the bound $\w q_t \le 1/2$ from~\eqref{eq:tvlb-02},
we have for all $t \in [T]$ that
\begin{equation}
	\textstyle
	\TV(w^\star, w_t)
	= 
	| 1- \delta - \w p_t|
	+ 
	| 1- \delta - \w q_t |
	\ge 
	| 1- \delta - \w q_t |
	\ge 
	| \frac{1}{2} - \delta | \;.
\end{equation}
As $\delta \le 1/6$ for all $T \ge 2$, 
it then follows that
\begin{equation*}
	\textstyle
	\TV(w^\star, w_t) \ge \frac{1}{3} 
	\;\;\text{for all $t \in [T]$.}
\end{equation*}
Finally, by the relationship
$\KL(w^\star, w_t) \ge \TV(w^\star, w_t)^2$
from Corollary~\ref{cor:distance-relations},
we also have 
\begin{equation*}
	\textstyle
	\KL(w^\star, w_t)	
	\ge 
	\TV(w^\star, w_t)^2
	\ge 
	\frac{1}{9} 
	\;\;\text{for all $t \in [T]$},
\end{equation*}
which concludes the proof.~\hfill $\blacksquare$

\subsection{Proof of Theorem~\ref{thm:uniform-lb-secondary}}

In this section we give the proof of 
Theorem~\ref{thm:uniform-lb-secondary}. 
We prove the two parts separately:

\subsubsection{Proof of Part 1}

The proof follows as a corollary 
of the \textit{universal} last-iterate lower bound
in $\KL$ from Theorem~\ref{thm:kl-last-lower}.
There, recall from the proof that 
for all $T \ge 3$, the game instance
is $A = A_{\delta, \delta}$ with
$\delta = 1/T$. 
We start by recalling the properties established
in the proof of Theorem~\ref{thm:kl-last-lower}:

\paragraph{Conclusions from
Proof of Theorem~\ref{thm:kl-last-lower}.}
Recall from~\eqref{eq:kl-lower-inits} that 
we define the set of initializations $\calW_{\delta, 0} \subset \relint(\calW)$
and the set $\calC_{\delta} \subset \relint(\calW)$ as:
\begin{equation*}
    \begin{aligned}
        \calW_{\delta, 0}
        &:= 
        \big\{
            (p, q) \in \calW \;:\;
            \exp(-3/\delta) \le p(2), q(2) \le \exp(-2/\delta)
        \big\} \\
        \calC_{\delta} 
        &:= 
        \big\{
            (p, q) \in \calW \;:\;
            p(2), q(2) \le \exp(-1/\delta)
        \big\} \;.
    \end{aligned}
\end{equation*}
The proof of Theorem~\ref{thm:kl-last-lower} established
the following two conclusions: 
for all initial $w_0 \in \calW_{\delta, 0}$ and 
for all times $t \in [T]$:
\begin{equation}
	w_t \in \calC_{\delta}\;\text{and}\;
	\KL(w^\star, w_t) 
	\ge 
	\KL(w^\star, w_0) 
	\cdot 
	\exp(-40 \eta^2 \cdot \exp(-1/\delta) \cdot t) \;.
	\label{eq:kllb2-claims}
\end{equation}
We now show further implications of these conclusions:

\paragraph{All iterates are convergent in TV and Duality Gap.}
By definition of the set $\calC_{\delta}$, 
the fact that $w_t \in \calC_{\delta}$ for all $t \in [T]$ implies that
\begin{equation}
	1-\exp(-1/\delta) \le p_t(1), q_t(1) \le 1
	\;\;\text{for all $t \in [T]$}.
	\label{eq:klwtbound}
\end{equation}
By the definition of $\TV(w^\star, w)$ from~\eqref{eq:tv-2x2-full},
it then follows for all $t \in [T]$ that:
\begin{equation*}
	\TV(w^\star, w_t)
	= |p_t(1) - 1-\delta| + |q_t(1) - 1- \delta|
	\le 
	2 \delta
	= \tfrac{2}{T} \;.
\end{equation*}
Here, the final equality is to due to the setting $\delta = 1/T$.
Moreover, using the relationship between
$\DG(w_t)$ and $\TV(w_t)$ from Corollary~\ref{cor:distance-relations},
and using the fact that $A_{\delta, \delta} \in [-1, 1]^{2\times 2}$,
we further conclude for all $t \in [T]$
that $\DG(w_t) \le \sqrt{2} \TV(w^\star, w_t) \le 2\sqrt{2}/T$. 
Thus it holds on this instance that 
\begin{equation*}
		\DG(w_t) \le O\big(T^{-1}\big)
		\;\text{and}\;\;
		\TV(w^\star,w_t) \le O\big(T^{-1}\big) \;\text{for all $t \in [T]$}.
\end{equation*}

\paragraph{All iterates remain far in KL.}
We now show that the conclusions from~\eqref{eq:kllower-conclusion}
imply that all iterates $\KL(w^\star, w_t)$ remain 
bounded below by a constant. 
First, we show a constant lower bound on 
$\KL(w^\star, w_0)$ for all initializations $w_0 \in \calW_{\delta_0}$.
For this, observe that $\calW_{\delta, 0} \subset \calC_{\delta}$,
and thus the bounds in~\eqref{eq:klwtbound}
also apply for any initialization $w_0 = (p_0, q_0) \in \calW_{\delta}$.
Then, using the definition of 
$\KL(w^\star, w) = \KL(p^\star, p_0) + \KL(q^\star, q_0)$
from~\eqref{eq:kl-2x2-full}, we have
\begin{align*}
    \KL(p^\star, p_0)
    &= 
    (1-\delta) \cdot \log\big(
        \tfrac{1-\delta}{p_0(1)}
    \big)
    + 
    \delta \cdot \log\big(
        \tfrac{\delta}{p_0(2)}
    \big) \\
    &\ge 
    (1-\delta) \cdot \log\big(
        \tfrac{1-\delta}{p_0(1)}
    \big)
    \ge 
     (1-\delta) \cdot \log\big(
        \tfrac{1-\delta}{1-\exp(-1/\delta)}
    \big)
    \ge 
    \tfrac{2}{3} \;.
\end{align*}
Here, the final inequality holds
when $\delta \le 1/2$, 
which is true by the setting of $\delta = 1/T$
for $T \ge 2$.
An identical calculation 
also gives $\KL(q^\star, q_0) \ge 2/3$,
and thus by definition 
$\KL(w^\star, w_0) \ge 4/3$ for all
initializations $w_0 \in \calW_{\delta, 0}$.
Further using that $T = 1/\delta$, 
it follows from~\eqref{eq:kllower-conclusion} that
for all $t \in [T]$:
\begin{align*}
	\KL(w^\star, w_t) 
	&\ge 
	\KL(w^\star, w_0)
	\cdot 
	\exp(- 40 \eta^2 \cdot \exp(-1/\delta) \cdot T) \\
	&\ge
	\tfrac{4}{3}
	\cdot 
	\exp(- 40 \eta^2 T/\exp(T)) \\
	&\ge 
	\tfrac{4}{3} \cdot \tfrac{1}{\exp(2)} 
	\ge \tfrac{1}{6} \;.
\end{align*}
Here, the penultimate inequality comes 
from the assumption that $\eta \le \tfrac{1}{2}$,
and since $T/\exp(T) \le \tfrac{1}{5}$ for all $T \ge 3$.
Thus we find on this instance that
$\KL(w^\star, w_t) \ge 1/6$.
\hfill ~ $\blacksquare$

\subsubsection{Proof of Part 2}

The structure of the proof is similar to that
of Theorem~\ref{thm:uniform-lb-main}.
First, we instantiate the payoff matrix $A$
and the set of initializations as follows:
\begin{itemize}[
	topsep=0em,
	itemsep=0em,
	leftmargin=2em,
]
\item
Set $A := A_{\delta_p, \delta_q} \in [-1, 1]^{2\times 2}$
from Definition~\ref{def:A-2x2} 
with $\delta := \delta_p = \delta_q = \frac{1}{\exp(T)}$. 

\item
Initialize $w_0 = (p_0, q_0) \in \calP_0 \times \calQ_0$,
where:
\begin{equation*}
	\calP_0 
	= 
	\{ p \in \relint(\Delta_2) \,:\, \tfrac{1}{2} \le p(1) \le \tfrac{3}{4} \} 
	\;\;\text{and}\;\;
	\calQ_0 
	= 
	\{q \in \relint(\Delta_2) \,:\, 1-\tfrac{\delta^2}{3} 
	\le q(1) \le 1- \tfrac{\delta^2}{6} \} \;.
\end{equation*}
\end{itemize}
Let $\calW_0 = \calP_0 \times \calQ_0 \subset \relint(\calW)$,
and observe that $\calW_0$ has positive Lebesgue measure.
Moreover, for every $T \ge 3$, 
we have $\delta < 1/8$, and thus 
all initializations $p_0 \in \calP_0$ satisfy $p_0 < 1-\delta$. 
On the other hand, note that 
$\calW_0$ does not contain 
the uniform initialization $p_0 = q_0 = (1/2, 1/2)$,
which is in contrast to 
the proof of Theorem~\ref{thm:uniform-lb-main}.

We now proceed to (1) track the trajectory
of the dual OMWU iterates over the time horizon $T$,
which (2) will lead to a constant lower bound on the
duality gap that holds over all iterates.
In more details:

\noindent
\textbf{Bounds on Trajectory.}\;
For every $t$, let $w_t = (p_t, q_t) \in \relint(\Delta_2)$
and $z_t = (x_t, y_t) \in \calZ$ denote the primal 
and dual iterate of OMWU, and
let $\w p_t = p_t(1) \in (0, 1)$, $\w q_t = q_t(1) \in (0, 1)$
and $\w x_t = x_t(1) \in \R$, $\w y_t = y_t(1) \in \R$.
Now under $A := A_{\delta, \delta}$, 
we can use Proposition~\ref{prop:omwu-pd-2x2}
to crudely bound for any $t \ge 0$
\begin{equation}
	\w y_{t+1} - \w y_t
	=
	- \eta \delta
	(1-\delta - \w p_t - (\w p_t  - \w p_{t-1}))
	\ge 
	- 2 \eta \delta 
	\ge
	- \delta \;.
	\label{eq:dgydrift}
\end{equation}
Here, the first inequality is due to
$\w p_t, \w p_{t-1}, \delta \in (0, 1)$,
and the second is by the assumption 
that $\eta < 1/2$.
Now under the initialization $q_0 \in \calQ_0$,
we have $\w q_0 \ge 1 - \delta^2/3$.
By Corollary~\ref{corr:2x2-pd-threshold},
this implies that 
\begin{equation}
	\w q_0 \ge 1 - \tfrac{\delta^2}{3}
	\;\implies\;
	\w y_0 \ge \delta \log\big(\tfrac{1-\delta^2/3}{\delta^2/3}\big) \;.
	\label{eq:dgylb}
\end{equation}
Then using~\eqref{eq:dgydrift}, it follows for
all $t \in [T]$ that we can bound
\begin{align}
	\w y_{t} 
	\ge \w y_0 - \delta \cdot t 
	\ge \w y_0 - \delta \cdot T 
	&\ge \delta \log\big(\tfrac{1-\delta^2/3}{\delta^2/3}\big) 
	- \delta \log\big(\tfrac{1}{\delta}\big) \nonumber \\
	&= 
	\delta \log\big( \tfrac{1-\delta^2/3}{\delta/3}\big)
	\ge 
	\delta \log \big( \tfrac{1-\delta/3}{\delta/3} \big) \;.
	\label{eq:dgytlb}
\end{align}
Here, the second inequality comes from
applying the lower bound on $\w y_t$ from~\eqref{eq:dgylb}
and using the setting $\delta = 1/\exp(T) \implies T = \log(1/\delta)$.
Again by Corollary~\ref{corr:2x2-pd-threshold}, it follows 
for all $t \in [T]$ that
\begin{equation}
	\w y_t \ge \delta \log \big( \tfrac{1-\delta/3}{\delta/3} \big)
	\;\implies\;
	\w q_t \ge 1-\delta/3 \;.
	\label{eq:dgqtlb}
\end{equation}
Moreover, using the fact that $\w p_0 \le 3/4 < 1-\delta$ for 
for all $p_0 \in \calP_0$, this further implies via
Part (ii) of Lemma~\ref{lemma:2x2-drift-bounds-unified}
that all $\w x_{t+1} - \w x_t < 0 \implies \w p_{t+1} - \w p_t < 0$. 
Thus for all $t \in [T]$ it holds that
\begin{equation}
	\w p_t \le \w p_0 \le 3/4 < 1-\delta \;.
	\label{eq:dgptub}
\end{equation}
We now derive the implications of the uniform
coordinate bounds of~\eqref{eq:dgqtlb}
and~\eqref{eq:dgptub} on the duality gap 
of the iterates.

\smallskip 

\noindent
\textbf{Bounds on Distance to Nash.}\;
For $A = A_{\delta, \delta}$, 
the unique NE $w^\star = (p^\star, q^\star)$ of $A$ is given 
by $p^\star = q^\star =(1-\delta, \delta)$.
Then by the definition of $\DG(w_t)$
from~\eqref{eq:dg-2x2-full},
and using the bounds $\w q_t \ge 1-\delta/3$
and $\w p_t < 1-\delta$, 
we have for all $t \in [T]$ that
\begin{align}
	\textstyle
	\DG(w_t)
	&= 
	(1-\delta - \w p_t) (1-\delta)
	+ (\w q_t - (1-\delta)) (1-\delta) \label{eq:dglbf01}\\
	&\ge 
	(1/4-\delta) \cdot (1-\delta) \\
	&\ge (1/4 - 1/20)\cdot (19/20) \ge 1/6 \;.
	\label{eq:dglbfinal}
\end{align}
The first inequality comes from the fact 
that the second term in~\eqref{eq:dglbf01} is non-negative,
and also by applying the upper bound $\w p_t \le 3/4$.
The second inequality comes
from the fact that $\delta = 1/\exp(T)$,
and thus $\delta \le 1/20$ for all $T \ge 3$.
As the bound of~\eqref{eq:dglbfinal} holds uniformly
for all $t \in [T]$, we conclude on this instance 
that $\min_{t \in [T]} \DG(w_t) \ge 1/6$,
which concludes the proof.
\hfill ~ $\blacksquare$

\section{Details on Best-Iterate Convergence Rate in Duality Gap
for 2x2 Setting}
\label{app:dg-upper}

In this section, we develop the proof of 
Theorem~\ref{thm:dg-best-2x2},
which establishes a fast uniform best-iterate
convergence rate upper bound in duality gap
for the $2\times 2$ setting. 
To restate the theorem:

\smallskip

\thmdgbest*

Note that this result improves over the prior 
best-known $O(T^{-1/6})$ rate of~\cite{C25}
for the same setting. 
This improvement is due to a refined proof
technique that leverages the new
analysis of OMWU in $\KL$ divergence
from Theorem~\ref{thm:kl-last-unified}.
We give a brief comparison with the analysis 
of~\cite{C25} and a high-level overview of our proof here:

\paragraph{Comparison 
with~\cite{C25}.}
We recall that the proof of~\cite{C25} uses a 
two-regime approach depending 
on the magnitude of the time horizon $T$
with respect to the minimum Nash coordinate
$\delta = \min(\delta_p, \delta_q)$.
They show a best-iterate bound in duality gap
scaling like 
$\min(\delta, T^{-1/4} \delta^{-1/2})$.
Depending on whether $T \ge \delta^{-1/6}$
or $T < \delta^{-1/6}$, this term always is 
at most $T^{-1/6}$, which results in their bound.
For large $T$, their analysis relies on  a
universal upper bound on the random-iterate
convergence in duality gap.
For small $T$, their analysis relies on 
tracking the trajectory of the OMWU iterates
in the primal space. 

\paragraph{High-level overview of proof.}
In contrast to the proof of~\cite{C25}, our proof uses a three-regime approach
that allows for much tighter bounds when
$T$ is large with respect to $\delta$.
In the largest regime, we use our new universal 
last-iterate analysis of OMWU in $\KL$ divergence 
(Theorem~\ref{thm:kl-last-unified})
to establish a fast \textit{last-iterate} bound in $\DG$. 
For the moderate and small $T$ regimes,
we directly track the OMWU trajectory
and obtain tighter best-iterate bounds 
of roughly $O(\delta^2)$ and $O(\max(1/T, \delta))$, respectively.
Together, using the definition of the three regimes, 
this leads to the overall $\widetilde O(T^{-1/2})$ 
best-iterate bound in duality gap.

We give more details on these regimes and their
consequences in the proof overview in the following section.
However, we first make two additional remarks regarding 
the assumption of uniform initialization,
and on the possibility of generalizing
our techniques to higher-dimensional settings.

\medskip

\begin{restatable}[On Uniform Initialization]{rem}{remdgtwoinitialization}
	\label{remark:2x2-dg-upper-initialize}
	The proof of Theorem~\ref{thm:dg-best-2x2} 
	assumes a uniform initialization
	(similar to the results of~\cite{DaskalakisP19}, 
	\cite{wei2021linear}, \cite{FGKLLZ24},~\cite{C25}).
	However, our proof also naturally extends to 
	a positive measure set of initializations
	and leads to the same rate,
	up to the leading absolute constants.
	Note that such initializations 
	are all interior and near the uniform distribution;
	as we showed in Part (ii) of
	Theorem~\ref{thm:uniform-lb-secondary},
	there is in general no quantitative uniform convergence rate
	when initialized near the boundary of the simplex.
\end{restatable}

\smallskip

\begin{restatable}[On Generalizing to Higher Dimensions]
{rem}{remdgtwogeneralize}
	\label{remark:2x2-dg-generalize}
	Even in the $2\times 2$ setting,
	the proof for obtaining the $\widetilde O(T^{-1/2})$
	uniform best-iterate convergence rate in duality gap
	is intricate and involves understanding 
	the interplay between (a)
	the trajectory of the OMWU iterates
	and (b) the non-uniform rate of energy dissipation 
	over the iterates. 
	In the $2\times 2$ setting, the low-dimensional
	property of the effective dual space
	(as detailed in 
	Section~\ref{app:2x2-setup:primal-dual:effective-dual})
	allows for obtaining fine-grained control of 
	the iterate trajectory.
	This helps for establishing
	a phase transition in the rate of energy dissipation 
	when $T$ is sufficiently large (see 
	Lemma~\ref{lem:best-dg-large-T} below).
	While such a transition empirically holds
	for higher-dimensional settings,
	establishing an analogous bound 
	remains a challenging and open
	technical question.
	Moreover, controlling the trajectory of the iterates
	in the $2\times 2$ setting
	also allows for proving that the primal iterates always enter
	a region of the simplex with a sufficiently small duality gap,
	even when far from Nash under a metric like TV distance
	(see Sections~\ref{app:2x2-setup:energy-distances}
	and~\ref{sec:2x2-levelsets-compare}). 
	Thus, extending the result to higher-dimensions
	likely requires some additional control 
	of the iterates (e.g., establishing the number of 
	steps spent in various regions of the primal
	and dual spaces, as a function of the
	location of the interior NE).
	We conjecture that OMWU still obtains
	a $\widetilde O(T^{-1/2})$ uniform best-iterate
	convergence rate in the general-dimension
	setting, and we leave this for future work. 
\end{restatable}

\smallskip

\subsection{Proof of Theorem~\ref{thm:dg-best-2x2}}

In this section, we give the proof of
Theorem~\ref{thm:dg-best-2x2}. The proof uses
several intermediate components and 
results that we outline here,
and which we prove in the later subsections.

\paragraph{Assumptions on NE.}
First, we assume without loss of generality 
that $A \in [-1, 1]^{2\times 2}$ is not a constant matrix.
Then by Part (ii) of Proposition~\ref{prop:2x2-interior-unique},
since $w^\star = (p^\star, q^\star)$ is an
interior NE of $A$, then $w^\star$ is the unique NE of $A$.
Let $p^\star = (1-\delta_p, \delta_p) \in \relint(\Delta_2)$
and $q^\star = (1-\delta_q, \delta_q) \in \relint(\Delta_2)$
for $\delta_p, \delta_q \in (0, 1)$.
In the remainder of the proof, we will assume without
loss of generality that
$0 < \delta_p \le \delta_q \le \frac{1}{10}$.
Otherwise, if $\delta_p, \delta_q \ge \frac{1}{10}$, 
then the universal last-iterate bound of 
Theorem~\ref{thm:kl-last-unified},
together with Corollary~\ref{cor:distance-relations},
would directly give $\KL(w^\star, w_T) \le O(T^{-1/2})$
for all $T \ge 1$.

\smallskip

\paragraph{Assumptions on Initialization.}
We assume that the initialization $w_0 = (p_0, q_0) \in \relint(\calW)$
is the uniform distribution $p_0 = (1/2, 1/2)$
and $q_0 = (1/2, 1/2)$.

\paragraph{Considering Canonical $2 \times 2$ Matrix Suffices.} 
The proof relies on the properties of the 
$2\times2$ setting introduced in Section~\ref{app:omwu-2x2}.
In particular, we recall from Proposition~\ref{prop:A-properties}
that, for $v \in \R$ and $0 < \gamma \le 4$, 
the matrix $A$ can be decomposed as
\begin{equation*}
	A = \gamma \cdot A_{\delta_p, \delta_q} + v \1 \1^\top \;.
\end{equation*}
Here, $A_{\delta_p, \delta_q} \in [-1,1]^{2\times 2}$ 
is the canonical $2\times 2$ matrix from Definition~\ref{def:A-2x2} 
with a unique and interior NE $w^\star = (p^\star, q^\star)$
(i.e., the same NE as $A$).
The key consequence of this decomposition is that,
to establish bounds on duality gap for the iterates
of OMWU on $A$, it suffices to analyze
the iterates of OMWU on $A_{\delta_p, \delta_q}$
with an appropriately scaled stepsize. 
In particular, we have the following equivalence:

\begin{restatable}[OMWU on Canonical Matrix]
{prop}{proptwotwosuffices}
	\label{prop:2x2-suffices}
	Fix $A \in [-1, 1]^{2\times 2}$
	with an interior Nash equilibrium $w^\star = (p^\star, q^\star)$,
	where $p^\star = (1-\delta_p, \delta_p)$
	and $q^\star = (1-\delta_q, \delta_q)$ for $\delta_p, \delta_q \in (0, 1)$.
	Let $A_{\delta_p, \delta_q} \in [-1, 1]^{2\times 2}$ 
	be the payoff matrix from Definition~\ref{def:A-2x2}.
	Fix the initialization $w_0 = (p_0, q_0)$
	for $p_0 = q_0 = (0.5,0.5)$.
	Then there exists $0 \le \gamma \le 4$ such that 
	the following holds:
	\begin{itemize}[
		label=-,
		itemsep=0em,
		topsep=0em,
		leftmargin=2em
	]
		\item
		Let $\{w_t\}$ be the iterates of~\eqref{eq:omwu-primal}
		on $A$ with stepsize $\eta > 0$,
		initialized from $w_0$.
		\item
		Let $\{w'_t\}$ be the iterates of~\eqref{eq:omwu-primal}
		on $A_{\delta_p, \delta_q}$ with stepsize 
		$\eta' = \eta \gamma > 0$,
		initialized from $w_0$.
	\end{itemize}
	Then $w_t = w'_t$ for all $t \ge 0$.
\end{restatable}

Moreover, we also have the following relationship 
between the duality gaps under $A$ and $A_{\delta_p, \delta_q}$:

\begin{restatable}[Relationship between Duality Gaps]
{prop}{propdgcompare}
	\label{prop:2x2-dg-compare}
	Let $\DG_A(\cdot)$ 
	denote the duality gap under matrix $A$,
	and let $\DG_{A'}(\cdot)$
	denote the duality gap under matrix $A'$,
	where $A = \gamma \cdot A' + v \1\1^\top$ 
	for constants $\gamma, v \in \R$ with  $\gamma > 0$.
	Then for any $w = (p, q) \in \calW$
	and value $B > 0$:
	\begin{equation*}
		\DG_{A'}(w) \le B 
		\;\implies\;
		\DG_{A}(w) \le \gamma B \;.
	\end{equation*}
\end{restatable}

The proofs of Proposition~\ref{prop:2x2-suffices}
and Proposition~\ref{prop:2x2-dg-compare}
are in Section~\ref{app:dg-upper:canonical}
and we grant them as true for now. 
As a result, the remainder of the proof focuses
on establishing bounds on the best-iterate in duality gap 
for the OMWU iterates when run 
on the matrix $A_{\delta_p, \delta_q}$
with scaled stepsize.
For this, we require the following slightly modified
assumption on the constant stepsize:

\paragraph{Assumptions on Stepsize for 
OMWU on Canonical Matrix.}
For OMWU run on any instantiation of the 
canonical matrix $A_{\delta_p, \delta_q} \in [-1, 1]^{2\times 2}$,
we assume that the stepsize $\eta$ satisfies the following:

\begin{restatable}[Stepsize on Canonical Matrix]
{ass}{assconstantdgupper}
	\label{ass:stepsize-dg-upper}
	$0 < \eta \le \frac{1}{4} \cdot \frac{1}{4(54\sigma_{\max} + 9)}$
	is an absolute constant.
\end{restatable}

Note that under Assumption~\ref{ass:stepsize-dg-upper},
we have that $\eta \cdot \gamma$ is bounded 
as in Assumption~\ref{ass:stepsize}
for $0 < \gamma \le 4$ (i.e., for the $\gamma$
arising in the equivalence of Proposition~\ref{prop:2x2-suffices}).
Thus, when analyzing the OMWU iterates on 
$A_{\delta_p, \delta_q}$ with stepsize 
$\eta' = \eta \gamma$ and $\eta$ satisfying 
Assumption~\ref{ass:stepsize},
then $\eta'$ also satisfies Assumption~\ref{ass:stepsize}.

We also reiterate Remark~\ref{remark:constant-stepsize}:
our proofs do not attempt to optimize
the constant stepsize constraint,
and the conclusions of this section likely hold under
even larger constant stepsize settings.

\paragraph{Uniform Best-Iterate Duality Gap Bounds
on Canonical $2 \times 2$ Matrix.}
Given the equivalence of Proposition~\ref{prop:2x2-suffices},
the core technical component of the proof 
is to establish uniform best-iterate bounds in duality gap
for OMWU run on the canonical matrix $A_{\delta_p, \delta_q}$.
For this, we prove three separate bounds 
that depend on the magnitude of the time horizon $T$
with respect to the minimum Nash coordinates $\delta_p$
and $\delta_q$. We refer to these regimes
as \textit{large $T$}, \textit{moderate $T$},
and \textit{small $T$}, and we give details below.

First, we make use of the values $V, L > 0$,
which we define as
\begin{equation}
	V = \delta_q \cdot \log\big(\tfrac{1-(\delta_q/3)}{\delta_q/3}\big) 
	\;\;\text{and}\;\;
	L = 40 \;.
	\label{eq:V-L-def}
\end{equation}
We then define the three regimes in terms
of $V, L$ and $\eta$, where
$\eta$ is the stepsize used for the iterates
of OMWU on the canonical matrix $A_{\delta_p, \delta_q}$.

\paragraph{Large $T$ regime.}
In the large $T$ regime, we assume that 
\begin{equation}
	\frac{1440 \cdot  L \cdot \log(2/V)}{\eta^2 (\delta_p \delta_q)^2}
	\le T \;,
	\label{eq:large-T}
	\tag{Large $T$}
\end{equation}
In this regime, we prove an upper bound on
the \textit{last-iterate} in $\KL$ divergence, which 
translates into an upper bound on the last-iterate
in duality gap. In particular, we prove the following:


\begin{restatable}[Large $T$ Last-Iterate Bound]
{lem}{lemlargeT}
	\label{lem:best-dg-large-T}
	Fix $0 < \delta_p, \delta_q \le \tfrac{1}{10}$,
	and let $A := A_{\delta_p, \delta_q}$. 
	Let $\{w_t\}$ be the iterates of OMWU on 
	$A$ with $\eta$ satisfying Assumption~\ref{ass:stepsize-dg-upper}.
	Let $T$ be as in \eqref{eq:large-T}.
	Then:
	\begin{equation*}
		\DG(w_T) 
		\le 
		\sqrt{4 \KL(w^\star, w_T)} 
		\le 
		O\Big( 
		\exp\big( - \tfrac{1}{2} \eta \sqrt{T} \big)
		\Big) \;.
	\end{equation*}
\end{restatable}
\noindent
The proof of the lemma is in Section~\ref{app:2x2-dg:large-T-proof}.

\paragraph{Moderate $T$ regime.}
In the moderate $T$ regime, we assume that 
\begin{equation}
	\frac{L}{\eta \delta_p \delta_q}
	\le T 
	< 
	\frac{1440 \cdot L \cdot \log(2/V)}{\eta^2 (\delta_p \delta_q)^2}
	 \;.
	\label{eq:mod-T}
	\tag{Moderate $T$}
\end{equation}
For this regime, we prove the following bound on the 
best-iterate in duality gap in terms of the NE parameters
$\delta_p$ and $\delta_q$. 
Formally, we prove:

\begin{restatable}[Moderate $T$ Best-Iterate Bound]
{lem}{lemmodT}
	\label{lem:best-dg-mod-T}
	Fix $0 < \delta_p, \delta_q \le \tfrac{1}{10}$,
	and let $A := A_{\delta_p, \delta_q}$. 
	Let $\{w_t\}$ be the iterates of OMWU on 
	$A$ with $\eta$ satisfying Assumption~\ref{ass:stepsize-dg-upper}.
	Let $T$ be as in \eqref{eq:mod-T}.
	Then there exists $t \in [T]$ such that
	\begin{equation*}
		\DG(w_t) 
		\le 
		O\big( \delta_p \delta_q \big) \;.
	\end{equation*}
\end{restatable}
\noindent
The proof of the lemma is in Section~\ref{app:2x2-dg:mod-T-proof}.

\paragraph{Small $T$ regime.}
Finally, in the small $T$ regime, we assume that 
\begin{equation}
	1 \le T
	<
	\frac{L}{\eta \delta_p \delta_q} \;.
	\label{eq:small-T}
	\tag{Small $T$}
\end{equation}
In this regime, we prove the following bound:

\begin{restatable}
[Small $T$ Best-Iterate Bound]
{lem}{lemsmallT}
	\label{lem:best-dg-small-T}
	Fix $0 < \delta_p, \delta_q \le \tfrac{1}{10}$,
	and let $A := A_{\delta_p, \delta_q}$. 
	Let $\{w_t\}$ be the iterates of OMWU on 
	$A$ with $\eta$ satisfying Assumption~\ref{ass:stepsize-dg-upper}.
	Let $T$ be as in \eqref{eq:small-T}.
	Then there exists a time $t \in [T]$ such that
	\begin{equation*}
		\DG(w_t) 
		\le
		O\Big(
		\max\Big\{\frac{1}{\eta T}, \delta_p \Big\}\Big) \;.
	\end{equation*}
\end{restatable}
\noindent
We note that the statement of the lemma
is similar to that of Theorem 5 of~\cite{C25}.
We give the proof in Section~\ref{app:2x2-dg:small-T-proof}.

\smallskip

Moreover, we note that, in the proofs of each of
Lemmas~\ref{lem:best-dg-large-T},~\ref{lem:best-dg-mod-T},
and~\ref{lem:best-dg-small-T},
the key technical component involves controlling
the amount of time until the iterates enter certain
regions of the primal space (defined relative to $w^\star$), 
and the number of iterations
spent within these regions. 
Using the tools introduced in
Section~\ref{app:2x2-setup:monotonicity},
we establish such bounds in 
Lemma~\ref{lemma:2x2-epoch-bounds},
which we state and 
prove in Section~\ref{app:2x2-dg:iterate-trajectory}.

\smallskip

\paragraph{Combining the Pieces.}
Using the preceding lemmas, the final bound
on the best-iterate in duality gap is straightforward.
For this, let $\{w_t\}$ be the iterates of running
OMWU on $A$ with stepsize $\eta$,
and let $\{w'_t\}$ be the iterates 
of running OMWU on $A_{\delta_p, \delta_q}$ 
with stepsize $\eta' = \gamma \eta$,
as described in Proposition~\ref{prop:2x2-suffices}.

\noindent
Then for the iterates $\{w'_t\}$, we have 
by Lemmas~\ref{lem:best-dg-large-T},~\ref{lem:best-dg-mod-T},
and~\ref{lem:best-dg-small-T} the following:
\begin{itemize}[
	leftmargin=1.5em
]
	\item
	If $T$ satisfies~\eqref{eq:large-T}, 
	then by Lemma~\ref{lem:best-dg-large-T},
	we have for all such $T$ that
	\begin{equation}
		\DG_{A_{\delta_p, \delta_q}}(w'_T) 
		\le O(\exp(-\tfrac{1}{2} \eta' T)) \le \frac{1}{\eta' \sqrt{T}}  \;.
		\label{eq:c-large}
	\end{equation}

	\item
	If $T$ satisfies~\eqref{eq:mod-T},
	then observe by definition that 
	\begin{equation*}
		T \ge \frac{L}{\eta' \delta_p \delta_q} \ge \frac{L}{\eta' \delta_p^2}
		\;\implies\;
		\delta_p \ge \sqrt{\frac{L}{T \eta'}} \;.
	\end{equation*}
	Moreover, by definition of $V$ from~\eqref{eq:V-L-def}, 
	it holds for $\delta_q \le 1/10$ that $V \ge \delta_q \ge \delta_p$,
	and thus
	$\log(2/V) \le \log(2/\delta_q) \le \log(2\sqrt{T\eta'/ L})
	\le \log(\sqrt{T/L}) \le \log(T/L)$,
	where the penultimate inequality is due to $\eta' \le 1/4$
	under Assumption~\ref{ass:stepsize-dg-upper}.
	Thus it further follows by the upper constraint on $T$ that
	\begin{equation*}
		\delta_p \delta_q 
		\le
		\frac{\sqrt{1440 L \log (2/V)}}{\eta' \sqrt{T}}
		\le 
		\frac{\sqrt{1440 L \log (T/L)}}{\eta' \sqrt{T}} \;.
	\end{equation*}
	Then using Lemma~\ref{lem:best-dg-mod-T}, we find
	\begin{equation}
		\min\nolimits_{t \in [T]}\,\DG_{A_{\delta_p, \delta_q}}(w'_t)
		\le 
		O(\delta_p \delta_q)
		\le O\Big( \frac{\sqrt{\log T}}{\eta' \sqrt{T}}\Big) \;.
		\label{eq:c-mod}
	\end{equation}

	\item
	Finally, if $T$ satisfies~\eqref{eq:small-T}, then
	observe that 
	\begin{equation*}
		T < \frac{L}{\eta' \delta_p \delta_q} \le \frac{L}{\eta' \delta_p^2}
		\;\;\implies\;\;
		\delta_p 
		\le 
		\sqrt{\frac{L}{\eta' T}}
		\le 
		\frac{\sqrt{L}}{\eta' \sqrt{T}} \;.
	\end{equation*}
	Then by Lemma~\ref{lem:best-dg-small-T}, we conclude
	\begin{equation}
		\min\nolimits_{t \in [T]}\,\DG_{A_{\delta_p, \delta_q}}(w'_t)
		\le 
		O\Big( 
			\max\Big\{ \frac{1}{\eta' T}, \delta_p \Big\}
		\Big)
		\le 
		O\Big(
			\frac{\sqrt{\log T}}{\eta' \sqrt{T}}
		\Big) \;.
		\label{eq:c-small}
	\end{equation}
\end{itemize}
Thus using the fact that $\eta' = \gamma \cdot \eta$, 
for any $T \ge 1$,
combining the conclusions of the three regimes 
in expressions~\eqref{eq:c-large},~\eqref{eq:c-mod},
and~\eqref{eq:c-small} together gives
\begin{equation*}
	\min\nolimits_{t \in [T]}\,\DG_{A_{\delta_p, \delta_q}}(w'_t)
	\le 
	O\Big(\frac{1}{\gamma \eta \cdot \sqrt{T}}\Big) \;.
\end{equation*}
Finally, by Proposition~\ref{prop:2x2-dg-compare}, 
it follows that
\begin{equation*}
	\min\nolimits_{t \in [T]}\,\DG_{A_{\delta_p, \delta_q}}(w'_t)
	\le 
	O\Big(\frac{\sqrt{\log T}}{\gamma \eta \cdot \sqrt{T}}\Big) 
	\;\implies\;
	\min\nolimits_{t \in [T]}\,\DG_{A}(w_t)
	\le 
	O\Big(\frac{\sqrt{\log T}}{\eta \cdot \sqrt{T}}\Big) \;.
\end{equation*}
Noting that $\eta$ is an absolute constant
under Assumption~\ref{ass:stepsize-dg-upper}
concludes the proof.
~\hfill $\blacksquare$

\subsection{Sufficiency of OMWU on 
Canonical $2 \times 2$ Matrix}
\label{app:dg-upper:canonical}

Here, we give the proofs of 
Propositions~\ref{prop:2x2-suffices} and \ref{prop:2x2-dg-compare},
which are used in the proof of Theorem~\ref{thm:dg-best-2x2}.

\subsubsection{Proof of Proposition~\ref{prop:2x2-suffices}}

\begin{proof}
	Let $A$ have entries $A = ((a, b), (c, d)) \in \R^{2\times 2}$.
	Let $v = \langle p^\star, A q^\star \rangle \in \R$.
	For readability, let $A' = A_{\delta_p, \delta_q}$. 
	For each $t \ge 0$, let $w_t = (p_t, q_t)$
	and $w'_t = (p'_t, q'_t)$.
	Now recall from the definition of~\eqref{eq:omwu-primal} 
 	and by Proposition~\ref{prop:oftrl-softmax},
	we have under the uniform initialization that for all $t \ge 0$:
	\begin{align*}
		&\begin{cases}
			p_t = \softmax(
				- \eta (\sum_{k=0}^{t-1} Aq_k + A q_{t-1} )
			) \\
			q_t = \softmax(
				\eta (\sum_{k=0}^{t-1} A^\top p_k + A^\top p_{t-1})
			)
		\end{cases} \\
		\text{and}\quad
		&\begin{cases}
		p'_t = \softmax(
			- \eta' (\sum_{k=0}^{t-1} Aq'_k + A q'_{t-1} )
		) \\
		q'_t = \softmax(
			\eta' (\sum_{k=0}^{t-1} A'^\top p'_k + A'^\top p'_{t-1})
		).
		\end{cases}
	\end{align*}
	We will show inductively that $p_t = p'_t$ and $q_t = q'_t$ for 
	all $t \ge 0$. The base case at time $t=0$ trivially holds
	by assumption in the proposition. 
	Now suppose the claim also holds for all $0 \le k \le t-1$.
	Now recall from Part (ii) of Proposition~\ref{prop:A-properties}
	that for $\gamma = \frac{|a - v|}{\delta_p \delta_q} \in (0, 4]$, 
	$A$ can be decomposed as 
	\begin{equation*}
		A = \gamma A' + v \1 \1^\top \;.
	\end{equation*}
	Thus at time $t$, observe that the cumulative payoff vector
	can be written as
	\begin{align}
		- \eta\Big(\sum_{k=0}^{t-1} Aq_k + Aq_{t-1}\Big)
		&=
		- \eta \Big(\sum_{k=0}^{t-1} (\gamma A' q_k + c \1) 
		+ \gamma A' q_{t-1} + c \1 \Big) \nonumber \\
		&= 
		- \eta \gamma 
		\Big(
		A' q'_k + A' q'_{t-1}
		\Big)
		- \eta c (t+1) \1 \;.
		\label{eq:sm-01}
	\end{align}
	Here, we applied in the second line the inductive
	hypothesis $q_k = q'_k$ for all $0 \le k \le t-1$.
	Then applying the softmax operator to~\eqref{eq:sm-01},
	it follows that 
	\begin{align*}
		\textstyle
		p_t &= \softmax\Big(
			- \eta \Big(\sum_{k=0}^{t-1} Aq_k + Aq_{t-1}\Big)
		\Big) \\
		&= 
		\softmax\Big(
			- \eta \gamma 
			\Big(\sum_{k=0}^{t-1} A'q'_k + A'q'_{t-1}\Big)
			- \eta c(t+1) \1
		\Big) \\
		&= 
		\softmax\Big(
			- \eta'
			\Big(\sum_{k=0}^{t-1} A'q'_k + A'q'_{t-1}\Big)
		\Big)
		= p'_t \;,
	\end{align*}
	where the penultimate equality comes from the
	invariance of softmax to constant shifts,
	and the final equality comes from the definition of $p'_t$.
	Repeating an identitical calculation, we similarly find
	$q_t = q'_t$, which establishes $w_t = w'_t$ for all $t \ge 0$
	and concludes the proof.
\end{proof}

\medskip

\begin{restatable}{rem}{remtwotwosuffices}
	\label{remark:2x2-suffices-softmax}
	We note that the statement of the 
	result is similar to Proposition 2 of~\cite{C25},
	but using a different ``base matrix'' 
	(see also Remark~\ref{rem:canonical-2x2-matrix}).
	Also, observe that the proof of
	Proposition~\ref{prop:2x2-suffices} only
	relies on the invariance of the softmax function 
	to constant shifts.
	Thus, the statement	 also holds under other MWU variants
	(e.g., standard MWU, c.f.
	Section~\ref{app:omwu-prelims:omwu-related}).
	Finally, we assume for simplicity that the initialization
	$w_0 = ((\frac{1}{2}, \frac{1}{2}),(\frac{1}{2}, \frac{1}{2}))$
	is uniform, but the proof also extends in a straightforward
	way to non-uniform initializations.
\end{restatable}

\subsubsection{Proof of Proposition~\ref{prop:2x2-dg-compare}}

\begin{proof}
	Recall definition of duality gap (i.e., from Section~\ref{sec:prelims})
	that for any $w = (p, q) \in \calW$,  we have
	\begin{equation*}
		\begin{aligned}
			\DG_A(w) 
			&= \max_{q' \in \Delta_2} \langle q', A^\top p \rangle
			- 
			\min_{p' \in \Delta_2} \langle p', A q \rangle \\
			\text{and}\;\;
			\DG_{A'}(w) 
			&= \max_{q' \in \Delta_2} \langle q', A'^\top p \rangle
			- 
			\min_{p' \in \Delta_2} \langle p', A' q \rangle \;.
		\end{aligned}
	\end{equation*}
	Now using the fact that 
	$A = \gamma A' + v \1\1^\top$ for $\gamma, v \in \R$
	and $\gamma > 0$, it follows that
	\begin{align*}
		\DG_{A}(w)
		&=
		\gamma \cdot 
		\max_{q' \in \Delta_2} \langle q', A' p \rangle + v
		- 
		\big(\gamma \cdot \min_{p' \in \Delta_2} \langle p', A' q \rangle + v\big) \\
		&= 
		\gamma \cdot 
		\big(
			\max_{q' \in \Delta_2} \langle q', A'^\top p \rangle
			- 
			\min_{p' \in \Delta_2} \langle p', A' q \rangle
		\big) \\
		&= \gamma \cdot \DG_{A'}(w) \;.
	\end{align*}
	Thus it follows that if $\DG_{A'}(w) \le B$ for any $B > 0$, 
	then $\DG_A(w) \le \gamma\cdot$.
\end{proof}

\subsubsection{Properties of Iterate Trajectory}
\label{app:2x2-dg:iterate-trajectory}

In this section, we state and prove two key technical lemmas
regarding the OMWU iterates on the canonical 
$2\times 2$ matrix $A_{\delta_p, \delta_q}$.
In the following lemma, we start by establishing
(somewhat crude) uniform bounds on the coordinates of 
the primal and dual iterates when initialized from
the uniform distribution.

\smallskip

\begin{restatable}[Energy and Coordinate Bounds
on Canonical 2x2 Matrix]
{lem}{lemtwotwouniformcoords}
	\label{lemma:2x2-max-coord-bound}
	Fix $0 < \delta_p \le \delta_q \le \frac{1}{10}$.
	Let $A := A_{\delta_p, \delta_q} \in [-1,  1]^{2\times2}$
	be the matrix from Definition~\ref{def:A-2x2}.
	Let $\{w_t\}$ and $\{z_t\}$ be the primal and dual iterates
	of OMWU for $\eta$ satisfying 
	Assumption~\ref{ass:stepsize-dg-upper}
	and initialized from the uniform distribution. 
	For each $t \ge 0$, let $z_t = (x_t, y_t)$
	and let $w_t = (p_t, q_t)$.
	Then for all $t \ge 0$, the following bounds hold:
	\begin{enumerate}[
		label={(\roman*)},
		itemsep=0em,
		topsep=0em,
		leftmargin=3em,
	]
		\item
		$F(z_t) < F(z_{t-1}) < \dots < F(z_1) \le
		\big(\tfrac{5}{4}\big) \cdot F(z_0) \le 2$.
		\item
		$- \big(\frac{\delta_p}{1-\delta_p}\big) \cdot 2 \le x_t(1) \le 2$ and 
		$- \big(\frac{\delta_q}{1-\delta_q}\big) \cdot 2\le y_t(1) \le 2$.
		\item
		$\min\{p_t(1), p_t(2)\} \ge \min\{\frac{1}{12}, 1-p_t(1)\}$ ~and~ 
		$\min\{q_t(1), q_t(2)\} \ge \min\{\frac{1}{12}, 1-q_t(1)\}$.
	\end{enumerate}
\end{restatable}

\begin{proof}
	First, under the uniform initialization $p_0 = (1/2, 1/2)$
	and $q_0 = (1/2, 1/2)$, observe by the 
	the primal-dual relationship in Part (iii) of
	Proposition~\ref{prop:2x2-energy-map-simplified} that
	$x_0 = (0, 0)$ and $y_0 = (0, 0)$.
	Thus by definition of $F(z_0)$, we have by Part (i) of 
	Proposition~\ref{prop:2x2-energy-map-simplified} that
	\begin{equation*}
		\textstyle
		F(z_0) = 
		x_0(1) + \log(1+\exp(\frac{-x_0(1)}{\delta_p})) + 
		y_0(1) + \log(1+\exp(\frac{-y_0(1)}{\delta_q}))  
		= 2 \log 2 \;.
	\end{equation*}
	Under the setting of $\eta$, recall by 
	Lemma~\ref{lem:energy-one-step-full} 
	that we have $F(z_t) - F(z_{t-1}) < 0$ for all $t \ge 2$.
	Moreover, by Lemma~\ref{lem:initial-change-energy},
	we further have 
	$F(z_1) \le \big(\tfrac{5}{4}\big) \cdot F(z_0)$
	Thus for all $t \ge 0$, we have 
	\begin{equation*}
		F(z_t) \le \big(\tfrac{5}{4}\big) \cdot F(z_0) 
		\le \tfrac{5}{2} \log 2 
		\le 2 \;.
	\end{equation*}
	This proves claim (i).

	Claim (ii) then follows by a direct application of 
	Proposition~\ref{prop:2x2-max-dual-width}
	using the bound $F(z_t) \le 2$.

	For claim (iii), observe from Proposition~\ref{prop:omwu-pd-2x2}
	(and also by Part (ii) of Proposition~\ref{prop:2x2-energy-map-simplified})
	that $p_t(1) = \sigmoid(x_t(1)/\delta_p)$ for all $t \ge 0$.
	Thus, using $x_t(1) \ge - 2 (\delta_p/(1-\delta_p))$, 
	we can further bound
	\begin{equation*}
		p_t(1) = 
		\frac{1}{1+\exp(2/(1-\delta_p))}
		\ge 
		\frac{1}{12} \;,
	\end{equation*}
	where the final inequality holds for all $\delta_p \le 1/10$.
	Thus 
	\begin{equation*}
		\textstyle
		\min\{p_t(1), p_t(2)\} = \min\{p_t(1), 1-p_t(1)\}
		\ge \min\{1/12, 1-p_t(1)\} \;.
	\end{equation*}
	Using an identitical calculation,
	we similarly find $\min\{q_t(1), q_t(2)\} \ge \min\{1/12, 1-q_t(2)\}$,
	which yields claim (iii) of the lemma and concludes the proof.
\end{proof}

\smallskip
The next lemma establishes a more fine-grained control
of the trajectory of the OMWU iterates on $A_{\delta_p, \delta_q}$.
In particular, the lemma proves two bounds in parallel:
first, it gives bounds on the amount of time 
until the iterates enter a region
of the primal space near the interior NE.
Second, it establishes a lower bound on the number 
of steps taken in regions where the minimum coordinates
of the primal iterates are uniformly lower bounded.
This latter property is used in the proof of Theorem~\ref{thm:dg-best-2x2} 
to establish a faster decay in $\KL$ divergence
(compared to the worst-case bound of Theorem~\ref{thm:kl-last-unified})
when $T$ is in the large regime.
Formally, we prove the following:

\smallskip

\begin{restatable}[Epoch invariants]{lem}{lemepochunified}
	\label{lemma:2x2-epoch-bounds}
	Fix $0 < \delta_p \le \delta_q \le \frac{1}{10}$,
	and let $A := A_{\delta_p, \delta_q} \in [-1, 1]^{2\times 2}$
	be the matrix from Definition~\ref{def:A-2x2}.
	Let $\{w_t\}$ be the iterates of OMWU
	on $A$ with stepsize $\eta > 0$ satisfying
	Assumption~\ref{ass:stepsize-dg-upper},
	and let $w_t = (p_t, q_t)$.
	Let $\tau = 40/(\eta \delta_p \delta_q)$.
	Then for any $t' \ge 0$:
	\begin{enumerate}[
		label={(\alph*)},
		itemsep=0em,
		topsep=0em,
		leftmargin=2.5em
	]
	\item 
	There exists at least one iteration $t' \le t \le t' + \tau$
	such that
	\begin{equation*}
		p_{t}(1) \ge 1- \delta_p
		\;\;\text{and}\;\;
		1 - 3 \delta_q \le q_t(1) \le 1-\delta_q\;.
	\end{equation*}

	\item
	There exist at least $\frac{1}{\eta}$ iterations
	$t' \le t \le t' + \tau$ such that
	\begin{equation*}
		p_{t}(1) \le 1- (\delta_p/3)
		\;\;\text{and}\;\;
		q_{t}(1) \le 1- (\delta_q/3) \;.
	\end{equation*}
	\end{enumerate}
\end{restatable}

\begin{proof}
To prove the theorem, we repeatedly apply the inequalities
from Lemma~\ref{lemma:2x2-drift-bounds-unified}
to control the trajectory of the dual (and thus primal)
OMWU iterates. 
To start, we set some notation: 
for any $t \ge 0$, let $z_t = (x_t, y_t) \in \calZ$
be the dual iterate of OMWU that evolves as in
Proposition~\ref{prop:omwu-pd-2x2}.
Let $\w x_t = x_t(1), \w y_t = y_t(1) \in \R$,
and let $\w p_t = p_t(1), \w q_t = q_t(1) \in (0, 1)$. 

We also make use of the following regions of 
$\calW = \Delta_2 \times \Delta_2$.
These are defined in terms of the leading coordinates 
$(\w p, \w q) \in [0, 1]^2$ for $(p, q) \in  \Delta_2 \times \Delta_2$
and $\w p = p(1), \w q = q(1)$:
\begin{equation}
\begin{aligned}
	\calW_{\rno}
	&= \{ (p, q) \in \calW \;:\; \w p \le 1-\delta_p / 3,\; \w q \le 1-3\delta_q\} \\
	\calW_{\rnt}
	&= \{ (p, q) \in \calW \;:\; \w p \ge 1-\delta_p/3,\; \w q \le 1-\delta_q/3\} \\
	\calW_{\rnth}
	&= \{ (p, q) \in \calW \;:\; \w p \ge 1-3\delta_p,\; \w q \ge 1-\delta_q/3\} \\
	\calW_{\rnf}
	&= \{ (p, q) \in \calW \;:\; \w p \le 1-3\delta_p,\; \w q \ge 1-3\delta_q\} \\
	\calI 
	&= 
	\{(p, q) \in\calW\;:\; 1-3\delta_p < \w  p <  \delta_p / 3\;,\;
	1-3\delta_q < \w q < \delta_q/ 3\} \;.
\end{aligned}
\label{eq:calW-defs}
\end{equation}
Without loss of generality, fix $t' = 0$.
We will consider two cases. In the first case, 
suppose for all iterations $t \in [\tau]$ that 
$w_t = (p_t, q_t) \notin \calI$.
If $w_0 \in \calW_{\rno}$, then using 
Part (i) of Lemma~\ref{lemma:2x2-drift-bounds-unified}, 
it follows for all $t \ge 1$ that 
while $\w q_{t-1} \le 1-3\delta_q$, then
\begin{equation*}
	\textstyle
	\w x_{t+1} - \w x_t \ge \frac{1}{2} \cdot \eta \delta_p \delta_q \;.
\end{equation*}
Now let $t' > 0$ be the first time such that $w_{t'} \notin \calW_{\rno}$.
Since we assume that $w_t \notin \calI$ for all $t \in [\tau]$,
then we must have $w_{t'} \in \calW_{\rnt}$.
Moreover, as $\w w_{t'} \in \calW_{\rnt}$ 
implies $\w p_{t'} \ge 1-\delta_p/3$, 
we have by the primal-dual relationship 
of Corollary~\ref{corr:2x2-pd-threshold}
that $\w x_{t'} \ge \delta_p \log((1-\delta_p/3)/(\delta_p/3))$ .

Further observe from Part (ii) of 
Lemma~\ref{lemma:2x2-max-coord-bound}
that $-2 \le \w x_t  \le 2$
for any $t \ge 0$,
which holds under the assumption that $\delta_p, \delta_q \le 1/10$.
Thus for any $t, t' \ge 0$, the maximum distance between 
any two dual iterates is bounded by $|\w x_t - \w x_{t'}| \le 4$.
Together with the bound 
$\w x_{t+1} - \w x_t \ge \tfrac{1}{2} \eta \delta_p \delta_q$, 
we must have 
\begin{equation*}
t' \le 4 / (\tfrac{1}{2} \eta \delta_p \delta_q)
=  8 / (\eta \delta_p \delta_q) = \tau / 5 \;.
\end{equation*}
Repeating a similar argument via parts (ii), (iii), and (iv)
of Lemma~\ref{lemma:2x2-drift-bounds-unified}, we 
further find that the iterates $\{w_t\}$ 
cycle from $\calW_{\rnt}$ to $\calW_{\rnth}$
within at most $\tau/5$ iterations, 
from $\calW_{\rnth}$ to $\calW_{\rnf}$
within at most $\tau/5$ iterations,
and from $\calW_{\rnf}$ back to $\calW_{\rno}$
within at most $\tau/5$ iterations. 

Now observe that, by definition of the regions 
in~\eqref{eq:calW-defs},  
between the first $t$ such that $w_t \in \calW_{\rno}$
and the first time $t' \ge t$ such that $w_{t'} \in \calW_{\rnf}$,
there must be an iterate $t''$ such that
$\w p_{t''} \ge 1-\delta_p$ and 
$1-3\delta_q \le \w q_{t''} \le 1-\delta_q$. 
Observe that such an iterate $t''$ satisfies property (a) of the lemma,
as we have from the previous calculation that
$0 \le t'' \le (4/5) \tau \le \tau$, as required.

Moreover, between the first time $t$ such that 
$w_t \in \calW_{\rnf}$ and the first time $t' \ge t$
such that $w_{t'} \in \calW_{\rnt}$,
there must be at least one iteration $t''$ such that
$\w q_{t''} \le 1-3\delta_q$ and $\w p_{t''} \le 1-3\delta_p$. 
Now observe from the dual OMWU update rule of 
Proposition~\ref{prop:omwu-pd-2x2} that we can bound
\begin{equation*}
	\w x_{t+1} - \w x_t 
	= 
	\eta \delta_p \cdot 
	(1-\delta_q - \w q_t - (\w q_t - \w q_{t-1}))
	\le 2 \eta \delta_p  \;.
\end{equation*}
Then letting $t''' > t''$ be the first time such that
$\w p_{t'''} \ge 1-\delta_p/3$, it holds via the
primal-dual relationship of Corollary~\ref{corr:2x2-pd-threshold} that
\begin{equation*}
	t''' - t'' \ge \frac{ \delta_p \cdot \log((9-3\delta_p)/(1-\delta_p/3))}{2 \eta \delta_p}
	\ge 
	\frac{1}{\eta} \;.
\end{equation*}
Here the final inequality comes from 
the fact that 
$ \log((9-3\delta_p)/(1-\delta_p/3)) \ge 2$ for all $\delta_p > 0$.
As $t''' \le 4/5 \tau \le \tau$ by the prior arguments,
this establishes property (b) of the lemma. 

Now, for the second case, suppose that
there exists at least one iteration $t \in [\tau]$ such that $w_t \in \calI$.
If the total number of iterations $w_t \in \calI$
is at least $\tau/5$, 
then property (b) of the lemma is trivially satisfied.
Note in this case that we must also have at least one iteration 
$t$ such that $\w p_t \ge 1-\delta_p$ and $\w q_t \le 1-\delta_q$
due to the strict convexity of the energy function $F$ over $\calZ$, 
and by the fact that $F(z_{t+1}) < F(z_{t})$ for all $t \ge 1$ (by Part (i)
of Lemma~\ref{lemma:2x2-max-coord-bound}).
Thus, property (a) of the lemma is also satisfied in
this subcase. 

On the other hand, suppose the total number of iterations $t \in [\tau]$
with $w_t \in \calI$ is at most $\tau/5$. 
In this case, note again by the strict convexity of $F$
over $\calZ$ and the fact that $F(z_{t+1}) < F(z_{t})$, 
the remaining at least $4\tau/5$ iterations
must at some point transition between the regions 
$\calW_{\rno} \to \calW_{\rnt}
\to \calW_{\rnth} \to \calW_{\rnf} \to \calW_{\rno}$.
Using the prior arguments for the first case, 
it then holds that properties (a) and (b) are both satisfied
within the remaining at least $4\tau/5$ iterations,
which concludes the proof.
\end{proof}

\subsection{Proof of Lemmas~\ref{lem:best-dg-large-T}, 
~\ref{lem:best-dg-mod-T}, and~\ref{lem:best-dg-small-T}}
\label{app:2x2-dg:lemmas-T}

\subsubsection{Proof of Lemma~\ref{lem:best-dg-large-T}}
\label{app:2x2-dg:large-T-proof}

To prove the lemma for the regime~\eqref{eq:large-T}, 
we analyze the iterates of OMWU in two phases.
In the first startup phase, the energy dissipation
(and thus contraction in $\KL$) is fast for at least $1/\eta$ steps
per every epoch of $O(1/(\delta_p\delta_q))$ iterations.
This follows from establishing a tighter lower bound on 
the minimum coordinates $p_{t, \min}$ and $q_{t, \min}$
than the worst-case bound of
Lemma~\ref{lem:wtmin-bound-uniform}.

After $O(1/(\delta_p \delta_q))$ epochs
in the startup phase, meaning
 $O(1/(\delta_p \delta_q)^2)$ total iterations,
the energy will then be sufficiently small enough establish
a uniform lower bound on $p_{t, \min}$ and $q_{t, \min}$
for all remaining iterations. Under the~\eqref{eq:large-T},
this then leads to a bound on the last-iterate 
$\KL(w^\star, w_T) \le O(\exp(- \eta \sqrt{T}))$,
which by~\eqref{cor:distance-relations} translates into a
last-iterate bound in duality gap. 
We proceed to detail these steps:

\paragraph{Energy threshold for uniformly bounded coordinates.} 
Let $V > 0$ 
be the threshold value from expression~\eqref{eq:V-L-def}.
By Proposition~\ref{prop:2x2-max-dual-width}, it holds
for any $z = (x, y) \in \calZ$ that if $F(z) \le V$, then 
\begin{equation*}
-(\delta_q/(1-\delta_q)) \cdot V \le x(1) \le  V
\;\;\text{and}\;\;
-(\delta_q/(1-\delta_q)) \cdot V \le y(1) \le  V \;.
\end{equation*}
By definition of $V$, and using the primal-dual relationships of 
Proposition~\ref{prop:2x2-energy-map-simplified},
it then follows for $w = (p, q) = \nabla F(z) \in \relint(\calW)$
that $p(1) \le 1 - \delta_q/3$ and $q(1) \le 1- \delta_q/3$.

Using Parts (i) and (iii) of 
Lemma~\ref{lemma:2x2-max-coord-bound},
we thus conclude that
if ever $F(z_{T'}) \le V$ for some $T' \ge 1$, then
$\min\{p_t(1), p_t(2)\}, \min\{q_t(1), q_t(2)\} \ge \delta_q/3$
for all subsequent $t \ge T'$
(which holds under the assumption that
$\delta_p \le \delta_q \le 1/10$).

Moreover, $T'' \ge 1$
is the first iterate such that $\KL(w^\star, w_{T''}) \le V$.
By Proposition~\ref{prop:calZ-minimizer}, 
we have $\KL(w^\star, w) \le F(z)$ for $w = \nabla F(z)$.
Thus if $T'$ is the first time such that $F(z_{T'}) \le V$,
then $T' \le T''$.
We will show in the following that $T'' \le T/2$,
which will then establish the uniform bound 
on $\min\{p_t(1), p_t(2)\}$ and $\min\{q_t(1), q_t(2)\}$
for all subsequent $t \ge T''$.

\paragraph{Energy decay in startup phase.}
To establish the upper bound on $T''$, 
recall by Property (b) of Lemma~\ref{lemma:2x2-epoch-bounds},
that simultaneously $p_t(1) \ge 1-\delta_p/3$
and $q_t(1) \ge 1-\delta_q/3$
for at least $1/\eta$ iterations $t$ every 
$\tau = \frac{L}{\eta \delta_p\delta_q}$ total steps,
where $L$ is as defined in~\eqref{eq:V-L-def}.
By Part (iii) of Lemma~\ref{lemma:2x2-max-coord-bound},
this means for all such $t$ that 
$p_{t, \min} = \min\{p_t(1), p_t(2)\} \ge \delta_p/3$ and
$q_{t, \min} = \min\{q_t(1), q_t(2)\} \ge \delta_q/3$. 
Now define the values $K$ and $T_{\text{startup}}$ by
\begin{equation}
	K = \frac{720 \cdot \log(2/V)}{\eta \delta_p \delta_q}
	\;\;\text{and}\;\;
	T_{\text{startup}}
	= 
	K \cdot \tau = K \cdot \frac{L}{\eta \delta_p \delta_q} \;.
	\label{eq:K-Tthresh-def}
\end{equation}
Further recall from Proposition~\ref{prop:rsv-canonical-2x2}
that for $A = A_{\delta_p, \delta_q}$, we have 
$\sigma_{\min} = 
\sigma_{\min}(A_{\delta_p, \delta_q}, \calS^\bot) = \frac{1}{2}$.
Then by the universal one-step multiplicative change 
in $\KL$ divergence from Theorem~\ref{thm:kl-last-unified}
(and in particular, using the expression from
Corollary~\ref{cor:kl-last-improved-one-step}, recall that
for every such $t$ where 
$p_{t, \min} \ge \delta_p/3$ and $q_{t, \min} \ge \delta_q/3$,
we have
\begin{align*}
	\KL(w^\star, w_{t+1})
	&\le 
	\KL(w^\star, w_t) \cdot 
	\exp\big(
		- \tfrac{1}{20} \eta^2 \sigma_{\min}^2 \cdot
		p_{t, \min} \cdot q_{t, \min}
	\big) \\
	&\le 
	\KL(w^\star, w_t) \cdot 
	\exp\big(
		- \tfrac{1}{720} \eta^2 \delta_p \delta_q
	\big) \;.
\end{align*}
By definition of $T_{\text{startup}}$, 
Property (b) of Lemma~\ref{lemma:2x2-epoch-bounds} 
implies that such a bound on 
$p_{t, \min}$ and $q_{t, \min}$
must hold for at least $K/\eta$ iterations.
Thus, by definition of $K$, 
we have at time $T_{\text{startup}}$ that
\begin{align*}
	\KL(w^\star, w_{T_{\text{startup}}})
	&\le 
	\KL(w^\star, w_0)
	\cdot 
	\exp\big(- 
		 \tfrac{1}{720} \eta^2 \delta_p \delta_q \cdot 
		 \tfrac{K}{\eta}
	\big) \\
	&\le 
	\KL(w^\star, w_0) 
	\cdot 
	\tfrac{V}{2} 
	\le 
	V \;.
\end{align*}
Here, the final inequality is due to $\KL(w^\star, w_0) \le 2$,
which follows from the fact that 
$F(z_0) \le 2$ 
(by  Part (i) of Lemma~\ref{lemma:2x2-max-coord-bound}),
and the fact that $\KL(w^\star, w_0) \le F(z_0)$ 
(by Proposition~\ref{prop:calZ-minimizer}).

Moreover, observe that in the regime~\eqref{eq:large-T}, 
we have by definition of $K$ and $T_{\text{startup}}$ that 
\begin{equation}
	T_{\text{startup}} 
	= 
	\frac{720 \cdot L \cdot \log(2/V)}
	{\eta^2 (\delta_p \delta_q)^2}
	\le \frac{T}{2} \;.
\end{equation}

\paragraph{Final last-iterate bound in KL Divergence and
Duality Gap.}
By the preceding arguments, it thus holds for 
all $T_{\text{startup}} \le t \le T$ that 
$p_{t, \min} \ge \delta_q/3$ 
and $q_{t, \min} \ge \delta_q/3$.
Again applying the bound on the one-step multiplicative 
change in $\KL$ from Theorem~\ref{thm:kl-last-unified},
we find
\begin{align}
	\KL(w^\star, w_T) 
	&\le 
	\KL(w^\star, w_{T_{\text{startup}}})
	\cdot 
	\exp
	\big(
		- \tfrac{1}{720} 
		\eta^2 \delta_q^2 \cdot
		(T - T_{\text{startup}})
	\big) \nonumber \\
	&\le 
	V
	\cdot 
	\exp\big(
		- \tfrac{1}{720}
		\eta^2 \delta_q \delta_p \cdot \tfrac{T}{2}
	\big) \;. 
	\label{eq:lt-01}
\end{align}
In the final inequality, we used the facts that 
$\delta_q \ge \delta_p$, that
$\KL(w^\star, w_{T_{\text{startup}}}) \le \KL(w^\star, w_0) \le V$,
and that $T-T_{\text{startup}} \ge T/2$. 
Now, by definition of the regime~\eqref{eq:large-T}, 
observe that 
\begin{equation*}
	\frac{1440 \cdot  L \cdot \log(2/V)}{\eta^2 (\delta_p\delta_q)^2} 
	\le T
	\;\;\implies\;\;
	\delta_p \delta_q 
	\ge 
	\frac{\sqrt{1440 \cdot  L \cdot \log(2/V)}}{\eta \sqrt{T}} 
	\ge 
	\frac{\sqrt{1440L}}{\eta \sqrt{T}}  \;.
\end{equation*}
Here, the last inequality follows from the fact 
that $V \le 1$ by definition of $V$ from~\eqref{eq:V-L-def},
and thus $\log(2/V) \ge 1$.
Then continuing from~\eqref{eq:lt-01}, 
we can further bound
\begin{align*}	
	\KL(w^\star, w_T) 
	\le 
	V
	\cdot 
	\exp\big(
	- \tfrac{\sqrt{2L}}{\sqrt{720}} \cdot \eta \sqrt{T}
	\big) 
	\le
	\exp\big(
	- \tfrac{\sqrt{2L}}{\sqrt{720}} \cdot \eta \sqrt{T}
	\big) \;.
\end{align*}
Now by the relationships between $\DG(w_T)$
and $\KL(w^\star, w_T)$ of Corollary~\ref{cor:distance-relations},
and using the fact that $A = A_{\delta_p, \delta_q} \in [-1,1]^{2\times 2}$,
we find
\begin{equation*}
 	\DG(w_T)
 	\le \sqrt{4 \KL(w^\star, w_T)} 
 	\le 
 	O\big(\exp(- \tfrac{1}{2} \cdot \eta \sqrt{T})\big) \;,
\end{equation*}
which concludes the proof.
~\hfill $\blacksquare$

\subsubsection{Proof of Lemma~\ref{lem:best-dg-mod-T}}
\label{app:2x2-dg:mod-T-proof}

For the regime~\eqref{eq:mod-T}, 
recall by Part (b) of Proposition~\ref{lemma:2x2-epoch-bounds}
that, within at most $\tau = \frac{L}{\eta \delta_p\delta_q}$
iterations (for $L$ defined in~\eqref{eq:V-L-def})
there exists at least one iteration $t \in [\tau]$ such that
\begin{equation}
	p_t(1) \ge 1-\delta_p
	\;\;\text{and}\;\; 1-3\delta_q \le q_t(1) \le 1-\delta_q \;.
	\label{eq:modt-01}
\end{equation}
For such an iterate $w_t = (p_t, q_t)$,
observe then that
\begin{equation*}
	p_t(1) - (1-\delta_p) \le 1
	\;\;\text{and}\;\;
	(1-\delta_q) - q_t(1) \le 2 \delta_q \;.
\end{equation*}
Then using the characterization of $\DG(\cdot)$
under $A_{\delta_p, \delta_q}$
from~\eqref{eq:dg-2x2-full} in Section~\ref{sec:2x2-dg},
this means that we can bound for this iterate
\begin{equation}
	\DG(w_t)
	\le 
	(p_t(1) - (1-\delta_p)) \cdot \delta_q
	+ (1-\delta_q - q_t(1) )\cdot \delta_p
	\le
	\delta_p \delta_q
	+ 2\delta_p \delta_q
	\le 
	3 \delta_p \delta_q \;.
	\label{eq:modt-02}
\end{equation}
For the regime~\eqref{eq:mod-T}, recall that we have 
$T \ge L / (\eta \delta_p \delta_q) = \tau$.
Thus by definition,
we are guaranteed the iterate $w_t$ satisfying~\eqref{eq:modt-01}
(and thus~\eqref{eq:modt-02})
occurs within $T$ total iterations.
~ \hfill $\blacksquare$

\subsubsection{Proof of Lemma~\ref{lem:best-dg-small-T}}
\label{app:2x2-dg:small-T-proof}

The proof distinguishes between two cases:

First, using the fact that initially $p_0 = q_0 = (1/2, 1/2)$,
we have by Part (iii) of Lemma~\ref{lemma:2x2-drift-bounds-unified} 
and the primal-dual relationship of 
Proposition~\ref{prop:omwu-pd-2x2}
the following: 
while $p_{t-1}(1), p_t(1) < 1-3\delta_p$, then 
\begin{equation}
	y_{t+1}(1) - y_t(1) < 0 
	\;\implies\;
	q_{t+1}(1) - q_t(1) < 0 \;.	
	\label{eq:st-01}
\end{equation}
Then the proof distinguishes between the following two cases:
\begin{enumerate}[
	label={(\alph*)}
]	
\item
There exists $1 \le t \le T$ such that
$p_t(1) \ge 1-3 \delta_p$. 
\item
It holds for all $1 \le t \le T$
that $p_t(1) < 1- 3\delta_p$.
\end{enumerate}

For Case (a), let $t$ be the first such time
where the iterate $w_t = (p_t,  q_t)$
satisfies $p_t(1) \ge 1-3\delta_p$.
By~\eqref{eq:st-01}, this further implies 
$q_t(1) \le q_t(0) \le 1/2 < 1- \delta_q$, 
where the final inequality holds under the assumption 
$\delta_q \le 1/10$.
Then by definition of $\DG(\cdot)$ under $A = A_{\delta_p, \delta_q}$
from~\eqref{eq:dg-2x2-full} in Section~\ref{sec:2x2-dg}, 
it holds that
\begin{equation}
	\DG(w_t) 
	= 
	(1-\delta_p - p_t(1)) (1-\delta_q)	
	+ 
	(1-\delta_q - q_t(1)) \delta_p 
	\le 
	2 \delta_p + \delta_p 
	= 3 \delta_p \;.
	\label{eq:st-casea}
\end{equation}

For Case (b), expression~\eqref{eq:st-01} further implies
that $q_t(1) \le q_{t-1}(1) \le q_0(1) = 1/2 \le 1-3\delta_q$
for all $t \in [T]$,
where the final inequality holds for $\delta_q \le 1/10$.
Now for each $t$, let $z_t = (x_t, y_t) \in \calZ$
be the corresponding dual OMWU iterate. 
By the update rule of Proposition~\ref{prop:omwu-pd-2x2},
it follows that 
\begin{align*}
	x_{t+1}(1) - x_t(1)
	&= 
	\eta \delta_p \cdot 
	\big(
		1-\delta_q - q_t(1) 
		- (q_t(1) - q_{t-1}(1))
	\big) \\
	&\ge 
	\eta \delta_p \cdot 
	\big(
		1-\delta_q - q_t(1) 
	\big)  
	&&\quad\text{(since $q_t(1) - q_{t-1}(1) \le 0$)}  \\
	&\ge 
	\eta \delta_p \cdot 
	\big(
		\tfrac{1}{2} - \delta_q
	\big)  
	&&\quad\text{(since $q_t(1) \le 1/2$)}  \\ 
	&\ge 
	\tfrac{1}{4} \cdot \eta \delta_p 
	&&\quad\text{(since $\delta_q \le 1/10 \le 1/4$)} \;.
\end{align*}
Thus for all $t \in [T]$ we find 
\begin{equation}
	x_{t}(1) \ge  x_0(1) + t \cdot \tfrac{\eta}{4} \delta_p 
	= t \cdot \tfrac{\eta}{4} {\delta_p}
	\;\;\implies\;\;
	p_t \ge 
	\sigmoid\big(\tfrac{t\eta \delta_p}{4\delta_p}\big)
	= 
	\sigmoid\big(\tfrac{t\eta}{4}\big) \;.
	\label{eq:st-02}
\end{equation}
Here, the final equality in the lefthand side
comes from $x_0(1) = 0$
under $p_0 = (1/2, 1/2)$,
and the implication is due 
to the primal-dual relationship of 
Proposition~\ref{prop:omwu-pd-2x2}.
Again using the definition of $\DG(\cdot)$ under $A = A_{\delta_p, \delta_q}$
from~\eqref{eq:dg-2x2-full}, we then find for all $t \in [T]$ that
\begin{align*}
	\DG(w_t)
	&= 
	(1-\delta_p - p_t(1))(1-\delta_q)
	+ 
	(1-\delta_q - q_t(1)) \delta_p \\
	&\le
	1-\delta_p - p_t(1)
	+ \delta_p  \\
	&= 
	1- p_t(1)
	\le
	1- \sigmoid(\tfrac{t \eta}{4}) 
	\le 
	\tfrac{4}{\eta t} \;,
\end{align*}
where the final inequality holds due to $t \eta > 0$.
Thus we conclude for Case (b) that
\begin{equation}
	\min\nolimits_{t \in [T]} \DG(w_t) \le \frac{4}{\eta T} \;.
	\label{eq:st-caseb}
\end{equation}
The statement of the lemma then follows by
combining the conclusions of 
expression~\eqref{eq:st-casea} for Case (a)
and expression~\eqref{eq:st-caseb} for Case (b).
~ \hfill $\blacksquare$

\end{document}